\documentclass[runningheads,envcountsame]{llncs}

\newif\ifeditmode
\editmodefalse

\newif\iflipics
\lipicsfalse

\title{Episodic Loops: Finitary Event Structures and Operational Semantics for
C11 Programs with Retries}
\titlerunning{Episodic Loops}

\iflipics
  \author{Christian Kissig}
         {School of Computing, University of Kent, Canterbury, UK}
         {c.kissig@kent.ac.uk}
         {0009-0000-9396-4084}
         {}
  \author{Jay Richards}
         {School of Computing, University of Kent, Canterbury, UK}
         {}{}{}
  \author{Mark Batty}
         {School of Computing, University of Kent, Canterbury, UK}
         {m.j.batty@kent.ac.uk}
         {0000-0001-7053-4364}
         {}
  \authorrunning{Kissig, Richards, and Batty}
  \Copyright{Christian Kissig, Jay Richards, and Mark Batty}
\else
  \author{Christian Kissig\inst{1}\orcidID{0009-0000-9396-4084} \and
          Jay Richards\inst{1} \and
          Mark Batty\inst{1}\orcidID{0000-0001-7053-4364}}
  \authorrunning{C. Kissig, J. Richards, and M. Batty}
  \institute{School of Computing, University of Kent, Canterbury, UK\\
    \email{\{c.kissig,m.j.batty\}@kent.ac.uk}}
\fi

\usepackage{amsmath, stmaryrd}
\iflipics\else
  \usepackage{amssymb} %
\fi
\usepackage{bussproofs}
\usepackage{cancel}
\usepackage{etoolbox}
\usepackage{color}
\usepackage[T1]{fontenc}
\usepackage{graphicx}
\usepackage{adjustbox}
\usepackage{varwidth}
\newcommand{\fitwidth}[1]{%
  \adjustbox{max width=\linewidth}{\begin{varwidth}{2\linewidth}#1\end{varwidth}}}
\usepackage{lineno}
\usepackage{microtype}
\usepackage[frozencache=true,cachedir=_minted]{minted}
\usepackage{subcaption}
\usepackage{thmtools, thm-restate}
\usepackage{tikz}
\usepackage{wasysym}
\usepackage[dvipsnames]{xcolor}
\iflipics
  \usepackage{hyperref}
\else
  \usepackage[hypertexnames=false,hidelinks]{hyperref} %
\fi
\usepackage{cleveref}

\usepackage{caption}
\usetikzlibrary{
  arrows,
  arrows.meta,
  automata,
  calc,
  cd,
  decorations.pathmorphing,
  decorations.pathreplacing,
  fit,
  matrix,
  positioning,
  shapes.geometric
}

\tikzset{
    ncbar angle/.initial=90,
    ncbar/.style={
        to path=(\tikztostart)
        -- ($(\tikztostart)!#1!\pgfkeysvalueof{/tikz/ncbar angle}:(\tikztotarget)$)
        -- ($(\tikztotarget)!($(\tikztostart)!#1!\pgfkeysvalueof{/tikz/ncbar angle}:(\tikztotarget)$)!\pgfkeysvalueof{/tikz/ncbar angle}:(\tikztostart)$)
        -- (\tikztotarget)
    },
    ncbar/.default=0.5cm,
  treenode/.style = {shape=rectangle, rounded corners, draw, align=center, top color=white, bottom color=blue!20},
  decision/.style = {treenode, diamond, aspect=2, bottom color=red!30},
  root/.style     = {treenode, font=\bfseries, bottom color=green!30},
  env/.style      = {treenode, bottom color=yellow!30},
  dummy/.style    = {circle,draw},
  square left brace/.style={ncbar=0.5cm},
  square right brace/.style={ncbar=-0.5cm}
}
\newcommand{\hide}[1]{}

\def\FWD{{\color{\fwdcolor} f}}
\def\WE{{\color{\fwdcolor} we}}
\def\fwdcolor{NavyBlue}
\def\ordop[#1]#2{\mathop{\mathord{#1}#2}}
\newcommand{\Allocs}{\mathcal A}

\newcommand{\Branches}{\mathcal B}
\newcommand{\CO}{\mathsf{co}}
\newcommand{\DP}{\dep}
\newcommand{\Deallocs}{\mathcal D}
\newcommand{\Effects}{\mathcal{W}^{\!*}}

\newcommand{\Downclosed}[1]{#1\mathord{\Downarrow}}
\newcommand{\ECO}{\mathsf{eco}}
\newcommand{\Events}{\mathcal E}

\newcommand{\Expressions}{\mathcal E}
\newcommand{\ProgExpressions}{\mathcal E_\prog}
\newcommand{\Registers}{\mathit{Reg}}
\newcommand{\FR}{\mathsf{fr}}
\newcommand{\Fences}{\mathcal F}
\newcommand{\HB}{\mathsf{hb}}

\newcommand{\Labels}{\mathit{CLabel}}

\newcommand{\Loc}{\mathit{Loc}}

\newcommand{\PSet}{\mathcal P}
\newcommand{\RF}{\mathsf{RF}}

\newcommand{\Reads}{\mathcal R}
\newcommand{\Set}{\text{Set}}

\newcommand{\Symbols}{\mathcal{S}}
\newcommand{\Threads}{T}

\newcommand{\Val}{\text{Val}}
\newcommand{\Var}{\text{Var}}
\newcommand{\Writes}{{\mathcal W}}

\newcommand{\acq}{\mathsf{acq}}

\newcommand{\bool}{\mathbb{B}}

\newcommand{\cas}{{\code{CAS}}}
\newcommand{\codefree}{\code{free}}

\newcommand{\codeskip}{\code{skip}}
\newcommand{\cond}{\mathsf{cond}}

\newcommand{\dep}{{\color{orange} \mathsf{dp}}}

\newcommand{\embeds}{\hookrightarrow}
\newcommand{\enterElse}{\code{enterElse}}
\newcommand{\enterThen}{\code{enterThen}}

\newcommand{\fadd}{\code{FAA}}
\newcommand{\freeze}{\hyperref[def:freeze]{\ifmmode\mathit{freeze}\else\textit{freeze}\fi}}
\newcommand{\funup}[3]{#1[#2:=#3]}

\newcommand{\fwdrel}[1][j]{\hyperref[def:fwd-ctx]{\xrightarrow{F_{#1}}}}

\newcommand{\horizon}{|}
\newcommand{\hp}[1]{\hyperref[hp#1]{\texttt{#1}}}
\newcommand{\idmap}{\text{id}}
\newcommand{\iter}{\hyperref[not:iter]{\text{iter}}}
\newcommand{\justifies}{\vdash}
\newcommand{\loc}{\mathsf{loc}}
\newcommand{\loopfun}{\hyperref[not:loopfun]{\mathsf{loops}}}
\newcommand{\innerloop}{\mathsf{inner}}

\newcommand{\memory}{\text{mem}}
\newcommand{\mordor}{\text{MoRDor}}
\newcommand{\nat}{\mathbb{N}}

\newcommand{\nta}{\mathsf{nta}}
\newcommand{\uafrel}{\mathsf{uaf}}
\newcommand{\poord}{\mathord{\po}}
\newcommand{\ppoord}{\mathord{\ppo}}
\newcommand{\origin}[1]{\mathit{O} (#1)}
\newcommand{\pc}{\mathsf{pc}}
\newcommand{\poinv}{\sqsupseteq}

\newcommand{\pormw}{\sqsubseteq^{\mathsf{rmw}}}

\newcommand{\po}{\sqsubseteq}

\newcommand{\ppoalias}[1][P]{\ordop[\ppo]{_\mathsf{alias}^{#1}}}

\newcommand{\ppormw}[1][P]{\ordop[\ppo]{_\mathsf{rmw}^{#1}}}

\newcommand{\pposync}{\ordop[\ppo]{_\mathsf{sync}}}
\newcommand{\ppo}{\hyperref[def:ppo]{\preceq}}
\newcommand{\pred}{\hyperref[def:ppo]{\mathsf{pred}}}

\newcommand{\prog}{\mathcal{P}}
\newcommand{\rcu}[1]{\hyperref[rcu#1]{\texttt{#1}}}
\newcommand{\relabeq}[2]{\hyperref[def:rel-eq]{\xrightarrow{#1,#2}}}
\newcommand{\rel}{\mathsf{rel}}
\newcommand{\remap}{\hyperref[def:fwd-ctx]{\ifmmode\mathit{remap}\else\textit{remap}~\fi}}
\newcommand{\restrictmap}{\mapsto}
\newcommand{\restrict}[2]{{{#1}\mathord{\upharpoonright_{#2}}}}
\newcommand{\restrictloop}[3]{{{#1}\mathord{\upharpoonright_{#2\restrictmap #3}}}}
\newcommand{\rf}{\mathsf{rf}}
\newcommand{\rlx}{\mathsf{rlx}}

\newcommand{\rmw}{\mathsf{rmw}}
\newcommand{\skipcmd}{\mathsf{skip}}
\newcommand{\smrd}{{\sc SMRD}}
\newcommand{\mrd}{{\sc MRD}}
\newcommand{\symbols}{\mathit{syms}}
\newcommand{\thread}{\mathsf{thread}}
\newcommand{\upclosed}[1]{#1\mathord{\uparrow}}
\newcommand{\valres}{\hyperref[def:es-prefix]{\mathsf v}}
\newcommand{\val}{\mathsf{val}}

\newcommand{\tstamps}{\mathbb{Q}}
\newcommand{\tst}{\mathsf{tst}}
\newcommand{\writeset}{\mathsf{W}}
\newcommand{\tview}[1][t]{\mathsf{tview}_{#1}}
\newcommand{\mview}[1][\hat w]{\mathsf{mview}_{#1}}
\newcommand{\OW}{\mathsf{OW}}
\newcommand{\anchors}{\mathsf{anc}}
\newcommand{\freshts}{\mathsf{fresh}}
\newcommand{\viewcomb}{\otimes}
\newcommand{\syncview}{\mathsf{sync}}
\newcommand{\Wrel}{\Writes_{\mathsf{rel}}}
\newcommand{\Racq}{\Reads_{\mathsf{acq}}}
\newcommand{\SW}{\mathsf{sw}}
\newcommand{\werel}[1][j]{\hyperref[def:fwd-rels]{\xrightarrow{\text{WE}_{#1}}}}

\renewcommand{\and}{\wedge}

\newcommand{\ifCond}{\code{ifCond}}

\tikzstyle{rf} = [->, dashed, thick, draw=white!50!Green]
\tikzstyle{po} = [->, draw=white!50!black]
\tikzstyle{ppo} = [->, draw=white!50!magenta]
\tikzstyle{dp} = [->, thick, draw=white!50!orange]
\tikzstyle{rmw} = [->, thick, draw=white!30!blue]
\tikzstyle{ppormw} = [->, densely dashed, draw=white!50!magenta]
\tikzstyle{conflict} = [-, draw=black!30!red, line width=0.9, decorate, decoration={zigzag, segment length=7, amplitude=2.6, pre=lineto, pre length=1pt, post=lineto, post length=1pt}]

\tikzstyle{lableft} = [xshift=-0.5cm, yshift=-0.1cm,color=green!50!black,font=\scriptsize]
\tikzstyle{labright} = [xshift=0.5cm, yshift=-0.1cm,color=green!50!black,font=\scriptsize]
\tikzstyle{basic} = [rectangle, rounded corners=3pt, thin,draw=black, fill=blue!5, sibling distance=10mm, minimum size=3ex,inner sep=0.25em]
\def\edgerf[#1]#2#3{\draw (#2) edge[style=rf,#1] node[] {\scriptsize\color{Green}rf} (#3);}
\def\edgerfopt[#1]#2#3#4{\draw (#2) edge[style=rf,#1] node[#4] {\scriptsize\color{Green}rf} (#3);}
\def\edgepo[#1]#2#3{\draw (#2) edge[style=po,#1] node[] {} (#3);}
\def\edgeppo[#1]#2#3{\draw (#2) edge[style=ppo,#1] node[] {\scriptsize\color{magenta}$\le$} (#3);}
\def\edgedp[#1]#2#3{\draw (#2) edge[style=dp,#1] node[] {\scriptsize\color{orange}dp} (#3);}
\def\edgeconflict[#1]#2#3{\draw (#2) edge[style=conflict,#1] node[] {\scriptsize\color{red}} (#3);}

\def\nodeRacq#1#2#3#4{\node (#1) at (#4) [label={[style=lableft]#1}, style=basic] {$R^\text{acq}~#2~#3$};}

\def\nodeWrel#1#2#3#4{\node (#1) at (#4) [label={[style=lableft]#1}, style=basic] {$W^\text{rel}~#2~#3$};}
\def\nodeB#1#2#3#4{\node (#1) at (#4) [label={[style=lableft]#1}, style=basic] {$[#3]$};}

\def\nodeP#1#2#3{\node (#1) at (#3) [style=basic] {#2};}

\def\semif#1#2#3{\ensuremath{\code{if}~(#1)~\{#2\}~\code{else}~\{#3\}}}
\def\semwhile#1#2{\ensuremath{\code{while}~(#1)~\{#2\}}}
\def\semseq#1#2{\ensuremath{#1\,\code{;}\,#2}}
\def\semregst#1#2{\ensuremath{\reg{#1}~\code{:=}~#2}}
\newcommand\semglobst[3][]{\ensuremath{#2~\code{:=}_{#1}~#3}}
\newcommand\semglobld[3][]{\ensuremath{\reg{#2}~\code{:=}_{#1}~#3}}
\def\semamp#1{\ensuremath{\code{\&}#1}}
\def\semderef#1{\ensuremath{\code{*}#1}}
\newcommand\semfence[1][]{\ensuremath{\code{fence}_{#1}}}
\def\semfadd[#1][#2]#3#4#5{\ensuremath{\reg{#3}~\code{:=}~\fadd_{#1,#2}(#4,#5)}}
\def\semcas[#1][#2]#3#4#5#6{\ensuremath{\reg{#3}~\code{:=}~\cas_{#1,#2}(#4,#5,#6)}}
\def\semcasf[#1][#2][#3]#4#5#6#7{\ensuremath{\reg{#4}~\code{:=}~\cas_{#1,#2,#3}(#5,#6,#7)}}
\def\semmalloc#1#2{\ensuremath{\reg{#1}~\code{:=}~\code{malloc}(#2)}}
\def\semfree#1{\ensuremath{\codefree(\reg{#1})}}
\def\sempar#1#2{\ensuremath{#1\parallel#2}}
\def\semskip{\ensuremath{\code{skip}}}

\makeatletter
\def\join[#1]#2{\@ifnextchar\bgroup{#2#1\join[#1]}{#2}}
\makeatother
\def\code#1{{\fontfamily{lmtt}\selectfont\textbf{\color{black}#1}}}
\def\reg#1{r_{#1}}
\def\expr{\varepsilon}
\def\bitand{\mathbin{\code{\&}}}
\def\bitor{\mathbin{\code{|}}}
\def\bitxor{\mathbin{\code{\textasciicircum}}}
\def\den#1{\llbracket#1\rrbracket}

\def\evalreg#1#2{\den{#1}_{#2}}
\def\evalenv#1#2{\den{#1}_{#2}}
\def\eval#1#2#3{\evalenv{#1}{[#2\,\mapsto\,#3]}}
\def\exeq{\hyperref[def:sem-equiv]{\equiv}}

\def\elab#1{\ensuremath{\textnormal{\textbf{G}}_\text{#1}}}

\newcommand{\ovrule}[2]{%
  \scalebox{0.#1}{%
    \begin{minipage}{\dimexpr\linewidth/#1*100\relax}%
      \centering#2%
    \end{minipage}}}
\newcommand{\ovsub}[5]{%
  \subcaptionbox{#4\label{#5}}[#1]{\ovrule{#2}{#3}}}

\iflipics\else
  \spnewtheorem{observation}{Observation}{\bfseries}{\rmfamily}
  \spnewtheorem{notation}{Notation}{\bfseries}{\rmfamily}
\fi

\iflipics\else
  \AtEndEnvironment{proof}{\ifvmode\noindent\fi\unskip\nobreak\hfill$\square$}
\fi

\begin{document}

\maketitle

\begin{abstract}

  Lock-free synchronisation algorithms are often implemented with fallible
  operations, such as compare-and-swap (CAS). Unbounded retry loops give these
  fallible operations an eventual success semantics. Verifying the correctness
  of such algorithms requires considering arbitrarily many failing iterations
  before a successful attempt. Retry loops thus yield large state spaces,
  compounded further by the interleavings of concurrent threads. Prior work
  discarded failing iterations, arguing that they leave no trace in the
  post-loop state.

  Modern compilers and hardware architectures optimise execution by reordering
  instructions; load-store reordering is of particular concern here.
  These reorderings may cross the boundaries of failing iterations, introducing
  subtle bugs in concurrent settings. We demonstrate one such bug, making
  use-after-free possible in a previously verified variant of Read-Copy-Update
  -- a synchronisation primitive widely adopted in the Linux Kernel -- and we
  provide and verify a fix.

  This motivates a closer study of retry loops under instruction reordering. We
  find that practical implementations of retries in lock-free algorithms adhere
  to a common pattern. We introduce \emph{episodic loops}, a semantic
  characterisation of unbounded retry loops which adhere to a syntactically
  recognisable pattern in many practical cases, and \emph{synchronisation
  points}, operations that bound both instruction reordering and state space
  within episodic loops. We show that \smrd{} -- a symbolic event structure
  semantics for C11 programs which allows for load-store reordering
  optimisations -- admits a finite representation in programs where unbounded
  loops are episodic. We further introduce a finitary operational semantics that
  allows safety properties to be verified in finitely many steps. For the
  use-after-free bug we demonstrate, this verification takes a single pass over
  the program, making it linear in the program size for safety properties of
  this shape.

  We provide a reference implementation of \smrd{} reproducing the
  use-after-free bug and verifying the fix, and mechanise the operational
  semantics together with the minimal bug and its fix in
  Isabelle/HOL~\cite{kissig2026rcuopsem}.

  \iflipics\else
  \fi

\end{abstract}

\newcommand{\conditionalpagebreak}{%
  \ifeditmode
    \pagebreak
  \fi
}

\section{Introduction}

Systems code must be fast, and its correctness is critical.
Tuning for performance leads to subtle behaviour that can harbour bugs:
non-blocking synchronisation primitives are preferred over simple locks,
compiler optimisations are applied, and the code is left to the whims of the
relaxed memory model of the target processor.
Verification can ensure correctness, but that requires a faithful specification
of the programming language. Right now, language specifications like that of
C/C++ are lacking: they avoid specifying behaviours introduced by compiler
optimisation, and even allow so-called \emph{out-of-thin-air} (OOTA) values,
that prevent reasoning at
all~\cite{Batty2015TheProblem,Vafeiadis2015CompilerOpt}.
Thin-air values arise from reordering loads with later stores, irrespective of
program dependencies.
In weak memory models like RC11~\cite{rc11}, thin-air values are excluded by
forbidding \emph{load buffering}, the execution shape in which each of two
threads reads the value the other stores later in its own program order.
A verification result is only as faithful as the model it builds on. Results
established over models that forbid load buffering, such as RC11~\cite{rc11}, do
not by themselves carry over to hardware: compiling relaxed loads and stores to
plain \code{ldr}/\code{str} instructions does not preserve the absence of load
buffering~\cite{Vafeiadis2015CompilerOpt,Geeson24Telechat}, so the
load-store reordering the model rules out can still be observed on modern
processors~\cite{AlglaveEtAl2014HerdingCats}. Such guarantees are not unsound but
\emph{conditional}: soundness with respect to hardware is recovered by a stronger
compilation scheme -- for instance a branch after the load to enforce the
dependency -- or by separately proving the program free of load-buffering races.

Low-level non-blocking algorithms are used for performance, but they require
atomic synchronisation primitives such as \emph{compare-and-swap} (\cas{}) and
\emph{fetch-and-add} (\fadd{}) operations~\cite{lawsoforder}, typically within
\emph{retry loops}. Retry loops
are a foundational programming pattern which can be found across the Linux
Kernel~\cite{linux_kernel}, the Boost Libraries~\cite{boost}, Meta's Folly
library~\cite{folly}, and beyond.
This paper shows that in the presence of compiler optimisations, the failed
iterations of retry loops have a bearing on the correctness of code.

The \cas{}-in-a-retry-loop idiom is pervasive in multi-threaded systems code:
lock-free data structures such as lock-free queues and deques update shared
state through a \cas{} operation in a retry loop. In the Linux 5.6 source
tree~\cite[v5.6]{linux_kernel} alone, 238 of the 1,274 static call sites of the
\code{cmpxchg} family of compare-and-swap primitives occur inside a retry
loop.\footnote{Identified by an AST-based scan over the \code{v5.6}
tag~\cite{kissig2026caspattern}: a \code{cmpxchg}/\code{try\_cmpxchg} call site
qualifies when it is lexically nested within a \code{while}/\code{do}/\code{for}
loop (209 sites) or within a backward \code{goto}-retry region (29 sites). The
remaining call sites are single-shot CAS, initialisers, and other non-loop
uses.}
\emph{Read-Copy-Update} (RCU)~\cite{mckenney2004exploiting} is a prominent
instance, used pervasively in the Linux Kernel; while the majority of RCU code
paths rely on quiescent-state-based reclamation, RCU-style lock-free structures
synchronise with exactly this \cas{} retry pattern.
The version of the RCU algorithm in Gotsman et al.~\cite{Gotsman23grace} has
been verified over a model of C++, but only under
RC11z~\cite{Semenyuk2023RCU}, an allocation-aware variant of RC11~\cite{rc11}
in which allocation and deallocation events are modelled through write events,
and which forbids load-buffering.
That exclusion is the hole: load-store reordering is exactly the behaviour C++
leaves unspecified and that survives compilation to plain loads and stores on
hardware, so a guarantee established without it is silent on the executions
where it occurs.

Symbolic Modular Relaxed Dependencies (\smrd{})~\cite{Richards25SMRD} is a model
of C++ concurrency that fills this hole, accounting for real-world compiler
optimisations by allowing load-buffering, while still forbidding thin-air
values. The behaviours \smrd{} newly admits include an execution in which a
failed retry attempt leads to a use-after-free (UAF), in the idiomatic use of
\cas{} in algorithms like RCU and spinlocks.
The bug arises from a load-store reordering that RC11z forbids but \smrd{}
admits -- one of several real-world optimisations \smrd{} validates, alongside
store-to-load forwarding, dead/redundant store elimination, and value-range
refinement. Under RC11z the offending execution cannot occur, so the prior
verification holds for that model; the defect surfaces only under \smrd{}.
Demonstrating it nonetheless lies beyond \smrd{}'s existing reasoning, since the
reordering occurs within an \emph{unbounded} retry loop, whose event
structure~\cite{Winskel87EventStructures} is unbounded. Supplying a finite
handle on such loops is the technical problem the remainder of this paper
solves.
Our findings do not imply a fault in the implementation of RCU in the Linux
kernel; practical implementations of lock-free synchronisation algorithms
typically employ more aggressive synchronisation already. Instead, our work
contributes to the foundation of a more principled optimisation of
synchronisation algorithms, with the UAF bug serving as a succinct working
example of our verification method. The scale of this space is suggested by RCU
alone, whose API accounts for 16,931 static uses across the Linux 5.6 source
tree~\cite{McKenney20RCU18}.

Our verification work introduces the concept of \emph{episodic loops}, a common
semantic property of unbounded loops that implement retry-logic. Episodic loops
can be found in lock-free synchronisation algorithms, including
RCU~\cite{Gotsman23grace}, hazard pointers~\cite{HazardPointers},
seqlocks~\cite{hemminger2002seqlock}, and
spinlocks~\cite[\S8.5]{hennessy1996computer}. In these and other practical
applications, episodic loops adhere to a syntactic and statically recognisable
pattern. When unbounded loops are episodic, they admit a finite representation
of their event structure semantics.

\medskip \noindent Our contributions are as follows.
\begin{enumerate}

   \item We enable formal reasoning about safety properties for
      code subject to load-buffering weak memory behaviour, reconciling the
      details of the allowed behaviour with programming patterns used in
      practice to build non-locking concurrent systems components.

    \item We exhibit an execution of a previously verified implementation of
      RCU~\cite{Semenyuk2023RCU} in which an idiomatic use of Compare-and-Swap
      leads to a use-after-free (Section~\ref{sec:bug}). \smrd{} admits the
      execution and RC11z forbids it, so the prior verification remains sound
      for the model it was carried out over, and is silent on this defect. C++
      permits the reorderings the execution depends on, and a production
      compiler already performs two of the three (Section~\ref{s:eval}).

    \item We define a criterion -- \emph{episodicity} -- for unbounded loops
      that captures the intent of a retry pattern typical in lock-free
      concurrent systems code (Section~\ref{s:episodic}). Many episodic loops
      can be classified purely syntactically, without consideration of the
      complex semantics of the program.

    \item We present a finitary operational semantics for the behaviour of
      programs with unbounded loops in the relaxed concurrent setting in
      Section~\ref{s:opsem}.
      The operational semantics provides the foundation for an ergonomic proof
      method for concurrent algorithms with episodic loops.
      We show that this method can be used to prove safety properties, such as
      the absence of use-after-free.

    \item We provide a tool, \mordor{}, described in Section~\ref{s:eval}, that
      takes a program and calculates the event structure semantics using a
      finite step-counter, applies sufficient episodicity criteria, and detects
      use-after-free by enumerating the valid executions of that bounded
      unravelling.

    \item We mechanise the finitary operational semantics together with the
      minimal use-after-free bug and its fix in
      Isabelle/HOL~\cite{kissig2026rcuopsem}, machine-checking that the bug is
      reachable in the buggy variant and excluded in the fixed one. Extending
      the mechanised development to an end-to-end no-use-after-free theorem for
      full RCU remains future work.

\end{enumerate}

\conditionalpagebreak
\section{CAS Semantics and a Use-After-Free
Bug}\label{sec:cas}\label{sec:bug}

\begin{figure}[t]
  \scalebox{0.8}{
  \begin{minipage}{0.45\textwidth}
      \input{snippets/rcu-inc-tidy}
    \captionof{figure}{Excerpt of Read-Copy-Update \code{inc()} with full
    listing in Appendix~\ref{app:rcu}}
    \label{fig:rcu-inc-code}
  \end{minipage}

  \hspace{0.2cm}

  \begin{minipage}{0.35\textwidth}
    \begin{center}
      \input{snippets/rcu-reclaim}
    \end{center}
    \captionof{figure}{Simplified and inlined \code{reclaim()} function with
    full listing in Appendix~\ref{app:rcu}}
    \label{fig:rcu-reclaim-code}
  \end{minipage}

  \hspace{0.2cm}

  \begin{minipage}{0.36\textwidth}
    \begin{flushright}
      \begin{tikzpicture}
\nodeRacq{read}{C}{\beta}{0, -1}
\nodeB{branch}{}{\beta = \alpha}{0, -2}
\nodeWrel{write}{C}{\varepsilon}{-1.25, -3}
\node (r3) at (1.25, -3) [label={[text=blue,font=\scriptsize]below:$\rho[r]:=\mathit{false}$}] {\vdots};
\node (r4) at (-1.25, -4) [label={[text=blue,font=\scriptsize]below:$\rho[r]:=\mathit{true}$}] {\vdots};
\edgepo[]{branch}{write}
\edgepo[]{read}{branch}
\draw (branch) edge[style=po] node[yshift=3pt,xshift=-5pt] {\scriptsize if} (write);
\draw (branch) edge[style=po] node[yshift=5pt,xshift=5pt] {\scriptsize else} (r3);
\edgepo[]{write}{r4}
\end{tikzpicture}
    \end{flushright}
    \vspace{-0.5cm}
    \captionof{figure}{Semantics of \code{r:=}\cas{}$^{\text{rel,acq}}$\code{(\&C,s,n)} where
    \code{s} has value $\alpha$ and \code{n} has $\varepsilon$. \cas{} reads
    $\beta$ from \code{C}.}
    \label{fig:cas-event-structure}
  \end{minipage}
}
\end{figure}

\cas{} is an atomic read-modify-write operation comprising a read instruction, a
branch, and a conditionally executed write instruction. The event structure
semantics of \cas{} with release-acquire memory order on success and acquire on
failure is depicted in Figure~\ref{fig:cas-event-structure}, where the
subsequent event structure, denoted \scalebox{0.6}{$\vdots$}, is interpreted with a register
state $\rho$ updated at the register \code{r}. Each instruction carries a memory
order annotation -- acquiring for the read and releasing for the write -- which
together define the memory semantics of \cas{} as a whole. Because the write
instruction only executes when the branch succeeds, the semantics is asymmetric
across the two outcomes: \cas{} cannot carry a release memory ordering in the
failure case~\cite[7.17.7.4]{iso-c11}. This allows memory operations to reorder
below the \cas{} in the failure case.

\begin{figure}[h]

  \begin{minipage}{\textwidth}
    \centering
    \scalebox{0.8}{\begin{tikzpicture}[node distance=1cm, auto]

\node[text width=2.5cm] (t1) at (0, 0) {Thread 1};
  \node[text width=2.5cm] (t1_sync) [below of=t1, yshift=-2cm]
  {\code{while}(\code{rcu}[$t_2$]);};
  \node[text width=2.5cm] (t1_free) [below of=t1_sync] {\code{free(s)}};

\node[text width=2.5cm] (t2) at (4, 0) {Thread 2};
  \node[text width=2.5cm] (t2_vs_placeholder) [below of=t2] {\textvisiblespace};
  \node[text width=2.5cm] (t2_cas) [below of=t2_vs_placeholder]
  {$\cas{}^\text{acq}(C, s, n_2)\text{\lightning}$};
  \node[text width=2.5cm] (t2_rcu_exit) [below of=t2_cas] {$\code{rcu}[t_2]:=0$};
  \node[text width=2.5cm] (t2_vs) [below of=t2_rcu_exit, text=red] {v:=*s};

\draw[->] (t2_rcu_exit) -- (t1_sync) node[midway, above] {$\rf$};
\draw[->] (t1_sync) -- (t1_free) node[midway, above] {};
\draw[->, color=red] (t1_free) -- (t2_vs) node[midway, above] {UAF};

\draw[decorate,decoration={brace,amplitude=5pt,mirror,raise=10pt}]
  (t2_cas.south east) -- (t2_vs_placeholder.north east)
  node[midway,right,xshift=13pt] {i};
\draw[decorate,decoration={brace,amplitude=5pt,mirror,raise=10pt}]
  (t2_rcu_exit.south east) -- (t2_rcu_exit.north east)
  node[midway,right,xshift=13pt] {i+1};
\draw[decorate,decoration={brace,amplitude=5pt,mirror,raise=10pt}]
  (t2_vs.south east) -- (t2_vs.north east) node[midway,right,xshift=13pt] {i};
\draw[decorate,decoration={brace,amplitude=5pt,raise=0pt},xshift=-3pt]
  (t1_free.south west) -- (t1_sync.north west) node[midway,left,xshift=-3pt]
  {{\small \rotatebox{90}{\texttt{reclaim(s)}}}};

\draw[->, color=blue] ([xshift=-15pt]t2_vs_placeholder.east) to[out=-45, in=45,
  looseness=0.8] ([xshift=-15pt]t2_vs.east);
\end{tikzpicture}}
    \par\smallskip
    \begin{minipage}{0.86\textwidth}
      \footnotesize
      The dereference \code{v := *s} can reorder ({\color{blue} blue} arrow) over the
      failing~(\lightning) \cas{} branch, the RCU exit command, and thus over
      \code{while(rcu[}$t_2$\code{])} in \code{reclaim} and \code{free(s)}.
      \code{programs/uaf-bug.lit} in \mordor{}~\cite{kissig2026mordor}.
    \end{minipage}
    \caption{Use-after-free bug}\label{fig:use-after-free-bug-overview}
  \end{minipage}

  \vspace{\dimexpr 13pt + \bigskipamount\relax}

  \begin{minipage}{\textwidth}
    \centering
    \scalebox{0.8}{\begin{tikzpicture}[node distance=1cm, auto]

\node[text width=2.5cm] (t1) at (0, 0) {Thread 1};
  \node[text width=2.5cm] (t1_sync) [below of=t1, yshift=-2cm]
  {$\code{while}(\code{rcu}[t_2]);$};
  \node[text width=2.5cm] (t1_free) [below of=t1_sync] {\code{free(s)}};

\node[text width=2.5cm] (t2) at (4, 0) {Thread 2};
  \node[text width=2.5cm] (t2_cas) [below of=t2] {$\cas{}^\text{acq}(C, s, n_2)\text{\lightning}$};
  \node[text width=2.5cm] (t2_vs) [below of=t2_cas] {\code{v:=*s}};
  \node[text width=2.5cm] (t2_rcu_exit) [below of=t2_vs] {$\code{rcu}[t_2]:=^{\text{\color{blue} rel}}0$};

\draw[->] (t2_rcu_exit) -- (t1_sync) node[midway, above] {$\rf$};
\draw[->] (t1_sync) -- (t1_free) node[midway, above] {};
\draw[->] (t2_vs) -- (t2_rcu_exit) node[midway] {$\ppo$};

\draw[decorate,decoration={brace,amplitude=5pt,mirror,raise=0pt},xshift=5pt]
  (t2_vs.south east) -- (t2_cas.north east) node[midway,right,xshift=10pt] {i};

\draw[decorate,decoration={brace,amplitude=5pt,mirror,raise=0pt},xshift=5pt] (t2_rcu_exit.south east) -- (t2_rcu_exit.north east) node[midway,right,xshift=10pt] {i+1};

\draw[decorate,decoration={brace,amplitude=5pt,raise=0pt},xshift=-1pt] (t1_free.south west) -- (t1_sync.north west) node[midway,left,xshift=-10pt] {{\small \rotatebox{90}{\texttt{reclaim(s)}}}};

\end{tikzpicture}}
    \par\smallskip
    \begin{minipage}{0.86\textwidth}
      \footnotesize
      As the RCU exit command has a release semantics, the dereference
      \code{v := *s} is $\ppo$-ordered before the RCU exit command, and thus
      ordered before
      \code{while(rcu[}$t_2$\code{])}. \code{programs/uaf-bug-fixed.lit} in
      \mordor{}~\cite{kissig2026mordor}.
    \end{minipage}
    \caption{Use-after-free bug fixed with releasing RCU exit}\label{fig:use-after-free-bug-fixed}
  \end{minipage}
\end{figure}

\paragraph{Overview of the UAF bug.} Suppose two threads concurrently attempt to
increment a shared value at a memory location \code{C} using the algorithm in
Figure~\ref{fig:rcu-inc-code}.
Each thread first retrieves the current location of \code{C} as \code{s} in the
fetch-and-add instruction \code{s:=}\fadd{}$^{\text{rel,acq}}$\code{(\&C, 0)},
updates the value, stores the updated value at a new memory location \code{n},
and attempts to swap the current memory location of \code{C} for \code{n} using
\cas{}. \fadd{} serves to add release semantics to the read through the atomic
read-don't-modify-write (RdMW) instruction, which orders the preceding
quiescent period before the \fadd{} and is thereby foundational for the overall
correctness of the RCU variant as a lock-free synchronisation primitive, and in
particular ABA-freedom, cf.~\cite{Semenyuk2023RCU}.
Each thread performs a \cas{} contending on \code{C}, and only one can succeed.
Here \code{s}, \code{n}, \code{v} and \code{r} are \emph{registers} --
thread-private local variables, set and read without memory operations -- as
opposed to the shared memory locations \code{C} and \code{rcu}.
The per-thread flag \code{rcu[tid]} marks whether thread \code{tid} is inside an
RCU read-side critical section, i.e.~currently accessing the object reached
through \code{C}; a concurrent \code{reclaim} waits,
\code{while(rcu[}$i$\code{])}, for every reader $i$ to clear its flag before
freeing. Each iteration closes the previous critical section and opens a new one
through the paired \code{rcu[tid]:=0; rcu[tid]:=1} at the top of the loop body.
The \code{rcu[tid]:=1} is not redundant: the \fadd{}, dereference, and \cas{}
that access the shared object all execute while the flag is set, and it is
cleared only by the next iteration's \code{rcu[tid]:=0} or by the final
\code{rcu[tid]:=0} after the loop. This per-iteration exit and re-entry is the
read-side quiescent point at each loop boundary.
If the \cas{} on Thread~1 succeeds, the \cas{} on Thread~2 must read from
Thread~1's \cas{} write, witnessing the intervening change to \code{C}, and so
must fail.
When Thread~2's \cas{} fails in this way, load-store reordering allows its
dereferencing instruction \code{v := *s} to move across the failing \cas{}
branch and the RCU exit command \code{rcu[}$t_2$\code{] := 0}, both on Thread~2.
Because Thread~1's \code{reclaim} waits for this RCU exit through
\code{while(rcu[}$t_2$\code{])} before it calls \code{free(s)}, the dereference
thereby reorders past \code{while(rcu[}$t_2$\code{])} and \code{free(s)} on
Thread~1, as indicated by the {\color{blue} blue} reordering arrow in
Figure~\ref{fig:use-after-free-bug-overview}.

Figure~\ref{fig:use-after-free-bug-overview} shows an execution that witnesses
the use-after-free bug. Thread~2 is within its loop where its \cas{} operation
fails (indicated by \lightning{}), leading to another iteration of the loop,
where \code{rcu[$t_2$]} is cleared and set, and the shared resource \code{C} is
accessed. The \cas{} on Thread~1 succeeds, and is followed by a call to
\code{reclaim}. \code{reclaim} reads 0 from \code{rcu[$t_2$]}, assumes that the
thread has completed its (failed) attempt to increment the value at \code{C},
and frees the previous location of \code{C}, which Thread~2 is about to
dereference through \code{s}, a use-after-free bug. We add the code with the
occurrence of use-after-free as \code{uaf-bug.lit} to the \mordor{} test suite.

The obvious fix is to prevent the dereferencing instruction from reordering
across the RCU exit operation, for instance by adding a release memory order
annotation to the RCU exit operation, as in
Figure~\ref{fig:use-after-free-bug-fixed}. We add the code with the fix as
\code{uaf-bug-fixed.lit} to the \mordor{} test suite.

\begin{figure}[h]
      \scalebox{0.8}{%
  \begin{minipage}{0.62\textwidth}
    \begin{center}
        \begin{tikzpicture}[node distance=1cm, auto, every node/.style={align=center}]

  \node[text width=3cm] (t1) at (0,0) {Thread 1};
  \node[text width=3cm] (t1_cas) [below of=t1, yshift=-1cm] {\cas};
  \node[text width=3cm] (t1_readrcu) [below of=t1_cas] {\code{r[i] := rcu[tid]}};
  \node[text width=3cm] (t1_free) [below of=t1_readrcu] {\code{free(s)}};

  \node[text width=2.5cm] (t2) at (4,0) {Thread 2};
  \node[text width=2.5cm] (t2_deref) [below of=t2] {\code{v := *s}};
  \node[text width=2.5cm] (t2_cas) [below of=t2_deref] {\cas};
  \node[text width=2.5cm] (t2_rcuexit) [below of=t2_cas] {\code{rcu[tid] := 0}};

  \draw[->] (t1_cas) -- (t2_cas) node[midway, above] {$\rf$};
  \draw[->, dashed, color=red] (t1_free) -- (t2_deref) node[midway, above]
  {$\uafrel$};
  \draw[->, color=red] (t2_rcuexit) -- (t1_readrcu) node[midway, below] {$\rf$};

  \draw[->] (t1_cas) -- (t1_readrcu) node[midway, left] {$\po$};
  \draw[->, color=red] (t1_readrcu) -- (t1_free) node[midway, left] {$\po$};

  \draw[->] (t2_deref) -- (t2_cas) node[midway, right] {$\po$};
  \draw[-{Stealth}, color=red] (t2_deref.south east) to[out=-45, in=45]
  (t2_rcuexit.north east) node[right] {$\po$};
  \draw[->] (t2_cas) -- (t2_rcuexit) node[midway, right] {$\po$};

\end{tikzpicture}
    \end{center}
    \captionof{figure}{\code{v := *s} is $\po$-ordered before \code{rcu[tid] := 0},
    so the {\color{red} red} $\poord\cup\rf$ chain places the dereference before
    \code{free(s)} in $\nta_{\text{RC11z}}$. The hypothesised use-after-free
    ({\color{red} red}, dashed) runs against that chain, and
    \texttt{no-thin-air} denies the resulting cycle, so RC11z forbids the UAF
    bug}
    \label{fig:uaf-bug-oota-cycle-rc11}
  \end{minipage}%
      }
  \hfill
      \scalebox{0.8}{%
  \begin{minipage}{0.62\textwidth}
    \begin{center}
        \begin{tikzpicture}[node distance=1cm, auto, every node/.style={align=center}]

  \node[text width=3cm] (t1) at (0,0) {Thread 1};
  \node[text width=3cm] (t1_cas) [below of=t1, yshift=-1cm] {\cas};
  \node[text width=3cm] (t1_readrcu) [below of=t1_cas] {\code{r[i] := rcu[tid]}};
  \node[text width=3cm] (t1_free) [below of=t1_readrcu] {\code{free(s)}};

  \node[text width=2.5cm] (t2) at (4,0) {Thread 2};
  \node[text width=2.5cm] (t2_deref) [below of=t2] {\code{v := *s}};
  \node[text width=2.5cm] (t2_cas) [below of=t2_deref] {\cas};
  \node[text width=2.5cm] (t2_rcuexit) [below of=t2_cas] {\code{rcu[tid] := 0}};

  \draw[->] (t1_cas) -- (t2_cas) node[midway, above] {$\rf$};
  \draw[->] (t1_free) -- (t2_deref) node[midway, above] {$\uafrel$};
  \draw[->] (t2_rcuexit) -- (t1_readrcu) node[midway, below] {$\rf$};

  \draw[->] (t1_cas) -- (t1_readrcu) node[midway, left] {$\ppo$};
  \draw[->] (t1_readrcu) -- (t1_free) node[midway, left] {$\DP$};

  \draw[->] (t2_cas) -- (t2_rcuexit) node[midway, right] {$\ppo$};
  \draw[->, -{Stealth}, dashed, color=blue] (t2_deref.south east) to[out=-45, in=45]
  (t2_rcuexit.north east);
  \draw[->,dashed,blue] (t2_deref) -- (t2_cas) node[midway, right] {};

\end{tikzpicture}
    \end{center}
    \captionof{figure}{\code{v := *s} is not ordered before \code{rcu[tid] := 0},
    indicated by the {\color{blue} blue} dashed arrows, so that no
    $\nta_{\smrd}=(\ppoord\cup\DP\cup\rf)^+$ chain runs from the dereference to
    \code{free(s)}, and the use-after-free contradicts no
    \texttt{thin-air-cycle} which would forbid the UAF bug in \smrd}
    \label{fig:uaf-bug-oota-cycle-smrd}
  \end{minipage}%
      }
  \end{figure}

\paragraph{Representing the use-after-free.}
An execution exhibits a use-after-free when it holds an access of a deallocated
location that it does not order before the deallocation -- the $\uafrel$-edge of
Figures~\ref{fig:uaf-bug-oota-cycle-rc11} and~\ref{fig:uaf-bug-oota-cycle-smrd}.
Nothing reads from a deallocation, so this is an ordering property and not a
read-from. The order is the relation the model's \texttt{no-thin-air} axiom
declares acyclic, written $\nta_{\text{RC11z}}~\triangleq~(\poord\cup\rf)^+$ for
RC11z, and $\nta_{\smrd}~\triangleq~(\ppoord\cup\DP\cup\rf)^+$ for \smrd{},
where $R^+\triangleq\bigcup_{n\geq1}R^n$ denotes the transitive closure of a
relation $R$, relating events joined by a chain of one or more $R$-steps.
Definition~\ref{def:uaf} states the property in full, and
Example~\ref{ex:uaf-rc11z} works the execution below through both orders, in
Appendix~\ref{app:defs}.

\paragraph{The UAF bug in RC11z.}
(\code{programs/uaf-bug-rc11.lit} in \mordor{}~\cite{kissig2026mordor}) The code
in Figures~\ref{fig:rcu-inc-code} and~\ref{fig:rcu-reclaim-code} has been taken
from the variant of RCU in Gotsman et al.~\cite{Gotsman23grace}. The same
variant has previously been verified~\cite{Semenyuk2023RCU} in RC11z. The UAF
bug does not appear in the program in RC11z, because $\nta_{\text{RC11z}}$
orders the dereference before the deallocation instruction.
Figure~\ref{fig:uaf-bug-oota-cycle-rc11} depicts the argument.

\paragraph{The UAF bug in \smrd{}.}
(\code{programs/\allowbreak uaf-bug-smrd.lit} in
\mordor{}~\cite{kissig2026mordor}) \smrd{}~\cite{Richards25SMRD} allows for
load-store-reorderings by relaxing the \texttt{no-thin-air} axiom, forbidding
instead \texttt{thin-air-cycles} in $\ppoord\cup\DP\cup\rf$. \emph{Preserved
program order} $\ppo$ and \emph{semantic dependency} $\DP$ are intra-thread
ordering relations that refine the program order $\po$: rather than enforcing
all of $\po$ as RC11 does, \smrd{} preserves only the edges dictated by the
chosen justification (Section~\ref{s:essem}). We provide the full constructive
definition of both taken from \cite{Richards25SMRD} in Appendix~\ref{app:defs}.
We argue in Figure~\ref{fig:uaf-bug-oota-cycle-smrd} that \smrd{} is sensitive
to the UAF bug, as $\nta_{\smrd}$ does not order the dereference before the
deallocation.

\paragraph{Reclamation in other models.}
\smrd{}'s first-class treatment of allocation and deallocation sets it apart
from three other ways of modelling reclamation. Allocation-aware models such as
RC11z keep deallocation out of the memory state altogether: \code{free} steps a
provenance-style allocation map $A\colon\Pi\rightarrow\Loc$ through a
$\mathsf{kill}$ action that leaves the RC11 state $\sigma$ untouched, so a
deallocation is neither a write event nor anything an $\rf$ can source, and
memory safety is discharged as an ownership invariant -- no thread retains a
read capability on a freed location -- rather than as a property of the
execution~\cite{Semenyuk2023RCU}. Separation logics instead track reclamation
through ghost state and per-location protocols, detecting the error as the
violation of an ownership invariant~\cite{Tassarotti2015RCU,Turon2014GPS}. Other
thin-air-free models omit reclamation altogether, and so cannot express a
use-after-free at all~\cite{Svendsen2018SepLogic}. In \smrd{}, by contrast,
allocation and deallocation are genuine events, allocation introducing fresh
symbolic locations under distinctness constraints (Definitions~\ref{def:gen-es}
and~\ref{def:freeze} in Appendix~\ref{app:defs}); the use-after-free is then an
access that $\nta_{\smrd}$ does not place before the deallocation
(Definition~\ref{def:uaf}).

\paragraph{Optimisation-induced bugs.}
The execution of RCU presented here shows that a hole in the C++ language
specification, where the impact of optimisation is disregarded, is hiding
erroneous program behaviours. The problem is not specific to RCU, but follows
from the use of \cas{} in a retry loop -- a code pattern that is common in
non-blocking concurrent code. Section~\ref{s:eval} reports which of the
reorderings this execution depends on a production compiler performs, and at
what rate the resulting defect is then observed.

\conditionalpagebreak
\section{Event Structure Semantics of Episodic
Loops}\label{s:essem}

\paragraph{Event Structure Semantics and Executions.}
Memory accesses, fences and branch conditions (collectively \emph{actions}) are
modelled by events, ordered by \emph{program order} $\po$, with \emph{conflict}
representing alternative branching outcomes. \smrd~\cite{Richards25SMRD}
supports dynamic memory management by giving a semantics to allocation and
deallocation instructions as first-class events.
\smrd{} extends \mrd{} with symbolic values: read events read symbolic values
and allocation events produce symbolic memory locations, each introducing a
fresh symbol by convention. Symbolic executions subsequently constrain these
symbols through value restrictions, $\rf$-relations, and a justification
mechanism for the write events in the execution. Justification of write events
has been introduced in \cite{Richards25SMRD}. We give a brief overview below,
and a formal definition in Appendix~\ref{s:app-defs-justs}.

\smrd{} weakens program order into the $\ppo$ and $\DP$ relations on each
thread, derived from the justification set of the execution. Valid
justifications are obtained through \emph{elaborations} that model standard
compiler optimisations: \emph{forwarding} models store-to-load forwarding and
redundant-load elimination, \emph{write elision} models dead-store elimination,
\emph{value assignment} models constant propagation, and \emph{lifting} models
hoisting a memory access out of a conditional branch. Because forwarding can
resolve the same accesses in several, potentially conflicting, ways, \smrd{}
admits several conflicting dependency relations over the same set of events.
A \emph{symbolic execution} -- which we refer to simply as an \emph{execution}
-- fixes one such choice, manifesting the dependencies over a maximal
conflict-free set of events constrained by the ordering relations and coherence
axioms of the underlying memory
model~\cite{pichon-pharabod-sewell-2016,castellan-2016}. One can picture a
symbolic execution as an equivalence class of traces, quotiented by the
reordering consistent with the independence relation induced by those dependency
relations, analogous to Mazurkiewicz traces. Executions must additionally
satisfy the memory model's consistency axioms, notably \emph{no-thin-air}, which
forbids cycles in $\ppoord\cup\DP\cup\rf$.

Concretely, \smrd~interprets a program compositionally in continuation-passing
style~\cite{Reynolds1972CPS}, where each command is interpreted relative to a
continuation constructed inductively from the end of terminating executions,
mapping register state and value restrictions to the event structure of the
remaining program. The \emph{register state} accumulates the assignment of
symbolic values to registers in \code{set}-commands; the \emph{value
restrictions} accumulate constraints evaluating branching conditions under
branching decisions.
\smrd{} therefore covers only terminating executions. In the presence of
unbounded loops, executions in the program semantics are in general unbounded.

\paragraph{Next enabled actions.}
To step through an execution operationally we define: a \emph{history}, $H$, as
a $\ppo$ and $\DP$ down-closed subset of an execution (i.e.~an execution
prefix), and the \emph{next enabled actions} of a history, $\horizon\Phi_{H}$,
as the minima of the $\ppo$ and $\DP$ later events.
Both relations order events of a single thread, so a future is a per-thread
order and each future set in $\Phi$ splits by thread. The one inter-thread
dependency of the model, $\rf$, does not contribute to $\Phi$, but constrains
$\Phi$ through the \texttt{no-thin-air} and \texttt{extended coherence} axioms.
Within the retry loops we consider, we will show that the next enabled actions
at a given line of code are equivalent in any two iterations.
We define episodicity to leverage this symmetry, enabling a finite
representation of loop semantics.

\subsection{Episodic Loops}\label{s:episodic}

In the source, we attach a \emph{loop identifier} to each loop. We assume
throughout that the loops of a program are indexed from $1$ upwards, so that $0$
is available for an event that lies outside every loop.
We support nested loops: for event $e$, $\loopfun(e)$ is the set of loop
identifiers for the loops nesting the line of code that performs $e$.
We define a loop iteration function $\iter(e):\loopfun(e)\to\nat$ that
identifies the iteration count of each nested loop.
The boundaries $\iter$ draws in the event structure do not have to align with
the start and end of the loop body in the syntax.
For each event, the combination of the program counter, event type, and
$\iter(e)$ is unique.
The loops nesting a given line of code are totally ordered by containment, so
$\loopfun(e)$ is a chain; we write it $\ell_1\prec\dots\prec\ell_k$ from the
outermost enclosing loop inwards. We require two properties of $\iter$.

\begin{enumerate}

  \item\label{iter:mono} \emph{Monotonicity.} In each execution, fix a loop
    $\ell$. For all events $e_1\po e_2$ with
    $\ell\in\loopfun(e_1)\cap\loopfun(e_2)$ that agree on the iteration count of
    every loop enclosing $\ell$, $\iter(e_1)(\ell)\leq\iter(e_2)(\ell)$.

  \item\label{iter:nesting} \emph{Compatibility with nesting.} For loops
    $\ell\prec\ell'$, the boundary $\iter$ draws for $\ell$ does not fall
    strictly within one execution of the body of $\ell'$: any two events
    performed by a single execution of the body of $\ell'$ agree on
    $\iter(\cdot)(\ell)$.

\end{enumerate}

\noindent Condition~\ref{iter:mono} compares the counts of a single loop, and
only between events of one iteration of all loops enclosing this loop.
Condition~\ref{iter:nesting} ensures that the iteration count of a nested loop
is reset at the start of the next iteration of an enclosing loop. Comparing
iterations in nested loops is equivalent to a lexicographic ordering of
iterations in all loops globally.
Where no loop nests within another, $\loopfun(e)$ is a singleton,
Condition~\ref{iter:nesting} does not apply and Condition~\ref{iter:mono} is the
pointwise comparison.

We now define episodicity as a semantic property, quantified over all
executions. Its conditions distil a careful study of retry loops, such as those
in RCU, generalised just far enough to yield the finitary event-structure
semantics of Theorem~\ref{t:finite-post-futures}.
Defining episodicity over executions, rather than over program syntax, is
essential: the $\rf$-relation is defined in executions, equivalence of memory
locations is defined relative to the constraints of an execution, branching
conditions evaluate relative to the constraints of the executions, and the
dependency relations between events are defined relative to executions, not the
event structure.
The syntactic, statically checkable patterns that witness episodicity in
practice are therefore \emph{sufficient conditions}, not its definition.

\medskip\noindent\begin{minipage}{\linewidth}
\begin{definition}[Episodic Loops]\label{def:episodic}

  A loop $\ell$ in a program is \emph{episodic} if all of the following
  conditions are met.

  \begin{enumerate}

    \item\label{episodic:reg} Registers are only accessed if written to
      $\po$-before within the same loop iteration, or before the loop.

    \item\label{episodic:mem}
      Reads within the loop must read from:

      \begin{enumerate}

      \item \label{episodic-casea} a $\po$-earlier write from the same
        iteration, or a write from before the loop,

      \item \label{episodic-caseb} a write $w$ on another thread which no
        \emph{earlier} iteration of $\ell$ reaches, that is, writing $t$ for the
        thread executing $\ell$ and $i$ for the iteration of the reading event,
        \[
          \forall e.~\thread(e)=t\wedge\iter(e)(\ell)<i
          \implies(e,w)\notin{(\DP\cup\rf)}^+
        \]
        or

      \item \label{episodic-casec} a \emph{read-don't-modify-write} -- an
        atomic read-modify-write that stores back the value it read unchanged,
          and is therefore unobservable -- whose value is read from a write
          satisfying case~\ref{episodic-casea}, \ref{episodic-caseb}, or
          \ref{episodic-casec}.

      \end{enumerate}

    \item\label{episodic:cond} The branching conditions of an iteration do not
      constrain values read before the loop. Writing $\varphi_\ell$ for the
      conjunction of the conditions of the branching events of one iteration of
      $\ell$, and $\restrict{\cdot}{\emptyset}$ for the restriction of a
      predicate to the symbols read before the loop, as defined in
      Appendix~\ref{s:app-proofs-restricted},
      \[
        \restrict{\varphi_\ell}{\emptyset}~=~\top
      \]
      The requirement is on the conjunction rather than on each condition
      separately, as conditions that pin no such value on their own may do so
      jointly.

    \item\label{episodic:events} Events from prior loop iterations are ordered
      before events of later loop iterations by the transitive closure of $\ppo$
      and $\DP$:
      \[
        \forall e_1,e_2.
        \iter(e_1)(\ell)<\iter(e_2)(\ell)\implies
        (e_1,e_2)\in(\ppoord\cup\DP)^+
      \]

  \end{enumerate}

\end{definition}
\end{minipage}\par\bigskip

Conditions~\ref{episodic:reg} and~\ref{episodic:mem} forbid passing a value
from one iteration to the next, both directly within the same thread and
indirectly through another thread. Within a thread values flow along $\DP$;
between threads they flow only along $\rf$, and case~\ref{episodic-caseb}
admits a read from another thread's write only when no earlier iteration of the
loop reaches that write. It therefore bounds where the origin of a value read in
the $i$-th iteration can lie: tracing $\DP$ and $\rf$ backwards from such a read
never reaches an event of the loop's own thread in an earlier iteration, and by
case~\ref{episodic-casea} the steps of that trace within the thread stay in the
$i$-th iteration or leave the loop altogether. No state is carried between
iterations through memory. A loop's own writes may still be read by other
threads -- episodicity restricts only what the loop reads, not who reads its
writes -- but the written value cannot be read back into the loop.

Expressions are evaluated against the register state.
Conditions~\ref{episodic:reg} and~\ref{episodic:events} ensure that the
expressions of corresponding events agree from one iteration to the next, up to
the renaming of the symbols each iteration introduces.
In \smrd, a branch is an event whose condition is recorded as a value
restriction on the execution, and an execution is admitted only if
its accumulated restrictions are satisfiable. Condition~\ref{episodic:cond}
forbids these restrictions from constraining \emph{non-retriable} reads, that
is values read before the loop. Otherwise a branch can pin such a value: in
\code{r:=*x; while(r>5)}, the symbol for \code{r} is read once, before the
loop, yet the condition forces $r\le5$ on every admitted execution -- the
semantics models only terminating executions and discards those in which the
loop never exits -- so a value fixed before the loop dictates the loop's
behaviour, breaking the retry pattern. A symbol re-read \emph{within} each
iteration carries no such information across the boundary.
The condition is asked of the conditions of an iteration jointly because
pinning a pre-loop value takes no single branch: with $\alpha$ read into
\code{r0} before the loop and $\beta$ into \code{r1} within it,
\code{if(r1==r0)} and a nested \code{if(r1==5)} each leave $\alpha$
unconstrained on their own, $\beta$ being free in each, while together they
force $\alpha=5$ -- a value fixed before the loop, pinned inside it, and
retained by a reset that goes by the symbols a restriction mentions.
Condition~\ref{episodic:events} precludes interference across loop boundaries,
ensuring that read events do not distribute across the boundary. Values
generated by read events within the loop are thus a function of the program
counter, and are constant across loop iterations.

The criteria of episodicity range over all executions. Establishing them
directly would therefore require the whole event structure semantics of a
program, which is unbounded in the presence of unbounded loops.
We provide sufficient conditions which can be checked successively on a
partially calculated event structure, using the inherent modularity of \smrd{},
for instance in \mordor{}.
In practice the episodicity conditions are met by code patterns that are
syntactically -- indeed statically -- checkable.
Below we discuss the individual conditions, and Section~\ref{s:eval} describes
the checks implemented in \mordor{}.

\subsection{Identifying Episodic Loops}\label{s:identifying}

We have confirmed -- by hand and using \mordor{} -- that the retry loops in all
of the following algorithms adhere to the episodicity criteria:
RCU (\code{rcu-1.lit}) following Gotsman et al.~\cite{Gotsman23grace} per code
listing in Appendix~\ref{app:rcu}, hazard pointers (\code{hp-1.lit}) adapted
from Folly~\cite{folly} by inlining all functions per Appendix~\ref{app:hp},
seqlock~\cite{hemminger2002seqlock} (\code{seqlock-1.lit}) as in
Appendix~\ref{app:seqlock}, and
spinlock~\cite[\S8.5]{hennessy1996computer} (\code{spinlock-1.lit}) as in
Appendix~\ref{app:spinlock}. The four programs are those of
\code{programs/episodicity/} in \mordor{}~\cite{kissig2026mordor}, which are
the ones \mordor{}'s test suite measures.
In RCU and hazard pointers they hold the increment operation alone. What follows
the retry loop -- \code{sync} and \code{reclaim} in RCU, \code{retire} and
\code{scan} in hazard pointers -- is elided, since the episodicity criteria are
conditions on the loops.

Condition~\ref{episodic:reg} can be verified statically on register variables.
The remaining conditions of episodicity require a case-by-case analysis. Static
verification is immediate in the four listings: RCU in Example~\ref{ex:rcu},
hazard pointers in Example~\ref{ex:hazptr}, the seqlock in
Example~\ref{ex:seqlock}, and the spinlock in Example~\ref{ex:spinlock}.

Condition~\ref{episodic:mem} requires testing the equivalence of the memory
locations of reads and writes. Memory locations are symbolic. Their concrete
equivalence depends on the constraints of the execution in context. The
equivalence holds trivially when a pair of accesses use one pointer variable,
and the variable is not modified. This is the case in all four algorithms: the
spinlock contends on the single location \code{mutex}
(Example~\ref{ex:spinlock}); the seqlock's \code{rseq} and \code{rdata} are
fixed at allocation and never reassigned (Example~\ref{ex:seqlock}); the hazard
pointer loop reads and writes the hazard slot through \code{hp[tid]} for a
fixed \code{tid} (Example~\ref{ex:hazptr}); and in RCU the pointer \code{s}
obtained from the \fadd{} is not reassigned before it is dereferenced
(Example~\ref{ex:rcu}). See Example~\ref{ex:episodic-rcu} below for a
discussion on RCU. These are instances of the fragment of
Lemma~\ref{l:loc-decidable}, in which the equivalence is decidable rather than
trivial: a pointer that is reassigned stays decidable as long as what is written
to it is again a location and a constant offset.

Case~\ref{episodic-caseb} of Condition~\ref{episodic:mem} is stated
semantically, over the $\DP$ and $\rf$ of a given execution, but it is testable
by a data flow analysis. By Definition~\ref{def:freeze} every $\DP$ edge runs from
$\origin\alpha$ to a write, allocation or deallocation whose justification
mentions $\alpha$, and symbols originate at read and allocation events, so each such edge is subsumed by a
def--use edge of the thread's own data and control flow. Every $\rf$ edge is
subsumed by the may-alias relation on shared locations. Writing $\leadsto$ for
the closure of those two syntactic relations,
${(\DP\cup\rf)}^+\subseteq~\leadsto$ on every execution, and a loop in which
$\leadsto$ admits no path from a write in the body back to a read in the body
across the loop's back edge therefore satisfies case~\ref{episodic-caseb}.

What this asks for is loop-carried dependence analysis, extended across threads
through shared locations: intra-thread def--use chains, a may-alias relation for
the cross-thread hops, and the question of whether a chain crosses the back
edge. Stating the case relative to earlier \emph{iterations} rather than to the
thread as a whole is what keeps it in that form; a formulation over the thread
would ask instead for reachability across the whole program.

One correlation is needed beyond the standard analysis. A conditional
read-modify-write whose write effect occurs only on the branch that leaves the
loop, as in the $\cas$ of Example~\ref{ex:rcu}, contributes no loop-carried
edge, because no further iteration follows it. Recovering that requires relating
the write effect to the result the loop condition tests, rather than treating
them as independent statements. An analysis without the correlation remains
sound -- it over-approximates -- but rejects the retry loops of
Appendix~\ref{app:rcu} and Appendix~\ref{app:hp}.

Condition~\ref{episodic:cond} requires calculating the constraints that the
branching conditions of an iteration imply together. In each of the algorithms
we consider, the loop condition compares two values read in the same iteration
and no other branch occurs in the loop, so that the episodicity condition is
trivially satisfied.

Condition~\ref{episodic:events} generally requires calculating the event
structure semantics of the program in order to establish a $\ppo$-ordering
between events in successive iterations of the loop. In the case of the
algorithms we considered, the ordering follows statically: the \fadd{} operation
in RCU and the memory fence in Hazard Pointers have release-acquire annotation.
In spinlock the read events reference the same literal memory location. In
seqlock every event in the failing case is release annotated.

Comparing the equality of symbolic memory locations requires tracing the
assignment of pointers through the program, and is undecidable in general.
Where a program takes no pointer offsets other than by constants, and compares
pointers only for equality, it is decidable, as
Appendix~\ref{s:app-proofs-locations} shows. All four algorithms considered here
are of that shape, and \mordor{} decides the query for them.

\subsection{Synchronisation Points}\label{s:sync-points}

\emph{Synchronisation points} are events that strongly order loop iterations.
\fadd{} in RCU and the memory fence in hazard pointers serve as
synchronisation points. They separate events occurring $\po$-earlier and
$\po$-later in program order $\po$ by preserved program order $\ppo$. When
placed at the beginning of each loop iteration, synchronisation points therefore
separate events across different iterations as drawn by $\iter$, making them
sufficient -- though not necessary -- for satisfying
Condition~\ref{episodic:events} of episodic loops.

\begin{definition}[Synchronisation Points]\label{def:sp}

  A \emph{synchronisation point} $e^\text{SP}$ is an event in the event
  structure such that $e\po e^\text{SP}$ iff $e\ppo e^\text{SP}$ and
  $e^\text{SP}\po e$ iff $e^\text{SP}\ppo e$ for all events $e$ in all
  executions.

\end{definition}

Consider as an example the \fadd{} instruction in RCU, which orders events
$\po$-earlier and $\po$-later by $\ppo$. Two of the three relations whose
closure forms $\ppo$ are at work. The first, $\pposync$, accounts for
statically declared memory order, carrying $\po$ across a releasing write or an
acquiring read. The second, $\ppormw$, accounts for the atomicity of
read-modify-write operations. The read and write events of an RMW are recorded
by the relation $\pormw\subseteq E\times\Expressions\times E$, which pairs
them with the condition under which the operation is atomic -- $\top$ for
\fadd{}, the branching condition for \cas{} -- so that $\fadd_R\pormw\fadd_W$
under $\top$.

$\pormw$ does not itself order the write event before the read event. It
contributes to $\ppo$ only through $\ppormw$, and there the write-to-read pair
occurs solely in composition with $\pposync$, on one side or the other.
$\ppormw$ thus extends $\pposync$ \emph{across} the RMW: events
$\pposync$-before the write are ordered ahead of the read, and the write ahead
of events $\pposync$-after the read, which is what keeps accesses to the same
location from being ordered between the read and the write of the RMW.
Appendix~\ref{s:app-defs-ppo} gives the formal definitions, and
Figures~\ref{fig:cas-rmw} and~\ref{fig:faa-rmw} there contrast the \cas{}
case, where $\ppormw$ holds on the succeeding branch only, with the
unconditional \fadd{} case.

\begin{example}\label{ex:rcu-sp}\emph{(Semantics of \fadd{} as a synchronisation
  point).}

  \noindent
  \begin{minipage}[c]{0.6\textwidth}
    In order to establish that \fadd{} separates events in executions by $\ppo$, we
  select two events $e_1$ and $e_2$, such that $e_1$ is $\po$-before \fadd,
  i.e.~$e_1\po\fadd^\rel_W$, and $e_2$ is $\po$-after \fadd,
  i.e.~$\fadd^\acq_R\po e_2$. By definition of the semantics of \fadd~in terms
  of $\pormw$, $\fadd^\rel_W\ppormw[\top] \fadd^\acq_R$, so that
    $e_1\ppo\fadd^\rel_W$ and $\fadd^\acq_R\ppo e_2$. Thus \fadd{} meets
  the condition of synchronisation points as the diagram on the right commutes.
  \end{minipage}%
  \hfill
  \scalebox{0.8}{
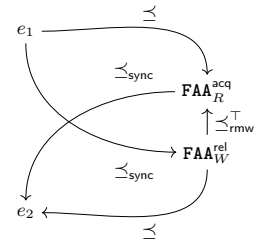
  \begin{minipage}[c]{0.35\textwidth}
  \begin{figure}[H]
    \centering
        \begin{tikzpicture}[node distance=3cm, auto]
  \node (e1) at (0,3) {$e_1$};
  \node (faar) at (3,2) {$\fadd^\acq_R$};
  \node (faaw) at (3,1) {$\fadd^\rel_W$};
  \node (e2) at (0,0) {$e_2$};

  \draw[<-] (faar) -- node[right] {$\ppormw[\top]$} (faaw);

  \draw[<-] (faar.north) to[out=90,in=0] node[above] {$\ppo$} (e1.east);
  \draw[<-] (faaw.west) to[out=180,in=270] node[below,xshift=1cm,yshift=-0.5cm] {$\pposync$} (e1.south);
  \draw[<-] (e2.east) to[out=0,in=270] node[below] {$\ppo$} (faaw.south);
  \draw[<-] (e2.north) to[out=90,in=180] node[above,xshift=1cm,yshift=0.5cm] {$\pposync$} (faar.west);
\end{tikzpicture}
  \end{figure}
  \end{minipage}
}
\end{example}

\subsection{Loop Boundaries}

Loop boundaries determine which $\po$-connected sections of executions fall into
one iteration as opposed to an earlier or later iteration. Once we allow
instructions to reorder across the loop condition -- breaking consistency with
$\po$ -- it becomes important to distinguish the syntactic loop boundaries, as
given in the abstract syntax tree, from the semantic loop boundaries given by
the $\iter$-function underlying our definition of episodicity above.

In the case of the RCU variant in this paper, the unbounded \code{while}-loop in
the \code{inc()} function is episodic under an $\iter$-function which draws the
loop boundary just before the \fadd{}, so that the release-acquire
synchronisation point at the \fadd{} separates iterations and
Condition~\ref{episodic:events} is met. Each iteration thus runs from the
\fadd{} through the RCU-exit \code{rcu[tid]:=0} and RCU-enter \code{rcu[tid]:=1}
that close it.

Figure~\ref{fig:rcu-loop-boundary} shows the event structure of the loop, each
event annotated with its code line; the \cas{} branches into a failing iteration
that continues the loop and a successful one that exits it. The braces on the
right, labelled $\iter(\cdot)=0$ and $\iter(\cdot)=1$, mark the first two
semantic iterations drawn by $\iter$. The {\color{black}black} brace on the left
marks one \emph{syntactic} loop body. It is offset from the semantic iteration
by the RCU-exit \code{rcu[tid]:=0} and the RCU-enter \code{rcu[tid]:=1}: the
syntactic loop boundary does not coincide with the semantic one drawn by
$\iter$.

\begin{figure}[t]
  \centering
  \begin{minipage}[t]{0.54\textwidth}
    \vspace{0pt}
    \centering
    \scalebox{0.82}{\begin{tikzpicture}[every node/.style={font=\small}]
  \node[align=left, text width=2.5cm] (a0) at (0,0)    {\code{rcu[tid]:=}$^{\text{\color{blue}rel}}$\code{0}};
  \node[align=left, text width=2.5cm] (a1) at (0,-0.8) {\code{rcu[tid]:=1}};
  \node[align=left, text width=2.5cm] (f1) at (0,-1.6) {\code{s:=FAA(\&C,0)}};
  \node[align=left, text width=2.5cm] (d1) at (0,-2.4) {\code{v:=*s}};
  \node[align=left, text width=2.5cm] (i1) at (0,-3.2) {\code{*n:=v+1}};
  \node[align=left, text width=2.5cm] (c1) at (0,-4.0) {\code{r:=CAS(\&C,s,n)}};
  \node[align=left, text width=2.5cm] (x1) at (0,-5.0) {\code{rcu[tid]:=}$^{\text{\color{blue}rel}}$\code{0}};
  \node[align=left, text width=2.5cm] (e1) at (0,-6.0) {\code{rcu[tid]:=1}};
  \node[align=left, text width=2.5cm] (f2) at (0,-7.0) {\code{s:=FAA(\&C,0)}};
  \node[align=center, text width=2.5cm] (vdF) at (0,-7.8) {$\vdots$};
  \node (vdS) at (1.75,-5.0) {$\vdots$};

  \node (pb0) at (-1.75,0)    {~};
  \node (pbc) at (-1.75,-4.0) {~};
  \node (pra) at (-2.7,-1.6)  {~};
  \node (prx) at (-2.7,-6.2)  {~};

  \draw[->] (a0)--(a1); \draw[->] (a1)--(f1); \draw[->] (f1)--(d1);
  \draw[->] (d1)--(i1); \draw[->] (i1)--(c1);
  \draw[->] (c1)--(x1) node[midway,left=1pt] {\scriptsize failure};
  \draw[->] (c1.south east) -- node[midway,above right=-2pt] {\scriptsize success} (vdS);
  \draw[->] (x1)--(e1); \draw[->] (e1)--(f2); \draw[->] (f2)--(vdF);

  \draw[decorate,decoration={brace,amplitude=4pt,raise=0pt}]
    (pbc.south) -- (pb0.north);
  \draw[red,decorate,decoration={brace,amplitude=4pt,raise=0pt}]
    (prx.south) -- (pra.north);
  \draw[decorate,decoration={brace,amplitude=4pt,mirror,raise=0pt}]
    (3.0,-5.2) -- (3.0,-0.8);
  \draw[decorate,decoration={brace,amplitude=4pt,mirror,raise=0pt}]
    (3.0,-8.0) -- (3.0,-5.7);

  \node[font=\scriptsize,rotate=90]     at (-2.1,-1.6)  {syntactic body};
  \node[red,font=\scriptsize,rotate=90] at (-3.05,-3.8) {syntactic body after
  rewriting};
  \node[font=\scriptsize,rotate=90]     at (3.45,-2.9)  {$\iter(\cdot)=0$};
  \node[font=\scriptsize,rotate=90]     at (3.45,-7.0)  {$\iter(\cdot)=1$};
\end{tikzpicture}}
    \captionof{figure}{Event structure of the \code{inc()} loop. The \cas{}
    branches into a failing iteration (continuing the loop) and a successful
    one. Each semantic \emph{iteration} drawn by $\iter$ coincides
    with the rewritten syntactic loop body ({\color{red}red}) but is offset from
    the original syntactic loop body ({\color{black}black}); all enclose the same
    events.}\label{fig:rcu-loop-boundary}
  \end{minipage}\hfill
  \begin{minipage}[t]{0.44\textwidth}
    \vspace{0pt}
    \centering
    \input{snippets/rcu-inc-episodic}
    \captionof{figure}{The \code{inc()} loop rewritten (changes in
    {\color{red} red}) so the syntactic loop boundary coincides with the
    semantic boundary drawn by $\iter$.}\label{fig:rcu-inc-episodic}
  \end{minipage}
\end{figure}

Note that if one rewrites the program code so that the syntactic loop boundaries
match the semantic loop boundaries, then the event structure is maintained.
Rewriting the program code is not necessary to meet episodicity with our
definition using $\iter$. We only present it for illustration.
Figure~\ref{fig:rcu-inc-episodic} shows such a rewrite: the RCU-exit
\code{rcu[tid]:=}$^{\text{rel}}$\code{0} is moved to before the loop and,
for the remaining iterations, relocated to the end of the loop body under an
\code{if} over the failing outcome of the \cas{}, while the RCU-enter
\code{rcu[tid]:=1} stays at the head of the body.

In the resulting event structure the syntactic loop body (the {\color{red}red}
brace in Figure~\ref{fig:rcu-loop-boundary}) coincides with the semantic
iteration, while the events themselves are unchanged.

\begin{example}\label{ex:episodic-rcu}
  \emph{(Episodic loop in RCU).}
  The \code{while}-loop in \code{inc()} is episodic. We choose an $\iter$ that
  places the loop boundary just before \fadd{}. Each of the conditions in
  Definition~\ref{def:episodic} of episodic loops holds:

  \begin{enumerate}

    \item[\ref{episodic:reg}.] The registers used in the \code{while}-loop are $s$,
      $v$, and $r$, all of which are assigned first in the same iteration of the
      loop.

    \item[\ref{episodic:mem}.] $C$ is the only memory location read in
      \code{inc()}: in $\fadd~(I_7)$, the dereferencing instruction \code{v:=*s}
      ($I_8$), and finally in $\cas~(I_{10})$. $C$ is written in the same
      loop iteration by $\fadd~(I_7)$, which does not observably modify the
      value at $C$.

    \item[\ref{episodic:cond}.] The loop condition only depends on the values
      read in \fadd~and \cas~in the same iteration of the loop. There is no
      additional branching in the loop, which could constrain symbols read
      outside of the loop.

    \item[\ref{episodic:events}.]
      \fadd~occurs at the beginning of the loop boundary defined by
      $\iter$. Its acquire and release annotations add to $\ppo$,
      making the \fadd{} a synchronisation point as defined above
      (per Example~\ref{ex:rcu-sp}),
      satisfying Condition~\ref{episodic:events}.

  \end{enumerate}
\end{example}

\subsection{De Bruijn-style Indexing of Symbols}\label{s:de-bruijn}

The finitary quotient of the event structure semantics of programs with
unbounded loops identifies next enabled actions across iterations through a
mapping $\gamma$ on events, where $\gamma$ is a direct extension of a de
Bruijn-style indexing~\cite{deBruijn1972indexing} of the symbols introduced by
read and allocation events, starting from the end of executions of the program.
The de Bruijn-style indexing is aware of the program structure in the sense that
it identifies symbols introduced in the same thread, and either outside of loops
or in the same iteration of a loop, counting from the end of the program across
executions. We first assign a base index $\iota_0$ to each such event
inductively from the end of the program, and then diagonalise it by loop index
and thread to obtain the de Bruijn index $\iota:\Symbols\rightharpoonup\nat$ on
the symbols these events introduce.

\begin{definition}[De Bruijn-style Indexing of Symbols]\label{def:de-bruijn}

  Let $\mathbb X$ be an execution over a set $X$ of events. Write
  $X^\sigma\subseteq X$ for the \emph{symbol-introducing} events of $X$, that is
  the read events $(e\colon~R~x~\alpha)$ and the allocation events
  $(e\colon~\Allocs~\alpha~\expr)$, each of which introduces a fresh symbol
  $\alpha$. We assign to each $e\in X^\sigma$ an index $\iota_0(e)\in\nat$
  inductively from the end of the execution $\mathbb X$ by traversing the
  execution in program order, $\po$, from the $\po$-largest event in descending
  order, enumerating symbol-introducing events $e$ categorised by $\loopfun(e)$.

  Write $\upclosed e$ for the smallest subset of $X$ upward-closed under $\po$
  and containing $e$, and let $\hat X_e$ collect the symbol-introducing events
  that follow $e$ within the same loops:
  \[
    \hat X_e~\triangleq~(\upclosed e\setminus\{e\})\cap
    \{e'\in X^\sigma\mid\loopfun(e)=\loopfun(e')\}
  \]
  \noindent The index $\iota_0(e)$ of events $e$ is then given by
  \[
    \iota_0(e)\triangleq
    \left\{
      \begin{array}{ll}
        0 & \text{if}~\hat X_e=\emptyset \\
        \max\left(\{\iota_0(e')~\mid~e'\in\hat X_e\}\right)+1 & \text{otherwise} \\
      \end{array}
    \right .
  \]

  We then obtain the de Bruijn indexing of symbol-introducing events, and thus
  of symbols, by diagonalising $\iota_0$ over the innermost nesting loop and the
  thread. Write $\innerloop(\loopfun(e))$ for the index of the inner most loop
  containing $e$, and $0$ when $\loopfun(e)=\emptyset$:

  \[
  \iota(e) \triangleq
  N\cdot\left((L+1)\cdot\iota_0(e)+\innerloop(\loopfun(e))\right) + \thread(e)
  \]

  \noindent where $L$ is the number of loops in $\prog$, $N$ is the number of
  threads, numbered $0,\ldots,N-1$, and $\thread(e)$ denotes the thread on which
  $e$ occurs. Since loops are indexed from $1$, $\innerloop(\loopfun(e))$ ranges
  over $0,\ldots,L$ -- one value per loop and one for the events outside every
  loop, which is the factor $L+1$ above. Each of $\iota_0(e)$,
  $\innerloop(\loopfun(e))$ and $\thread(e)$ is therefore recovered from
  $\iota(e)$, so that $\iota$ is injective.
\end{definition}

The $\po$-largest event exists because $\poinv$ is a total order in executions.
The up-closure $\upclosed e$ exists because $\poinv$ is well-founded in event
structures generated for a finite step counter, these being finite.

\begin{corollary}[Properties of De Bruijn Indexing]\label{l:de-bruijn-props}
  Constructed as above, the following properties hold for a program
  $\prog$ with a loop $\ell$:
  \begin{enumerate}

    \item\label{prop:de-bruijn-tail} let $e$ be a symbol-introducing event in
      $\langle\prog\rangle_{n+1~\emptyset~\lambda\rho\,\varphi.\emptyset~\top}$ with
      $\iter(e)(\ell)=i+1$ and $e'$ a symbol-introducing event in
      $\langle\prog\rangle_{n~\emptyset~\lambda\rho\,\varphi.\emptyset~\top}$ with
      $\iter(e')(\ell)=i$, agreeing on the program counter and on the iteration
      count of every loop other than $\ell$. Then $e$ and $e'$ introduce the
      same symbol.

    \item\label{prop:de-bruijn-head} the symbols introduced at the same program
      counter before the loop are identical between executions for different
      step-counters.

  \end{enumerate}
\end{corollary}

\noindent Property~\ref{prop:de-bruijn-tail} identifies a symbol-introducing
event by its program counter together with the iteration counts of the loops
other than $\ell$ that nest it, and not by its program counter alone. The
distinction is immaterial unless a loop nests within $\ell$: a program counter
inside such a loop is reached once per iteration of it, so within one iteration
of $\ell$ it names as many events as that loop performs, and the counts of the
nested loops are what tell them apart. Where $\ell$ is the only loop nesting the
events in question the qualification is empty.

Figure~\ref{fig:rcu-two-executions-debruijn} illustrates the indexing on two
executions of the simplified RCU \code{inc()} loop under step-counters $n$
(left) and $n+1$ (right), the latter one iteration longer. Only the reads of
\fadd{} and \cas{} are shown; the code before the loop and after its exit is
elided as \scalebox{0.6}{$\vdots$}, as is the dereference between the \fadd{}
and the \cas{}. Program text is set in typewriter, distinguishing the pointer
\code{n} written by the \cas{} from the step-counter $n$. Each read carries the
symbol it introduces, and the number in brackets is the index $\iota_0$ counted
from the end. Under this indexing the extra iteration of the longer execution
appears at the \emph{start}, so symbols in the tails agree by de Bruijn index --
the identity of symbols asserted by Property~\ref{prop:de-bruijn-tail} of
Corollary~\ref{l:de-bruijn-props}.

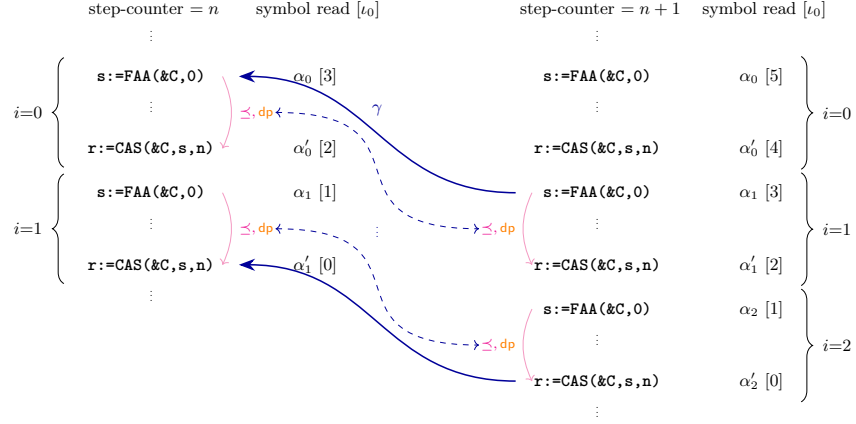
\begin{figure}[htbp] \centering
  \scalebox{0.72}{\begin{tikzpicture}[scale=1.0]
\tikzset{
    lbrace/.style={decorate, decoration={brace, mirror, amplitude=6pt}},
    rbrace/.style={decorate, decoration={brace, amplitude=6pt}},
    lbracelabel/.style={midway, left=8pt, font=\small},
    rbracelabel/.style={midway, right=8pt, font=\small}
}

\matrix (exA) [matrix of nodes, row sep=0.28cm, column sep=0.9cm,
  nodes={align=left}, column 1/.style={nodes={minimum width=3.1cm, align=left}},
  anchor=north west] at (0,0) {
    \smash{\scalebox{0.6}{$\vdots$}}                & \smash{$\phantom{\alpha_0'\ [0]}$} \\
    \code{s:=}\fadd\code{(\&C,0)}   & $\alpha_0\ [3]$  \\
    \smash{\scalebox{0.6}{$\vdots$}}                & \\
    \code{r:=}\cas\code{(\&C,s,n)}  & $\alpha_0'\ [2]$ \\
    \code{s:=}\fadd\code{(\&C,0)}   & $\alpha_1\ [1]$  \\
    \smash{\scalebox{0.6}{$\vdots$}}                & \\
    \code{r:=}\cas\code{(\&C,s,n)}  & $\alpha_1'\ [0]$ \\
    \smash{\scalebox{0.6}{$\vdots$}}                & \\
};

\matrix (exB) [matrix of nodes, row sep=0.28cm, column sep=0.9cm,
  nodes={align=left}, column 1/.style={nodes={minimum width=3.1cm, align=left}},
  anchor=north west] at (8.2,0) {
    \smash{\scalebox{0.6}{$\vdots$}}                & \smash{$\phantom{\alpha_0'\ [0]}$} \\
    \code{s:=}\fadd\code{(\&C,0)}   & $\alpha_0\ [5]$  \\
    \smash{\scalebox{0.6}{$\vdots$}}                & \\
    \code{r:=}\cas\code{(\&C,s,n)}  & $\alpha_0'\ [4]$ \\
    \code{s:=}\fadd\code{(\&C,0)}   & $\alpha_1\ [3]$  \\
    \smash{\scalebox{0.6}{$\vdots$}}                & \\
    \code{r:=}\cas\code{(\&C,s,n)}  & $\alpha_1'\ [2]$ \\
    \code{s:=}\fadd\code{(\&C,0)}   & $\alpha_2\ [1]$  \\
    \smash{\scalebox{0.6}{$\vdots$}}                & \\
    \code{r:=}\cas\code{(\&C,s,n)}  & $\alpha_2'\ [0]$ \\
    \smash{\scalebox{0.6}{$\vdots$}}                & \\
};

\coordinate (hdr) at (0,0.28);
\node[anchor=base, font=\small] at (exA-1-1.center |- hdr) {step-counter $=n$};
\node[anchor=base, font=\small] at (exA-1-2.center |- hdr) {symbol read $[\iota_0]$};
\node[anchor=base, font=\small] at (exB-1-1.center |- hdr) {step-counter $=n+1$};
\node[anchor=base, font=\small] at (exB-1-2.center |- hdr) {symbol read $[\iota_0]$};

\draw[lbrace] ($(exA-2-1.north west)+(-0.1,0.1)$) -- ($(exA-4-1.south west)+(-0.1,-0.1)$)
  node[lbracelabel] {$i{=}0$};
\draw[lbrace] ($(exA-5-1.north west)+(-0.1,0.1)$) -- ($(exA-7-1.south west)+(-0.1,-0.1)$)
  node[lbracelabel] {$i{=}1$};

\draw[rbrace] ($(exB-2-2.north east)+(0.1,0.1)$) -- ($(exB-4-2.south east)+(0.1,-0.1)$)
  node[rbracelabel] {$i{=}0$};
\draw[rbrace] ($(exB-5-2.north east)+(0.1,0.1)$) -- ($(exB-7-2.south east)+(0.1,-0.1)$)
  node[rbracelabel] {$i{=}1$};
\draw[rbrace] ($(exB-8-2.north east)+(0.1,0.1)$) -- ($(exB-10-2.south east)+(0.1,-0.1)$)
  node[rbracelabel] {$i{=}2$};

\draw[-{Stealth[length=3mm]}, thick, blue!60!black]
  (exB-5-1.west) to[out=180,in=0,looseness=1.15] (exA-2-1.east);
\draw[-{Stealth[length=3mm]}, thick, blue!60!black]
  (exB-10-1.west) to[out=180,in=0,looseness=1.15] (exA-7-1.east);
\node[font=\itshape, blue!60!black]
  at ($(exA-2-1.east)!0.5!(exB-5-1.west)+(0,0.45)$) {$\gamma$};

\coordinate (ppoA0a) at ($(exA-2-1.east)+(-0.3,0)$);
\coordinate (ppoA0b) at ($(exA-4-1.east)+(-0.3,0)$);
\coordinate (ppoA1a) at ($(exA-5-1.east)+(-0.3,0)$);
\coordinate (ppoA1b) at ($(exA-7-1.east)+(-0.3,0)$);
\coordinate (ppoB1a) at ($(exB-5-1.west)+(0.3,0)$);
\coordinate (ppoB1b) at ($(exB-7-1.west)+(0.3,0)$);
\coordinate (ppoB2a) at ($(exB-8-1.west)+(0.3,0)$);
\coordinate (ppoB2b) at ($(exB-10-1.west)+(0.3,0)$);
\def\edgeppodp[#1]#2#3{\draw (#2) edge[style=ppo,#1]
  node[] {\scriptsize\color{magenta}$\ppo,\DP$} (#3);}
\edgeppodp[bend left=25, every node/.append style={fill=white, inner sep=1pt, name=pA0, anchor=west, xshift=3pt}]{ppoA0a}{ppoA0b}
\edgeppodp[bend left=25, every node/.append style={fill=white, inner sep=1pt, name=pA1, anchor=west, xshift=3pt}]{ppoA1a}{ppoA1b}
\edgeppodp[bend right=25, every node/.append style={fill=white, inner sep=1pt, name=pB1, anchor=east, xshift=-3pt}]{ppoB1a}{ppoB1b}
\edgeppodp[bend right=25, every node/.append style={fill=white, inner sep=1pt, name=pB2, anchor=east, xshift=-3pt}]{ppoB2a}{ppoB2b}

\draw[<->, dashed, blue!60!black] (pA0) to[out=0, in=180, looseness=1.8] (pB1);
\draw[<->, dashed, blue!60!black] (pA1) to[out=0, in=180, looseness=1.8] (pB2);

\coordinate (gm1) at ($(exA-2-1.east)!0.5!(exB-5-1.west)$);
\coordinate (gm2) at ($(exA-7-1.east)!0.5!(exB-10-1.west)$);
\node[blue!60!black] at ($(gm1)!0.5!(gm2)$) {\scalebox{0.6}{$\vdots$}};

\end{tikzpicture}}
  \caption{
    De Bruijn indexing of symbols in executions with step counter, and the
    $\gamma$ mapping between iterations, identifying symbols by de Bruijn index.
    The diagram shows only a simplified presentation of the execution; the index
    $\iota_0$ shown here enumerates the drawn reads, and the actual de Bruijn
    index $\iota$ is the result of diagonalising $\iota_0$ by loop index and
    thread. The $\ppo,\DP$ arrows are likewise simplified: each stands for the
    $\ppo$ and $\DP$ dependency edges within one iteration of the loop, not for
    a single edge between the two events it connects.
}\label{fig:rcu-two-executions-debruijn}
\end{figure}

\subsection{Finite Bound on Event Structure
Semantics}\label{s:essem-finite-bound}

\smrd{}~\cite{Richards25SMRD} uses step-counters to give a semantics to loops,
decrementing the counter on loop iteration. We take the counter \emph{per loop},
each loop carrying its own bound: a single counter shared between nesting levels
leaves a loop nested under $k$ iterations of an enclosing loop with a smaller
unravelling than under $k-1$, and the results below need successive iterations
of a loop to carry the same unravelling of the loops nested within them. The two
readings agree unless loops are nested; the refinement is set out in
Appendix~\ref{s:app-defs}.
The episodicity criteria guarantee a structural symmetry of the event structure
semantics between loop iterations. The following exploits the symmetry to
establish a finite bound on the next enabled actions in programs with episodic
loops.

\begin{lemma}\label{l:post-futures-narrowed}

  In programs where all unbounded loops are episodic, next enabled actions
  monotonically narrow down over loop iterations: in an episodic loop $\ell$,
  for any history $H$ ending in the $i+1$-st iteration in $\ell$,
  $\horizon\Phi_{H_{i+1}}\subseteq\horizon\Phi_{H_i}$ where $H_{i+1} =
  H\setminus\{e\mid i+1\leq\iter(e)(\ell)\}$ and $H_i = H\setminus\{e\mid
  i\leq\iter(e)(\ell)\}$.

\end{lemma}

The unabridged proof is available in Appendix~\ref{s:app-proofs}. The following
is an outline of the proof.

The proof hinges on the de Bruijn-style indexing of the symbols introduced in
executions, counted from the end (Definition~\ref{def:de-bruijn}).
Consider two executions of the program: the first with step-counter $n$, and the
second, an extension of the first, with step-counter $n+1$. Under the de
Bruijn-style indexing scheme, the additional loop iteration afforded by
step-counter $n+1$ appears at the start of the existing iterations under
step-counter $n$ (Figure~\ref{fig:rcu-two-executions-debruijn}).
We construct a partial mapping, $\gamma$, between the events of the two
executions, extending the identity map on de Bruijn indices, such that events of
the $i+1$-st loop iteration under step-counter $n+1$ are mapped to events of the
$i$-th iteration under step-counter $n$. $\gamma$ preserves program counter,
event type, and follows branching decisions. As $\pc$, event type, and $\iter$
jointly uniquely identify events in executions, $\gamma$ is injective. $\gamma$
is undefined on the additional first iteration afforded by the larger
step-counter. As $\gamma$ follows branching decisions, it preserves the
iteration count of inner nested loops.

Using the episodicity criteria and the inductive constructions of the dependency
relations $\ppo$ and $\DP$, it follows that $\gamma$ preserves and reflects
$\ppo$, and reflects $\DP$. $\gamma$ does not preserve $\DP$ everywhere: value
restrictions record the branching conditions of every earlier iteration, and the
first iteration in the larger event structure is outside the domain of $\gamma$,
so it sources $\DP$-edges that $\gamma$ does not carry over. Those edges are
discarded in the posterior futures, whose sources lie outside the history, so
$\gamma$ still establishes an equivalence of next enabled actions across
successive iterations in separate event structures generated for successive
step-counters.

The step-counter semantics of unbounded loops defines their behaviour using
nested \code{if}-statements.
A larger step-counter extends this nesting, so the event structure generated for
a given step-counter embeds into those of larger ones.
The least fixed point in the lattice of all event structures with
$\subseteq$-inclusion is the event structure of all terminating executions.
Every event structure generated for a finite step-counter embeds into the fixed
point, and so does $\gamma$ as a partial map on events.
Figure~\ref{fig:es-embedding-fixpoint} shows the three structures side by side.
Each \cas{} that the step-counter does not cut off branches into the loop exit
$\bot$, taken when it succeeds, and the next iteration's \fadd{}, taken when it
fails; the embeddings $\mathit{id}$ and $\mathit{id}^{\mathrm{lim}}_k$ identify
each structure with an initial part of the next and of the fixed point, while
$\gamma$ runs the other way, back from an iteration to its predecessor.

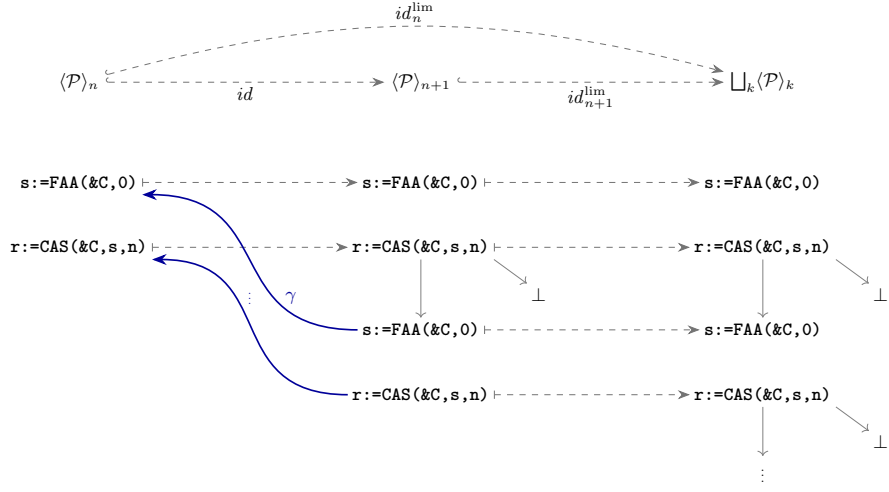
\begin{figure}[htbp]
  \centering
  \scalebox{0.78}{\begin{tikzpicture}[
  ev/.style={inner sep=2pt},
  es/.style={->, draw=white!50!black},
  emb/.style={{Hooks[right]}-{Stealth[length=2mm]}, dashed, draw=black!55},
  map/.style={{Bar[width=4pt]}-{Stealth[length=2mm]}, dashed, draw=black!55},
  emblabel/.style={above, font=\small, inner sep=2pt},
  emblabelb/.style={below, font=\small, inner sep=2pt},
  gam/.style={-{Stealth[length=2.5mm]}, thick, blue!60!black},
  panel/.style={font=\small},
]

\node[ev] (a1f) at (0,0)      {\code{s:=}\fadd\code{(\&C,0)}};
\node[ev] (a1c) at (0,-1.1)   {\code{r:=}\cas\code{(\&C,s,n)}};
\node[panel] (lblA) at (0,1.7) {$\langle\prog\rangle_n$};

\node[ev] (b1f) at (5.8,0)    {\code{s:=}\fadd\code{(\&C,0)}};
\node[ev] (b1c) at (5.8,-1.1) {\code{r:=}\cas\code{(\&C,s,n)}};
\node[ev] (b1x) at (7.8,-1.9) {$\bot$};
\node[ev] (b2f) at (5.8,-2.5) {\code{s:=}\fadd\code{(\&C,0)}};
\node[ev] (b2c) at (5.8,-3.6) {\code{r:=}\cas\code{(\&C,s,n)}};
\draw[es] (b1c.south east) -- (b1x.north west);
\draw[es] (b1c) -- (b2f);
\node[panel] (lblB) at (5.8,1.7) {$\langle\prog\rangle_{n+1}$};

\node[ev] (c1f) at (11.6,0)    {\code{s:=}\fadd\code{(\&C,0)}};
\node[ev] (c1c) at (11.6,-1.1) {\code{r:=}\cas\code{(\&C,s,n)}};
\node[ev] (c1x) at (13.6,-1.9) {$\bot$};
\node[ev] (c2f) at (11.6,-2.5) {\code{s:=}\fadd\code{(\&C,0)}};
\node[ev] (c2c) at (11.6,-3.6) {\code{r:=}\cas\code{(\&C,s,n)}};
\node[ev] (c2x) at (13.6,-4.4) {$\bot$};
\node[ev] (c3)  at (11.6,-4.9) {\scalebox{0.7}{$\vdots$}};
\draw[es] (c1c.south east) -- (c1x.north west);
\draw[es] (c1c) -- (c2f);
\draw[es] (c2c.south east) -- (c2x.north west);
\draw[es] (c2c) -- (c3);
\node[panel] (lblC) at (11.6,1.7) {$\bigsqcup_k\langle\prog\rangle_k$};

\draw[map] (a1f.east) -- (b1f.west);
\draw[map] (a1c.east) -- (b1c.west);
\draw[map] (b1f.east) -- (c1f.west);
\draw[map] (b1c.east) -- (c1c.west);
\draw[map] (b2f.east) -- (c2f.west);
\draw[map] (b2c.east) -- (c2c.west);

\draw[emb] (lblA) -- node[emblabelb] {$\mathit{id}$}   (lblB);
\draw[emb] (lblB) -- node[emblabelb] {$\mathit{id}^{\mathrm{lim}}_{n+1}$} (lblC);
\draw[emb] (lblA) to[bend left=15] node[emblabel] {$\mathit{id}^{\mathrm{lim}}_{n}$} (lblC);

\draw[gam] (b2f.west) to[out=180, in=0, looseness=1.4] (a1f.south east);
\draw[gam] (b2c.west) to[out=180, in=0, looseness=1.4] (a1c.south east);
\node[font=\itshape, blue!60!black] at (3.6,-1.95) {$\gamma$};

\coordinate (gmA) at ($(a1f.south east)!0.5!(b2f.west)$);
\coordinate (gmB) at ($(a1c.south east)!0.5!(b2c.west)$);
\node[blue!60!black] at ($(gmA)!0.5!(gmB)$) {\scalebox{0.6}{$\vdots$}};

\end{tikzpicture}}
  \caption{Event structures for successive step-counters and their embedding
  into the least fixed point. Only the \fadd{} and \cas{} of each iteration are
  shown, and $\bot$ marks the loop exit. A larger step-counter extends the
  nesting, so $\langle\prog\rangle_n$ embeds into $\langle\prog\rangle_{n+1}$
  and every finite structure into the fixed point. Both embeddings are the
  identity on events: $\mathit{id}$ into the structure for the next
  step-counter, and $\mathit{id}^{\mathrm{lim}}_k$, its iteration to the limit,
  from $\langle\prog\rangle_k$ into the fixed point. $\gamma$ is a partial map
  within the fixed point.}\label{fig:es-embedding-fixpoint}
\end{figure}

Every execution generated for a finite step-counter that exits its loops is an
execution of the fixed point, so each $\gamma$ relates two executions of that
one event structure and induces an embedding of next enabled actions within it.
The embedding maps the next enabled actions of the next iteration into the
current iteration.
The embedding is monotonic over the loop iteration count. The embedding is not
surjective, as the $\rf$-relation in an execution may entail the termination of
the loop. In our example of RCU with a strong \cas{}, each thread's \cas{} may
either read another thread's \cas{} and fail, or succeed otherwise. In each
iteration of the loop, one of the threads must succeed. The correspondence
establishes the property of Lemma~\ref{l:post-futures-narrowed}.
The narrowing, and hence the bound, follows from the monotonicity of the
correspondence alone; it does not rely on any iteration eventually succeeding.
Furthermore, the monotonicity of the correspondence identifies the next enabled
actions in the first iteration as an upper bound on those of all later
iterations, a property captured by the following theorem. Note that the theorem
only provides a finite bound on next enabled actions, not on the execution
suffixes themselves, which continue to termination and are still generally
unbounded.

\begin{restatable}[Finite Bound on Next Enabled Actions]{theorem}{finitenextactions}\label{t:finite-post-futures}

  In a program $\prog$ where unbounded loops are episodic, there are only
  finitely many next enabled actions in the event structure semantics.

\end{restatable}

\conditionalpagebreak
\section{Finitary Operational Semantics of Episodic Loops}\label{s:opsem}

In the previous section, we showed that episodicity gives rise to a finitary
event-structure semantics, which lets us reason uniformly about failing
iterations through the lens of next enabled actions.
The semantics loses track of individual iterations, restricting reasoning to
safety properties, that is, invariants that can be tested to hold independently
of past failing iterations.
We verify the fix to the UAF bug through such a safety property, stated in a
propositional logic over predicates in an ownership-based instance of
Owicki-Gries logic. These properties admit reasoning in an operational semantics
whose steps are derived from next enabled actions -- making that semantics
finitary as well.
The same applies to any safety property expressible in this ownership logic --
over register and memory values and the ownership predicates of the proof state;
use-after-free freedom is the instance we treat here, and ABA-freedom is
another. Liveness properties, such as a loop's eventual termination, lie outside
the scope of this semantics, which abstracts away the iteration count.

This is a \emph{symbolic} operational semantics, reusing the machinery of the
SMRD primer (Section~\ref{s:essem}): reads produce fresh symbols rather than
concrete values, memory locations are symbolic, and values are constrained in
the predicate $\sigma.\varphi$ in the program state $\sigma$, accumulated over
derivations. Transitions follow the preserved-program-order and dependency
relations $\ppo$ and $\DP$, stepping through an execution's next enabled
actions, while thread-local computation is carried in the register state $\rho$.

\emph{Derivations} in the operational semantics are contextual in the program
$\prog$. In the rules below, $\overline\prog$ denotes the \emph{atomic set
unravelling} of $\prog$ by projecting $\prog$ by thread $t\in\Threads$,
unravelling loops as nested \code{if}-statements, and assigning a unique label
to each step in $\prog$. We write $\overline\prog(t)$ for the projection onto
thread $t$. The labels are a composite of program counter, action type, and
iteration index per loop -- matching the uniqueness property of event labels in
the event structure semantics. The next instruction from the program is unpacked
from the atomic set unravelling.
Derivations track transitions
$(\sigma,\rho,H)\xrightarrow{\overline{a}}(\sigma',\rho',H')$ between
configurations consisting of \emph{program state} $\sigma$, \emph{register
state} $\rho$, and history $H$. Histories in the operational semantics match
histories in the event structure semantics, as events are labelled actions.

The \emph{one-step semantics} of commands is defined in terms of a one-step
semantics of actions. Actions modify the program state and commands modify the
register state.
Figure~\ref{fig:opsem:overview} collects the complete set of rules -- the
command one-step rules and action one-step rules together with the future
stepping rules -- for reference; the remainder of this section introduces and
explains them individually. Throughout, $\funup{f}{x}{v}$ denotes the
\emph{partial update of a function} $f$ at $x$ by $v$, such that
$(\funup{f}{x}{v})(y)$ is $v$ if $y=x$ and $f(y)$ otherwise; the rules use it
to update register states, viewfronts, and the components of $\sigma$.

\begin{figure}[p]
  \centering
  \small
  \captionsetup[sub]{font=small,skip=7pt}
  \captionsetup{skip=26pt}
  \setlength{\tabcolsep}{2pt}

  \framebox[\linewidth]{\textbf{Command one-step semantics}\hfill
    $\sigma\vdash\left(l\colon c,\rho\right)\xrightarrow{\overline a}_t\left(\skipcmd,\rho'\right)$\hfill}

  \bigskip

  \begin{tabular}{@{}cc@{}}
    \ovsub{0.48\linewidth}{90}{\input{rules/commands_set}}{Set register}{fig:command:set} &
    \ovsub{0.48\linewidth}{90}{\input{rules/commands_read-var}}{Read variable}{fig:command:read-var} \\[32pt]
    \ovsub{0.48\linewidth}{90}{\input{rules/commands_read-ref}}{Read reference}{fig:command:read-ref} &
    \ovsub{0.48\linewidth}{90}{\input{rules/commands_read-ptr}}{Read pointer}{fig:command:read-ptr} \\[32pt]
    \ovsub{0.48\linewidth}{90}{\input{rules/commands_write-var}}{Write variable}{fig:command:write-var} &
    \ovsub{0.48\linewidth}{84}{\input{rules/commands_write-ptr}}{Write pointer}{fig:command:write-ptr} \\[32pt]
    \ovsub{0.48\linewidth}{90}{\input{rules/commands_alloc}}{Allocate}{fig:command:alloc} &
    \ovsub{0.48\linewidth}{90}{\input{rules/commands_free}}{Free}{fig:command:free} \\[32pt]
    \ovsub{0.48\linewidth}{90}{\input{rules/commmands_fence}}{Fence}{fig:command:fence} &
    \ovsub{0.48\linewidth}{86}{\input{rules/commands_faa}}{FAA}{fig:command:faa} \\[32pt]
    \ovsub{0.48\linewidth}{77}{\input{rules/commands_cas-success}}{CAS success}{fig:command:cas-success} &
    \ovsub{0.48\linewidth}{80}{\input{rules/commands_cas-failure}}{CAS failure}{fig:command:cas-failure} \\
  \end{tabular}

  \caption{Overview of the finitary operational semantics: command one-step
  rules (updating the register state $\rho$, this page), action one-step rules
  (updating the program state $\sigma$) and the future stepping rules that drive
  derivations along the next enabled actions in $\horizon\Phi$ (next page). Each
  rule is introduced and explained individually in
  Section~\ref{s:opsem}.}\label{fig:opsem:overview}
\end{figure}

\begin{figure}[p]
  \ContinuedFloat
  \centering
  \small
  \captionsetup[sub]{font=small,skip=7pt}
  \captionsetup{skip=26pt}
  \setlength{\tabcolsep}{2pt}

  \framebox[\linewidth]{\textbf{Action one-step semantics}\hfill
    $\sigma\overset{a}{\rightsquigarrow}_t\sigma'$\hfill}

  \bigskip

  \begin{tabular}{@{}cc@{}}
    \ovsub{0.48\linewidth}{76}{\input{rules/actions_write}}{Write}{fig:action:write} &
    \ovsub{0.48\linewidth}{71}{\input{rules/actions_read}}{Read}{fig:action:read} \\[32pt]
    \ovsub{0.48\linewidth}{66}{\input{rules/actions_allocate}}{Allocate}{fig:action:allocate} &
    \ovsub{0.48\linewidth}{88}{\input{rules/actions_deallocate}}{Deallocate}{fig:action:deallocate} \\[32pt]
    \multicolumn{2}{c}{\ovsub{0.6\linewidth}{82}{\input{rules/actions_fence_branch}}{Fence / branch}{fig:action:fence-branch}} \\
  \end{tabular}

  \vspace{1.2cm}

  \framebox[\linewidth]{\textbf{Future stepping}\hfill
    $\overline\prog\vdash\left(\sigma,\rho,H\right)\rightarrow\left(\sigma',\rho',H'\right)$\hfill}

  \bigskip

  \begin{tabular}{@{}c@{}}
    \ovsub{\linewidth}{82}{\input{rules/step_non-lb_non-branch}}{Stepping outside loop boundaries or branching}{fig:step:non-lb_non-branch} \\[28pt]
    \ovsub{\linewidth}{62}{\input{rules/step_lb}}{Stepping at loop boundaries}{fig:step:lb} \\[28pt]
    \ovsub{\linewidth}{80}{\input{rules/step_branch-then}}{Branching: \code{then}}{fig:step:branch-then} \\[28pt]
    \ovsub{\linewidth}{80}{\input{rules/step_branch-else}}{Branching: \code{else}}{fig:step:branch-else} \\
  \end{tabular}

  \caption{Overview of the finitary operational semantics (continued): action
  one-step rules and future stepping rules.}
\end{figure}

\paragraph{Timestamped writes and viewfronts.} The program state $\sigma$
records the write actions performed so far, each stamped with a rational
\emph{timestamp} from $\tstamps$, in $\sigma.\writeset$. The timestamps totally
order the writes at each symbolic memory location, and that order is the
modification order (equivalently, the coherence order $\CO$) of the event
structure semantics: we write $\tst(\hat w)$ for the timestamp of a timestamped
write $\hat w=(w,q)$, and take $(w_1,w_2)\in\CO$ exactly when $w_1$ and $w_2$
write to locations equivalent under $\sigma.\varphi$ and
$\tst(\hat w_1)<\tst(\hat w_2)$.

Which of those writes a thread may read is not a global matter: writes are not
propagated to all threads at once, so each thread carries its own
\emph{viewfront}. The program state holds two families of viewfronts, both
mapping symbolic memory locations to timestamped writes:

\begin{itemize}

  \item $\sigma.\tview$, the viewfront of thread $t$. Thread $t$ may read any
    write at $x$ whose timestamp is not earlier than that of
    $\sigma.\tview(x)$, so the writes observable to $t$ at $x$ are
    \[
      \sigma.\OW(t,x)~\triangleq~
      \left\{
        (w,q)\in\sigma.\writeset
        \mid
        \loc(w)\equiv_{\sigma.\varphi}x
        \wedge
        \tst(\sigma.\tview(x))\leq q
      \right\}.
    \]
    Symbolic locations are compared under the constraint $\sigma.\varphi$ of the
    program state, as visible writes were: $\OW$ is the symbolic counterpart of
    the single last visible write, and collapses to it when every location is
    concrete and every thread's view is current.

  \item $\sigma.\mview[\hat w]$, the viewfront of the write $\hat w$ -- the
    viewfront its writing thread had when it performed $\hat w$. It is what a
    thread acquires when it synchronises with $\hat w$.

\end{itemize}

A command writing to a global variable and the underlying write action are given
by Figures~\ref{fig:command:write-var} and~\ref{fig:action:write}. The write
action picks a fresh timestamp $q'$ immediately after the timestamp $q$ of some
observable write, where
\[
  \sigma.\freshts(x,q,q')~\triangleq~
  q<q'\wedge
  \left(\forall (w,q'')\in\sigma.\writeset.~
    \loc(w)\equiv_{\sigma.\varphi}x\wedge q<q''\Rightarrow q'<q''\right),
\]
adds $(a,q')$ to $\sigma.\writeset$, advances the writing thread's viewfront to
it, and records that viewfront as the new write's $\mview$.

The read instruction $r :=_o x$ with memory order $o$ is interpreted in terms of
a read action $\Reads_o~x~\alpha$ reading the value of $x$ as a symbol $\alpha$
and updates the register state $\rho$ at $r$. The semantics of the read action
picks any observable write $\hat w\in\sigma.\OW(t,x)$ other than an anchor of
$t$ -- a write of $t$'s own that a loop boundary retained, introduced with the
reset rule below -- records the pair in $\sigma.\RF$, and adds a constraint equating $\alpha$
with the value written by that write to $\sigma.\varphi$. The command and the underlying
read action are given by Figures~\ref{fig:command:read-var}
and~\ref{fig:action:read}.

\paragraph{Synchronisation.} Reading also moves the reading thread's viewfront,
and this is the one place where information crosses threads. If the read is
plain, $\sigma.\tview$ merely advances at $x$ to the write that was read. If a
releasing write is read by an acquiring read -- $w\in\Wrel$ and $a\in\Racq$ --
the two \emph{synchronise}, and the reader additionally takes on the writer's
viewfront: the two viewfronts are combined by
\[
  (v_1\viewcomb v_2)(x)~\triangleq~
  \begin{cases}
    v_1(x) & \text{if }\tst(v_2(x))\leq\tst(v_1(x)) \\
    v_2(x) & \text{otherwise,}
  \end{cases}
\]
which takes the later write at every location. Writing
\[
  \sigma.\syncview_t(\hat w,a)~\triangleq~
  \begin{cases}
    \sigma.\tview\viewcomb\sigma.\mview
      & \text{if }w\in\Wrel\text{ and }a\in\Racq \\
    \sigma.\tview & \text{otherwise,}
  \end{cases}
\]
for the reader's viewfront after synchronising with $\hat w=(w,q)$, the read
rule advances it at $x$ to the write it read,
$\funup{\sigma.\syncview_t(\hat w,a)}{x}{\hat w}$. Everything the writer had
observed when it released is therefore observed by the reader afterwards, and
the stale writes it had already passed are no longer in $\OW$ for the reader.
This is what carries a release/acquire handshake -- the RCU reclaimer observing
a reader's quiescent exit write, and with it every write that reader made before
exiting.

Futures are built from $\ppo$ and $\DP$, both intra-thread, and $\rf$ -- the one
inter-thread dependency of the model -- does not contribute to the dependency
relations. Keeping $\rf$ out of $\Phi$ is what lets $\Phi$ split per thread, and
hence what lets the Owicki-Gries proof of Section~\ref{s:opsem-og} decompose
into a per-thread invariant. Cross-thread ordering is carried by the viewfronts
instead, and the synchronises-with edges they realise extend happens-before to
$\HB\triangleq(\DP\cup\ppoord\cup\SW)^+$ with
$\SW\triangleq\rf\cap(\Wrel\times\Racq)$.

Fences are transparent to viewfronts in this fragment
(Rule~\ref{fig:action:fence-branch}~(branch)): release and acquire are carried by the
annotations on the accesses themselves, as in the operational semantics of
\cite{Wright2023OpSem}, which likewise omits fences. A fence-based handshake
would need a further viewfront component in $\sigma$ recording the view at the
last read, and we do not treat it.

\paragraph{Future stepping.} The operational rules keep transitions consistent
with the dependency relations $\ppo$ and $\DP$ by following the next enabled
actions in $\horizon\Phi$ in future stepping rules. We write $(l\colon\overline
a)\in^*\horizon\Phi_H$ when the actions of a list $\overline a$ are enabled in
sequence -- each a next enabled action once its predecessors have been added to
the history $H$. Branching and fence actions are skipped in that test and not
recorded in $H$, since executions, and so $\Phi$ and histories, contain no such
events. Outside of loop boundaries and branching, future stepping is
given by Rule~\ref{fig:step:non-lb_non-branch}~(non-lb/non-branch).

\paragraph{Loop boundaries.} Loop boundaries are identified through a change of
iteration between history $H$ and control label $l$, such that $l$ is in the
loop $\ell$ and $\iter(e)(\ell)<\iter(l)(\ell)$ for all events $e$ in $H$.
$\iter(H)<\iter(l)$ denotes the lexicographic extension of $<$ to all loop
indices in $\loopfun(l)$. At loop boundaries, program state $\sigma$, register
state $\rho$, and history $H$ are reset to the beginning of the loop, where
$|^t_{\setminus\ell}$ removes terms added by the current thread $t$ in the loop.
On a program state it acts componentwise, on $\RF$ and $\varphi$ as before and
on the timestamped state by dropping the writes $t$ made in the loop from
$\sigma.\writeset$ together with their $\mview$s -- all but the one of greatest
timestamp at each location, which is retained as an \emph{anchor}. Episodicity
is what makes the removal harmless: by Condition~\ref{episodic:mem} no read
after the boundary observes a write of an earlier iteration, and the read rule
accordingly denies $t$ its own anchors, while leaving them observable to the
other threads as they were before the boundary. The anchor is what keeps $t$'s
viewfront on a write of the state; an entry of any other viewfront that pointed
at a removed write keeps that write's timestamp, entries being compared only by
timestamp. No viewfront therefore moves back, and the reset widens no thread's
reach (Section~\ref{s:opsem-sound-complete}).
Future stepping at loop boundaries is given by Rule~\ref{fig:step:lb}~(lb).

\paragraph{Branching.} Unlike the command rules above, which are directed by the
syntax of the instruction being executed, the two rules for branching are future
stepping rules: they resolve a branch taken by the executing thread $t$ rather
than interpreting an $\code{if}$-statement. The reason is that executions in the
event structure semantics, and thus the future set $\Phi$, do not contain
branching events, so there is no next enabled action for the rules to follow at
a branch. They are instead non-deterministic, and derivations evaluate both
branches simultaneously, as in Figures~\ref{fig:step:branch-then}
and~\ref{fig:step:branch-else}.
In these rules constraints from branching are accumulated in the predicate
$\sigma.\varphi$ under the program state $\sigma$. Impossible states with
unsatisfiable constraints are pruned.

The rules use the auxiliary functions $\ifCond$, $\enterThen$ and $\enterElse$
to detect branching instructions and extract the branching condition. $\ifCond$
extracts the expression in the $\code{if}$-statement. $\enterThen$ and
$\enterElse$ detect, for the executing thread $t$, the change from $H$ into
either the $\code{then}$ or the $\code{else}$ branch of an
$\code{if}$-statement, where all events in $H$ are $\po$-before the
$\code{if}$-statement, and $l$ is in the $\code{then}$ or the $\code{else}$
branch, respectively. $\enterThen$ and $\enterElse$ perform a similar task to
the iteration test $\iter(H)<\iter(l)$ in Rule~\ref{fig:step:lb}~(lb) above.
Example~\ref{ex:opsem-rcu} below applies both rules to the
$\code{if}$-statements of the RCU $\code{inc()}$ loop on a named thread.

\paragraph{Atomic read-modify-write (RMW) instructions.} RMW instructions such
as \fadd~and \cas~are interpreted in terms of a sequence of events. The
\fadd-instruction is interpreted in terms of a read and a write action. The
actions of a conditional RMW instruction such as \cas~depend on the outcome of
the condition. The failing branch of the \cas-instruction yields a read action,
the successful branch a read action followed by a write action, as given by
Figures~\ref{fig:command:cas-success} and~\ref{fig:command:cas-failure}.
Between the two sits a branching action carrying the outcome of the test --
$\alpha=\evalreg{\expr_s}{\rho}$ for success, its negation for failure -- which
the branch rule of Rule~\ref{fig:action:fence-branch}~(branch) adds to
$\sigma.\varphi$, as the rules for $\code{if}$-statements do for theirs. Both
\cas{} rules are therefore tried, and an outcome inconsistent with
$\sigma.\varphi$ is pruned. The test is decided after the read action, so
against a constraint that already equates $\alpha$ with the value of the write
that was read.

\begin{example}\label{ex:opsem-rcu}

  We demonstrate the operational semantics of \code{while}-loops on the example
  of \code{inc()} in Figure~\ref{fig:rcu-inc-code}.

  From $\overline\prog(t)$ each application of a future stepping rule selects a
  pair $l\colon c$ of label $l$ and instruction $c$. By convention $l$ is
  a composite of $\pc$ and $\iter$. For simplicity we denote $\iter$ by the
  iteration of the \code{while} loop in \code{inc()} only, and use the line
  numbers from Appendix~\ref{app:rcu} for $\pc$.

  Then consider a history $H_0$ including all events up to the write event in
  $I_{9}$. The register state $\rho_0$ at this point will have symbolic values
  for memory locations $\alpha_n$ and $\alpha$ for $n$ and $s$, respectively.
  \[ H_0=\Downclosed{\left((I^W_7,0),\Writes_\rel~C~\alpha\right)} \cup
  \left\{((I_8,0),\Reads~s~\beta),((I_{9},0),\Writes~n~\beta+1)\right\} \]
  We have chosen $H_0$ deliberately to include the events at $I_8$ and $I_9$, so
  that both outcomes, success and failure, are possible for \cas,
  i.e.~$(I_{10},\overline a_\text{success})\in^*\horizon\Phi_{H_0}$ with
  $\overline a_\text{success}=\{a_r,a_b,a_w\}$ and $(I_{10},\overline
  a_\text{failure})\in^*\horizon\Phi_{H_0}$ where $\overline
  a_\text{failure}=\{a_r,a_b\}$. The branching action $a_b$ is skipped by
  $\in^*$ and recorded in $\sigma.\varphi$ rather than in the history. At the
  history $H_1=H_0\cup\{(I_{10},a_r)\}$ the posterior future set
  $\Phi_{H_1}$ contains a future of a successful execution, in whose horizon the
  write $(I^W_{10},a_w)$ is a next enabled action, and a future of a failing
  execution, in whose horizon it is not. Both cases are covered by
  Figures~\ref{fig:command:cas-success} and~\ref{fig:command:cas-failure} as
  follows.

  The read action $a_r=\Reads~C~\alpha'$ in \cas~reads the latest memory
  location of $C$ as $\alpha'$, the branching action
  $a_b=\Branches~(\alpha=\alpha')$ compares the values of $s$ and $C$, and the
  write event $a_w=\Writes~C~\alpha_n$ swaps the location of $C$ for $n$ if the
  comparison $\alpha=\alpha'$ succeeds. The branching action adds
  $\alpha=\alpha'$ to $\varphi$ in the success case and $\neg(\alpha=\alpha')$
  in the failure case, each admitted if satisfiable with $\varphi_0$, the
  constraint after $H_0$ extended by the read action $\Reads~C~\alpha'$ of
  Rule~\ref{fig:action:read}~(read). The register state is updated by the
  \cas~result $r$, i.e.~$\rho_1=\funup{\rho_0}{r}{\top}$ in the success case and
  $\rho_1=\funup{\rho_0}{r}{\bot}$ in the failure case.

  The next command in $\overline\prog(t)$ is the \code{if}-statement around the
  quiescent period at the end of the loop body. Branching events are explicitly
  excluded from executions, and thus from the next-enabled actions in
  $\horizon\Phi$. If the rcu-exit command \code{rcu[tid]:=0} is selected,
  Rule~\ref{fig:step:branch-then}~(then) applies, and evaluates the rcu-exit and
  rcu-enter, \code{rcu[tid]=1} immediately after using
  Rule~\ref{fig:command:write-var}~(write-var) via
  Rule~\ref{fig:step:non-lb_non-branch}~(non-lb/non-branch).

  The next command is the \code{if}-statement obtained from unrolling the
  \code{while}-loop into nested \code{if}-statements. In the success case, by
  Rule~\ref{fig:step:non-lb_non-branch}~(non-lb/non-branch) via
  Rule~\ref{fig:step:branch-else}~(else), $H_2=H_1\cup\{(I_{10},a_w)\}$ the next
  command enabled in $\horizon\Phi_{H_2}$ is $(I_{11},\code{rcu[tid]:=0})$. The
  updated register state contains memory locations of $s$ and $n$, and the
  \cas~result $r=\top$.

  In the failure case where Rule~\ref{fig:step:branch-then}~(then) applies, the
  next command is the \fadd-instruction at label $(I_7,1)$. The label has a
  higher $\iter$ number than all events in the history at that point triggering
  Rule~\ref{fig:step:lb}~(lb) for future stepping at loop boundaries.
  Rule~\ref{fig:step:lb}~(lb) resets the timestamped writes in
  $\sigma.\writeset$ bar the anchors, constraint $\sigma.\varphi$, and read-from
  relations $\sigma.\RF$ in the program state $\sigma$, and register state
  $\rho$. The command \code{s:=}\fadd{}$^{\text{rel,acq}}$\code{(\&C,0)} is
  interpreted as two actions $\overline
  a=\Reads_\acq~C~\alpha,\Writes_\rel~C~\alpha$. The new history $H'$ then only
  contains events up to $\overline a$ at $(I_7,0)$ in the first iteration of the
  loop.

  \begin{figure}[htbp]
    \centering
    \scalebox{0.79}{%
      \begin{minipage}{1.265\textwidth}%
    \begin{prooftree}
      \AxiomC{($I_7\colon \code{s:=}\fadd^{\text{rel,acq}}\code{(\&C,0)})\in\overline\prog(t)$}
        \noLine
        \UnaryInfC{($I_7\colon\overline
        a)\in^*\horizon\Phi_{\{\ldots,(I_6:\Writes~\code{rcu[tid]}~1)\}}$}
      \AxiomC{$\sigma\overset{\overline a}{\rightsquigarrow}_t\sigma'$}
        \noLine
        \UnaryInfC{$\sigma'.\RF=\{\}$}
        \noLine
      \AxiomC{$\sigma'.\writeset=\sigma_\text{init}.\writeset\cup\sigma'.\anchors_t$}
        \noLine
        \UnaryInfC{$\sigma'.\varphi=\top$}
        \noLine
      \BinaryInfC{$\sigma\vdash(I_7\colon \code{s:=}\fadd^{\text{rel,acq}}\code{(\&C, 0)},\rho){\overset{\overline a}{\longrightarrow}}_t(\skipcmd,\rho')$}
        \AxiomC{$\rho'=\{n\mapsto\alpha_n\}$}
        \noLine
        \UnaryInfC{$H'=\{\ldots,\overline a\}
        $}
      \RightLabel{lb}
      \TrinaryInfC{$
        \overline\prog\vdash
        \left(\sigma,\rho,H\right)
        \rightarrow
        \left(\sigma',\rho',H'\right)
      $}
    \end{prooftree}
      \end{minipage}%
    }

    \caption{Future stepping at loop boundary in RCU}
    \vspace{0.3cm}
  \end{figure}
\end{example}

\subsection{Well-definedness and Finite Bound on Operational
Semantics}

\paragraph{Well-definedness.} The future stepping rules in the operational
semantics are contextual in the future set $\Phi$, calculated from executions in
the event structure semantics. For programs with unbounded loops, the future set
$\Phi$ is infinite. As shown in the previous section, the posterior future
horizons in $\Phi$ are symmetric between iterations in episodic loops: by
Corollary~\ref{c:gamma-history-reset} the horizons of the later iterations are
images of those of the first under $\gamma$. It therefore suffices to consider
the unravelling of the event structure up to step-counter 2, which includes, at
the first iteration, both a failing and a succeeding attempt at the fallible
operation of the retry loop. On that bounded unravelling the future stepping
rules are well defined. The bound agrees with the resets at loop boundaries in
 Rule~\ref{fig:step:lb}~(lb), which return the operational semantics to the first
iteration rather than carrying a third.

\paragraph{Finite bound.} Resetting histories, register states, and program
states at the loop boundaries makes the operational semantics finite, as the
following theorem shows. The formal proof is in Appendix~\ref{app:opsem-finite}.

\begin{restatable}{theorem}{opsemfinite}\label{t:opsem-finite}

  For any program $\prog$ with only episodic loops and associated future set
  $\Phi$, there are only finitely many configurations reachable from the initial
  configuration $\left(\sigma_0, \rho_0, H_0\right)$ in the operational
  semantics, up to order-isomorphism of timestamps.

\end{restatable}

\paragraph{Timestamps up to order-isomorphism.}
Timestamps are drawn from $\tstamps$, so Rule~\ref{fig:action:write}~(write)
has infinitely many choices of a fresh $q'$ at each step and the reachable
configurations are literally infinite in number.
The assignments of timestamps for a memory location yield a ranking of writes
by comparing timestamps by $\leq$. Comparing timestamp assignments by their
ranked form yields an equivalence relation. Timestamp assignments then fall
into equivalence classes.
These equivalence classes are finitely many because the writes and the viewfront
entries are, which follows from episodicity: the resets of
 Rule~\ref{fig:step:lb}~(lb) leave one iteration's writes at a time, and the
anchors they retain are denied to the thread that wrote them, so no read reaches
an iteration that thread has closed.
Ranking is stable under the rules because they read timestamps only through
$\leq$ -- in $\OW$, in $\freshts$, and in $\viewcomb$ -- and $\freshts(x,q,q')$
asks only for a point strictly between two adjacent timestamps at $x$, which the
density of $\tstamps$ supplies, so a step of one configuration is replayed in
any configuration ranked like it (Lemma~\ref{l:opsem-ts-congruence}).
What remains is a count of ranked states, which the finiteness of the writes,
register states, value restrictions and histories supplies; the timestamps
contribute the orbits of the order-preserving bijections of $\tstamps$ on tuples
of a given length, the Fubini number many, the group being oligomorphic
(Lemma~\ref{l:opsem-ts-count}).
This is the region construction of timed
automata~\cite{AlurDill1994TimedAutomata}, and orbit-finiteness in the sense of
nominal sets~\cite{BojanczykKlinLasota2014NominalAutomata}.
The theorem bounds the state space and not the length of derivations, the
semantics having no final configuration and imposing no fairness condition.

\subsection{Collapsing Cross-Thread Interleavings with Symbolic
Read-Froms}

The assignment of read-from edges can multiply states across threads. The
following example illustrates this for \fadd~and \cas~in RCU as in
Figure~\ref{fig:rcu-inc-code}. For the verification of RCU, it only matters
whether the \cas{} is reading from \fadd{} in the same iteration, leading to
success, or not, leading to failure. The distinction of the two cases is
symbolic, uniform across all iterations and threads, and verified outside of the
operational semantics.

\begin{example}[$\rf$-relation across threads]

  Consider two threads $t_1$ and $t_2$ performing a value increment on a shared
  counter $C$ using RCU, focusing on \fadd{} and \cas{} only. Let $C$ be
  initialised to the symbolic memory location $\beta_0$. Threads $t_1$ and $t_2$
  have pointers $n$ with symbolic memory locations $\beta_1$ and $\beta_2$
  respectively. By global assumptions from memory allocation
  $\beta_1\neq\beta_2$ and both are different from the initial value $\beta_0$
  of $C$.

  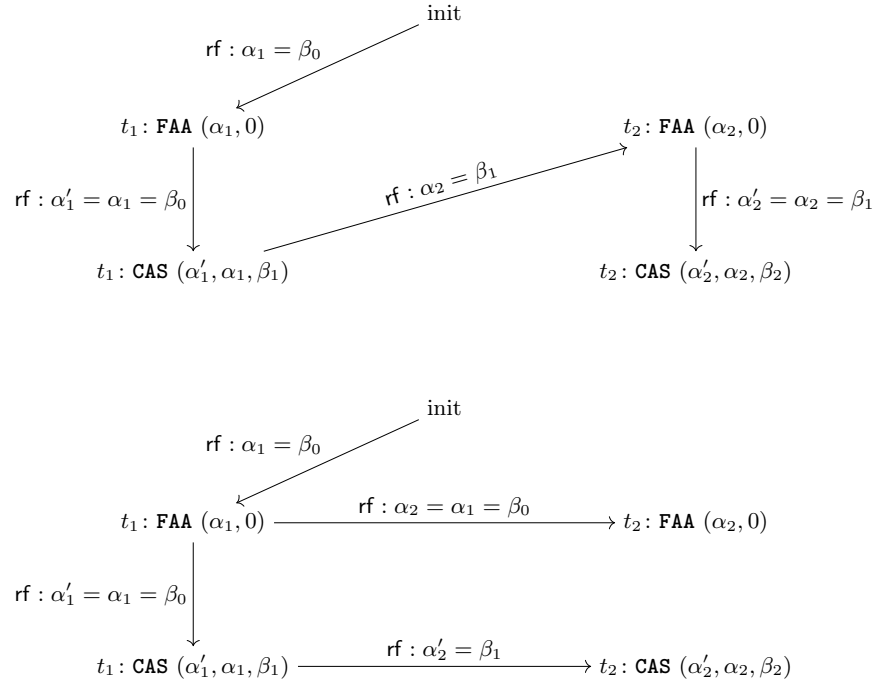
\begin{figure}[htbp]
    \centering
    \scalebox{0.95}{\begin{tikzpicture}[node distance=3cm, auto,
    lbl/.style={font=\footnotesize, inner sep=2pt}]

  \node (init-s) at (3.5,8.6) {$\text{init}$};
  \node (faa-s-1) at (0,7) {$t_1\colon\fadd~(\alpha_1, 0)$};
  \node (cas-s-1) at (0,5) {$t_1\colon\cas~(\alpha'_1, \alpha_1, \beta_1)$};
  \node (faa-s-2) at (7,7) {$t_2\colon\fadd~(\alpha_2, 0)$};
  \node (cas-s-2) at (7,5) {$t_2\colon\cas~(\alpha'_2, \alpha_2, \beta_2)$};

  \draw[->] (init-s) to node[lbl,above left] {$\rf: \alpha_1=\beta_0$} (faa-s-1);
  \draw[->] (faa-s-1) to node[lbl,left] {$\rf: \alpha'_1=\alpha_1=\beta_0$} (cas-s-1);
  \draw[->] (cas-s-1) to node[lbl,above,sloped,pos=0.5] {$\rf: \alpha_2=\beta_1$} (faa-s-2);
  \draw[->] (faa-s-2) to node[lbl,right] {$\rf: \alpha'_2=\alpha_2=\beta_1$} (cas-s-2);

  \node (init-f) at (3.5,3.1) {$\text{init}$};
  \node (faa-f-1) at (0,1.5) {$t_1\colon\fadd~(\alpha_1, 0)$};
  \node (cas-f-1) at (0,-0.5) {$t_1\colon\cas~(\alpha'_1, \alpha_1, \beta_1)$};
  \node (faa-f-2) at (7,1.5) {$t_2\colon\fadd~(\alpha_2, 0)$};
  \node (cas-f-2) at (7,-0.5) {$t_2\colon\cas~(\alpha'_2, \alpha_2, \beta_2)$};

  \draw[->] (init-f) to node[lbl,above left] {$\rf: \alpha_1=\beta_0$} (faa-f-1);
  \draw[->] (faa-f-1) to node[lbl,above] {$\rf: \alpha_2=\alpha_1=\beta_0$} (faa-f-2);
  \draw[->] (faa-f-1) to node[lbl,left] {$\rf: \alpha'_1=\alpha_1=\beta_0$} (cas-f-1);
  \draw[->] (cas-f-1) to node[lbl,above] {$\rf: \alpha'_2=\beta_1$} (cas-f-2);
\end{tikzpicture}}
    \caption{Read-from combinations and \cas{} outcomes}\label{f:rf-mult-example}
  \end{figure}

  Suppose $t_1$ performs the \fadd{} first, reading the initial value of $C$ as
  $\alpha_1$. Then there are two possible outcomes, shown in
  Figure~\ref{f:rf-mult-example}:

  \begin{enumerate}

    \item (top) $t_1$ performs the \cas{} first, which reads from $t_1$'s
      \fadd{}. The operational semantics adds $\alpha'_1=\alpha_1$ to $\varphi$
      during the read rule. Relative to $\varphi$, the \cas{} then
      succeeds, writing $\beta_1$ to $C$. The write during $t_1\colon\cas$ is
      then the only visible write to $C$, and the \fadd{} of $t_2$ will then
      read the result $\beta_1$ of the swap in $t_1\colon\cas$. The read-from
      edge adds $\alpha_2=\beta_1$ to $\varphi$. $t_2$ then performs the \cas{},
      which reads from $t_2$'s \fadd{}. The read action adds
      $\alpha'_2=\alpha_2$ to $\varphi$. Relative to $\varphi$ the \cas{} then
      succeeds.

    \item (bottom) $t_1\colon\fadd$ and $t_2\colon\fadd$ both read the initial
      value of $C$, so that $\alpha_2=\alpha_1=\beta_0$. $t_1$ performs the
      \cas{} first. The \cas{} reads from $t_1\colon\fadd$, which adds
      $\alpha'_1=\alpha_1$ to $\varphi$. Relative to $\varphi$, the \cas{}
      succeeds, writing $\beta_1$ to $C$. When $t_2$ reads $C$ during \cas{}, it
      will read $\beta_1$ from $C$, adding $\alpha'_2=\beta_1$ to $\varphi$. As
      $\beta_1$ and $\beta_0$ are necessarily distinct -- subsequently allocated
      memory locations without intermediate deallocation -- the comparison of
      $\alpha'_2$ against $\alpha_2$ fails, and the \cas{} of $t_2$ fails.

  \end{enumerate}

  Note that $t_2$ cannot distinguish if it reads from \code{init} or
  $t_1\colon\fadd$: as a read-don't-modify-write, $t_1\colon\fadd$ writes back
  the value $\beta_0$ it read from \code{init}. The assignment of $\rf$ then
  leads to two incompatible predicates $\varphi$ which correspond to the two
  incompatible outcomes of \cas{}.

\end{example}

For verifying the use-after-free property, this symbolic treatment replaces an
enumeration of the cross-thread interleavings of read-from assignments with the
finitely many symbolic outcomes that affect the property -- here the two
predicates $\varphi$ above. Since episodicity makes the reasoning for a single
iteration stand for all iterations, verifying the property then requires only a
single pass over the program, and is thus linear in the program size. This is
the reduction claimed in the abstract; it is specific to safety properties such
as the use-after-free bug and not a general linear bound, as the number of
interleavings remains exponential in the number of threads in the worst case.

\subsection{Correspondence with Event Structure
Semantics}\label{s:opsem-sound-complete}

We relate the operational semantics and the event structure semantics with
soundness and completeness results. The full proofs are in
Appendix~\ref{app:opsem-sound-complete}.

The future stepping rules in the operational semantics are defined along the
next enabled actions.
Our proof refines the proof of Wright et al.~\cite{Wright2023OpSem}. Their
operational semantics keeps track of all subsequent actions following the point
of execution, a choice that would lead to infinite derivations for loops.
By Lemma~\ref{l:post-futures-narrowed}, the next enabled actions of the body of
an episodic loop are consistent between loop iterations, so we can track them
for one iteration only, and reset the configuration at the loop boundary,
keeping derivations finite.

\begin{restatable}[Completeness]{theorem}{opsemcompleteness}\label{t:opsem-complete}

  For every complete execution in a program where unbounded loops are episodic
  there is a corresponding derivation in the operational semantics following the
  actions of the execution.

\end{restatable}

\begin{restatable}[Soundness]{theorem}{opsemsoundness}\label{t:opsem-sound}

  For every derivation in the operational semantics there is a corresponding
  execution in the event-structure semantics such that the derivation follows
  the actions of the execution.

\end{restatable}

The remaining crucial point of the proofs is to show that the resets at
boundaries of loop iterations in Rule~\ref{fig:step:lb}~(lb) accurately reflect the
symmetry between next enabled actions in successive loop iterations. This
follows from (1) the resets at boundaries of loop iterations in
 Rule~\ref{fig:step:lb}~(lb) are compatible with executions, and (2)
read-from relations establish $\RF$-pairs in the operational semantics in
 Rule~\ref{fig:action:read}~(read).

\emph{(1) Boundaries of loop iterations.} In order to prove the resets in
 Rule~\ref{fig:step:lb}~(lb) correct, we need to show that they do not restrict
configurations in a way that breaks compatibility with the event structure
semantics. Therefore, we need to show that the resets in the rule either
subsume episodicity conditions in Definition~\ref{def:episodic} of episodic
loops or correspond to $\gamma$, which identifies states across loop iterations
in the event structure semantics.

\emph{(2) Read-from relations.} The read-from relation $\rf$ in complete
executions assigns a visible write to every read.
Rule~\ref{fig:action:read}~(read) selects from visible writes in the
configuration.

The future stepping
Rules~\cref{fig:step:non-lb_non-branch,fig:step:lb,fig:step:branch-then,fig:step:branch-else}
proceed along next enabled actions. Following the next enabled actions in a
derivation constructs a history which is consistent with the dependency
relations under the branching decisions and assignments of visible writes to
read actions. Exhaustively following the next enabled actions constructs a
maximal such set, which is an execution in the event structure.

\subsection{Owicki-Gries Logic, Ownership, and RCU Verification}\label{s:opsem-og}

We can now verify the fix for the UAF bug in our operational semantics as a
safety property. Concretely, we verify the invariant that the deallocating
thread has exclusive ownership over the shared memory location at the time of
deallocation, given that the RCU flag together with the memory-order annotations
on memory operations in the critical section suffice to correctly transfer
ownership from the competing thread to the deallocating thread.

Figure~\ref{fig:use-after-free-bug-overview} presents the use-after-free bug.
The \cas{} in Thread~1 reads from the \fadd{} in either thread and succeeds with
a read and write action. The \cas{} in Thread~2 then reads from the write of
Thread~1's \cas{} and fails, and executes only a read action. The bug then
occurs after the following history:
\[ H_0= \left\{ \ldots, {t_1\colon\fadd^{\rel,\acq}_{R,W}},
{t_2\colon\fadd^{\rel,\acq}_{R,W}}, \ldots, {t_1\colon\cas^{\rel,\acq}_{R,B,W}},
{t_2\colon\cas^\acq_{R,B}} \right\} \]

Without release annotation on the RCU exit $t_2\colon\code{rcu[}t_2\code{]:=0}$,
the RCU exit is enabled after $H_0$, that is
$(I_5\colon t_2\colon\code{rcu[}t_2\code{]:=0})\in\horizon\Phi^\text{bug}_{H_0}$.
Thread~1 reading $\code{rcu[}t_2\code{]}$ from the RCU exit in Thread~2 then
enables the deallocation action in \code{free(s)}, which allows deallocation to be ordered
before dereferencing leading to use-after-free.

With the release annotation on the RCU exit in the fix, the RCU exit is not
enabled after $H_0$,
$(t_2\colon\code{rcu[}t_2\code{]:=0})\not\in\horizon\Phi^\text{fix}_{H_0}$
and becomes only enabled after a history including the dereferencing
instruction, e.g.~$H_0\cup\{t_2\colon\code{v:=*s}\}$. Thread~1 can only read
from the RCU exit $t_2\colon\code{rcu[}t_2\code{]:=0}$ after a history which has $s$
dereferenced, preventing the bug.
The two steps of that argument come from different places, and it is worth
separating them. That the exit follows the dereference is an ordering
\emph{within} Thread~2, and it is $\Phi$ that supplies it: $\ppo$ preserves the
edge into a releasing write, so the future of the fixed program orders the two
and the future of the buggy program does not. That Thread~1 sees the exit at
all, on the other hand, is not in $\Phi$ -- no future carries an edge between
threads. It is the program state that supplies it: the exit is in
$\sigma.\OW(t_1,\code{rcu[}t_2\code{]})$ only once Thread~2 has performed it,
and reading it with an acquiring read folds Thread~2's viewfront into
Thread~1's, so everything Thread~2 had done before releasing -- the dereference
among it -- is thereafter observed by Thread~1. The release annotation is what
makes that fold happen; a relaxed exit would still be read, and still leave
Thread~1's view of $s$ stale.

\medskip \paragraph{UAF bug freedom in RCU as an ownership-based safety
property.} Semenyuk et al.~\cite{Semenyuk23Ownership} introduced a canonical
format for Owicki-Gries local and global conditions based on ownership
predicates, describing which actor currently holds ownership over a resource,
and how ownership is transferred between actors. Applied to RCU, threads $t_1$
and $t_2$ are actors, and provably distinct memory locations are independent
resources, such as $n$ and $C$ before the \cas{}. Ownership constraints which
establish local correctness and global non-interference are chosen manually, but
arise naturally from the command semantics, as in the following examples:
\begin{figure}[htbp]
  \scalebox{0.8}{
  \begin{subfigure}{0.6125\textwidth}
  \begin{prooftree}
    \AxiomC{}
    \RightLabel{\shortstack{own:\\free}}
    \UnaryInfC{$
      \{\forall t'.t'\neq t_1\implies t'\not\in\text{own}(s)\}
      t_1\colon\code{free}(s)
      \{\}
    $}
  \end{prooftree}
\end{subfigure}
  \begin{subfigure}{0.6125\textwidth}
  \begin{prooftree}
    \AxiomC{}
    \RightLabel{\shortstack{own:\\deref}}
    \UnaryInfC{$
      \{t_2\in\text{own}(s)\}
      t_2\colon\semglobld{}{\semderef s}
      \{\}
    $}
  \end{prooftree}
  \end{subfigure}
}
\end{figure}

Ownership transitioning ensures that the dereferencing does not occur after the
$\code{free}$. A thread can only assume exclusive ownership of a resource after
any other thread has relinquished ownership. From individual commands it is not
evident when threads assume or relinquish ownership, as ownership is subject to
global visibility and depends on the program context. As subsequently allocated
memory locations without intermediate deallocation must be distinct, we can
assume that threads instantly assume ownership of newly allocated $n$. Similarly
threads have read ownership after reading the shared data structure $C$ in the
\fadd{} instruction. A thread cannot assume it has ownership of the reclaimable
previous location of $C$ after \cas, as other threads may still dereference the
old memory location as in the UAF bug example. We choose the RCU exit after the
failed \cas~and the $\code{free}$ instruction after the successful \cas~as the
points where the respective threads relinquish ownership over the memory
location $C$.
\begin{figure}[htbp]
  \scalebox{0.8}{
  \begin{subfigure}{0.6125\textwidth}
  \begin{prooftree}
    \AxiomC{}
    \RightLabel{rcu-exit}
    \UnaryInfC{$
      \{\}
      t_2\colon\code{rcu[}t_2\code{]:=}^\text{rel}\code{0}
      \{t_2\not\in\text{own}(s)\}
    $}
  \end{prooftree}
\end{subfigure}
  \begin{subfigure}{0.6125\textwidth}
  \begin{prooftree}
    \AxiomC{}
    \RightLabel{free}
    \UnaryInfC{$
      \{t_1\in\text{own}(s)\}
      t_1\colon\code{free(s)}
      \{t_1\not\in\text{own}(s)\}
    $}
  \end{prooftree}
\end{subfigure}
  }
\end{figure}

The UAF occurs if Thread~1 frees the memory location $s$ by executing
{\color{blue} $t_1\colon\code{free(s)}$}, while Thread~2 has not yet
relinquished ownership over $s$, as it has yet to dereference $s$ by executing
{\color{blue} $t_2\colon\code{v:=*s}$}. Thread 2 assumes ownership of $s$ in the
post-condition of the fetch-and-add instruction, that is: $\{\}~t_2\colon
\code{s:=}\fadd^{\text{rel,acq}}\code{(\&C,0)}~\{t_2\in\text{own}(s)\}$. Thread~1 is only safe to free $s$
if it meets the precondition $t_2\not\in\text{own}(s)$, that is:
$\{t_2\not\in\text{own}(s)\}~t_1\colon\code{free(s)}~\{\}$. A derivation in the
calculus follows future stepping and thus adheres to the dependencies between
events. The Hoare triples below show the sequencing of ownership transfer in the
fixed program.

\medskip
\noindent\scalebox{0.8}{%
  \begin{minipage}{1.25\textwidth}
\[
  \begin{array}{lll}
    \text{pre} & \text{action} & \text{post} \\
    \{\} & t_2\colon \code{s:=}\fadd^{\text{rel,acq}}\code{(\&C,0)} & \{t_2\in\text{own}(s)\} \\
    \{t_2\in\text{own}(s)\} & t_2\colon \code{v:=*s} & \{\} \\
    \{\} & t_2\colon\code{rcu[}t_2\code{]:=}^\text{rel}\code{0} & \{t_2\not\in\text{own}(s)\} \\
    \{\} & t_1\colon\code{sync()} & \{t_2\not\in\text{own}(s)\} \\
    \{t_2\not\in\text{own}(s)\} & t_1\colon\code{free(s)} & \{\} \\
  \end{array}
\]
  \end{minipage}%
}
\medskip

In the fixed variant of RCU, this pattern of ownership transfer is sufficient
to rule out the use-after-free.
The run below is the \emph{buggy} sequencing, in which the RCU exit of
Thread~2, $t_2\colon$\code{rcu[}$t_2$\code{]:=0}, precedes the dereference
$\code{v:=*s}$. It is a run and not a pair of triples: each step takes an action
from the horizon of the history reached so far, as future stepping does, and
carries the ownership assertion along with it. The exit carries no release here
-- that is what the bug is -- so $\ppo$ leaves the dereference unordered against
it and the horizon offers the exit first.

\begin{figure}[htbp]
  \centering
  \scalebox{0.8}{%
    \begin{minipage}{1.25\textwidth}%
    \centering
    \[
      \begin{array}{l@{\qquad}l@{~}l@{~}l}
        \text{action enabled in} & \text{pre} & \text{action} &
          \text{post}\\[3pt]
        \horizon\Phi^\text{bug}_{H_0} & \{\} &
          t_2\colon\code{rcu[}t_2\code{]:=$^\text{rel}$0} &
          \{t_2\not\in\text{own}(s)\}\\[2pt]
        \horizon\Phi^\text{bug}_{H_1} & \{t_2\not\in\text{own}(s)\} &
          t_1\colon\code{sync()} & \{t_2\not\in\text{own}(s)\}\\[2pt]
        \horizon\Phi^\text{bug}_{H_2} & \{t_2\not\in\text{own}(s)\} &
          t_1\colon\code{free(s)} & \{\}\\[2pt]
        \horizon\Phi^\text{bug}_{H_3} &
          \{{\color{red}t_2\in\text{own}(s)}\} & t_2\colon\code{v:=*s} &
          \mbox{\color{red}\lightning}\\
      \end{array}
    \]
    \medskip
    {\footnotesize Each step extends the history by the action it takes,
     $H_{i+1}=H_i\cup\{(l\colon a)\mid a\in\overline a_i\}$. The
     \lightning{} marks the step whose precondition the run has not
     established.}
\end{minipage}%
  }
\end{figure}

\noindent Every step is enabled: the dereference is still in
$\horizon\Phi^\text{bug}_{H_3}$, because no $\ppo$ edge orders it against a
relaxed exit, so $\Phi$ does not rule this run out. What rules it out is the
ownership assertion the run carries. The exit at the first step leaves
$t_2\not\in\text{own}(s)$, and nothing restores it before the last step, which
requires $t_2\in\text{own}(s)$.

In the fixed program the run cannot start at all: the exit is releasing, so
$(t_2\colon\code{rcu[}t_2\code{]:=}^\text{rel}\code{0})\not\in
\horizon\Phi^\text{fix}_{H_0}$, and the dereference is taken first.

We carry out this verification entirely in the operational semantics, as an
ownership-based Owicki-Gries proof. Because the operational semantics is sound
and complete with respect to the event structure semantics of
Section~\ref{s:essem} (Theorems~\ref{t:opsem-sound} and~\ref{t:opsem-complete}),
the safety guarantee it establishes holds in the event structure semantics as
well.

\subsection{Complexity of the UAF Bug Verification}
$\rf$ bounds executions of programs with episodic loops. In the case of a
2-threaded implementation of the RCU writer with a strong \cas{}, the thread
that succeeds first does not retry, while the other retries at most once, that
is, when the competing thread updates the value in the meantime. By a similar
argument, in a 3-threaded implementation the first thread to succeed does not
retry, the second retries at most once, and the third at most twice. We write each contention scenario as the multiset of the threads'
retry counts, in non-decreasing order; e.g.\ $(0,1,2)$ denotes one thread
succeeding without retrying, one retrying once, and one twice. As successes on
the shared location are serialised, the thread that succeeds $j$-th can have
failed at most $j-1$ times, once against each earlier success, so the realisable
patterns are exactly the non-decreasing sequences $0=c_1\le c_2\le\dots\le c_N$
with $c_j\le j-1$. A single-threaded program thus has one pattern, $(0)$; a
2-threaded program two, $(0,0)$ and $(0,1)$; and a 3-threaded program five,
$(0,0,0)$, $(0,0,1)$, $(0,0,2)$, $(0,1,1)$, and $(0,1,2)$. Their number is the
Catalan number $C_N=\frac{1}{N+1}\binom{2N}{N}$~\cite{stanley2015catalan}, which grows exponentially in
the number $N$ of threads. Moreover, each retry pattern represents a large
number of program traces, themselves exponential in the size of the program. In
our operational semantics, we can verify UAF-freedom in a single thread,
symbolically accounting for the scenario of a successful execution of \cas{},
and one where a competing thread wins -- here it does not matter for the
verification which of the other threads wins. Similarly, each failing iteration
of the incrementing loop is indistinguishable from the others, and only has to
be covered once -- making the verification finite even in programs where the
competing threads increase the shared counter arbitrarily often.

\conditionalpagebreak
\section{Implementation and Evaluation}\label{s:eval}

\paragraph{MoRDor.}
\mordor{}~\cite{kissig2026mordor} is a reference implementation of \smrd{},
developed for this work, with the corrected \cas{} semantics of
Section~\ref{sec:cas}, the finite step-counter semantics of
Section~\ref{s:essem}, use-after-free evaluation, and a checker for the
episodicity criteria of Definition~\ref{def:episodic} of episodic loops. It
confirms the paper's use-after-free example -- both the bug and its fix --
end-to-end on the event-structure semantics, by enumerating the valid symbolic
executions of \smrd{} for a given step-counter. \mordor{} admits allocation and
deallocation events as sources of $\rf$-edges. That is a detection device of the
tool, not a feature of \smrd{}, in which nothing reads from a deallocation: it
surfaces a use-after-free, in the sense of Definition~\ref{def:uaf}, and a read
from uninitialised memory, as an edge the tool can report. The full development
history is publicly available~\cite{kissig2026mordor}.

\paragraph{From the semantics to the tool.}
Each program is given in a custom \code{.lit} language, equivalent to a subset
of C with simplified thread declarations. From a program, \mordor{} builds the
\smrd{} event structure with a finite step counter, following the
event-structure semantics of \smrd{} literally (Section~\ref{s:essem}). It
enumerates the symbolic executions as the maximal conflict-free sets of events,
each augmented with a valid combination of justifications of its writes, again
following the paper's definitions, and computes the dependency relations per
symbolic execution. From the set of all symbolic executions \mordor{} computes
the future sets $\Phi$ that drive the operational semantics of
Section~\ref{s:opsem}.

A use-after-free is the property of Definition~\ref{def:uaf}: an access of a
deallocated location that $\nta$ does not order before the deallocation. Where
the access is a write, \mordor{} tests that directly, with $\nta$ computed as
the transitive closure of $\DP\cup\ppo\cup\rf$. Where it is a read, as in
Section~\ref{sec:bug}, \mordor{} admits the deallocation as a source of $\rf$
and reports the resulting edge. The edge pins the read to the freed location,
and it is what lets the memory model speak: \mordor{} computes $\ppo$ and $\DP$
after \smrd{} whichever model is selected, and the selection is a coherence
filter on executions, so an execution in which a read takes its value from a
deallocation is one RC11z's coherence axiom rejects and \smrd{}'s admits. That
is the difference Table~\ref{tab:eval-uaf} records.

Episodicity is a semantic property defined against the valid executions of the
program. \mordor{} makes use of the inherent modularity~\cite{Paviotti2020MRD}
of \smrd{} in order to check episodicity using the partially calculated
dependencies. The check is one-sided: where it passes the loop is episodic,
where it fails the loop may still be. It is incomplete for two independent
reasons.
First, the relations of Definition~\ref{def:episodic} are semantic. Equivalence
of memory locations, in particular, is the satisfiability of
$\loc(e_1)=\loc(e_2)$ under the constraints of an execution, as in $\ppoalias$
of Definition~\ref{def:ppo}. \mordor{} replaces these with the syntactic
over-approximations of Section~\ref{s:identifying}, as any implementation must.
Second, Case~\ref{episodic-caseb} of Condition~\ref{episodic:mem} quantifies
over every earlier iteration of the loop. Condition~\ref{episodic:events}
quantifies over iterations too, but collapses to a single loop boundary by
transitivity; no such reduction is known for Case~\ref{episodic-caseb}, since
Condition~\ref{episodic:events} orders iterations by ${(\ppoord\cup\DP)}^+$ while
Case~\ref{episodic-caseb} forbids reachability in ${(\DP\cup\rf)}^+$, and the
former does not bound the latter. A check on a bounded unravelling must
therefore approximate, and \mordor{} does so in the restriction to a loop's own
dependencies and in the treatment of writes in a last iteration, both of which
err towards reporting a violation.
The semantic loop boundary drawn by $\iter$ need not align with the syntactic
loop, so \mordor{} accommodates \emph{floating loop boundaries}: it tries every
bisection of the loop body compatible with the program's syntactic structure,
and reports a loop episodic if some bisection satisfies the conditions.

\paragraph{Results.}
We report two use-after-free results (Table~\ref{tab:eval-uaf}) and two
episodicity results (Table~\ref{tab:eval-episodic}).

\smallskip\noindent\emph{Use-after-free under \smrd{}.} On the minimal RCU
bug/fix programs, \mordor{} finds the use-after-free reachable under
\smrd{}~(\code{uaf-bug.lit}) and excluded after the fix of
Section~\ref{sec:bug}~(\code{uaf-bug-fixed.lit}).

\smallskip\noindent\emph{Use-after-free across memory models.} On the same buggy
program, \mordor{} finds the use-after-free reachable under
\smrd{}~(\code{uaf-bug.lit}) but forbidden under
RC11z~(\code{uaf-bug-rc11.lit}), confirming that the buggy execution is absent
under RC11z (Section~\ref{sec:bug}).

\begin{table}[htbp]
  \vspace{-8pt}
  \centering
  \begin{tabular}{lll}
    \hline
    Program & Memory model & Use-after-free \\
    \hline
    \code{uaf-bug.lit}       & \smrd{} & reachable \\
    \code{uaf-bug-fixed.lit} & \smrd{} & excluded (fix) \\
    \code{uaf-bug-rc11.lit}  & RC11z   & forbidden \\
    \hline
  \end{tabular}
  \caption{Use-after-free detection. The bug is reachable under \smrd{},
  excluded by the fix, and forbidden under RC11z.}
  \label{tab:eval-uaf}
  \vspace{-10pt}
\end{table}

\smallskip\noindent\emph{Episodicity under \smrd{}.} \mordor{} finds the retry
loops of all four algorithms of Section~\ref{s:identifying} episodic:
RCU~\cite{Gotsman23grace}~(\code{rcu-1.lit}), hazard
pointers~\cite{folly}~(\code{hp-1.lit}),
seqlock~\cite{hemminger2002seqlock}~(\code{seqlock-1.lit}), and
spinlock~\cite[\S8.5]{hennessy1996computer}~(\code{spinlock-1.lit}), all in
\code{programs/episodicity/}. The RCU and hazard pointer programs hold the
increment operation alone; the reclamation that follows the retry loop is
elided, as Section~\ref{s:identifying} describes.

\smallskip\noindent\emph{Non-episodic examples.} \mordor{} also rejects
non-episodic loops. Rejection does not itself imply non-episodicity; the
violations below are confirmed by hand. A \code{for}-loop whose counter passes a value from one
iteration to the next is reported as non-episodic
(Condition~\ref{episodic:reg}), as is a family of small counterexamples that each
isolate the violation of one of the remaining conditions of
Definition~\ref{def:episodic} of episodic loops: a read whose source lies
outside the loop's
permitted writes (Condition~\ref{episodic:mem}), a branch that constrains a value
read before the loop (Condition~\ref{episodic:cond}), and reads that are not
separated across iterations (Condition~\ref{episodic:events}). Because a single
satisfying bisection suffices for episodicity, each counterexample must violate
its condition under \emph{every} compatible loop boundary; otherwise moving the
boundary masks the intended violation and the loop is reported episodic.

\begin{table}[htbp]
  \vspace{-8pt}
  \centering
  \begin{tabular}{ll}
    \hline
    Program & Retry loops episodic \\
    \hline
    RCU (\code{rcu-1.lit})                          & yes \\
    Hazard pointers (\code{hp-1.lit})               & yes \\
    Seqlock (\code{seqlock-1.lit})                  & yes \\
    Spinlock (\code{spinlock-1.lit})                & yes \\
    \hline
  \end{tabular}
  \caption{Episodicity under \smrd{}: the retry loops of the four algorithms are
  episodic. Non-episodic contrast examples are discussed in the text.}
  \label{tab:eval-episodic}
  \vspace{-10pt}
\end{table}

\paragraph{Mechanised use-after-free.}
We have additionally mechanised the operational semantics of
Section~\ref{s:opsem} in Isabelle/HOL~\cite{kissig2026rcuopsem}. On the paper's
minimal bug/fix client we prove, with no \code{sorry}, that the use-after-free
is reachable in the buggy program (\code{uaf\_reachable\_bug}) and excluded in
the fixed program (\code{uaf\_excluded\_fixed}), both over a future set $\Phi$
that the mechanisation defines using the output of \mordor{}.

\paragraph{From the model to the machine.}
The use-after-free above is established in \smrd{}; we also wanted to know
whether it occurs outside it. The composite execution is far too rare to measure
directly, so we broke it into the three reorderings it depends on -- the
dereference \code{v := *s} sinking below the failing \cas{}, below the loop-back
branch, and below the RCU exit -- and put each to a memory model, a compiler and
hardware~\cite{kissig2026uaflitmus}. Under \code{herd7} C11 allows all three and
RC11 forbids all three, so whatever the machine does, the gap between the two is
where the defect lives. The AArch64 architecture's own model permits all three
individually, the composite included~\cite{kissig2026uaflitmus}.
GCC performs two of them on ordinary accesses: the hop over the RCU exit from
\code{-O1} upward, and the hop over the loop-back branch under
\code{-fallow-store-data-races}, which \code{-Ofast} implies. Both are then
witnessed running, in exactly the binaries carrying them -- the RCU exit $13$
times in $10^9$ rounds and the branch $71$ in $10^9$ on x86-64
(Intel Core i7-13700H), and the RCU exit $33$ in $10^9$ on Graviton2
and Graviton3 (Neoverse-N1 and Neoverse-V1). Clang performs neither, and neither
survives the fix of Section~\ref{sec:bug}.
The third, over the failing \cas{}, is performed by no compiler in the matrix
and was never observed, so the chain as a whole was not witnessed. The claim the
paper makes is accordingly one about the model. A proof over RC11z forbids all
three reorderings outright, so it is silent on the defect, and would stay silent
if a compiler took up the third tomorrow -- as it has already taken up two.
Rarity does not take away from the result in this paper. A defect appearing a
few dozen times in $10^9$ iterations is past what testing reaches, so a proof is
the only instrument that covers it, and it has to range over the executions in
which the defect arises.

\conditionalpagebreak
\section{Related Work}

This paper synthesises and extends three prior works.
First, Richards et al.\ introduce \smrd{}~\cite{Richards25SMRD}, a relaxed
concurrency model that accommodates compiler optimisations, more closely
matching the intent of the C++ specification~\cite{CPP11}.
Second, Wright et al.\ provide an Owicki-Gries logic~\cite{Wright2023OpSem} and
operational semantics, built above MRD~\cite{Paviotti2020MRD}, a concrete-valued
precursor to \smrd{}.
Third, Semenyuk et al.\ use an ownership-based proof system to verify a variant
of RCU over RC11z~\cite{Semenyuk2023RCU}.
Synthesising these works, we extend the ownership-based proof system to work
over an Owicki-Gries logic and operational semantics built above \smrd{}. Using
\smrd{}'s symbolic nature, we provide a finite bound on the verification of
programs with unbounded loops that follow the specific code shape of episodic
loops. With this finitary reasoning, we verify the fix of a bug arising from
load-store reordering in a failing iteration of a retry loop -- an execution
that \smrd{} admits but RC11z forbids, so the prior verification over RC11z
remains sound for that model, and is silent on this defect.

\paragraph{Work that relies on an absence of program-order reads-from cycles.}
The Owicki-Gries method is unsound under weak memory.
Lahav and Vafeiadis showed that its non-interference check implicitly assumes
that an interfering thread shares the asserting thread's view of memory, and
repaired it for the release-acquire fragment of C11 by quantifying stability
over every value a thread may read at a non-later
point~\cite{Lahav2015OwickiGries}.
The Owicki-Gries logics we build on descend from that repair, whose soundness
follows from forbidding $\poord\cup\rf$~cycles: the fragment makes every access
releasing or acquiring, so all of program order is preserved, and the load-store
reordering the bug of Section~\ref{sec:bug} relies on cannot arise.
Their verification of RCU is accordingly sound for that fragment and silent on
the defect, as is the verification of Semenyuk et al.\ over
RC11z~\cite{Semenyuk2023RCU}.
Much related work makes the same
assumption~\cite{Tassarotti2015RCU,Kaiser2017Iris,Dalvandi2020RAR,Dalvandi2022Isabelle,Doherty2019Operational,Doko2017FSL}
to forbid thin-air values: all program order is enforced, even when there is no
semantic dependency. This means either enforcing that order by inserting
additional memory fences, or leaving reasoning unsound over C++, where the
standard does not enforce the assumption, but instead appeals to something
weaker~\cite{CPP11}.
Our verification indicates, however, that the ordering these approaches enforce
is stronger than correctness requires: some accesses can be relaxed, avoiding
their accompanying performance cost. The release-acquire fragment is the
sharpest case, enforcing program order at every access whether or not a semantic
dependency needs it.

\paragraph{Alternative thin-air-free models, and their program logics.}
There are now several prospective solutions to the out-of-thin-air
problem~\cite{Chakraborty2019ThinAir,Kang2017Promising,Lee2020Promising2,Jeffrey2019EventStructures,Paviotti2020MRD,PichonPharabod2016}.
Our verification is based on \smrd{}~\cite{Richards25SMRD} because it supports
RCU's C-style dynamic memory use, because it allows $\poord\cup\rf$-cycles and the
optimisations that rely on this, and because justified executions provide
dependency relations that allow us to leverage symmetry between retry loop
iterations to identify a finite representation.

Svendsen et al.\ provide a separation logic built above the Promising
Semantics~\cite{Svendsen2018SepLogic}, which -- like \smrd{} -- permits the
load-store reordering our bug relies on, but lacks first-class dynamic
allocation, so it cannot express the UAF bug (Section~\ref{sec:bug}).

\paragraph{Replacing retry loops with blocking primitives.}
Prior work substitutes retry loops with syntactic blocking constructs. Lahav and
Margalit introduce a \emph{blocking CAS}~\cite{LahavMargalit2019} -- a language
primitive denoting a compare-and-swap busy-wait whose eventual success is
assumed -- so that their robustness analysis disregards the benign stale reads
of the spin loop, yielding a more precise notion of robustness that avoids
inserting unnecessary fences. VSync similarly restricts \emph{await loops} via
bounded-effect and bounded-length principles~\cite{Oberhauser2021VSync}.
Episodicity is more general: it is a semantic property of the failing iterations
themselves, so it covers non-blocking retry loops without re-expressing them as
a blocking primitive and, unlike either of these prior works, it holds over a
model that forbids thin-air cycles.

\paragraph{Retry loops in deployed systems code.}
The work in this paper formalises retry loops around fallible operations like
\cas{}, and is motivated by deployed systems code. We show in
Section~\ref{s:identifying} that four algorithms adhere to this pattern.
Hazard pointers show how far the episodicity criteria reach.
Like RCU, hazard pointers~\cite{HazardPointers} defer memory reclamation and
retry a \cas{} until it succeeds; they are provided by Meta's Folly
library~\cite{folly} and by the Haphazard library in Rust.
The loops in Folly's algorithm appear to be episodic, but they are more complex
than RCU's critical section loop because they are nested, and our approach
applies to them because our indexing scheme enumerates nested loop bodies
appropriately (Appendix~\ref{app:hp}). Furthermore, the \code{reset\_hazptr}
function appears to be the synchronisation point of the critical section, and it
is releasing, which most likely avoids a use-after-free bug similar to the one
described in this work.
We leave the verification of hazard pointers to future work.

\section{Conclusions}

We have examined the correctness of concurrent systems code in the presence of
compiler optimisations that violate program ordering, demonstrating that
reorderings across loop boundaries can introduce subtle bugs invisible to prior
verification approaches. The central contribution is the notion of episodic
loops -- a semantic characterisation of the retry loops ubiquitous in lock-free
systems code, recognised in many practical applications by a sufficient
syntactic condition -- together with a proof that their behaviour admits a
finite representation even under load-store reordering.

We identified this pattern concretely in four widely used algorithms
(Section~\ref{s:identifying}): a variant of Read-Copy-Update, Folly's hazard
pointers, a seqlock, and a basic spinlock.
Applying our framework to the previously verified RCU variant of Gotsman et al.,
we discovered a use-after-free bug arising from reordering across failing \cas{}
iterations -- a bug invisible under RC11z but observable in the presence of
load-store reordering.
We provided and verified a fix via a finitary operational semantics grounded in
Owicki-Gries ownership reasoning.
The fix is a release annotation on the RCU exit, not a fence. It suffices
because the bug needs one direction of reordering only: the dereference must not
sink below the exit, which is what a release forbids, while operations issued
after the exit remain free to move above it. Enforcing program order wholesale,
as the approaches above do, orders both directions and so rules out reorderings
the algorithm never relied upon; the annotation constrains strictly less and is
still enough to exclude the use-after-free.
The \mordor{} tool automates a sufficient
check for episodic loops and the detection of UAF errors. We provided a Python script to match
potential occurrences of the UAF bug in real-world C++ programs.

\paragraph{Future work.}
Building on our Isabelle/HOL mechanisation of the operational semantics and of
the minimal use-after-free bug and its fix~\cite{kissig2026rcuopsem}, we plan to
extend the mechanised development to an end-to-end no-use-after-free theorem for
full variants of RCU, and to formalise additional safety properties -- in
particular, ABA freedom for RCU and hazard pointers.

\vfill
\clearpage
\pagebreak
\iflipics
  \bibliographystyle{plainurl}
\else
  \bibliographystyle{splncs04}
\fi
\bibliography{main}

\vfill

\pagebreak
\appendix
\section{Appendix: Definitions}\label{s:app-defs}\label{app:defs}

The programs we consider are of a subset of the C programming language,
augmented with top-level thread-parallel composition. We use functions as
syntactic sugar and assume functions are implicitly inlined.

\subsection{Expressions in \smrd{}}

\begin{definition}[Expressions in Programs]\label{def:expressions}
  \emph{Expressions} $\expr$ can be arithmetic expressions $\expr_A$
  and boolean expressions $\expr_B$
  \[
    \begin{array}{rcl}
      \expr & ::= & \expr_A\mid\expr_B \\
      \expr_A & ::= & r\mid
      *\expr_A\mid\nat\mid\expr_A+\expr_A\mid
      \expr_A-\expr_A\mid\expr_A*\expr_A\mid
      \expr_A/\expr_A\mid \\
      & & \expr_A\bitand\expr_A\quad\mid\quad
      \expr_A\bitxor\expr_A\quad\mid\quad
      \expr_A\bitor\expr_A \\
      \expr_B & ::= & \expr_A=\expr_A~|~\expr_A\leq
      \expr_A~|~\neg \expr_B~|~\expr_B \wedge
      \expr_B~|~\expr_B \vee \expr_B
    \end{array}
  \]
\end{definition}

Expressions in the event structure semantics of programs differ from the
program expressions of Definition~\ref{def:expressions}: they contain no
registers and no pointer dereferences, which the semantics resolves against the
register state and evaluates to memory locations respectively, but they do
contain the symbols $\alpha$ introduced by read and allocation events.

\begin{definition}[Expressions in Event Structures]\label{def:es-expressions}
  \emph{Expressions} $\expr$ can be arithmetic expressions $\expr_A$
  and boolean expressions $\expr_B$
  \[
    \begin{array}{rcl}
      \expr & ::= & \expr_A \mid \expr_B \\
      \expr_A & ::= & \alpha\mid\nat\mid\Var\mid \expr_A+\expr_A\mid
      \expr_A-\expr_A\mid \expr_A\times \expr_A\mid
      \expr_A/\expr_A\mid \\
      & & \expr_A\bitand\expr_A\quad\mid\quad
      \expr_A\bitxor\expr_A\quad\mid\quad
      \expr_A\bitor\expr_A \\
      \expr_B & ::= &  \bool\mid\expr_A = \expr_A~|~\expr_A \leq
      \expr_A~|~\neg \expr_B~|~\expr_B \wedge
      \expr_B~|~\expr_B \vee \expr_B
    \end{array}
  \]

  A global variable $x\in\Var$ occurs as the location it denotes, and
  $\bool=\{\top,\bot\}$. Both are needed for the register states of
  Definition~\ref{def:gen-es}, which hold a location at
  $\semregst{}{\semamp{x}}$ and a boolean at the two outcomes of $\cas$.

  We denote by $\symbols(\expr)$ the set of symbols in an expression
  $\expr$.
\end{definition}

The two grammars call for two interpretations, which we write alike and
distinguish by their subscript: $\den{-}_\rho$ resolves the registers of a
program expression of Definition~\ref{def:expressions} against a register
state, while $\den{-}_f$ substitutes for the symbols of an expression already
one of Definition~\ref{def:es-expressions}. Only the latter is quantified over
in Definitions~\ref{def:sem-equiv} and~\ref{def:ppo} below.

\begin{definition}[Interpretation of Program Expressions]\label{def:prog-expr-sem}
  The \emph{interpretation of program expressions} $\expr$ of
  Definition~\ref{def:expressions} is defined as a function
  $\den{-}_\rho:\ProgExpressions\to\Expressions$ relative to a
  \emph{register state} $\rho:\Registers\to\Expressions$, homomorphic on the
  operators and with
  \[
    \begin{array}{rclrcl}
      \evalreg{n : \nat}{\rho}&=&n &\qquad
      \evalreg{r : \Registers}{\rho}&=&\rho(r)
    \end{array}
  \]
  Pointer dereferences $*\expr_A$ are not resolved here but by the read and
  write events the semantics of Definition~\ref{def:gen-es} generates for them.
\end{definition}

\begin{definition}[Semantics of Expressions]\label{def:expr-sem}
  The \emph{interpretation of expressions} $\expr$ from
  $\Expressions$ is defined as a function
  $\den{-}_f:\Expressions\to\Expressions$ relative to an environment
  $f:\Symbols\rightharpoonup\Expressions$ mapping \emph{symbols} from
  $\Symbols$ to \emph{expressions} from $\Expressions$, such that
  \[
    \begin{array}{rclrcl}
      \evalenv{v : \Val}{f}&=&v & \\

      \evalenv{\alpha:\Symbols}{f}&=&f(\alpha)~\text{if}~\alpha\in\text{Dom}(f)&
      \evalenv{\alpha:\Symbols}{f}&=&\alpha~\text{if}~\alpha\not\in\text{Dom}(f)\\

      \evalenv{\expr_{A_1}+\expr_{A_2}}{f}&=&\evalenv{\expr_{A_1}}{f} + \evalenv{\expr_{A_2}}{f} &\qquad
      \evalenv{\expr_{A_1}-\expr_{A_2}}{f}&=&\evalenv{\expr_{A_1}}{f} - \evalenv{\expr_{A_2}}{f}\\

      \evalenv{\expr_{A_1}\times \expr_{A_2}}{f}&=&\evalenv{\expr_{A_1}}{f} \times \evalenv{\expr_{A_2}}{f} &
      \evalenv{\expr_{A_1}/\expr_{A_2}}{f}&=&\evalenv{\expr_{A_1}}{f} / \evalenv{\expr_{A_2}}{f}\\

      \evalenv{\expr_{A_1} = \expr_{A_2}}{f}&=&\evalenv{\expr_{A_1}}{f} = \evalenv{\expr_{A_2}}{f} &
      \evalenv{\expr_{A_1}\leq \expr_{A_2}}{f}&=&\evalenv{\expr_{A_1}}{f}\leq\evalenv{\expr_{A_2}}{f}\\

      \evalenv{\neg \expr_B}{f}&=&\neg\evalenv{\expr_B}{f} &&& \\

      \evalenv{\expr_{B_1}\wedge \expr_{B_2}}{f}&=&\evalenv{\expr_{B_1}}{f}\wedge\evalenv{\expr_{B_2}}{f} &
      \evalenv{\expr_{B_1}\vee \expr_{B_2}}{f}&=&\evalenv{\expr_{B_1}}{f}\vee\evalenv{\expr_{B_2}}{f}\\
    \end{array}
  \]

  The bitwise operators act on the binary representation of a natural, which is
  finite, so they are total on $\nat$ and fix no word width.

  The \emph{values} $\Val\subseteq\Expressions$ are the naturals $\nat$, the
  locations $\Var$ and the booleans $\bool$, and, the interpretation being an eager
  partial evaluation~\cite{Richards25SMRD}, we identify a closed expression
  with the value it denotes. Hence $\evalenv{\expr}{f}\in\Val$ exactly when $f$ maps
  every symbol of $\expr$ to a closed expression, and in particular
  $\evalenv{\expr}{f}\in\Val$ for every $f$ defined on $\symbols(\expr)$
  with values in $\Val$.
\end{definition}

\begin{definition}[Semantics of Equivalence]\label{def:sem-equiv}
  \[
    (\expr_1\equiv\expr_2)~\triangleq~\forall f.
    \evalenv{\expr_1}{f},\evalenv{\expr_2}{f}\in\Val
    \Rightarrow
    \evalenv{\expr_1}{f}=\evalenv{\expr_2}{f}
  \]
  \[
    (\expr_1\equiv_P\expr_2)~\triangleq~(P\Rightarrow\expr_1=\expr_2)\equiv\top
  \]
\end{definition}

\subsection{Program Semantics in Event Structures}

\begin{figure}[htbp]
\[
\begin{aligned}
	\prog ::=&~~~ \join[~|~]
  {\codeskip}
  {\semseq{\prog}{\prog}}
  {\sempar{\prog}{\prog}}
  {\semif{\expr_B}{\prog}{\prog}}
  {\semwhile{\expr_B}{\prog} \\ &}
  {\semregst{i}{\expr}}
  {\semglobld[o]{i}{x}}
  {\semglobst[o]{x}{\expr}}
  {\semregst{i}{\semamp{x}}}
	{\semglobld[o]{i}{\semderef{\expr}}}
  {\semglobst[o]{\semderef{\expr_1}}{\expr_2} \\ &}
  {\semfence[o]}
  {\semfadd[o_r][o_w]{i}{x}{\expr}}
	{\semcas[o_r][o_w]{i}{x}{\expr_1}{\expr_2} \\ &}
  {\semmalloc{i}{\expr}}
  {\semfree{i}}
\end{aligned}
\]
  \caption{Program syntax (\smrd{} fragment)}\label{fig:grammar}
\end{figure}

\begin{definition}[Program Syntax]\label{def:prog-syntax}
  The syntax of programs $\prog$ is given in Figure~\ref{fig:grammar} as a
  subset of the C programming language with: $r$ ranging over registers holding
  thread local state; $x$ ranging over global variables; $\expr$ and $\expr_B$
  ranging over expressions and boolean expressions respectively, as given in
  Definition~\ref{def:expressions} of expressions in programs; and $o \in
  \{\rlx, \rel, \acq \}$ of C atomic memory orderings \cite[7.17.3 Order and
  consistency]{iso-c11}.

  Programs are built from the following constructs:

  \begin{itemize}

    \item the empty command $\codeskip$;

    \item sequential composition $\semseq{\prog}{\prog}$;

    \item parallel composition $\sempar{\prog}{\prog}$;

    \item branching $\semif{\expr_B}{\prog}{\prog}$;

    \item unbounded looping $\semwhile{\expr_B}{\prog}$;

    \item register assignment $\semregst{i}{\expr}$ and
      $\semregst{i}{\semamp{x}}$;

    \item atomic memory accesses $\semglobld[o]{i}{x}$,
      $\semglobld[o]{i}{\semderef{\expr}}$, $\semglobst[o]{x}{\expr}$, and
      $\semglobst[o]{\semderef{\expr_1}}{\expr_2}$;

    \item memory fence operations $\semfence[o]$;

    \item atomic read-modify-write operations
      $\semfadd[o_r][o_w]{i}{x}{\expr}$ and\\ $\semcas[o_r][o_w]{i}{x}{\expr_1}{\expr_2}$;

    \item and dynamic memory operations $\semmalloc{i}{\expr}$ and
      $\semfree{i}$.

  \end{itemize}
\end{definition}

\paragraph{Symbols.}
Reads from memory locations can yield arbitrary values, which may be further
constrained by the program. Symbolic MRD abstracts values as \emph{symbols}, and
defines constraints on symbols through value restrictions local to executions
and global guarantees.

The semantics of programs $\prog$ in event structures is defined inductively
over the \emph{atomic set unravelling} $\overline\prog$ of the program $\prog$.

\begin{definition}[Atomic Set
  Unravelling]\label{def:unrolling-prog}\label{def:atomic-set-unravel}
  The atomic set unravelling $\overline\prog$ of a program $\prog$ is defined
  through the following transformation steps.

  \begin{enumerate}

    \item Projecting the program by thread

    \item Unravelling \code{while} loops as nested \code{if}-statements of a
      maximal depth given by a step-counter $n$, as given by
      Equation~\ref{eq:while-sem} of the event structure semantics below

    \item Assigning to each step in $\overline\prog$ a unique label from a set
      $\Labels$

  \end{enumerate}

\end{definition}

Unravelling \code{while}-loops, we lose information about the structure of the
program. The labels preserve some of the information. We define the
syntax-derived functions from labels in the atomic set unravelling:

\begin{itemize}

  \item\label{not:loopfun} $\loopfun:\Labels\rightarrow\PSet(\nat)$ mapping
    control labels to the set of indices of the loops nesting them, and

  \item\label{not:iter} $\iter:\Labels\rightarrow\text{Pfn}(\nat,\nat)$ mapping
    control labels to loop iterations indexed by loop indices -- partial
    functions from $\text{Pfn}(\nat,\nat)$ on $\nat$.

\end{itemize}

Because the atomic set unravelling $\overline\prog$ of a program $\prog$ is
acyclic, the event structure semantics
$\langle\prog\rangle_{n~\rho~\kappa~\varphi}$ of the program can then be defined
inductively over $\overline\prog$, such that each operation is interpreted as
one or more events prefixing the event structure interpreting the tail of
$\overline\prog$ by Definition~\ref{def:gen-es} of the event structure semantics
below.

\begin{itemize}

  \item The register state $\rho$ maps register names to expressions over
    symbolic values.

  \item The continuation $\kappa$ maps a register state and a value restriction
    to an event structure interpreting the tail of the program. $\kappa$ is
    well-defined, as executions are limited by the step-counter to terminating
    executions.

  \item The predicate $\varphi$ accumulates value restrictions over branching
    statements inductively from the start of the program.

\end{itemize}

The resulting event structures are non-confluent tree structures over events
ordered by program order $\po$. In Section~\ref{sec:cas} we have defined the
semantics of a read-modify-write operation in terms of a structure of events.
Atomicity is defined in terms of an additional relation $\pormw$. Branching
introduces conflict between events following alternative outcomes of the
branching condition. Value restrictions $\valres(e)$ are predicates which
accumulate the outcomes of branching decisions up to the events $e$. Events are
in conflict if their value restrictions are incompatible. Because value
restrictions are unique, event structures cannot be confluent, and unbounded
loops cannot be modelled by recursion in event structures.

\begin{definition}[Symbolic Event Structures]\label{def:event-structures}
  A \emph{symbolic event structure} is a tuple
  $\mathbb E~=~(E,\po,\pormw,\valres)$ comprising:
  \begin{itemize}

    \item a set $E$ of events,

    \item a relation $\po~\subseteq E\times E$ denoting \emph{program order},

    \item a ternary relation $\pormw~\subseteq E\times\Expressions\times E$
      modelling \emph{atomicity of read-modify-write operations} such as $\cas$
      and $\fadd$, and

    \item a function $\valres:E\rightharpoonup\Expressions$ mapping events to boolean
      expressions denoting \emph{value restrictions}

  \end{itemize}
  where all three are defined in Definition~\ref{def:gen-es} of the event
  structure semantics below: $\po$ through the prefixing operation of event
  structures, $\pormw$ through the semantics of read-modify-write instructions,
  and $\valres$ through the predicate $\varphi$ accumulating branching
  conditions.
\end{definition}

The interpretation of a C11 program in event structures is defined through the
semantics of commands, sequential and parallel composition as follows.

\begin{definition}[Event Structure Semantics]\label{def:gen-es}
  \noindent\textbf{Semantics of Commands}:

  In the following let $n\in\nat$ be a finite step-counter, $\rho$ a register
  state mapping registers to expressions, $\kappa$ a continuation mapping
  register states and value restrictions to event structures, and $\varphi$
  accumulating value restrictions over branching conditions.

  \begin{center}
  \fitwidth{$\displaystyle
    \begin{array}{rcl}
			\langle \semregst{}{\expr}\rangle_{n~\rho~\kappa~\varphi} & \triangleq & \kappa(\rho[r\mapsto\evalreg{\expr}{\rho}],\varphi) \\
			\langle \semglobld[o]{}{x}\rangle_{n~\rho~\kappa~\varphi} & \triangleq &
      (e\colon~\Reads_o~x~\alpha)[\varphi] \cdot \kappa(\rho[r\mapsto\alpha],\varphi) \\
			\langle \semglobst[o]{x}{\expr}\rangle_{n~\rho~\kappa~\varphi} &
      \triangleq & (e\colon~\Writes_o~x~\evalreg{\expr}{\rho})[\varphi] \cdot \kappa(\rho,\varphi) \\
			\langle \semregst{}{\semamp{x}}\rangle_{n~\rho~\kappa~\varphi} & \triangleq & \kappa(\rho[r\mapsto x],\varphi) \\
			\langle \semglobld[o]{}{\semderef{\expr}}\rangle_{n~\rho~\kappa~\varphi} & \triangleq &
      (e\colon~\Reads_o~\evalreg{\expr}{\rho}~\alpha)[\varphi] \cdot
      \kappa(\rho[r\mapsto\alpha],\varphi) \\
			\langle \semglobst[o]{\semderef{\expr_1}}{\expr_2}\rangle_{n~\rho~\kappa~\varphi} &
      \triangleq &
      (e\colon~\Writes_o~\evalreg{\expr_1}{\rho}~\evalreg{\expr_2}{\rho})[\varphi]
      \cdot \kappa(\rho,\varphi) \\
			\langle\semfadd[o_r][o_w]{}{x}{\expr}\rangle_{n~\rho~\kappa~\varphi} &
      \triangleq &
      \rmw(\left(
        \begin{array}{ll}
          (e_r\colon~\Reads_{o_r}~x~\alpha)[\varphi] & \cdot \\
          (e_w\colon~\Writes_{o_w}~x~(\alpha+\evalreg{\expr}{\rho})[\varphi])[\varphi] & \cdot
          \\
          \kappa(\rho[r\mapsto\alpha+\evalreg{\expr}{\rho}],\varphi) & \\
        \end{array}
      \right)
      ,e_r,e_w,\top)
        \\
			\langle\semcas[o_r][o_w]{}{x}{\expr_1}{\expr_2}\rangle_{n~\rho~\kappa~\varphi}
      & \triangleq &
      \rmw(
      \left(
      \begin{array}{ll}
        (e_r\colon~\Reads_{o_r}~x~\alpha)[\varphi] & \cdot \\
        (e_c:[\alpha=\evalreg{\expr_1}{\rho}])[\varphi]& \cdot \\
        \left(
        \begin{array}{lll}
          & (e_w\colon~\Writes_{o_w}~x~\evalreg{\expr_2}{\rho})[\varphi_\top] &
          \cdot \\
          & \kappa(\rho[r\mapsto\top],\varphi_\top) & \\
          \\
          + & \kappa(\rho[r\mapsto\bot],\varphi_\bot) & \\
        \end{array}
        \right) & \\
      \end{array}
      \right)
      ,e_r,e_w,\alpha=\evalreg{\expr_1}{\rho}) \\
			\langle\semfence[o]\rangle_{n~\rho~\kappa~\varphi} & \triangleq & (e\colon~\Fences_o)[\varphi] \cdot \kappa(\rho,\varphi) \\
			\langle\semmalloc{}{\expr}\rangle_{n~\rho~\kappa~\varphi} &\triangleq& (e\colon~\Allocs~\alpha~\evalreg{\expr}{\rho})[\varphi]\cdot\kappa(\rho[r\mapsto\alpha],\varphi) \\
			\langle\semfree{}\rangle_{n~\rho~\kappa~\varphi} &\triangleq& (e\colon~\Deallocs~\evalreg{r}{\rho})[\varphi]\cdot\kappa(\rho,\varphi) \\
			\langle\semif{b}{\prog_1}{\prog_2}\rangle_{n~\rho~\kappa~\varphi} & \triangleq & (e \colon \evalreg{b}{\rho}) \cdot (\langle\prog_1\rangle_{n~\rho~\kappa~(\varphi\wedge\evalreg{b}{\rho})}) \\
			&&~+~(\langle \prog_2\rangle_{n~\rho~\kappa~(\varphi\wedge\neg\evalreg{b}{\rho})}) \\
  \end{array}
  $}
  \end{center}

  \noindent where $\varphi_\top$ and $\varphi_\bot$ abbreviate the two outcomes
  of the $\cas{}$ test,

  \[
    \varphi_\top~\triangleq~\varphi\wedge(\alpha=\evalreg{\expr_1}{\rho})
    \qquad
    \varphi_\bot~\triangleq~\varphi\wedge\neg(\alpha=\evalreg{\expr_1}{\rho})
  \]

  \noindent and the event structure prefix $\cdot$ is given by

  \[
    e[\varphi_e]\cdot(E,\po,\pormw,\valres)~\triangleq~(\{e\}\cup E,\poord\cup(
    \{e\}\times(\{e\}\cup E)),\pormw,\valres[e\mapsto\varphi_e])
  \]\label{def:es-prefix}

  \noindent and the coproduct $\mathbb E_1+\mathbb E_2$ of event structures
  $\mathbb E_1~=~(E_1,\po_1,\pormw_1,\valres_1)$ and
  $\mathbb E_2~=~(E_2,\po_2,\pormw_2,\valres_2)$ by

  \[
    \mathbb E_1+\mathbb E_2~\triangleq~(E_1\uplus E_2,\po_1\uplus\po_2,
    \pormw_1\uplus\pormw_2,\valres_1\uplus\valres_2)
  \]

  \medskip
  \noindent\textbf{While Loops}: The semantics of while loops is defined through a step-counter as follows.

  \begin{equation}
    \begin{array}{rcl}
      \langle\prog\rangle_{0~\rho~\kappa~\varphi} & \triangleq & \emptyset \\
      \langle\semskip\rangle_{n~\rho~\kappa~\varphi} & \triangleq & \kappa(\rho,\varphi) \\
			\langle\semwhile{b}{\prog}\rangle_{n~\rho~\kappa~\varphi} & \triangleq & \langle\semif{b}{\semseq{\prog}{\semwhile{b}{\prog}}}{\semskip}\rangle_{n-1~\rho~\kappa~\varphi} \\
    \end{array}
    \label{eq:while-sem}
  \end{equation}

  \medskip
  \noindent\textbf{Sequential Composition}:
  \begin{equation}
		\langle\semseq{\prog_1}{\prog_2}\rangle_{n~\rho~\kappa~\varphi}~\triangleq~\langle\prog_1\rangle_{n~\rho~(\lambda\rho\,\varphi.\langle\prog_2\rangle_{n~\rho~\kappa~\varphi})~\varphi}
  \end{equation}

  \medskip\noindent\textbf{Parallel Composition}:
  \begin{equation}
		\langle\sempar{\prog_1}{\prog_2}\rangle_{n~\rho~\kappa~\varphi}~\triangleq~\langle\prog_1\rangle_{n~\rho~\kappa~\varphi}+\langle\prog_2\rangle_{n~\rho~\kappa~\varphi}
  \end{equation}

  \medskip\noindent\textbf{Branching}: The branches of the \code{if} command
  above are combined with the coproduct of event structures. Events in the
  \code{then} branch are in conflict with events in the \code{else} branch,
  through the extension of the value restriction with $\evalreg{b}{\rho}$ and
  $\neg\evalreg{b}{\rho}$, respectively.

\end{definition}

\paragraph{The step-counter per loop.}
The step-counter of Equation~\ref{eq:while-sem} bounds the depth of the
unravelling, and thus ensures that $\kappa$ -- which is constructed from the end
of the program -- is well-defined.
A global step-counter would suffice for that, but it is shared between nesting
levels: a loop nested under $k$ iterations of an enclosing loop is unravelled
with counter $n-k$, so successive iterations carry successively smaller
unravellings of the loop nested within them.
As the proofs in Appendix~\ref{s:app-proofs} compare iterations of the same
loop, we instead read the subscript in Equation~\ref{eq:while-sem} as a map in
$\text{Pfn}(\nat,\nat)$ from loop indices to bounds, as for $\iter$, of which
the \code{while} rule decrements only the component of the loop it unravels, and
write $\langle\prog\rangle_n$ for the uniform choice assigning the bound $n$ to
every loop. The base case is read along with it: unravelling stops where the
component of the loop being unravelled is exhausted, that is
$\langle\semwhile{b}{\prog}\rangle_{n~\rho~\kappa~\varphi}\triangleq\emptyset$
whenever $n(\ell)=0$ for the loop $\ell$ of that \code{while}, and not where
every component is. As under a global step-counter, this drops the executions
that would need more than $n(\ell)$ iterations of $\ell$, while the exits
unravelled at smaller depths remain. Every iteration of a loop then contains the
same unravelling of the loops nested within it.
Executions generated under per-loop step-counters are contained in those
generated under a global step-counter, obtained by multiplying the counters of
nesting loops and taking the maximum over loops in sequence, so $\kappa$ remains
well-defined.
The two readings agree unless loops nest, and so differ on none of the
algorithms considered here except hazard pointers (Appendix~\ref{app:hp}), whose
retry loops are nested.

\paragraph{Ordering of control labels and
symbols.}\label{par:order-clab}\label{par:order-symb}
Without loss of generality, we assume that labels are enumerated from the start
of the program, so that $\Labels$ is ordered by a partial well-founded order.
Furthermore, we assume that control labels are assigned depth-first preferring
the path with an earlier termination of a loop.
By convention each read event and each allocation event introduces a fresh
symbol. We assume an enumeration of symbols by control label, so that symbols in
event structures are enumerated from the start of the program.

\paragraph{Auxiliary event accessor functions.}
We use functions $\loc$, $\val$, and $\cond$. $\loc$ accesses the memory
location in write, read, allocation, and deallocation events; for an allocation
it is the symbol $\alpha$ the event introduces, and for a deallocation the
expression naming the location it frees. $\val$ accesses the value expression:
the value written in a write event, the symbol a read event introduces, and the
size expression in an allocation event. A deallocation carries no value, and
$\val$ is the empty expression there, so that $\symbols(\val(d))=\emptyset$ for
every deallocation $d$; the same holds for fence and branching events. $\cond$
accesses the branching condition in branching events.

The functions $\loopfun$ and $\iter$ above extend naturally to events.
Additionally, we introduce a function $\pc:\Events\to\nat$ from events to
program counters. Program counters differ from lines of code in particular for
composite operations such as RMW operations.

\begin{example}
  Consider the following program
  \begin{minted}{c}
    r:=0;r:=r+1;x:=r;
  \end{minted}

  Starting with step-counter $n=1$, the register state is constructed from the
  start of the program as follows:

  \[
    \begin{array}{rc}
      \langle r:=0;r:=r+1;x:=r\rangle_{1~\mathbf{\{\}}~\lambda\rho\,\varphi.\emptyset~\top}&=\\
      \langle r:=0\rangle_{1~\mathbf{\{\}}~(\lambda\rho\,\varphi.\langle r:=r+1;x:=r\rangle_{1~\rho~\lambda\rho\,\varphi.\emptyset~\varphi})~\top}&=\\
      \langle r:=r+1;x:=r\rangle_{1~\mathbf{\{r\mapsto 0\}}~\lambda\rho\,\varphi.\emptyset~\varphi}&=\\
      \langle r:=r+1\rangle_{1~\mathbf{\{r\mapsto 0\}}~(\lambda\rho\,\varphi.\langle x:=r\rangle_{1~\rho~\lambda\rho\,\varphi.\emptyset~\varphi})~\varphi}&=\\
      \langle x:=r\rangle_{1~\mathbf{\{r\mapsto 1\}}~\lambda\rho\,\varphi.\emptyset~\varphi}&=\\
      (e:\mathcal W~x~1)[\top]\cdot\emptyset
    \end{array}
  \]

  The continuations are constructed cumulatively from the end of the program as
  follows:

  \[
    \begin{array}{l}
      \lambda\rho\,\varphi.\lambda\rho\,\varphi.\lambda\rho\,\varphi.(e\colon
      \mathcal W~x~\evalreg{r}{\rho})[\varphi]\cdot\emptyset~(\{r\mapsto
      1\},\top)(\{r\mapsto 0\},\top)(\{\},\top)\\[2pt]
      \qquad=~(e\colon\mathcal W~x~1)[\top]\cdot\emptyset
    \end{array}
  \]
\end{example}

\paragraph{Atomicity guarantees of read-modify-write operations.}
Atomicity of read-modify-write operations is captured by the ternary relation
$\pormw\subseteq E\times\Expressions\times E$ through the $\rmw$-function below

\[
  \rmw((E,\po,\pormw,\valres),e_r,e_w,b)~\triangleq~(E,\po,\pormw\cup\{(e_r,
  b,e_w)\},\valres)
\]

$\pormw$ does not order the write before the read of an RMW operation in
executions: it contributes to $\ppo$ only via $\ppormw$, and there the
write-to-read pair occurs solely in composition with $\pposync$
(Definition~\ref{def:ppo} of preserved program order). $\ppormw$ extends
$\pposync$ across the RMW, ordering events before the write ahead of the read
and the write ahead of events after the read, thereby preventing memory accesses
to the same location from being ordered between the read and write of the RMW.
How this works is illustrated for \fadd{} in Example~\ref{ex:rcu-sp} in
Section~\ref{s:sync-points}.

Figures~\ref{fig:cas-rmw} and~\ref{fig:faa-rmw} contrast the two directions on
the event structures of \cas{} and of \fadd{} respectively. In both, $\pormw$
runs from the read to the write, and the dashed $\ppormw$ arrow runs from the
write back to the read. The latter is drawn between the two events for
legibility only: it enters $\ppo$ solely in composition with $\pposync$, so it
contributes order to an execution only together with an event $\pposync$-before
the write or $\pposync$-after the read. The two instructions differ in the
condition carried by $\pormw$. For \fadd{} it is $\top$, so $\ppormw[\top]$
holds in every execution. For \cas{} it is the branching condition
$\beta=\alpha$, so $\ppormw$ holds only where that condition is $P$-equivalent
to $\top$, that is on the succeeding branch -- the asymmetry of conditional
read-modify-write operations noted above.

\begin{figure}[htbp]
  \centering
  \begin{minipage}[t]{0.48\textwidth}
    \centering
    \scalebox{0.85}{\begin{tikzpicture}
  \nodeP{casr}{$e_r\colon R^\text{acq}~C~\beta$}{0, 0}
  \nodeP{casb}{$e_c\colon[\beta=\alpha]$}{0, -1.2}
  \nodeP{casw}{$e_w\colon W^\text{rel}~C~\varepsilon$}{-1.4, -2.4}
  \node (case) at (1.4, -2.4) {$\vdots$};
  \edgepo[]{casr}{casb}
  \draw (casb) edge[style=po] node[above left,font=\scriptsize,inner sep=1pt] {if} (casw);
  \draw (casb) edge[style=po] node[above right,font=\scriptsize,inner sep=1pt] {else} (case);
  \draw[style=rmw] ([yshift=-1.5mm]casr.west) to[out=225,in=135,looseness=1.0]
    node[font=\scriptsize,fill=white,inner sep=1pt,align=center]
    {$\pormw$\\[1pt]$\scriptstyle\beta=\alpha$} ([yshift=1.5mm]casw.west);
  \draw[style=ppormw] ([yshift=-1.5mm]casw.west) to[out=180,in=180,looseness=2.2]
    node[font=\scriptsize,fill=white,inner sep=1pt] {$\ppormw$} ([yshift=1.5mm]casr.west);
\end{tikzpicture}}
    \captionof{figure}{Atomicity guarantees of
    \code{r:=}\cas{}$^{\text{rel,acq}}$\code{(C,s,n)}, where \code{s} has value
    $\alpha$, \code{n} has $\varepsilon$, and \cas{} reads $\beta$ from
    \code{C}. $\pormw$ carries the branching condition, and $\ppormw$ holds on
    the succeeding branch only.}
    \label{fig:cas-rmw}
  \end{minipage}
  \hfill
  \begin{minipage}[t]{0.48\textwidth}
    \centering
    \scalebox{0.85}{\begin{tikzpicture}
  \nodeP{faar}{$e_r\colon R^\text{acq}~C~\alpha$}{0, 0}
  \nodeP{faaw}{$e_w\colon W^\text{rel}~C~\alpha+\varepsilon$}{0, -2.4}
  \edgepo[]{faar}{faaw}
  \draw[style=rmw] ([yshift=-1.5mm]faar.west) to[out=225,in=135,looseness=1.0]
    node[font=\scriptsize,fill=white,inner sep=1pt,align=center]
    {$\pormw$\\[1pt]$\scriptstyle\top$} ([yshift=1.5mm]faaw.west);
  \draw[style=ppormw] ([yshift=-1.5mm]faaw.west) to[out=180,in=180,looseness=3.0]
    node[font=\scriptsize,fill=white,inner sep=1pt] {$\ppormw[\top]$} ([yshift=1.5mm]faar.west);
\end{tikzpicture}}
    \captionof{figure}{Atomicity guarantees of
    \code{s:=}\fadd{}$^{\text{rel,acq}}$\code{(\&C,}$\varepsilon$\code{)}, where
    \fadd{} reads $\alpha$ from \code{C}. $\pormw$ carries the condition
    $\top$, so $\ppormw[\top]$ holds unconditionally.}
    \label{fig:faa-rmw}
  \end{minipage}
\end{figure}

\subsection{Justifications}\label{s:app-defs-justs}

In a program semantics which allows for out-of-order execution of instructions,
traces form equivalence classes with events reordered up to certain dependency
relations: the semantic dependency relation $\DP$, the preserved program order
$\ppo$, and the read-from relation $\rf$. $\DP$ and $\ppo$ are thread-local,
$\rf$ constrains event ordering across threads through the axiomatic memory
consistency model. $\rf$ is populated from write events visible at the point of
read events. $\DP$ and $\ppo$ are defined relative to a set of justifications.

\begin{definition}[Justifications]\label{def:justs}

  Let $\mathbb E~=~(E,\po,\pormw,\valres)$ be an event structure. Write
  $\Effects~\triangleq~\Writes\cup\Allocs\cup\Deallocs$ for the
  \emph{memory-effectful} events -- those that act on memory rather than observe
  it. A justification of a memory-effectful event $w\in\Effects$ in $\mathbb E$
  is a tuple $j\colon(P,D)\justifies^\delta w$ consisting of:
  \begin{itemize}

    \item a predicate $P$,

    \item a set $D\subseteq E$ of events, and

    \item a forwarding context $\delta~=~(\FWD,\WE)$ consisting of a
      forwarding relation $\FWD\subseteq E\times E$ and a write elision relation
      $\WE\subseteq E\times E$.

  \end{itemize}

  The event $w$ recorded in a justification is a \emph{copy} of the event of
  $\mathbb E$ it justifies, carrying that event's control label together with
  its own location and value expressions. Over \smrd{}~\cite{Richards25SMRD},
  where only writes are justified, this admits allocations and deallocations as
  well: a deallocation depends on the expression naming the location it frees
  and on the control flow reaching it, and an allocation on its size expression,
  and neither dependency is expressible as a justification of a write.
  Elaborations rewrite the expressions of the copy; the event structure is
  generated once by Definition~\ref{def:gen-es} and is thereafter fixed. As no
  elaboration changes the label, and as labels are unique in $\overline\prog$
  and a conflict-free set of events therefore contains at most one event per
  label, the event of $\mathbb E$ that a justification speaks about is recovered
  from that label wherever this is needed. Applied to the write of a
  justification, $\loc$ and $\val$ denote the expressions of the copy; applied
  to an event of $\mathbb E$, they denote that event's own.

  Given a justification $j$ as above, define $P_j=P$, $D_j=D$, and
  $\delta_j=\delta$. For a set $D\subseteq E$ of events we write
  \[
    \symbols(D)~\triangleq~\{\alpha\in\Symbols\mid\origin\alpha\in D\}
  \]
  \noindent for the symbols whose origins lie in $D$, so that
  $\origin{\symbols(D)}=D$ whenever every event of $D$ is the origin of a
  symbol.
\end{definition}

The \emph{justification set} $\mathbb J$ of all such justifications over
$\mathbb E$ is generated inductively from a set $\mathbb J_0$ of
pre-justifications which justify memory-effectful events against the origin of
symbols used in their own location, value and size expressions, by closing
$\mathbb J_0$ under elaborations Value Assignment, Forwarding, Write Elision,
Lifting, Strengthening and Weakening as given in Definition~\ref{def:gen-just}
of the generation of justifications in Appendix~\ref{s:app-defs}.

\subsection{Preserved Program Order}\label{s:app-defs-ppo}

The preserved program order $\ppo$ is defined in the context of justifications.
Note that some elaborations are defined dependent on $\ppo$ and $\pred$, such as
\hyperref[def:elab-lift]{Lifting}. The definition of $\ppo$ becomes part of
the inductive definition of justifications in Definition~\ref{def:gen-just}.

Over the definition of \smrd~in~\cite{Richards25SMRD}, we change $\ppormw$ to
reflect the asymmetric nature of memory order in conditional read-modify-write
operations such as $\cas$, where the failing branch has an acquire, the
succeeding branch a release-acquire semantics.

\begin{definition}[Preserved Program Order]\label{def:ppo}\label{def:pred}
  \emph{Preserved program order} is defined relative to a predicate $P$
  and a \hyperref[def:fwd-ctx]{forwarding context} $\delta=(\FWD,\WE)$.

  The preserved program order $\ppo$ is defined as the closure of three basic
  relations $\pposync$, $\ppormw$, and $\ppoalias$ under the forwarding context
  $\delta$.

  \[
    \ordop[\ppo]{^P_\delta}~\triangleq~\remap_\delta\left(\pposync\cup\ppormw\cup\ppoalias\right)
  \]

  \noindent $\pposync$ accounts for statically declared memory order.

  \[
      \pposync~\triangleq~{\left(\po;\Delta_{\mathcal W_{\text{rel,sc}}}\cup\poord;\Delta_{\mathcal F_{\text{rel,sc}}};\po_{\setminus\mathcal R}\cup\Delta_{\mathcal R_{\text{acq,sc}}};\poord\cup\poord_{\setminus\mathcal W};\Delta_{\mathcal F_{\text{acq,sc}}};\po\right)}_{\setminus\mathcal F\cup\mathcal B}
  \]

  \noindent $\ppormw$ accounts for atomicity of read-modify-write operations.

  \[\begin{array}{rcl}
    \ppormw&\triangleq&\pposync;\left\{(e_w,e_r)\mid(e_r,c,e_w)~\in~\pormw~\land~c\exeq_P\top\right\}\\
    && \cup\left\{(e_w,e_r)\mid(e_r,c,e_w)~\in~\pormw~\land~c\exeq_P\top\right\};\pposync
  \end{array}
  \]

  \noindent $\ppoalias$ accounts for order induced by memory location.

  \[
			\ppoalias~\triangleq~\left\{(e_1,e_2)~\in~\sqsubseteq~\mid~\exists f.\evalenv{P\land\loc(e_1)~=~\loc(e_2)}{f}\exeq\top\right\}
  \]

  \noindent Additionally, we define a predecessor relation $\pred$, which
  relates events with no other event $\ppo$-between.

  \[
    \begin{array}{r@{~}l}
      \pred_\delta(e, P)~\triangleq~\{ e'~\mid & e'\ordop[\ppo]{^P_\delta} e~\wedge~e\neq e'~\wedge{} \\
      & \forall e''.e'\ordop[\ppo]{^P_\delta} e''\ordop[\ppo]{^P_\delta} e\Rightarrow (e'=e''\vee e''=e)\}
    \end{array}
  \]

  \noindent $R_{\setminus X}$ is shorthand for $R\cap{(E\setminus X) }^2$ for a
  binary relation $R\subseteq E\times E$, and $\Delta_X=\{(x,x)\mid x\in X\}$ is
  the diagonal relation on $X$.

  In the context of an \hyperref[def:executions]{execution} with a set $J$ of
  justifications, we write $e_1\ppo_{\!\!J}~e_2$ or $e_1\ppo e_2$, if $J$ is
  clear from context, as a shorthand for $e_1\ppo^{P_J}_{\delta_J}e_2$, where
  $\delta_J~=~\bigcup\limits_{(j\colon(P,D)\justifies^\delta w)\in J}\delta$ and
  $P_J~=~\bigwedge\limits_{(j\colon(P,D)\justifies^\delta w)\in J}P$.

\end{definition}

\subsection{Generating Justifications}\label{s:app-defs-gen-justs}

\begin{definition}[Generating Justifications]\label{def:gen-just}

  Justifications are generated from \emph{pre-justifications}

  \[
    \begin{array}{rcl}
      \mathbb J_0&\triangleq&\left\{
        (\valres(w),\origin x\cup\origin\expr)\justifies^{(\emptyset,\emptyset)}
        (w\colon W~x~\expr)\mid w\in\Writes\right\}\\[2pt]
      &\cup&\left\{
        (\valres(a),\origin\expr)\justifies^{(\emptyset,\emptyset)}
        (a\colon\Allocs~\alpha~\expr)\mid a\in\Allocs\right\}\\[2pt]
      &\cup&\left\{
        (\valres(d),\origin\expr)\justifies^{(\emptyset,\emptyset)}
        (d\colon\Deallocs~\expr)\mid d\in\Deallocs\right\}
    \end{array}
  \]

  \noindent in each case where the value restriction is satisfiable,
  $\valres(\cdot)\not\equiv\bot$, and where

  \begin{itemize}
    \item
      $\origin\expr~\triangleq~\{\origin\alpha~\mid~\alpha\in\symbols(\expr)\}$
    \item $\valres(e)$ is the value restriction accumulated over the event
      structure up to the memory-effectful event $e$
  \end{itemize}

  \noindent and by elaborations Value Assignments $\elab{va}$, Forwarding
  $\elab{fwd}$, Write Elision $\elab{we}$, Lifting $\elab{lift}$, Strengthening
  $\elab{str}$ and Weakening $\elab{weak}$ inductively, such that

  \[
    \begin{array}{l}
      \mathbb J_{i+1} \triangleq \mathbb J_i\cup\bigl\{j\mid
        \exists j_1,j_2\in\mathbb J_i.\\[2pt]
      \qquad G(j_1,j)~\text{where}~
        G\in\{\elab{va},\elab{fwd},\elab{we},\elab{str},\elab{weak}\}
        ~\text{or}~\elab{lift}(j_1,j_2,j)\bigr\}
    \end{array}
  \]

  \noindent and $P_j\wedge\Omega\not\equiv\bot$.

\end{definition}

\begin{definition}[Forwarding Context]\label{def:fwd-ctx}
  In the context of an execution $\mathbb X~=~(X, J, \rf)$, with justifications
  $j\colon(P,D)\justifies^\delta w$, a \emph{forwarding context} $\delta$ is a
  pair $(\FWD,\WE)$ of binary relations on $X$:
  \begin{itemize}
    \item a forwarding relation $\FWD\subseteq X\times X$ with edges introduced by
      \hyperref[def:elab-fwd]{forwarding elaborations} as in
      Definition~\ref{def:elab-fwd}
    \item a write elision relation $\WE\subseteq X\times X$ with edges
      introduced by \hyperref[def:elab-we]{write elision elaborations} as in
      Definition~\ref{def:elab-we}
  \end{itemize}
  \noindent For a forwarding context $\delta=(\FWD,\WE)$, we define a predicate
  \[
    \psi_\delta~\triangleq~\bigwedge\limits_{(e_1,e_2)\in\FWD}\val(e_1)=\val(e_2)
  \]
  \noindent and a recursively closed remapping function
  \[
      \mathrm{remap}_\delta(e)~\triangleq~\begin{cases}
        \mathrm{remap}_\delta(e_1) & \text{where}~(e_1,e)\in(\FWD\cup\WE) \\
        e & \text{otherwise}
      \end{cases}
  \]
\end{definition}

\begin{definition}[Value Assignments]\label{def:elab-va}
  \[
    \begin{array}{rcl}
      \elab{va}(j_1,j)&\triangleq&
      j_1\colon(P,D)\justifies^\delta(w\colon W_O~x~\expr)~\wedge~
      j\colon(P,D')\justifies^\delta(w\colon W_O~x'~\expr')~\wedge\\
      &&{\hyperref[def:sem-equiv]{\alpha\equiv_P v}}~\wedge~v\in\Val~\wedge\\
			&&D'=\origin{\symbols(x')}\cup\origin{\symbols(\expr')}~\wedge\\
			&&x'=\eval{x}{\alpha}{v}~\wedge~\expr'=\eval{\expr}{\alpha}{v}
    \end{array}
  \]
\end{definition}

Value assignment substitutes $v$ for $\alpha$ in the location and value
expressions of the write, and so removes the data dependency on $\alpha$, but
it leaves the predicate $P$ unchanged, as in~\cite[Definition
4.10]{Richards25SMRD}.

\begin{definition}[Strengthening]\label{def:elab-str}
  \[
    \begin{array}{rcl}
      \elab{str}(j_1,j)&\triangleq&j_1\colon(P,D)\justifies^\delta
      w~\wedge~j\colon(P',D)\justifies^\delta w~\wedge\\
      &&S=\origin{\symbols(P')}\setminus\origin{\symbols(P)}~\wedge~\remap_\delta(S)=S~\wedge\\
      &&\forall e\in S.\left((e\sqsubseteq w\vee w\sqsubseteq e)~\wedge~w\not\ppo^P_\delta e\right)~\wedge\\
      &&P'=P'~\wedge~P~\wedge\bigwedge\limits_{e\in S}v(e)\\
    \end{array}
  \]
\end{definition}

\begin{definition}[Forwarding Relations]\label{def:fwd-rels}
  \[
    \begin{array}{rcl}
      e_1 \xrightarrow{F'_j} e_2~&\triangleq&j\colon(P, D)\justifies^{\delta} w~\wedge
      e_1 \in \pred_{\delta}(e_2, P)~\wedge
      \loc(e_1)\equiv_{P\wedge\psi_\delta}\loc(e_2)\\
      e_1 \fwdrel e_2&\triangleq&e_1\xrightarrow{F'_j}e_2~\wedge
      (e_1, e_2) \in
      (\Writes\times\Reads_{\text{rlx}}\cup\Writes\times\Writes_{\text{rlx}}\cup\Reads\times\Reads)\\
      e_1 \werel e_2&\triangleq&e_1\xrightarrow{F'_j}e_2~\wedge (e_1, e_2) \in
      (\Writes\times\Writes)\\
    \end{array}
  \]

  \noindent We name the three shapes of $\fwdrel$: \emph{store
  forwarding} for $\Writes\times\Reads_\rlx$, \emph{store-store forwarding}
  for $\Writes\times\Writes_\rlx$, and \emph{load forwarding} for
  $\Reads\times\Reads$.
\end{definition}

\begin{definition}[Forwarding]\label{def:elab-fwd}
  \[
    \begin{array}{rcl}
      \elab{fwd}(j_1,j)&\triangleq&j_1\colon(P, D)\justifies^{(f,we)}
      (w\colon W_0~x~\expr)~\wedge\\
      && j\colon(P', D')\justifies^{\left(f
      \cup \left\{(e_1,e_2)\right\},we\right)} (w\colon W_0~x'~\expr')~\wedge\\
      && e_1 \fwdrel[j_1] e_2~\wedge\\
			&& g = [\text{val}(e_2) \mapsto \text{val}(e_1)]~\wedge\\
      && P' = \evalenv{P}{g}~\wedge~\expr' = \evalenv{\expr}{g}~\wedge~x' = \evalenv{x}{g}~\wedge\\
      && D' = \origin{\symbols(\expr')}\cup\origin{\symbols(x')}\\
    \end{array}
  \]
\end{definition}

Neither $\elab{va}$ nor $\elab{fwd}$ introduces an event. Both rewrite the
location and value expressions of the justification's own copy of $w$, and both
leave its control label, and the event structure, untouched.

\begin{definition}[Write Elision]\label{def:elab-we}
  \[
      \elab{we}(j_1,j)~\triangleq~j_1\colon(P,
      D)\justifies^{(\FWD,\WE)} w~\wedge~j\colon(P, D)\justifies^{\left(\FWD,\WE
      \cup\left\{(e_2,e_1)\right\}\right)}w~\wedge~e_1 \werel[j_1] e_2
  \]
\end{definition}

\begin{observation}\label{o:elab-fwd-value-eq}

  $\delta$ is composed during the forwarding elaboration $\elab{fwd}$ and the
  write elision elaboration $\elab{we}$. The forwarding elaboration models
  load forwarding, store forwarding, and store-store forwarding.
  $\elab{fwd}$ does not itself assert equality of values commonly required for
  forwarding optimisations. The values depend on the execution, and are only
  known once $\rf$ is assigned and are known to be consistent with the
  constraints and forwarding context of the justification set. The predicate
  $\varphi$ in Definition~\ref{def:freeze} asserts the equivalence of values.

\end{observation}

\begin{definition}[Closed Relabel-Equivalence]\label{def:rel-eq}
  In the following let a \emph{relabelling}
  $\Lambda:\Symbols\rightharpoonup\Symbols$ be an environment mapping symbols
  one-to-one from one branch to another.

  Then $P_1\colon\expr_1\relabeq\Lambda\delta P_2\colon\expr_2$ is
  a \emph{relabel equivalence on expressions} if
  \[
    \exists\expr.\left(
			\evalenv{P_1\Rightarrow\expr_1=\expr}{\Lambda}\wedge
      (P_2\Rightarrow\expr_2=\expr)
      \equiv_{\hyperref[def:fwd-ctx]{\psi_\delta}}\top
    \right)
  \]
  \noindent Then $P_1:e_1\relabeq{\Lambda}{\delta}P_2:e_2$ is a
  \emph{relabel equivalence} if
  \[
    \left\{
      \begin{array}{ll}
        P_1:\loc(e_1)\relabeq{\Lambda}{\delta}P_2:\loc(e_2)~\wedge
        & (e_1,e_2)\in(\mathcal R^2\cup\mathcal W^2\cup\mathcal A^2)\\
        ~~P_1:\val(e_1)\relabeq{\Lambda}{\delta}P_2:\val(e_2) & \\
        P_1:\loc(e_1)\relabeq{\Lambda}{\delta}P_2:\loc(e_2) &
        e_1,e_2\in\mathcal D \\
        \top & e_1,e_2\in\mathcal F \\
        \bot & \text{otherwise} \\
      \end{array}
    \right.
  \]
  \noindent and $P_1\colon e_1\relabeq{\Lambda}{\delta}^* P_2\colon e_2$ is
  a \emph{closed relabel equivalence} if
  \[
    \begin{array}{l}
      P_1\colon e_1\relabeq{\Lambda}{\delta} P_2\colon e_2~\wedge\\
      \emptyset~=~\pred_\delta(e_1,P_1)
      \Leftrightarrow
      \emptyset~=~\pred_\delta(e_2,P_2)~\wedge\\
      \forall e'_1\in\pred_\delta(e_1,P_1),e'_2\in\pred_\delta(e_2,P_2).
      P_1\colon e_1'\relabeq{\Lambda}{\delta}^* P_2\colon e'_2
    \end{array}
  \]

  \noindent The two clauses read $\loc$ and $\val$ of different things. Where
  $e_1$ and $e_2$ are the writes of justifications, as they are where
  \hyperref[def:elab-lift]{Lifting} below invokes the equivalence, the first
  clause compares the expressions of the copies each records, per
  Definition~\ref{def:justs}. The recursion through $\pred_\delta$ that closes
  the equivalence then descends into events of $\mathbb E$, which carry no
  justification of their own and are compared on their own expressions.
\end{definition}

\begin{definition}[Lifting]\label{def:elab-lift}
  \[
    \begin{array}{rcl}
      \elab{lift}(j_1,j_2,j)&\triangleq&
        j_1\colon(P_1,D_1)\justifies^\delta w_1~\wedge
        j_2\colon(P_2,D_2)\justifies^\delta w_2~\wedge\\
      && j\colon(\evalenv{P_1}{\Lambda}\vee P_2,D_2)\justifies^\delta w_2~\wedge\\
			&& P_1\colon w_1\relabeq{\Lambda}{\delta}^*P_2\colon w_2~\wedge\\
      && \{\origin{\evalenv{\alpha}{\Lambda}}\mid\alpha\in\symbols(D_1)\}=D_2~\wedge\\
      && \forall\alpha\in\symbols(D_1).~
        P_1\colon\origin\alpha
        \relabeq{\Lambda}{\delta}^*
        P_2\colon\origin{\evalenv{\alpha}{\Lambda}}\\
    \end{array}
  \]
\end{definition}

\begin{definition}[Weakening]\label{def:elab-weak}
  \[
    \begin{array}{rcl}
      \elab{weak}(j_1,j)&\triangleq&j_1 : (P'\wedge P, D)\justifies^\delta w~\wedge\\
      &&j\colon(P', D)\justifies^\delta w~\wedge\\
      &&\Omega\Rightarrow P\\
    \end{array}
  \]
\end{definition}

\noindent Weakening is the only elaboration that consults the global guarantees
$\Omega$. For the sake of a simple presentation we ignore Weakening in this
paper: by Observation~\ref{obs:elab-weak} it neither enables other
elaborations nor adds dependencies during freezing, so omitting it changes
neither the dependency relations nor the executions the semantics admits.

\noindent Note that
\begin{enumerate}
  \item\label{note:elabs-modifying-p} only Strengthening, Forwarding, Lifting, and Weakening modify $P$
  \item\label{note:elabs-modifying-delta} Forwarding and Write Elision modify the forwarding context
\end{enumerate}

\begin{observation}
  Strengthening and weakening take a special role among the elaborations:
  \begin{enumerate}
    \item\label{obs:elab-str}
      \hyperref[def:elab-str]{Strengthening} enables other elaborations.
      For instance, strengthening a predicate $P$ by $\alpha=v$ enables
      \hyperref[def:elab-va]{value assignment} for
      $\alpha\equiv_{P\wedge\alpha=v}v$. As defined, $\mathbb J$ is monotone
      in the sense that adding a justification $j$ to $\mathbb J_i$ for some
      $i$ will lead to a larger $\mathbb J$, and not remove any justifications.
      Applying \hyperref[def:elab-str]{Strengthening} adds more justifications,
      and thus more constrained executions via \hyperref[def:freeze]{freezing}.
    \item\label{obs:elab-weak}
      \hyperref[def:elab-weak]{Weakening} on its own does not enable other
      elaborations and does not add dependencies during \hyperref[def:freeze]{%
        freezing}, assuming that $\Omega$ is not false, i.e.~$\Omega\neq\bot$.
     In particular, consider a set $\mathbb J_n$ of elaborations closed under
      all elaborations except Weakening,
      $\elab{$\setminus$weak}(\mathbb J_n)=\mathbb J_n$.
      Applying Weakening yields a set $\mathbb J_{n+1}=\elab{weak}(\mathbb
      J_n)$, which is again closed under all elaborations except Weakening,
      $\elab{$\setminus$weak}(\mathbb J_{n+1})=\mathbb J_{n+1}$.

  \end{enumerate}
\end{observation}

\begin{proof}
  Observation~\ref{obs:elab-weak} follows through case distinction over the
  set of elaborations. Per Definition~\ref{def:elab-weak} of weakening, let
  $j_1\colon(P'\wedge P)\justifies^\delta w$ and
  $j\colon(P',D)\justifies^\delta w$, where the global guarantees $\Omega$
  imply $P$.

  \medskip
  \hyperref[def:elab-va]{\textbf{Value assignment}} $\elab{va}$: If
  $\alpha\equiv_{P'}v$, then $\alpha\equiv_{P'\wedge P}v$.

  \medskip
  \hyperref[def:elab-fwd]{\textbf{Forwarding}} $\elab{fwd}$:
  $e_1\xrightarrow{F_j}e_2$ iff
  $e_1\in\pred_\delta(e_2,P')$ and
  $\loc(e_1)\equiv_{P'\wedge\psi_\delta}\loc(e_2)$, then
  $e_1\in\pred_\delta(e_2,P'\wedge P)$ and
  $e_1\equiv_{P'\wedge P\wedge\psi_\delta}e_2$ iff
  $e_1\xrightarrow{F_{j_1}}e_2$.

  \medskip
  \hyperref[def:elab-we]{\textbf{Write elision}} $\elab{we}$:
  $e_1\xrightarrow{WE_j}e_2$ iff
  $e_1\in\pred_\delta(e_2,P')$ and $e_1\equiv_{P'\wedge\psi_\delta}e_2$,
  then $e_1\in\pred_\delta(e_2,P'\wedge P)$ and
  $e_1\equiv_{P'\wedge P\wedge\psi_\delta}e_2$ iff
  $e_1\xrightarrow{WE_{j_1}}e_2$.

  \medskip
  \hyperref[def:elab-lift]{\textbf{Lifting}} $\elab{lift}$: Let
  $j_1\colon(P'_1\wedge P_1,D_1)\justifies^\delta w_1$ and
  $j_2\colon(P'_2\wedge P_2,D_2)\justifies^\delta w_2$, where
  $\Omega\implies P_1$ and $\Omega\implies P_2$. If
  $P'_1\colon w_1\relabeq{\Lambda}{\delta}^*P'_2\wedge P_2\colon w_2$,
  $P'_1\wedge P_1\colon w_1\relabeq{\Lambda}{\delta}^*P'_2\colon w_2$, or
  $P'_1\colon w_1\relabeq{\Lambda}{\delta}^*P'_2\colon w_2$, then
  $P'_1\wedge P_1\colon w_1\relabeq{\Lambda}{\delta}^*P'_2\wedge P_2\colon w_2$,
  as $P'_1\implies\expr_1=\expr$ implies
  $P'_1\wedge P_1\implies\expr_1=\expr$,
  and $P'_2\implies\expr_2=\expr$ implies
  $P'_2\wedge P_2\implies\expr_2=\expr$.
\end{proof}

\subsection{Executions in Event Structures}

An execution records the memory effects of a run: its events are the reads,
writes, allocations and deallocations of $\Reads$, $\Writes$, $\Allocs$ and
$\Deallocs$. Branching and fence events are excluded, not because they are
inert, but because what they contribute is already recorded elsewhere and
including them would state it twice.

A branching event carries its condition, which Definition~\ref{def:gen-es} of
the event structure semantics accumulates into the value restriction $\valres$
of the events below it. That restriction reaches the execution twice over: the
definition below admits only justification sets consistent with
$\bigwedge_{e\in X}\valres(e)$, and pre-justifications take
$P=\valres(w)$ per Definition~\ref{def:gen-just}. The branching event itself
would add nothing beyond it.

A fence carries ordering rather than data, and that ordering is folded into
$\pposync$ in Definition~\ref{def:ppo} of preserved program order. Its
definition uses fences as intermediate points -- the summands
$\po;\Delta_{\Fences_{\text{rel,sc}}};\po_{\setminus\Reads}$ and
$\po_{\setminus\Writes};\Delta_{\Fences_{\text{acq,sc}}};\po$ -- and then
restricts the union to $_{\setminus\Fences\cup\Branches}$, so that the pairs
it contributes relate the accesses \emph{around} a fence and never the fence
itself. Excluding $\Fences$ and $\Branches$ from $X$ below is that same
decision, stated on the events rather than on the order. Consistently with
fences carrying no data, Definition~\ref{def:rel-eq} of closed
relabel-equivalence relates any two of them.

\begin{definition}[Executions in Symbolic Event Structures]\label{def:executions}

  An \emph{execution} in an event structure $\mathbb E~=~(E,\po,\pormw,
  \valres)$ is a tuple $\mathbb X~=~(X,J,\rf)$ comprising:
  \begin{itemize}
    \item a maximal conflict-free set $X\subseteq E\setminus(\Branches\cup
      \Fences)$ of events,
    \item a consistent set $J\subseteq\mathbb J$ of
      \hyperref[def:justs]{justifications} consistent with
      $\bigwedge\limits_{e\in X}\valres(e)$, and
    \item an injective read-from relation $\rf\subseteq\Writes\times\Reads$
      linking read events with write events.
  \end{itemize}
  with $J$ such that all memory-effectful events $w\in X\cap\Effects$ not elided by
  the forwarding context the justifications of $J$ share are uniquely justified
  by a justification in $J$, a justification being matched to the memory-effectful
  event of $X$ carrying the label of its copy of $w$ as in
  Definition~\ref{def:justs}. The exemption is the $\dagger$ of
  Definition~\ref{def:freeze}, and it reaches writes only, as those are the
  events write elision elides.

  An execution is \emph{complete} if additionally $\rf$ is surjective, that is
  if $\rf$ assigns each read event a write event.
\end{definition}

Given a set $J$ of justifications over a maximal conflict-free set $X$ with a
read-from relation, $\DP$ and $\ppo$ are defined through freezing with
additional constraints as follows.

\begin{definition}[Freezing Justifications]\label{def:freeze}

  Let

  \begin{itemize}

    \item $X$ be a maximal conflict-free set of events

    \item $J\subseteq\mathbb J$ be a set of justifications
      $j\colon(P,D)\justifies^\delta w$ of memory-effectful events in $X$, such that
      all justifications in $J$ share a common forwarding context $\delta$

    \item $\dagger\subseteq X$ be the events in $X$ elided by
      $\delta=(\FWD,\WE)$, i.e.~$\dagger=\pi_1\WE$

    \item all memory-effectful events in $X\cap\Effects\setminus\dagger$ not elided
      by $\delta$ be uniquely justified in $J$, where only writes are ever
      elided, $\dagger\subseteq\Writes$

    \item $\rf\subseteq X^2\cap(\Writes\times\Reads)$ be a read-from relation
      which assigns to read events unique write events

  \end{itemize}

  \noindent where a justification is matched to the memory-effectful event of $X$
  carrying the label of its copy of $w$, as in Definition~\ref{def:justs}, and
  $w_j$ denotes that event of $X$ throughout. Then define

  \begin{equation}
    \freeze(X, J, \rf)~\triangleq~(\DP, \ppo, \varphi)
    \label{eq:def:freeze}
  \end{equation}

  \noindent with the set $X$ of events, binary relations $\DP,
  \ppo~\subseteq~X\times X$ and a predicate $\varphi$, such that

  \begin{equation*}
    \begin{array}{rcl}
      \DP&\triangleq&\bigcup\limits_{j\in
      J}\{(\origin\alpha,w_j)\mid\alpha\in\symbols(P_j)\cup\symbols(D_j)\}\\
      P&\triangleq&\bigwedge\limits_{j\in J}
      P_j\land\psi_{\delta_j}\\
      \ppo&\triangleq&\bigcup\limits_{j\in J}\ppo_\delta^{P_j}\cap\left\{e\in X\mid
      e\po w_j\vee e=w_j\right\}^2\\
    \varphi&\triangleq&
      {\left(\pi_1(\rf\cup\DP)\cap\dagger=\emptyset\right)}^{(1)}\land
      {\left(\pi_2(\rf)=(X\cap\mathcal R)\right)}^{(2)}\land
      {\left(P\land\varphi_\rf\not\equiv\bot\right)}^{(3)}\\
    \end{array}
  \end{equation*}

  \noindent where $\varphi_\rf$ establishes

  \begin{itemize}

    \item the equality of locations and values of each read-write pair in $\rf$,
      i.e.~$\loc(r)=\loc(w)~\wedge~\val(r)=\val(w)$ for all $(r,w)\in\rf$. As
      $\rf$ relates events of $X$, these are the events' own expressions and not
      those of the copy recorded by the justification of $w$. Nothing is lost by
      that reading: of the elaborations only $\elab{va}$ and $\elab{fwd}$ rewrite
      the copy, the first under $\alpha\equiv_Pv$ and the second under
      $\psi_\delta$, and $P$ above conjoins both, so copy and event carry
      equivalent expressions under $P$

    \item enforces the disjointness of memory locations: distinct global
      locations denote distinct locations, and the symbolic location introduced
      by an allocation event is distinct from the global locations and from the
      symbolic location of any other allocation event without an intermediate
      deallocation event.

  \end{itemize}

  \noindent Then $\varphi$ establishes that

  \begin{itemize}

    \item (1) no events in the domain of $\rf$ and $\DP$ are elided by $\delta$,
      and in particular reads only read from visible writes

    \item (2) $\rf$ assigns a read-from to all the read events in $X$

    \item (3) $P\wedge\varphi_\rf$ is satisfiable

  \end{itemize}

\end{definition}

\paragraph{Axiomatic Memory Consistency Model.}\label{def:mem-model-axiom}
Consistent executions are subject to constraints from the axiomatic memory
consistency model. In this paper we consider MRD+C11 with

\begin{enumerate}

  \item\label{def:mem-model-axiom:nta} \texttt{No-Thin-Air} constraint:
    $\DP\cup\ppoord\cup\rf$ is acyclic. We write
    $\nta\triangleq\left(\DP\cup\ppoord\cup\rf\right)^+$ for the
    \emph{no-thin-air order} the constraint declares acyclic, so that the
    constraint states exactly that $\nta$ is a strict partial order. Here, as
    in Section~\ref{sec:cas}, $R^+\triangleq\bigcup_{n\geq1}R^n$ denotes the
    transitive closure of a relation $R$.

  \item\label{def:mem-model-axiom:co} \texttt{Coherence} axiom: $\ECO\cup\HB$ is
    acyclic, where the extended coherence order is
    $\ECO\triangleq(\rf\cup\CO\cup\FR)^+$ with $\FR\triangleq\rf^{-1}\circ\CO$,
    and happens-before is
    $\HB\triangleq\left(\DP\cup\ppoord\cup\SW\right)^+$ with the
    synchronises-with edges $\SW\triangleq\rf\cap(\Wrel\times\Racq)$ relating a
    releasing write to an acquiring read that takes its value.

\end{enumerate}

Each of the models compared in Section~\ref{sec:bug} carries such an order,
differing only in the intra-thread relation closed over $\rf$: \smrd{} takes
$\DP\cup\ppoord$, whereas RC11z takes the whole of program order,
$\nta_{\text{RC11z}}=(\poord\cup\rf)^+$. The two presentations of RC11z reach
that order by different routes. Declarative RC11 states it as an axiom,
forbidding cycles in $\poord\cup\rf$~\cite{rc11}. The RC11-RAR semantics that
Semenyuk et al.~\cite{Semenyuk2023RCU} verify RCU against is operational, with
no acyclicity axiom at all: a thread steps its commands in program order, so no
execution it generates reorders an access past a later one, and
$\poord\cup\rf$-acyclicity holds of the model rather than being imposed on it.
Either route yields the containment Example~\ref{ex:uaf-rc11z} needs.

$\SW$ is the only inter-thread constituent of $\HB$: $\DP$ and $\ppo$ are
per-thread, which is why the operational semantics of Section~\ref{s:opsem} can
take the whole of the future set $\Phi$ from $\DP$ and $\ppo$ and carry $\SW$ in
the viewfronts of the program state instead.

The following observation ensures that read events can only read from writes
observable to the reading thread, which will become important when we construct
the operational semantics from the event structure semantics.

\begin{lemma}\label{l:co-axiom-rf-viswrite}
  By Axiom~\ref{def:mem-model-axiom:co}, $\rf$ assigns to a read only a write
  that no $\HB$-earlier write at the same location is $\CO$-after -- that is,
  a write observable to the reading thread.
\end{lemma}

\begin{proof}

  A write $w_2$ covers a write $w_1$, if it writes to a location equivalent
  relative to latent constraints, and $w_1$ is ordered before $w_2$. Suppose
  towards contradiction, that a read event $r$ after $w_2$, that is $w_2$
  happens before $r$, reads from $w_1$. Then the extended coherence order $\ECO$,
  instantiated as $\FR;\CO$, forms a cycle with the $\HB$ relation as in the
  following diagram.

  
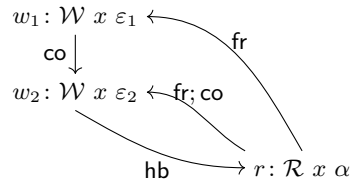
\begin{figure}[H]
    \begin{tikzpicture}[node distance=3cm, auto]

      \node (w1) at (1,3) {$w_1\colon\Writes~x~\expr_1$};
      \node (w2) at (1,2) {$w_2\colon\Writes~x~\expr_2$};
      \node (r) at (4,1) {$r\colon\Reads~x~\alpha$};

      \draw[->] (w1) to node[left] {$\CO$} (w2);
      \draw[->] (r.north) to[out=120,in=0] node[above] {$\FR$} (w1.east);
      \draw[->] (w2.south) to[out=330,in=180] node[below] {$\HB$} (r.west);
      \draw[->] (r.north west) to[out=150,in=0] node[above] {$\FR;\CO$} (w2.east);

    \end{tikzpicture}
  \end{figure}

\end{proof}

\paragraph{Use-after-free.} The bug of Section~\ref{sec:bug} is a property of an
execution. In \smrd{} deallocation is a first-class action, and UAF is an
ordering property with the no-thin-air order $\nta$ of
Axiom~\ref{def:mem-model-axiom:nta}. As each of the models we consider, \smrd{}
and RC11z, supplies a no-thin-air order, we define UAF parametric in the model.

\begin{definition}[Use-after-free]\label{def:uaf}

  Let $M$ be a memory consistency model with no-thin-air order $\nta_M$. An
  execution $\mathbb X~=~(X,J,\rf)$ with
  $\freeze(X,J,\rf)~=~(\DP,\ppo,\varphi)$ \emph{exhibits a use-after-free} if
  there are a deallocation $d\in X\cap\Deallocs$ and an access $a\in
  X\cap(\Reads\cup\Writes)$ of the deallocated location,
  $\loc(a)\equiv_\varphi\loc(d)$, such that $(a,d)\notin\nta_M$, that is the
  access is not ordered before the deallocation.

\end{definition}

Intermediate reallocation of memory addresses is handled implicitly as follows.
Definition~\ref{def:freeze} forces symbolic memory locations apart only for
allocations \emph{without} an intervening deallocation, so an allocation may
reuse the location a deallocation released, and an access $\nta_M$-ordered after
that reallocation reads memory that is live again. Exempting it would take an
$a'\in X\cap\Allocs$ at $\loc(d)$ with $(d,a')\in\nta_M$ and $(a',a)\in\nta_M$,
which can exist only when $(d,a)\in\nta_M$; Definition~\ref{def:uaf} reports
that case too, and so over-approximates on programs that free and reallocate.
The use-after-free of Section~\ref{sec:bug} is not of that kind: there $a$ and
$d$ are unordered, no such $a'$ exists, and the two readings agree.

\begin{example}[The UAF bug under the two orders]\label{ex:uaf-rc11z}

  Take the execution of Figures~\ref{fig:rcu-inc-code}
  and~\ref{fig:rcu-reclaim-code} in which Thread~2's \cas{} fails, Thread~1's
  succeeds, and reclaim's guard read \code{r[i] := rcu[tid]} returns the $0$
  written by Thread~2's RCU exit \code{rcu[tid] := 0}. Four events matter:

  \[
    \text{\code{v := *s}}~,~\text{\code{rcu[tid] := 0}}~,~
    \text{\code{r[i] := rcu[tid]}}~,~\text{\code{free(s)}}
  \]

  \noindent the first two on Thread~2 and the last two on Thread~1, and the
  question is whether \code{v := *s} is $\nta$-before \code{free(s)}.

  Under $\nta_{\text{RC11z}}=(\poord\cup\rf)^+$ it is
  (Figure~\ref{fig:uaf-bug-oota-cycle-rc11}). All of program order
  counts, so \code{v := *s} $\po$ \code{rcu[tid] := 0} on Thread~2 and
  \code{r[i] := rcu[tid]} $\po$ \code{free(s)} on Thread~1, and the guard read
  takes the exit's $0$:

  \[
    \text{\code{v := *s}}~\po~\text{\code{rcu[tid] := 0}}
    ~\xrightarrow{~\rf~}~
    \text{\code{r[i] := rcu[tid]}}~\po~\text{\code{free(s)}}
  \]

  \noindent So $(\code{v := *s},\code{free(s)})\in\nta_{\text{RC11z}}$, the
  dereference is ordered before the deallocation, and
  Definition~\ref{def:uaf} is not met: no use-after-free.

  Under $\nta_{\smrd}=(\ppoord\cup\DP\cup\rf)^+$ it is not
  (Figure~\ref{fig:uaf-bug-oota-cycle-smrd}). Thread~2's \cas{} fails, and on
  the failing branch it carries no release, so \code{v := *s} is neither $\ppo$-
  nor $\DP$-before \code{rcu[tid] := 0}: the first edge of the chain is missing.
  \code{v := *s} is a read and has no outgoing $\rf$, so no other route to
  \code{free(s)} is open either, and $(\code{v :=
  *s},\code{free(s)})\notin\nta_{\smrd}$. Definition~\ref{def:uaf} is met, with
  \code{v := *s} the access and \code{free(s)} the deallocation: the
  use-after-free of Section~\ref{sec:bug}. Annotating the RCU exit with a
  release, as in Figure~\ref{fig:use-after-free-bug-fixed}, puts \code{v := *s}
  $\ppo$ \code{rcu[tid] := 0} back and restores the chain.

\end{example}

\mordor{} reports the violation as an edge out of the deallocation event, which
is a convenience of the tool: the property it witnesses is the one defined here.

\subsection{Futures}

\begin{definition}[Program Futures in Event Structures]\label{def:futures}

  The \textit{futures} $\phi_{\mathbb X}$ for an execution $\mathbb X~=~(X,J,
  \rf)$ in an event structure $\mathbb E$ is the set of pairs

  \[
    \phi_{\mathbb X}~\triangleq~X^2\cap(\ppoord\cup\DP)
  \]

  \noindent Both relations are per-thread, and a future is accordingly the order
  in which one thread must execute its own events. The third relation
  Axiom~\ref{def:mem-model-axiom:nta} keeps acyclic, $\rf$, is deliberately not
  among them: it is the only inter-thread dependency of the model, and admitting
  it would stop $\Phi$ from splitting per thread, which is what the
  Owicki-Gries decomposition of Section~\ref{s:opsem-og} rests on. The ordering
  a release/acquire handshake induces is carried by the viewfronts of the
  operational semantics instead (Section~\ref{s:opsem}), not by $\Phi$. Since
  $\ppoord\cup\DP\subseteq\ppoord\cup\DP\cup\rf$, acyclicity of $\phi_{\mathbb X}$
  still follows from the axiom, so the minima of Definition~\ref{def:horizons}
  exist.

  The future set $\Phi$ in an event structure $\mathbb E$ is the set of all
  futures $\phi_{\mathbb X}$ for \hyperref[def:executions]{executions} $\mathbb
  X$.

\end{definition}

\begin{definition}[Histories]\label{def:histories}
  A \emph{history} in an execution $\mathbb X~=~(X,J,\rf)$ is
  a set $H\subseteq X$ of events downward-closed in $\ppo_J\cup\DP_J$,
	i.e.~$\Downclosed{H}~=~H$. Equivalently, $H$ is downward-closed in the future
  $\phi_{\mathbb X}$ of Definition~\ref{def:futures}: a history is a prefix in
  the ordering a future induces.
\end{definition}

\begin{definition}[Posterior Futures Set]\label{def:post-futures}
  Let $\phi_{\mathbb X}$ be a future for an execution $\mathbb X~=~(X,J,\rf)$ in
  an event structure $\mathbb E$, and let $H$ be a
  history in $\mathbb X$.

  The \emph{posterior future} $\phi_{\mathbb X,H}$ is the
  set of pairs $(e_1,e_2)\in\phi_{\mathbb X}$ such that $e_1\not\in H$.

  The \emph{posterior future set} $\Phi_H$ is the set of posterior futures
  $\phi_{\mathbb X,H}$ for all complete executions $\mathbb X$ in $\mathbb E$.
\end{definition}

\begin{definition}[Future Horizons]\label{def:horizons}
  The \emph{future horizon} $\horizon\phi$ of the posterior
  future $\phi=\phi_{\mathbb X,H}$ for a history $H$ in an execution
  $\mathbb X~=~(X,\mathbb J,\rf)$ is the set of events $e\in X\setminus H$, such
  that there is no event $e'\in X\setminus H$ before $e$
  w.r.t.~$\ppoord\cup\DP$, i.e.~$(e',e)\in\phi$.

  $\horizon\Phi$ defines the set of future horizons $\horizon\phi$ for all
  future sets $\phi\in\Phi$.
\end{definition}

\vfill

\pagebreak
\section{Appendix: Proof of Finite Bound on Posterior Future Horizons}\label{s:app-proofs}

We prove a finitary bound on posterior future horizons, discharging
Theorem~\ref{t:finite-post-futures} by the argument outlined in
Section~\ref{s:essem-finite-bound}.

The \emph{next enabled actions} $\horizon\Phi_H$ of a history $H$, introduced in
Section~\ref{s:essem}, are the horizons of the posterior future set $\Phi_H$ in
the sense of Definition~\ref{def:horizons}: both denote the minima of the events
that follow $H$ under $\ppo$ and $\DP$, and we write $\horizon\Phi_H$ for
either. The main text uses the operational reading, \emph{next enabled actions},
and this appendix the event-structure one, \emph{posterior future horizons}.

\finitenextactions*

In general, programs with unbounded loops yield an unbounded number of posterior
future horizons. We prove the above theorem by showing that in programs where
unbounded loops are episodic, posterior future horizons are symmetric between
iterations and may in fact narrow down depending on the loop condition.

The symmetry between iterations is witnessed by a map $\gamma$ from the events
of the $i+1$-st iteration of an episodic loop, in an execution of the event
structure generated for a program with step-counter $n+1$, to those of the
$i$-th iteration of an execution of the event structure generated with
step-counter $n$. $\gamma$ is defined as an extension of the de Bruijn indexing
of symbols in the respective executions per Section~\ref{s:de-bruijn}. The map
$\gamma$ will preserve and reflect dependency relations, and thus posterior
future horizons.

As a result of the symmetry between iterations we can treat failing iterations
of episodic loops uniformly, for executions of arbitrary, though finite, depth.
No single event structure $\langle\prog\rangle_n$ holds all of them for a finite
$n$, and thus none holds all of the mappings $\gamma$ either. We therefore first
prove the event structures generated for increasing finite step-counters
monotonic, and establish a new event structure as the limit of the monotonically
increasing chain. The finiteness result is then obtained in that limit.

\subsection{Monotonicity of event structures}

In the previous work~\cite{Richards25SMRD}, a global step-counter was introduced
to restrict the semantics of programs with unbounded loops to the fragment of
terminating executions. Throughout the rest of the appendix we assume a
step-counter taken per loop, which allows us to compare iterations in executions in
such programs for varying step-counters.

The event structure semantics of programs $\prog$ is monotonic in the
step-counter $n$. Monotonicity is witnessed by identity maps embedding event
structures generated for increasing step-counters, i.e.
$\idmap_n:\langle\prog\rangle_n\embeds\langle\prog\rangle_{n+1}$. The identity
maps $\idmap_n$ identify events by control label from $\Labels$. Control labels
and symbols introduced are enumerated from the start of the program as in
Paragraph~\ref{par:order-clab}. Note that symbols are \emph{not} reassigned
under the de Bruijn indexing from Section~\ref{s:de-bruijn}.

\begin{lemma}[Event structures are monotonic in
  step-counters]\label{l:es-mono}

  For all $n$ let

  \begin{equation}
    \mathbb E_n =
    \langle\prog\rangle_{n~\emptyset~\lambda\rho\,\varphi.\emptyset~\top} =
    (E_n,\po_n,\pormw_n,\valres_n)
  \end{equation}

  \medskip\noindent then $\mathbb E_n\subseteq\mathbb E_{n+1}$ with
\medskip
  \begin{equation}
    E_n\subseteq
    E_{n+1},~~\po_n\subseteq\po_{n+1},~~\pormw_n\subseteq\pormw_{n+1},~~\text{and}~~\valres_n(e)=\valres_{n+1}(e)~\text{for}~e\in E_n
  \end{equation}
\end{lemma}

\begin{proof}

  The proof proceeds by induction over $n$ and the program structure. In the
  base case $\mathbb E_0\subseteq\mathbb E_1$ as $\mathbb E_0$ is empty. In
  particular for \code{while} statements,\\
  $\langle\semwhile{b}{\prog}\rangle_{0~\rho~\kappa~\varphi} = \langle
  \semwhile{b}{\prog}\rangle_{1~\rho~\kappa~\varphi}$ by definition of the
  semantics of \code{while} by unravelling into \code{if}-statements
  in Equation~\ref{eq:while-sem}.

  The semantics of all statements $c$ per Definition~\ref{def:gen-es} of the
  event structure semantics is monotonic in $\kappa$, so that whenever
  $\kappa(\rho,\varphi)\subseteq\kappa'(\rho,\varphi)$ for all $\rho$ and
  $\varphi$, then $\langle
  c\rangle_{n~\rho~\kappa~\varphi}\subseteq\langle
  c\rangle_{m~\rho~\kappa'~\varphi}$. The product and coproduct of event
  structures are monotonic as well, so that whenever $\mathbb
  E_1\subseteq\mathbb E'_1$ and $\mathbb E_2\subseteq\mathbb E'_2$, then
  $\mathbb E_1+\mathbb E_2\subseteq\mathbb E'_1+\mathbb E'_2$ and $\mathbb
  E_1\times\mathbb E_2\subseteq\mathbb E'_1\times\mathbb E'_2$.

  Thus $\langle\prog\rangle_{n~\rho~\kappa~\varphi}\subseteq\langle\prog
  \rangle_{n+1~\rho~\kappa'~\varphi}$ for all $\prog$ and $n$.

\end{proof}

The monotonic embeddings have a fixed point in the class of all event
structures by the Knaster-Tarski fixed point theorem.

\begin{corollary}\label{c:es-mono-fp}

  The identity mapping embedding
  \[
    \langle\prog\rangle_{n~\emptyset~\lambda\rho\,\varphi.\emptyset~\top} \embeds
    \langle\prog\rangle_{n+1~\emptyset~\lambda\rho\,\varphi.\emptyset~\top}
  \]
  then has a fixed point

  \begin{equation}
    \mathbb E_\infty =
    \bigcup\limits_{n\in\nat}
      \langle\prog\rangle_{n~\emptyset~\lambda\rho\,\varphi.\emptyset~\top}
    \triangleq
    \left(
      \bigcup\limits_n E_n,
      \bigcup\limits_n\po_n,
      \bigcup\limits_n\pormw_n,
      \bigcup\limits_n\valres_n
    \right)
  \end{equation}

  \noindent such that each event structure
  $\langle\prog\rangle_{n~\emptyset~\lambda\rho\,\varphi.\emptyset~\top}$ embeds into the
  limit via the identity mapping.

\end{corollary}

Executions transfer along the chain of event structures in
Corollary~\ref{c:es-mono-fp} too, which allows us to lift $\gamma$ to the limit
event structure and to read $\gamma$ as a map between executions of arbitrary
depth.

\begin{lemma}[Executions transfer along the chain]\label{l:execs-transfer}

  Let $\mathbb X~=~(X,J,\rf)$ be an execution in $\mathbb E_n$ in which every
  loop of $\prog$ has been exited. Then $\mathbb X$ is an execution in $\mathbb
  E_{n+1}$, and hence in the fixed point $\mathbb E_\infty$ of
  Corollary~\ref{c:es-mono-fp}.

  Conversely, every execution in $\mathbb E_\infty$ is an execution in $\mathbb
  E_n$ for some $n$.

\end{lemma}

\begin{proof}

  For the first claim, $X\subseteq E_n\subseteq E_{n+1}$ and
  $\po_n\subseteq\po_{n+1}$ by Lemma~\ref{l:es-mono}, and $J\subseteq\mathbb
  J_n\subseteq\mathbb J_{n+1}$, as the generation of justifications in
  Definition~\ref{def:gen-just} is monotone in the step-counter:
  pre-justifications are generated for the memory effectful events, and
  $E_n\subseteq E_{n+1}$. The read-from relation, and the no-thin-air and
  coherence axioms remain unchanged.

  We show that $X$ is \emph{maximal} conflict-free in $\mathbb E_{n+1}$: We use
  the assumption that every loop has been exited. The events of
  $E_{n+1}\setminus E_n$ must belong to a further iteration afforded by the
  larger step-counter. As $\mathbb X$ exits each loop, $X$ contains for each
  loop the branching decision on which it exits, and by
  Definition~\ref{def:gen-es} of the event structure semantics the value
  restrictions of the additional events are incompatible with the value
  restrictions of events in $X$, so that $X$ remains maximal and $\mathbb X$ is
  an execution in $\mathbb E_{n+1}$, $\mathbb E_m$ for all $m>n$ and hence in
  $\mathbb E_\infty$.

  For the converse, let $\mathbb X$ be an execution in $\mathbb E_\infty$. As
  $\mathbb X$ terminates, it enters each loop finitely often; let $n$ exceed the
  greatest iteration count over the loops of $\prog$. Every event of $X$ then
  lies in $E_n$, and $X$ is maximal there, as any event of $\mathbb E_\infty$
  compatible with $X$ would already have been in $X$ by maximality in $\mathbb
  E_\infty$.

\end{proof}

\noindent Consequently two executions of $\mathbb E_\infty$ are executions of
the $\mathbb E_n$ for a step-counter $n$ that exit every loop in each execution,
and we may read Definition~\ref{def:gamma-constr} of $\gamma$ as relating two
executions of $\mathbb E_\infty$.

\subsection{Restricted Predicates}\label{s:app-proofs-restricted}

Predicates -- accumulated in the value restrictions of events
(Definition~\ref{def:gen-es}) and in the elaborations of justifications
(Definition~\ref{def:gen-just}) -- decide branching and define the dependency
relations.
The proofs below show branching decisions and dependency relations symmetric
between iterations of an episodic loop, and repeatedly rely on the fact that
only the symbols known before the loop or added in the current iteration of the
loop determine the outcome of branching decisions and justifications.
Formally, this is captured by \emph{predicate restrictions}. We define the
restriction $\restrictloop{P}{\ell}{i}$ of a predicate $P$ to the set

\begin{equation}
  \begin{array}{rcl}
    \Sigma_{\ell\mapsto \emptyset} & \triangleq &
    \{\alpha\in\Symbols\mid\iter(\origin\alpha)(\ell)=\text{undefined}\}\\
    \Sigma_{\ell\mapsto i} & \triangleq &
    \Sigma_{\ell\mapsto \emptyset}\cup\{\alpha\in\Symbols\mid\iter(\origin\alpha)(\ell)=i\} \\
  \end{array}
\end{equation}

\noindent of symbols read in the $i$-th iteration or before the loop,
$i=\emptyset$, as follows

\begin{equation}
  \restrictloop{P}{\ell}{i}~\triangleq~\bigwedge
  \left\{p\mid\symbols(p)\subseteq\Sigma_{\ell\mapsto i}\wedge P\implies p\right\}
\end{equation}

The subscript names both the loop $\ell$ and the iteration $i$ in $\ell$:
$\restrictloop{P}{\ell}{i}$ restricts to the $i$-th iteration \emph{of the loop
$\ell$}. For a nested or sibling loop $\ell'$ we write
$\restrictloop{P}{\ell'}{i}$; as we fix a single episodic loop $\ell$
throughout, we abbreviate $\restrictloop{P}{\ell}{i}$ to $\restrict{P}{i}$.

$\restrict{P}{i}$ is well-defined, as $P$ and $\Sigma_{\ell\mapsto i}$ are
finite, and the formulas are in a propositional logic over the program state.

The register Condition~\ref{episodic:reg} in Definition~\ref{def:episodic} of
episodic loops provides a syntactic criterion to decide episodicity in programs.
In the event structure semantics Condition~\ref{episodic:reg} implies the
following constraint on expressions in events.

\begin{lemma}\label{l:episodic-symb}\label{l:restrict-syms}

  In an episodic loop $\ell$, the symbols $\alpha$ in the memory location
  $\loc(e)$ or value $\val(e)$ of an event $e$ in the $i$-th iteration of $\ell$
  are read in the same iteration of the loop, i.e.

  \[
\symbols(\loc(e))\cup\symbols(\val(e))\subseteq\Sigma_{\ell\mapsto
  \iter(e)(\ell)}
  \]

\noindent where

  \[
  \Sigma_{\ell\mapsto i}=
  \left\{
      \iter(\origin\alpha)(\ell)=i~\vee\\
      \iter(\origin\alpha)(\ell)=\text{undefined}
  \right\}
  \]

\end{lemma}

\begin{proof}

  By Definition~\ref{def:gen-es} of the event structure semantics, the location
  and value expressions of an event are $\evalreg{\expr}{\rho}$ for the register
  state $\rho$ accumulated from the start of the program up to that event.
  Symbols therefore reach the expressions of an event only through $\rho$, and
  they enter $\rho$ only at read and allocation events, each introducing a fresh
  symbol by the convention of Paragraph~\ref{par:order-symb}. The lemma follows
  from an invariant on the register state, which we establish by induction over
  the derivation of the event structure. Condition~\ref{episodic:reg} of
  Definition~\ref{def:episodic} asserts that registers are written to before use
  in the same iteration of the loop.

  \medskip\noindent\textbf{Invariant.} For every register $r$ and every point in
  the $i$-th iteration of $\ell$ at which $\rho(r)$ is read,
  $\symbols(\rho(r))\subseteq\Sigma_{\ell\mapsto i}$.

  The induction then proceeds over the syntax of the programming language
  following Definition~\ref{def:gen-es} of the event structure semantics.

  \begin{itemize}

    \item $\semregst{}{\expr}$ sets $\rho'=\rho[r\mapsto\evalreg{\expr}{\rho}]$,
      so that $\symbols(\evalreg{\expr}{\rho})$ is the union of
      $\symbols(\rho(r'))$ over the registers $r'$ occurring in $\expr$. Each
      such $r'$ is accessed at this point, so by Condition~\ref{episodic:reg} it
      was written $\po$-before within the same iteration or before the loop, and
      the induction hypothesis gives
      $\symbols(\rho(r'))\subseteq\Sigma_{\ell\mapsto i}$.

    \item The read commands $\semglobld[o]{}{x}$ and
      $\semglobld[o]{}{\semderef\expr}$ set $\rho'=\rho[r\mapsto\alpha]$ for a
      symbol $\alpha$ introduced by a read event in the $i$-th iteration, so
      $\iter(\origin\alpha)(\ell)=i$ and $\alpha\in\Sigma_{\ell\mapsto i}$.
      Their location expressions are covered by the previous case.

    \item $\semmalloc{}{\expr}$ likewise binds a fresh symbol introduced by an
      allocation event of the $i$-th iteration.

    \item $\semregst{}{\semamp{x}}$ sets $\rho'=\rho[r\mapsto x]$ for the
      address of a global variable, which carries no symbols.

    \item $\semfadd[o_r][o_w]{}{x}{\expr}$ sets $\rho'=\rho[r\mapsto\alpha]$ for
      the symbol $\alpha$ of its read event, again of the $i$-th iteration. The
      value written, $\alpha+\evalreg{\expr}{\rho}$, is covered by the first
      case together with $\alpha\in\Sigma_{\ell\mapsto i}$.

    \item $\semcas[o_r][o_w]{}{x}{\expr_1}{\expr_2}$ sets
      $\rho'=\rho[r\mapsto\top]$ on the succeeding and
      $\rho'=\rho[r\mapsto\bot]$ on the failing branch, neither of which carries
      symbols. The symbol $\alpha$ of its read event instead occurs in the
      condition $\alpha=\evalreg{\expr_1}{\rho}$ of its branching event, which
      is of the $i$-th iteration alongside the read, so
      $\alpha\in\Sigma_{\ell\mapsto i}$; the value written,
      $\evalreg{\expr_2}{\rho}$, is covered by the first case.

    \item Write commands, $\semfence[o]$ and $\semfree{}$ leave $\rho$
      unchanged, and their expressions are covered by the first case.

    \item Sequential and parallel composition and branching thread $\rho$
      without introducing symbols.

  \end{itemize}

  \medskip\noindent\textbf{Loop boundaries.} Passing from the $i$-th iteration
  to the $i+1$-st, $\rho$ still holds the registers written in the $i$-th, whose
  values lie in $\Sigma_{\ell\mapsto i}$ and not in general in
  $\Sigma_{\ell\mapsto i+1}$. The invariant is re-established because
  Condition~\ref{episodic:reg} forbids reading such a register in the $i+1$-st
  iteration unless it is written there first, and a fresh write returns it to
  $\Sigma_{\ell\mapsto i+1}$ by the cases above. Here
  Condition~\ref{episodic:reg} is to be read with respect to the iterations
  drawn by $\iter$ and not to the syntactic loop body, which as
  Section~\ref{s:essem} discusses need not coincide with them; in RCU the writes
  to $\code{rcu}[\code{tid}]$ lie precisely in the offset between the two.

  \medskip\noindent For an event $e$ of the $i$-th iteration, $\loc(e)$ and
  $\val(e)$ are $\evalreg{\expr}{\rho}$ for expressions $\expr$ over registers
  read at $e$, so the invariant gives
  $\symbols(\loc(e))\cup\symbols(\val(e))\subseteq\Sigma_{\ell\mapsto i}$, which
  is the claim.

\end{proof}

Lemma~\ref{l:episodic-symb} above is stated of an event of $\mathbb E$. The
lemmas below apply it to an event justified in a justification, which by
Definition~\ref{def:justs} is a copy carrying its own location and value
expressions, and the bound has to be carried across that distinction. For a
write these are the location written and the value written; for an allocation,
the symbol it introduces and its size expression; for a deallocation, the
expression naming the location it frees, a deallocation carrying no value, so
that $\symbols(\val(\cdot))$ is empty there.

\begin{corollary}\label{c:episodic-symb-copy}

  Let $j\colon(P,D)\justifies^\delta w$ be a justification generated per
  Definition~\ref{def:gen-just} of an effectful event of the $i$-th iteration
  of an episodic loop $\ell$. Then the bound of Lemma~\ref{l:episodic-symb}
  holds of the copy $w$ records,
  \begin{equation}
    \symbols(\loc(w))\cup\symbols(\val(w))~\subseteq~\Sigma_{\ell\mapsto i}
  \end{equation}
  \noindent unless $\elab{fwd}$ has been applied along a load forwarding edge
  between two iterations.

\end{corollary}

\begin{proof}

  By induction over Definition~\ref{def:gen-just} of the generation of
  justifications. A pre-justification records the expressions of the event
  itself, where Lemma~\ref{l:episodic-symb} applies directly. Of the
  elaborations only $\elab{va}$ and $\elab{fwd}$ rewrite the expressions of the
  copy: $\elab{str}$, $\elab{we}$ and $\elab{weak}$ justify the same event as
  their premise, and $\elab{lift}$ passes on the event of $j_2$ unchanged. The
  two that do rewrite are stated on a justified write, so for an allocation or a
  deallocation the pre-justification case is the whole argument.

  $\elab{va}$ substitutes a value $v\in\Val$ for a symbol, removing symbols
  without introducing any. $\elab{fwd}$ substitutes $\val(e_1)$ for the symbol
  $\val(e_2)$ along an edge $e_1\fwdrel[j_1]e_2$, and both are expressions of
  events. Where $e_1$ and $e_2$ lie in one iteration,
  Lemma~\ref{l:episodic-symb} puts both over $\Sigma_{\ell\mapsto i}$. Where
  they span a boundary, Appendix~\ref{s:app-proofs-fwd} excludes store
  forwarding outright and fixes the values of store-store forwarding over
  $\Sigma_{\ell\mapsto \emptyset}$, which is contained in every
  $\Sigma_{\ell\mapsto i}$. Load forwarding is the exception claimed, and is the
  one Lemma~\ref{l:restrict-elabs} already carries.

\end{proof}

In contrast to the location and value expressions, the predicate in
justifications is not iteration-pure, in that it may use symbols read in a prior
iteration of the loop. Responsible are Strengthening and Lifting, which can
relate symbols of different iterations. Remarks~\ref{r:str-earlier}
and~\ref{r:lift-earlier} below consider their effect on the ordering across
iterations.

\begin{remark}[Strengthening across iterations]\label{r:str-earlier}

  Let $\ell$ be an episodic loop, and let
  $j\colon(P',D)\justifies^\delta w$ arise by
  Definition~\ref{def:elab-str} of Strengthening from
  $j_1\colon(P,D)\justifies^\delta w$, with $S$ the origins it newly constrains.
  Let $e\in S$ be read in an earlier iteration of $\ell$ than $w$, that is
  $\ell\in\loopfun(e)\cap\loopfun(w)$ and $\iter(e)(\ell)<\iter(w)(\ell)$.
  In an execution $(X,J,\rf)$ with $j\in J$, Definition~\ref{def:freeze} of
  freezing adds the edge $(e,w)$ to $\DP$. The strengthening only adds an
  ordering which is already there by Condition~\ref{episodic:events} of
  Definition~\ref{def:episodic}.

  Let $J_1$ be $J$ with $j_1$ in place of $j$. Then $(X,J_1,\rf)$ is again an
  execution: $j_1$ justifies the same event with the same dependency set, so
  every effectful event of $X$ remains uniquely justified; $P'\implies P$, so
  consistency with the value restrictions of $X$ and
  Condition~(3) of freezing are preserved; the relation $\DP_1$ frozen from
  $J_1$ is contained in $\DP$, so Condition~(1) is preserved; and $\rf$ is
  unchanged. Condition~\ref{episodic:events} applied to $(X,J_1,\rf)$ gives
  \[
    (e,w)~\in~{(\ppoord_1\cup\DP_1)}^+
  \]
  \noindent where $\ppo_1$ is frozen from $J_1$. So $e$ precedes $w$ in the
  execution without the strengthening, and the edge $(e,w)$ it contributes
  orders nothing that execution leaves unordered.

  The argument appeals to Condition~\ref{episodic:events} on $(X,J_1,\rf)$ and
  not on $(X,J,\rf)$. On the latter it would be circular: the path witnessing
  the condition for $(e,w)$ may be the edge the strengthening adds.
  It covers origins read in an earlier iteration only. Origins read before the
  loop, for which $\iter(e)(\ell)$ is undefined, are not ordered by
  Condition~\ref{episodic:events}, and origins read in the iteration of $w$ lie
  in $\Sigma_{\ell\mapsto\iter(w)(\ell)}$ already.

  Definition~\ref{def:elab-str} also admits $w\sqsubseteq e$, and so an origin
  read in a \emph{later} iteration of $\ell$ than $w$. Such a strengthening is
  never used in a consistent execution. Condition~\ref{episodic:events} orders
  $w$ before $e$ in ${(\ppoord\cup\DP)}^+$, while freezing the strengthened
  justification puts $(e,w)$ in $\DP$, and the two close a cycle in
  $\DP\cup\ppoord\cup\rf$, which Axiom~\ref{def:mem-model-axiom:nta} forbids.
  The side condition $w\not\ppo^P_\delta e$ of Definition~\ref{def:elab-str}
  does not already exclude it, as it forbids a $\ppo$-edge and not a path.

\end{remark}

\begin{remark}[Lifting across iterations]\label{r:lift-earlier}

  Definition~\ref{def:elab-lift} of Lifting places no condition on the
  iterations in which $w_1$ and $w_2$ lie, and so admits lifting across the
  boundaries of iterations of a loop $\ell$. Condition~\ref{episodic:events} of
  Definition~\ref{def:episodic} does not exclude this, as it constrains the
  executions of a program and not the generation of its justifications. It
  does, however, prevent such a lifting from reordering events across a
  boundary.

  Let $\ell$ be an episodic loop and $(X,J,\rf)$ an execution with a
  justification $j\in J$ of $w$, generated with a lifting among its
  elaborations, which drops the dependency of $w$ on an origin $e$ read in an
  earlier iteration of $\ell$ than $w$, that is
  $\ell\in\loopfun(e)\cap\loopfun(w)$, $\iter(e)(\ell)<\iter(w)(\ell)$ and
  $(e,w)\notin\DP$. Condition~\ref{episodic:events} applied to $(X,J,\rf)$
  gives
  \[
    (e,w)~\in~{(\ppoord\cup\DP)}^+
  \]
  \noindent so another path orders $e$ before $w$. The lifting removes the edge
  $(e,w)$ from $\DP$, but not the ordering.

  Unlike Remark~\ref{r:str-earlier}, the argument appeals to
  Condition~\ref{episodic:events} on the execution with the elaboration rather
  than without it. It is not circular, as the ordering in question is the one
  of the execution containing the lifting, which is what the condition
  constrains.
  Lifting where Condition~\ref{episodic:events} is silent can still remove
  dependencies: of a write outside $\ell$, for which $\iter(w)(\ell)$ is
  undefined, as when the RCU exit after the loop, reached when the \cas{}
  succeeds, is lifted with the RCU exit reached when it fails. It also changes
  $\DP$ itself, and not only its closure with $\ppo$.
  Condition~\ref{episodic:events} constrains ordering and not predicates. It
  does not extend purity to lifted predicates, whose disjunction relates
  symbols of several iterations.

\end{remark}

\noindent Note that the predicate $\psi_\delta$ of Definition~\ref{def:fwd-ctx}
of the forwarding context is not pure in the iterations of $\ell$ either, and
for a boundary-crossing pair is genuinely not iteration-pure. It is not part of
$P$ during the generation of justifications; Definition~\ref{def:freeze} of
freezing justifications conjoins it only at freezing.

The symbols of events are bounded by iteration; those of predicates, as above,
are not. Instead the proofs below use the weaker property that the restricted
predicates suffice.

\begin{lemma}\label{l:restrict-props}

  For any two events $e_1$ and $e_2$ with $i=\iter(e_1)(\ell)$ if
  $\iter(e_1)(\ell)=\iter(e_2)(\ell)$ in an episodic loop $\ell$ and $i=\emptyset$
  otherwise, and for all expressions $\expr$, $\expr_1$, $\expr_2$ over symbols
  in $\Sigma_{\ell\mapsto i}$.

  \begin{equation}
    \begin{array}{rcl}
      \loc(e_1)\equiv_P\loc(e_2) & \text{iff} &
      \loc(e_1)\equiv_\restrict{P}{i}\loc(e_2) \\
      \val(e_1)\equiv_P\val(e_2) & \text{iff} &
      \val(e_1)\equiv_\restrict{P}{i}\val(e_2) \\
      \exists f.\evalenv{P\wedge\loc(e_1)=\loc(e_2)}{f}\equiv\top
      & \text{iff} &
      \exists
      f.\evalenv{\restrict{P}{i}\wedge\loc(e_1)=\loc(e_2)}{f}\equiv\top
      \\
      \expr\equiv_P\top & \text{iff} & \expr\equiv_\restrict{P}{i}\top \\
      \exists\expr.\left(
      \begin{array}{l}
        \evalenv{P_1\Rightarrow\expr_1=\expr}{\Lambda} \\
        \wedge(P_2\Rightarrow\expr_2=\expr)
      \end{array}
      \equiv_{\hyperref[def:fwd-ctx]{\psi_\delta}}\top
    \right)
      & \text{iff} &
    \exists\expr.\left(
      \begin{array}{l}
        \evalenv{\restrict{P_1}{i}\Rightarrow\expr_1=\expr}{\Lambda} \\
        \wedge(\restrict{P_2}{i}\Rightarrow\expr_2=\expr)
      \end{array}
      \equiv_{\hyperref[def:fwd-ctx]{\psi_\delta}}\top
    \right)
    \end{array}
  \end{equation}

\end{lemma}

\begin{proof}

  By Lemma~\ref{l:episodic-symb} the location, value and branching condition
  expressions contain only symbols from $\Sigma_{\ell\mapsto i}$, so that in each of the
  equivalences above the conjuncts of $P$ discarded by $\restrict{P}{i}$ share
  no symbol with either side.

\end{proof}

Each of the properties above is of the following structure: Given an
implication $P\implies p$, where $p$ is a predicate over symbols $\Sigma$, then
there is an interpolant $P\upharpoonright_\Sigma$ with $P\implies
P\upharpoonright_\Sigma\implies p$. Given that all predicates are defined in a
propositional logic, and $P$ and $\Sigma$ are finite, the interpolant exists and
is decidable, making the above an instance of the Craig interpolation
theorem~\cite{harrison2009handbook}.

\begin{corollary}
  Given a finite predicate $P$, $P\upharpoonright_\Sigma$ for a finite $\Sigma$
  is finite and finitely decidable.
\end{corollary}

Two properties follow from the restriction being the strongest consequence over
its vocabulary, and are used for the elaborations that form a disjunction or
rename symbols. It distributes over disjunction,
\[
  \restrict{(P_1\vee P_2)}{i}~\equiv~\restrict{P_1}{i}\vee\restrict{P_2}{i}
\]
\noindent as $\restrict{P_1}{i}\vee\restrict{P_2}{i}$ is a predicate over
$\Sigma_{\ell\mapsto i}$ implied by $P_1\vee P_2$, and conversely every such
predicate implied by $P_1\vee P_2$ is implied by each disjunct, hence by
$\restrict{P_1}{i}$ and by $\restrict{P_2}{i}$. And it commutes with a
relabelling $\Lambda$ that maps $\Sigma_{\ell\mapsto i}$ onto itself,
\[
  \restrict{\evalenv{P}{\Lambda}}{i}~\equiv~\evalenv{\restrict{P}{i}}{\Lambda}
\]
\noindent as $\Lambda$ is one-to-one on symbols, so a predicate $p$ over
$\Sigma_{\ell\mapsto i}$ is implied by $\evalenv{P}{\Lambda}$ exactly when
$\evalenv{p}{\Lambda^{-1}}$, again over $\Sigma_{\ell\mapsto i}$, is implied by
$P$. Neither uses the shape of the predicates.

Firstly, we establish that the restricted predicates suffice to support
elaborations. Forwarding requires care, as load forwarding across loop
boundaries can leak symbols into the next iteration of the loop. We treat it in
Appendix~\ref{s:app-proofs-fwd} and use the result here. The exception is
legitimised using the $\rf$-relation when freezing the justification sets as
symbolic executions later on: by Observation~\ref{o:elab-fwd-value-eq} the two
symbols a load forwarding identifies are already equated by $\varphi_\rf$
whenever the two reads take their value from a common write, so the load
forwarding shape introduces no dependency that a $\rf$ assignment could not.

\begin{lemma}\label{l:restrict-elabs}

  For any elaboration $\elab{}$, excepting load forwarding across a loop
  boundary, justifications $j_1$ and $j_2$ of memory-effectful events $w_1$ and
  $w_2$ in the $i$-th iteration of the episodic loop $\ell$ the following
  commutes.

  \begin{figure}[htbp]
  \center
  \begin{tikzcd}
    \left\{
      j_k\colon(P_k,D_k)\justifies^{\delta_k} w_k
    \right\}_{k\in\{1,2\}}
    \ar[r, "\elab{}"]
    \ar[d, "\restrict{}{i}"]
    &
    j\colon(P,D)\justifies^{\delta}w
    \ar[d, "\restrict{}{i}"]
    \\
    \left\{
      j_k\colon(\restrict{P_k}{i},D_k)\justifies^{\delta_k} w_k
    \right\}_{k\in\{1,2\}}
    \ar[r, "\elab{}"]
    &
    j\colon(\restrict{P}{i},D)\justifies^{\delta}w
  \end{tikzcd}
    \caption{Predicate restrictions commute with
    elaborations}\label{fig:restrict-elabs}
  \end{figure}

\end{lemma}

\begin{proof}

  The proof follows the elaborations in Section~\ref{s:app-defs-gen-justs}.

  Two of the cases arise for writes only. Value Assignment and Forwarding are
  stated on a justified write, rewriting its location and value expressions, and
  by Definition~\ref{def:gen-just} do not apply to justifications of
  allocations or deallocations at all; the arguments below accordingly speak of
  $\loc(w)$ and $\val(w)$. The remaining four are generic in the justified
  event: Strengthening argues on the $\DP$-edges the restriction discards,
  Lifting on the disjunction, Weakening on the removed conjunct, and Write
  Elision constrains only the pair of writes it elides and not the event being
  justified, so none of them turns on $w$ being a write.

  \medskip \hyperref[def:elab-va]{\textbf{Value assignments}} $\elab{va}$
  applies constraints $\alpha=v$ implied by $P$ replacing symbols $\alpha$ with
  values $v$. Diagram~\ref{fig:restrict-elabs} fixes the dependency set and the
  write across its vertical arrows, so the two rows must agree not only on the
  restricted predicate but on the value substituted, the write of the elaborated
  justification carrying $x'=\eval{x}{\alpha}{v}$ and
  $\expr'=\eval{\expr}{\alpha}{v}$ and its dependency set being read off those.
  We distinguish two cases on $\alpha$.

  If $\alpha\in\Sigma_{\ell\mapsto i}$, then unfolding $\alpha\equiv_Pv$ by
  Definition~\ref{def:sem-equiv} and instantiating Lemma~\ref{l:restrict-props}
  at the expression $\alpha=v$, which is over $\Sigma_{\ell\mapsto i}$ as
  $v\in\Val$ carries no symbols, gives $\alpha\equiv_Pv$ iff
  $\alpha\equiv_\restrict{P}{i}v$. Both rows therefore substitute the same value
  for the same symbol, so $x'$, $\expr'$ and $D$ agree, and both rows leave the
  predicate unchanged.

  If $\alpha\notin\Sigma_{\ell\mapsto i}$, the predicate is again unchanged in
  both rows, as $\elab{va}$ does not modify it. Here $\loc(w)$ and $\val(w)$ are
  the expressions of the copy the justification records, which
  Corollary~\ref{c:episodic-symb-copy} bounds by $\Sigma_{\ell\mapsto i}$ save
  under the load forwarding of the case below, so $\alpha$ occurs in neither of
  them and $x'=x$, $\expr'=\expr$ and $D$ is unchanged. The bottom row is a
  no-op and the diagram commutes for that reason.

  \medskip \hyperref[def:elab-fwd]{\textbf{Forwarding}} $\elab{fwd}$ modifies
  the predicate $P_1$ by a map $g:\val(e_2)\mapsto\val(e_1)$ for events
  $e_1\xrightarrow{\fwdrel{j_1}}e_2$, and correspondingly the location and value
  expressions of the write. We distinguish whether the two events lie in one
  iteration, and if not, which of the shapes of Definition~\ref{def:fwd-rels} of
  the forwarding relations applies.

  If $e_1$ and $e_2$ lie in the same iteration then
  $\loc(e_1)\equiv_{P\wedge\psi_\delta}\loc(e_2)$ iff
  $\loc(e_1)\equiv_{\restrict{P}{i}}\loc(e_2)$, as both location expressions use
  only symbols from that iteration or from before the loop. The map $g$ likewise
  substitutes an expression over $\Sigma_{\ell\mapsto i}$ for a symbol of
  $\Sigma_{\ell\mapsto i}$, so that
  $\restrict{\evalenv{P}{g}}{i}=\evalenv{\restrict{P}{i}}{g}$ and
  Diagram~\ref{fig:restrict-elabs} commutes.

  Suppose then that $e_1$ lies in the $k$-th iteration and $e_2$ in the
  $k+1$-st, and recall Observation~\ref{o:fwd-crosses-boundaries} that the
  episodicity criteria do not by themselves preclude this.

  For store forwarding the case does not arise. Forwarding from a write is the
  constraint $\varphi_\rf$ that of Definition~\ref{def:freeze} of freezing
  justifications imposes when $e_2$ reads from $e_1$, and
  Condition~\ref{episodic:mem} of Definition~\ref{def:episodic} of episodic
  loops admits no read from a write of the preceding iteration on the same
  thread.

  For store-store forwarding the values $\val(e_1)$ and $\val(e_2)$ are equated
  by the predicate $\varphi$ of Definition~\ref{def:freeze}, and by the
  interpolation argument of Appendix~\ref{s:app-proofs-fwd} both are then fixed
  by symbols of $\Sigma_{\ell\mapsto \emptyset}$ or are constants. The map $g$
  therefore substitutes within $\Sigma_{\ell\mapsto \emptyset}$, which is
  contained in every $\Sigma_{\ell\mapsto i}$, and the diagram commutes as in
  the same-iteration case.

  Load forwarding is the exception noted in the statement of the lemma. Here $g$
  replaces the symbol $\val(e_2)\in\Sigma_{\ell\mapsto k+1}$ by
  $\val(e_1)\in\Sigma_{\ell\mapsto k}$, and the diagram does not commute:
  restricting to $\Sigma_{\ell\mapsto k+1}$ after the substitution discards the
  conjuncts that mention $\val(e_1)$, whereas substituting after the restriction
  retains them, and symmetrically at $\Sigma_{\ell\mapsto k}$. This is not a
  defect of the argument but of the shape: the substitution transports precisely
  the information across the loop boundary that the restriction is there to
  discard, and no restriction of a single window can be insensitive to it.

  The exception does not propagate. What Lemma~\ref{l:gamma-pres-justs} requires
  is that the restrictions of $P'$ and $P$ at corresponding windows of $\mathbb
  E_{n+1}$ and $\mathbb E_n$ agree. Those windows carry identical symbols by
  Property~\ref{prop:de-bruijn-tail} of Corollary~\ref{l:de-bruijn-props}, and
  the substitution is the same map on both sides. The argument is given in the
  forwarding case of the proof of Lemma~\ref{l:gamma-pres-justs}, and rests on
  the purity of $P$ in the iterations of $\ell$. A pre-justification's predicate
  is pure: it is $\valres(w)$, the conjunction of the branching conditions
  accumulated along $\po$ up to $w$ per Definition~\ref{def:gen-es}, and each
  conjunct is $\evalreg{b}{\rho}$ for a branching event, whose symbols lie in
  $\Sigma_{\ell\mapsto k}$ for the iteration $k$ of that event by
  Lemma~\ref{l:episodic-symb}, so every conjunct is pure in one iteration or in
  $\Sigma_{\ell\mapsto\emptyset}$. Forwarding preserves purity: it applies
  $g=[\val(e_2)\mapsto\val(e_1)]$, which within an iteration maps a conjunct
  pure in it to another, and for $e_1$ in the $k$-th iteration and $e_2$ in the
  $k+1$-st maps a conjunct pure in the $k+1$-st to one pure in the $k$-th,
  moving it between factors rather than mixing them.

  \medskip \hyperref[def:elab-we]{\textbf{Write-elision}} $\elab{we}$ does not
  modify the predicate, making Diagram~\ref{fig:restrict-elabs} commute
  trivially.

  \medskip \hyperref[def:elab-lift]{\textbf{Lifting}} $\elab{lift}$ forms
  $\evalenv{P_1}{\Lambda}\vee P_2$, and passes on the write $w_2$ and the
  dependency set $D_2$ of its second premise unchanged, so the two rows agree on
  those. For the predicate, let $\Lambda$ map $\Sigma_{\ell\mapsto i}$ onto
  itself. By the two properties of restriction noted after
  Lemma~\ref{l:restrict-props}, \[ \restrict{\left(\evalenv{P_1}{\Lambda}\vee
  P_2\right)}{i}~\equiv~
  \evalenv{\restrict{P_1}{i}}{\Lambda}\vee\restrict{P_2}{i} \] \noindent which
  is the predicate the bottom row forms from $\restrict{P_1}{i}$ and
  $\restrict{P_2}{i}$. The side conditions of Definition~\ref{def:elab-lift},
  closed relabel-equivalence of the two writes and of the origins of $D_1$,
  carry to the restricted predicates by Lemma~\ref{l:restrict-relabeq}, so the
  bottom row is again a lifting and Diagram~\ref{fig:restrict-elabs} commutes.

  Neither step looks inside the predicates: unlike the other cases, Lifting
  needs no separation of $P_1$ and $P_2$ by iteration, and the restriction
  commutes with the disjunction whether or not the two premises diverge at
  branches of several iterations.

  The hypothesis on $\Lambda$ holds where the lifted writes and the branch they
  follow lie in one iteration of $\ell$, as the symbols $\Lambda$ relates are
  then read in that iteration by Lemma~\ref{l:episodic-symb}.

  \medskip \hyperref[def:elab-str]{\textbf{Strengthening}} $\elab{str}$ takes
  $j_1\colon(P_1,D)\justifies^\delta w$ to $j\colon(P,D)\justifies^\delta w$,
  where $P$ is closed as $P=P\wedge P_1\wedge\bigwedge\limits_{e\in S}v(e)$ for
  the origins $S=\origin{\symbols(P)}\setminus\origin{\symbols(P_1)}$ it newly
  constrains. The write and the dependency set are unchanged, so the two rows
  agree on those, and the bottom row is again a strengthening. Its side
  conditions are on $S$, on $\remap_\delta$ and on the side condition
  $w\not\ppo^P_\delta e$ of Definition~\ref{def:elab-str}, and that condition
  carries to $\restrict{P_1}{i}$ by cases on the origin $e$. Where $e$ is read
  in the iteration of $w$, Claim~\ref{restrict-ppo:same} of
  Lemma~\ref{l:restrict-ppo} gives the equivalence. Where $e$ is read in an
  earlier iteration, $w\not\ppo_\delta e$ holds under either predicate, as
  $\ppo$ refines $\sqsubseteq$ and $e$ does not follow $w$. Where $e$ is read in
  a later iteration, Remark~\ref{r:str-earlier} shows the strengthened
  justification is not used in a consistent execution. The predicate the bottom
  row conjoins is the part of $P$ over $\Sigma_{\ell\mapsto i}$.

  What the two rows need not share is the rest of $P$. Restriction keeps the
  value restrictions $v(e)$ of origins $e\in S$ read in the iteration of $w$,
  whose symbols lie in $\Sigma_{\ell\mapsto i}$, and discards those of origins
  read elsewhere, and with them the edges $(e,w)$ that those origins contribute
  to $\DP$ on freezing. We argue that no ordering is lost, using that
  strengthening only adds $\DP$-edges: $P$ retains $P_1$ as a conjunct, so
  $\symbols(P)\supseteq\symbols(P_1)$ and freezing adds the edges from $S$ and
  removes none.

  For an origin $e$ read in an earlier iteration of $\ell$ than $w$,
  Remark~\ref{r:str-earlier} gives $(e,w)\in{(\ppoord\cup\DP)}^+$ in the
  execution without the strengthening, so the edge orders nothing that execution
  leaves unordered. For an origin read in a later iteration, the same remark
  shows that the strengthened justification is not used in any consistent
  execution. Origins read before the loop lie in $\Sigma_{\ell\mapsto i}$, so
  their conjuncts are retained.

  \medskip \hyperref[def:elab-weak]{\textbf{Weakening}} $\elab{weak}$ takes
  $j_1\colon(P'\wedge P,D)\justifies^\delta w$ to
  $j\colon(P',D)\justifies^\delta w$, where the global guarantees imply the
  conjunct it removes, $\Omega\implies P$. The write and the dependency set are
  unchanged, so the two rows agree on those. Restriction is monotone, so
  $\restrict{(P'\wedge P)}{i}\implies\restrict{P'}{i}$, and the bottom row is
  again a weakening as $\restrict{(P'\wedge P)}{i}$ is $\restrict{P'}{i}$
  conjoined with something the global guarantees imply. Two cases give that.

  If the removed conjunct is over $\Sigma_{\ell\mapsto i}$, then
  $\restrict{(P'\wedge P)}{i}\equiv\restrict{P'}{i}\wedge P$: a predicate $p$
  over $\Sigma_{\ell\mapsto i}$ is implied by $P'\wedge P$ exactly when
  $P\implies p$ is implied by $P'$, and $P\implies p$ is again over
  $\Sigma_{\ell\mapsto i}$. The bottom row removes the same conjunct $P$, which
  $\Omega$ implies.

  If the removed conjunct shares no symbol with $P'$, then $\restrict{(P'\wedge
  P)}{i}\equiv\restrict{P'}{i}\wedge\restrict{P}{i}$, as the two restrictions
  constrain disjoint vocabularies, and the bottom row removes $\restrict{P}{i}$,
  which $\Omega$ implies through $P$.

\end{proof}

Using the above properties we show that the restricted predicate
$\restrict{P}{i}$ suffices for the preserved program order $\ppo^P_\delta$
between events in the same iteration of the loop.

\begin{lemma}\label{l:restrict-ppo}

  Let $e_1$ and $e_2$ be events in the same loop $\ell$,
  $\ell\in\loopfun(e_1)\cap\loopfun(e_2)$.

  \begin{enumerate}

    \item\label{restrict-ppo:same} Where the two events lie in one iteration,
      $i=\iter(e_1)(\ell)=\iter(e_2)(\ell)$, then $e_1\ppo^P_\delta e_2$ iff
      $e_1\ppo^\restrict{P}{i}_\delta e_2$.

    \item\label{restrict-ppo:other} For any $i$, $\pposync$ is the same relation
      under $P$ and under $\restrict{P}{i}$, every pair of $\ppoalias^P$ is one
      of $\ppoalias^\restrict{P}{i}$, and every pair of
      $\ppormw^\restrict{P}{i}$ is one of $\ppormw^P$.

  \end{enumerate}

\end{lemma}

\begin{proof}

  The proof proceeds over the Definition~\ref{def:ppo} of $\ppo$. Only
  $\ppormw$ and $\ppoalias$ are parametric in $P$, and $\remap$ is not.

  Claim~\ref{restrict-ppo:same}. That $e_1\ppormw^P e_2$ iff
  $e_1\ppormw^\restrict{P}{i}e_2$ follows from $c\equiv_P\top$ iff
  $c\equiv_\restrict{P}{i}\top$ in Lemma~\ref{l:restrict-props} and $c$ being
  generated from expressions over symbols from $\Sigma_{\ell\mapsto i}$ only in
  the semantics of $\cas$ and $\fadd$. That $e_1\ppoalias^Pe_2$ iff
  $e_1\ppoalias^\restrict{P}{i}e_2$ from Lemma~\ref{l:restrict-props} alone.

  Claim~\ref{restrict-ppo:other} holds irrespective of the two events.
  $\pposync$ reads memory order annotations and not the predicate.
  $P\implies\restrict{P}{i}$, so a satisfiable $P\wedge\loc(e_1)=\loc(e_2)$ is a
  satisfiable $\restrict{P}{i}\wedge\loc(e_1)=\loc(e_2)$, which is $\ppoalias$,
  and a condition $c$ entailed by $\restrict{P}{i}$ is entailed by $P$, which is
  $\ppormw$. The two inclusions run in opposite directions, which is why
  Claim~\ref{restrict-ppo:same} cannot be had for a pair whose events lie in
  different iterations: there the restriction weakens the predicate, and the two
  relations move apart rather than together.

\end{proof}

\begin{remark}[Nested loops]\label{r:restrict-nested}

  The hypothesis $\ell\in\loopfun(e_1)\cap\loopfun(e_2)$ holds in the presence
  of nested loops with restrictions as follows:
  Let $\ell\prec\ell'$ with $\ell$ episodic, let $e_1$ lie in the $k$-th
  iteration of $\ell$ but before the inner loop $\ell'$, and let $e_2$ lie in
  the $m$-th iteration of $\ell'$ within that same iteration of $\ell$.

  Then $\ell$ in Lemma~\ref{l:restrict-ppo} means the outer loop
  $\ell\in\loopfun(e_1)\cap\loopfun(e_2)$ and
  $\iter(e_1)(\ell)=\iter(e_2)(\ell)=k$, so $i=k$. By
  Condition~\ref{iter:nesting} no iteration of $\ell'$ straddles a boundary of
  $\ell$, so every symbol $e_2$ reads inside $\ell'$ has
  $\iter(\origin\alpha)(\ell)=k$ and lies in $\Sigma_{\ell\mapsto k}$, which is
  what Lemma~\ref{l:restrict-props} asks of the expressions it compares.

  Read of the inner loop the hypothesis fails, as
  $\ell'\notin\loopfun(e_1)$ and $\iter(e_1)(\ell')$ is undefined. Falling
  through to $i=\emptyset$ would not recover the statement:
  $\restrictloop{P}{\ell'}{\emptyset}$ discards every conjunct mentioning a
  symbol read inside $\ell'$, among them the symbols of $\loc(e_2)$ and
  $\val(e_2)$, so the expressions compared are not over
  $\Sigma_{\ell'\mapsto\emptyset}$ and Lemma~\ref{l:restrict-props} does not
  apply. Weakening $P$ only adds satisfying assignments, so $\ppoalias$ is
  preserved but need not be reflected: where $P$ pins $\loc(e_2)$ to a symbol
  read in $\ell'$ and $\loc(e_1)$ to another location,
  $\loc(e_1)=\loc(e_2)$ is unsatisfiable with $P$ and satisfiable with
  $\restrictloop{P}{\ell'}{\emptyset}$.

  The same reading applies to a boundary-crossing pair of a single loop, with
  $\iter(e_1)(\ell)=k$ and $\iter(e_2)(\ell)=k+1$. Such a pair gets
  Claim~\ref{restrict-ppo:other} of Lemma~\ref{l:restrict-ppo} and no
  equivalence, and this is why: a restriction that holds the factor of neither
  iteration weakens the predicate, and $\ppoalias$ and $\ppormw$ answer a
  weakening in opposite directions, the first gaining pairs and the second
  losing them.
\end{remark}

Using the above result, the restricted predicate $\restrict{P}{i}$ is sufficient
to establish the predecessor $\pred_\delta(\cdot, P)$ relation between events in
the same iteration of the loop.

\begin{lemma}\label{l:restrict-pred}

  For any two events $e_1$ and $e_2$,
  \[
    e_1\in\pred_\delta(e_2, P)
    \quad\text{iff}\quad
    e_1\in\pred_\delta(e_2, \restrict{P}{i}),
  \]
  where $\ell$, as in Lemma~\ref{l:restrict-ppo}, is an episodic loop with
  $\ell\in\loopfun(e_1)\cap\loopfun(e_2)$ and
  $i=\iter(e_1)(\ell)=\iter(e_2)(\ell)$.

\end{lemma}

\begin{proof}

  An event $\ppo$-between $e_1$ and $e_2$ lies between them in $\po$, as $\ppo$
  refines $\po$, and so in the same iteration $i$; the three events are
  therefore covered by Claim~\ref{restrict-ppo:same} of
  Lemma~\ref{l:restrict-ppo} alike.

  Let $e_1\in\pred_\delta(e_2, P)$, then $e_1\ppo^P_\delta e_2$ by definition of
  $\pred$, and thus $e_1\ppo^\restrict{P}{i}_\delta e_2$ by
  Lemma~\ref{l:restrict-ppo}. Let there be an event $e$ with
  $e_1\ppo^\restrict{P}{i}_\delta e$ and $e\ppo^\restrict{P}{i}_\delta e_2$.
  Thus $e_1\ppo^P_\delta e$ and $e\ppo^P_\delta e_2$, so that $e_1=e$ or $e=e_2$
  by definition of $\pred$. The converse direction follows analogously.

\end{proof}

\begin{lemma}\label{l:restrict-relabeq}

  Let $e_1$ and $e_2$ be events with justifications
  $j_1\colon(P_1,D_1)\justifies^\delta e_1$ and
  $j_2\colon(P_2,D_2)\justifies^\delta e_2$, let $\ell$ be an episodic loop with
  $\ell\in\loopfun(e_1)\cap\loopfun(e_2)$, and let
  $i=\iter(e_1)(\ell)=\iter(e_2)(\ell)$, then

  \begin{equation}
    \begin{array}{rcl}
      P_1\colon e_1\relabeq\Lambda\delta P_2\colon e_2
      & \text{iff} &
      \restrict{P_1}{i}\colon e_1\relabeq\Lambda\delta\restrict{P_2}{i}\colon e_2 \\
      P_1\colon e_1\relabeq\Lambda\delta^* P_2\colon e_2
      & \text{iff} &
      \restrict{P_1}{i}\colon e_1\relabeq\Lambda\delta^*\restrict{P_2}{i}\colon e_2 \\
    \end{array}
  \end{equation}

\end{lemma}

\begin{proof}

  As $e_1$ and $e_2$ are the writes the two justifications record,
  Corollary~\ref{c:episodic-symb-copy} gives $\symbols(\loc(e_1)),
  \symbols(\loc(e_2)), \symbols(\val(e_1)),
  \symbols(\val(e_2)) \subseteq \Sigma_{\ell\mapsto i}$, so that $\exists\expr.\left(
  \evalenv{P_1\Rightarrow\expr_1=\expr}{\Lambda}
  \wedge(P_2\Rightarrow\expr_2=\expr)
  \equiv_{\hyperref[def:fwd-ctx]{\psi_\delta}}\top \right)$ iff \\
  $\exists\expr.\left(
  \evalenv{\restrict{P_1}{i}\Rightarrow\expr_1=\expr}{\Lambda}
  \wedge(\restrict{P_2}{i}\Rightarrow\expr_2=\expr)
  \equiv_{\hyperref[def:fwd-ctx]{\psi_\delta}}\top \right)$ by
  Lemma~\ref{l:restrict-props} for $\expr_1=\loc(e_1)$ and $\expr_2=\loc(e_2)$,
  and $\expr_1=\val(e_1)$ and $\expr_2=\val(e_2)$. Thus
  $P_1,e_1\relabeq\Lambda\delta P_2,e_2$ iff $\restrict{P_1}{i}\colon
  e_1\relabeq\Lambda\delta\restrict{P_2}{i}\colon e_2$. By
  Lemma~\ref{l:restrict-pred}, $e'_1\in\pred_\delta(e_1,P_1)$ iff
  $e'_1\in\pred_\delta(e_1,\restrict{P_1}{i})$ for $e'_1$ with
  $\iter(e'_1)(\ell)=i$ and similar for $e'_2$ and $e_2$. Thus, through an
  inductive argument the previous result generalises to complete relabel
  equivalences, so that $P_1\colon e_1\relabeq\Lambda\delta^* P_2\colon e_2$.

\end{proof}

\subsection{Forwarding across Loop Boundaries}\label{s:app-proofs-fwd}

Lemma~\ref{l:restrict-elabs} above covers the forwarding elaboration
$\elab{fwd}$, excepting load forwarding across a loop boundary, and
Lemma~\ref{l:gamma-pres-justs} below relies on that. This subsection supplies
the argument. We first observe that the episodicity criteria do not by
themselves prevent forwarding across a loop boundary,  and then take the three
shapes of Definition~\ref{def:fwd-rels} of the forwarding relations in turn.

\begin{observation}\label{o:fwd-crosses-boundaries}

  The conditions of Definition~\ref{def:episodic} of episodic loops do not
  preclude a forwarding edge between events of consecutive iterations.

\end{observation}

By Definition~\ref{def:fwd-rels} of the forwarding relations, $e_1\fwdrel[j]e_2$
requires $e_1\in \pred_\delta(e_2,P)$, that is $e_1\ppo^P_\delta e_2$ with no
event $\ppo$-strictly between, together with
$\loc(e_1)\equiv_{P\wedge\psi_\delta} \loc(e_2)$ and a matching pair of event
types. Condition~\ref{episodic:events} of Definition~\ref{def:episodic} of
episodic loops orders events of distinct iterations by ${(\ppoord\cup\DP)}^+$,
that is by the transitive closure; it does not assert that any event lies
between them, and a single $\ppo$ step satisfies it while leaving $\pred$
intact. Conditions~\ref{episodic:reg} and~\ref{episodic:cond} constrain
expressions and branching conditions, and Condition~\ref{episodic:mem}
constrains $\rf$; none of them mentions $\pred$.

What does separate the iterations is an event lying strictly between them, which
is what a synchronisation point in the sense of Definition~\ref{def:sp}
provides. Synchronisation points are, however, sufficient and not necessary for
Condition~\ref{episodic:events}, as noted in Section~\ref{s:sync-points}, so a
loop may be episodic without one.

Note also that of the three shapes of Definition~\ref{def:fwd-rels} of the
forwarding relations, only store-load and store-store forwarding carry a
memory-order side condition. Load forwarding carries none, so an acquire
annotation on the reads does not exclude it.

\medskip\noindent We now take the shapes in turn. In each case the conclusion is
that the forwarding leaks no value across the boundary.

\medskip\noindent\textbf{Store forwarding.}
$\elab{fwd}$ applies the substitution $g=[\val(e_2)\mapsto\val(e_1)]$, replacing
the symbol read by $e_2$ with the value expression written by $e_1$. This is the
constraint that $\varphi_\rf$ of Definition~\ref{def:freeze} imposes when $e_2$
reads from $e_1$, so forwarding from a write is indistinguishable from reading
from it. Condition~\ref{episodic:mem} of Definition~\ref{def:episodic} of
episodic loops admits a read only from a $\po$-earlier write of the same
iteration, from a write before the loop, or from another thread under
case~\ref{episodic-caseb} or~\ref{episodic-casec}. A write of the preceding
iteration on the same thread is none of these, so the shape is excluded.

\medskip\noindent\textbf{Store-store forwarding.}
By Observation~\ref{o:elab-fwd-value-eq} the predicate $\varphi$ of
Definition~\ref{def:freeze} asserts the equality of values. In any execution in
which the shape occurs we therefore have $\val(e_1)\equiv\val(e_2)$. By
Lemma~\ref{l:episodic-symb}, $\val(e_1)$ is an expression over
$\Sigma_{\ell\mapsto i}$ and $\val(e_2)$ over $\Sigma_{\ell\mapsto i+1}$, and
these share only the symbols $\Sigma_{\ell\mapsto \emptyset}$ read before the
loop. An entailed equality between expressions over vocabularies that meet only
in $\Sigma_{\ell\mapsto \emptyset}$ has an interpolant over $\Sigma_{\ell\mapsto
\emptyset}$, by the same Craig interpolation that makes $\restrict{P}{i}$ well
defined in Appendix~\ref{s:app-proofs-restricted}. Both values are therefore
fixed by symbols read before the loop, or are constants, and the substitution
carries no symbol of the preceding iteration into the current one.

\medskip\noindent\textbf{Load forwarding.}
Here $\elab{fwd}$ replaces the symbol read by $e_2$ with the symbol read by
$e_1$. The same identification arises without forwarding whenever $e_1$ and
$e_2$ read from a common write $w$: then $\varphi_\rf$ contributes
$\val(e_1)=\val(w)$ and $\val(e_2)=\val(w)$, and hence $\val(e_1)=\val(e_2)$.
The two are thus indistinguishable in their constraints, and the shape
introduces no dependency that a read-from assignment could not.

Whether these shapes arise at all is a property of the program, and in the
algorithms this paper treats they do not, as Example~\ref{ex:fwd-four-algorithms}
checks.

\begin{example}\label{ex:fwd-four-algorithms}

  In the four algorithms of Section~\ref{s:identifying}, the shapes are excluded as
  follows.

  \begin{itemize}

    \item In RCU (Appendix~\ref{app:rcu}) the \fadd{} at the head of each
      iteration is a synchronisation point by Example~\ref{ex:rcu-sp}, and so
      lies $\ppo$-between the events of one iteration and those of the next. No
      pair across the boundary is $\pred$-adjacent.

    \item In hazard pointers (Appendix~\ref{app:hp}) the memory fence plays the
      same role.

    \item In spinlock (Appendix~\ref{app:spinlock}) the loop body is a single
      $\cas$. A failing $\cas$ contributes no write event, so consecutive
      iterations contribute their read events alone. These reference the same
      literal memory location, which is how
      Condition~\ref{episodic:events} is met, and are therefore
      $\pred$-adjacent: load forwarding applies across the
      boundary. It is harmless by the load-load case above. The shapes
      involving writes do not arise, as a failing iteration has none.

    \item In seqlock (Appendix~\ref{app:seqlock}) the read-side loop contains no
      write events at all, so only load forwarding can arise. The
      sample of the sequence counter closing one iteration and the sample
      opening the next are adjacent, as branching events are excluded from
      executions by Definition~\ref{def:executions}. Again the load-load case
      applies.

  \end{itemize}

  RCU and hazard pointers are thus excluded by a synchronisation point, and
  spinlock and seqlock by the load-load case.

\end{example}

\medskip\noindent The three cases above concern the substitution that
$\elab{fwd}$ performs on the predicate and on the location and value
expressions. $\elab{fwd}$ additionally extends the forwarding context $\delta$
by the pair $(e_1,e_2)$, and by Definition~\ref{def:fwd-ctx} $\remap_\delta$
then identifies $e_2$ with $e_1$, so that $\ppo^P_\delta$ is computed on the
quotient. Forwarding and write elision are the only elaborations that modify
$\delta$.

The effect of a boundary-crossing pair on $\delta$ is accounted for by the
commutation of the elaborations with the value restrictions,
Lemma~\ref{l:restrict-elabs}, which grounds the preservation and reflection of
justifications along $\gamma$ modulo restriction in
Lemma~\ref{l:gamma-pres-justs}. Extending Lemma~\ref{l:restrict-elabs} to
$\elab{fwd}$ leaves a single case. A pair $(e_1,e_2)$ spanning the $k$-th and
$k+1$-st iterations with $k\geq 1$ lies wholly within the domain of $\gamma$,
and is carried to the pair spanning the $k-1$-st and $k$-th, so $\remap$ in the
image mirrors $\remap$ in the domain. Only a pair from the first iteration into
the second has no image, $\gamma$ being undefined on the first. This is the same
residue as in Claim~\ref{gamma-dp:pres} of Lemma~\ref{l:gamma-pres-dp}, where
the $\DP$-edges $\gamma$ does not preserve are exactly those sourced in the
first iteration, and it is discharged in the same way, the posterior futures of
Definition~\ref{def:post-futures} discarding them once the history covers that
iteration.

\subsection{Mapping between Iterations}

In the following we define the partial embedding $\gamma$ from events in an
execution $\mathbb
E_{n+1}=\langle\prog\rangle_{n+1~\emptyset~\lambda\rho\,\varphi.\emptyset~\top}$ to
events in an execution $\mathbb
E_n=\langle\prog\rangle_{n~\emptyset~\lambda\rho\,\varphi.\emptyset~\top}$. $\gamma$ will
restrict an execution $\mathbb X$ to the iterations after the first, mapping the
$i+1$-st iteration of the loop to the $i$-th. $\gamma$ is undefined on the
events of the first iteration, which the de Bruijn indexing places at the start
of $\mathbb E_{n+1}$.

For convenience we define $\gamma$ as a two parted identity, with a first
matching events up until the beginning of the loop and a second matching events
from the $i+1$-st iteration in $\mathbb E_{n+1}$ to events from the $i$-th
iteration in $\mathbb E_n$. $\gamma$ is undefined on the events of the first
iteration in $\mathbb E_{n+1}$, which have no counterpart in $\mathbb E_n$.

From Definition~\ref{def:es-prefix} of event structure prefixing it follows
through an inductive argument over the generation of event structures that
events up to the start of the loop are identical in $\mathbb E_n$ and in
$\mathbb E_{n+1}$.

Recall that event structures $\mathbb E_n$ are defined from register state
$\rho$ inductively from the start of the program and from continuations $\kappa$
recursively from the end of the program. The event structure rooted in an event
$e$ is thus the continuation $\kappa_e$ accumulated from the end of the program
until $e$ applied to the register state $\rho_e$ accumulated from the start of
the program until $e$, that is $\kappa_e(\rho_e)$. It follows from a recursive
argument that $\kappa_e$ in $\mathbb E_n$ is identical to $\kappa_{e'}$ in
$\mathbb E_{n+1}$, where $e$ is the $\po$-least such event in $\mathbb E_n$ and
$e'$ the $\po$-least such event in $\mathbb E_{n+1}$. Using the register
Condition~\ref{episodic:reg} of episodic loops in Definition~\ref{def:episodic},
the register state $\rho_{e'}$ coincides with $\rho_{e}$ on all registers read
in the $i+1$-st iteration in $\mathbb E_{n+1}$ and in the $i$-th iteration in
$\mathbb E_n$, so that $\kappa_{e'}(\rho_{e'})=\kappa_e(\rho_e)$, that is the
event structures rooted in $e'$ for step-counter $n+1$ and $e$ for step-counter
$n$ are identical.

\begin{definition}\label{def:gamma-constr}\label{def:gamma}

  Given an execution $\mathbb X'$ over a set $X'$ of events in an event
  structure $\langle\prog\rangle_{n+1~\emptyset~\lambda\rho\,\varphi.\emptyset~\top}$,
  define $\gamma$ as the identity map on actions in $X'$ under the de
  Bruijn-style indexing of Definition~\ref{def:de-bruijn}, restricted to

  \begin{equation}
    \mathrm{dom}(\gamma)~=~\left\{e\in X'\mid\iter(e)(\ell)\neq 0\right\}
    \label{eq:gamma-dom}
  \end{equation}

  \noindent that is, to all events except those of the first iteration of
  $\ell$. Events before the loop, after the loop, and on threads not executing
  $\ell$ carry no $\iter(\cdot)(\ell)$ and are in the domain, where $\gamma$ is
  the identity.

  The de Bruijn-style indexing of Definition~\ref{def:de-bruijn} is derived
  from control labels, $\iota$ being built from $\iota_0$,
  $\innerloop(\loopfun(e))$ and $\thread(e)$. We therefore read $\gamma$ on
  the copy of a write carried by a justification, as in
  Definition~\ref{def:justs}, by the label that copy carries, $\gamma$ being
  the identity on its action; this is the sense in which $\gamma j$ and $\gamma
  w$ are written below for a justification $j$ of $w$.

\end{definition}

The first iteration is the only obstruction. The step-counter does not
contribute a second one: the de Bruijn indexing aligns the executions at their
\emph{ends}, so the truncation of the deepest unrolling in $\mathbb E_{n+1}$
matches that in $\mathbb E_n$, and the last iteration of $\mathbb E_{n+1}$ maps
to the last iteration of $\mathbb E_n$. For a program with several episodic
loops, $\gamma_\ell$ is undefined on the first iteration of each $\ell$
separately.

\begin{lemma}[$\gamma$ preserves expressions]\label{l:gamma-pres-expr}
  $\gamma$ preserves expressions, that is
  \begin{enumerate}
    \item $\loc(e)=\loc(\gamma e)$
    \item $\val(e)=\val(\gamma e)$

  \end{enumerate}
  \noindent for all events $e$ in the domain of $\gamma$.
\end{lemma}

\begin{proof}

  By the de Bruijn-style indexing of symbols, an event $e$ in the domain of
  $\gamma$ introduces the same symbol as $\gamma e$, so that $\val(e)=\val(\gamma
  e)$ holds for all read events $e$ in the domain of $\gamma$, and likewise
  $\loc(e)=\loc(\gamma e)$ for allocation events. Any location expression
  $\loc(e)$ and value expression $\val(e)$ in an event $e$ in the domain of
  $\gamma$ will only read register values set in the same loop iteration or
  before the loop by Condition~\ref{episodic:reg} of episodic loops in
  Definition~\ref{def:episodic}. The proof then proceeds by induction over the
  derivation of event structures and thus construction of expressions. Using
  that $\gamma$ preserves $\pc$, expressions in values of registers match
  between the domain and the range of $\gamma$.

\end{proof}

\begin{corollary}\label{c:gamma-constr}

  For all events $e$ in the domain of $\gamma$, that is all events with
  $\iter(e)(\ell)\neq 0$ and all events outside the loop,

  \begin{equation}
    \begin{array}{rcll}
      \pc(\gamma e)&=&\pc(e)&\\
      \loopfun(\gamma e)&=&\loopfun(e)&\\
      \iter(\gamma e)(\ell)&=&\iter(e)(\ell)-1& and\\
      \gamma e\in\mathcal T&\iff&e\in\mathcal T&\text{for all event types}~
      \mathcal T\in\{\Reads,\Writes,\Allocs,\Deallocs,\Fences,\Branches\} \\
    \end{array}
    \label{gamma-constr:pc}
  \end{equation}
\end{corollary}

\begin{lemma}\label{l:gamma-restr-execs}

  Let $\mathbb X'~=~(X',J',\rf')$ be an execution in $\mathbb
  E_{n+1}=\langle\prog\rangle_{n+1~\emptyset~\lambda\rho\,\varphi.\emptyset~\top}$.
  Then $\gamma\mathbb X'=(\gamma X',\gamma J',\gamma\rf')$ is itself an
  execution in $\mathbb
  E_n=\langle\prog\rangle_{n~\emptyset~\lambda\rho\,\varphi.\emptyset~\top}$.

\end{lemma}

\begin{proof}

  $\gamma\mathbb X'$ consists of the events of $\mathbb X'$ before the loop,
  together with the events of the iterations after the first, each mapped one
  iteration earlier. It is an execution of $\prog$ under step-counter $n$: it
  has one iteration fewer than $\mathbb X'$, and by
  Condition~\ref{episodic:reg} and Condition~\ref{episodic:mem} of
  Definition~\ref{def:episodic} of episodic loops no register value and no read
  of the iterations after the first depends on the first iteration, so the
  events dropped by $\gamma$ constrain none of the events retained. The value
  restrictions of the retained events are therefore satisfiable without the
  branching conditions of the first iteration, and the $\rf$-assignments of
  $\mathbb X'$ restrict to the retained events by Condition~\ref{episodic:mem},
  which admits no read in a later iteration from a write of the first.

  Note that, unlike the extension of an execution to a further iteration, the
  restriction requires no reachability assumption: the iterations of $\mathbb
  X'$ after the first are given, and $\gamma$ merely re-indexes them.

\end{proof}

This is what is meant by $\gamma$ embedding into the fixed point.
There is no family of maps relating executions of different event structures to
be reconciled: each $\gamma$ of Definition~\ref{def:gamma-constr} of $\gamma$ is
already a map between two executions of $\mathbb E_\infty$, and the embedding of
next enabled actions is the map it induces on their horizons by
Lemma~\ref{l:gamma-pres-post-futures}.

$\gamma$ reflects the $\DP$ relations, and preserves those outside the first
iteration of the loop in its domain. $\DP$ relations are defined by freezing
justifications per Definition~\ref{def:freeze} from the \emph{unrestricted}
predicates $P_j$, which accumulate the branching conditions of every earlier
iteration of the loop. The restricted predicates $\restrict{P}{i}$ of
Appendix~\ref{s:app-proofs-restricted} therefore do not suffice for $\DP$, as
they do for $\ppo$ in Lemma~\ref{l:restrict-ppo}. As the events of the first
iteration in $\mathbb E_{n+1}$ have no image under $\gamma$, neither do the
$\DP$-edges they source, and $\gamma$ does not preserve $\DP$ everywhere.
Lemma~\ref{l:gamma-pres-dp} below makes this precise. It suffices for the
narrowing argument, which uses $\DP$ only through the \emph{posterior} futures
of Definition~\ref{def:post-futures} of the posterior future set: there the
$\DP$-edges outside the domain of $\gamma$ are discarded, because their source
lies in the history.

The following lemma establishes that $\gamma$ preserves justifications in the
sense that

\begin{equation}
  \begin{array}{rcl}
  j'\colon(P',D')\justifies^{\delta'} w'
  &\text{iff}&
  j\colon(P,\gamma D')\justifies^{\gamma\delta'}\gamma w'
  \end{array}
\end{equation}

\noindent Here and throughout the remainder of this appendix, primed objects
are those of the domain of $\gamma$, in $\mathbb E_{n+1}$, and unprimed ones
their images in $\mathbb E_n$, following
Lemma~\ref{l:gamma-restr-execs}.

Due to the de Bruijn indexing, events in the domain and range of $\gamma$ have
identical actions, so that $\gamma D'=D'$, $\gamma\delta'=\delta'$, and $\gamma
w'=w'$. We keep $\gamma$ to make clear that the events lie in the range of
$\gamma$.

Even though $\gamma$ preserves symbols and expressions, $P$ is not a function of
$P'$ alone, at least because value restrictions of events in the range of
$\gamma$ contain constraints established in the previous iterations of the loop.
Recall from Definition~\ref{def:gen-just} that justifications are generated from
pre-justifications through refinement by elaborations. Pre-justifications depend
only on the actions, which are preserved verbatim by $\gamma$ assuming the de
Bruijn indexing of symbols. Predicates in pre-justifications are value
restrictions with constraints from previous loop iterations. Elaborations modify
existing justifications and depend on the predicates of the justifications they
modify. The proof of the following lemma uses the insight that elaborations only
use the restrictions of predicates to symbols introduced in the current
iteration of the loop or before the loop. The latter uses
Condition~\ref{episodic:cond} of Definition~\ref{def:episodic} of episodic
loops, and constitutes an invariant over the generation of justifications.

That $\gamma$ preserves and reflects justifications follows through an inductive
argument, where, formally, the image of the predicate of a justification under
$\gamma$ is itself inductively defined as follows. The preservation and
reflection is then established incrementally below.

\begin{lemma}[$\gamma$ preserves and reflects
  justifications]\label{l:gamma-pres-justs}\label{l:gamma-pres-just}

  \medskip\textbf{Pre-justifications.}
  Let $w'$ be a memory-effectful event outside the first iteration of $\ell$,
  where $\gamma$ is defined, in an execution $\mathbb X'~=~(X',J',\rf')$ in
  $\langle\prog\rangle_{n+1~\rho~\kappa~\varphi}$. Write $\Xi(w')$ for the
  dependency set Definition~\ref{def:gen-just} gives a pre-justification of
  $w'$: $\origin{x'}\cup\origin{\expr'}$ for a write $(w'\colon W~x'~\expr')$,
  $\origin{\expr'}$ for an allocation $(w'\colon\Allocs~\alpha'~\expr')$, and
  $\origin{\expr'}$ for a deallocation $(w'\colon\Deallocs~\expr')$. If \[
    j'\colon(\valres(w'), \Xi(w'))\justifies^{(\emptyset,\emptyset)} w' \]
  \noindent is a pre-justification for $w'$, then \[ j\colon(\valres(\gamma w'),
  \gamma\,\Xi(w'))\justifies^{(\emptyset,\emptyset)}\gamma w' \] \noindent is a
  pre-justification for $\gamma w'$, and vice versa.

  \medskip\textbf{Elaborations.}
  For the induction hypothesis let
  $j'_1\colon(P'_1,D'_1)\justifies^{\delta'_1}w'_1$ and
  $j'_2\colon(P'_2,D'_2)\justifies^{\delta'_2}w'_2$ be justifications of
  memory-effectful events in $\mathbb J'_i$. Then $\gamma
  j'_1\colon(P_1,D_1)\justifies^{\delta_1}\gamma w'_1$ and $\gamma
  j'_2\colon(P_2,D_2)\justifies^{\delta_2}\gamma w'_2$ are justifications in
  $\mathbb J_i$, such that $e'_1\ppo^{P'_1}_{\delta'_1} e'_2$ iff $\gamma
  e'_1\ppo^{P_1}_{\delta_1}\gamma e'_2$ and similar for $\gamma j'_2$.

  Then, there is a justification $j'\colon(P',D')\justifies^{\delta'}w'$ in
  $\mathbb J'_{i+1}$ in
  $\langle\prog\rangle_{n+1~\rho~\kappa~\varphi}$ iff there is a justification
  $j\colon(P,D)\justifies^{\delta}w$ in $\mathbb
  J_{i+1}$ such that $G(\gamma j'_1,j)$ for all $G\in\{\elab{va},
  \elab{str}, \elab{fwd}, \elab{we}, \elab{weak}\}$ or
  $\elab{lift}(\gamma j'_1,\gamma j'_2,j)$, with $i+1$ the minimal such index.
  We denote $j$ as the image of $j'$ under $\gamma$, and write it as $\gamma j'$.

  The correspondence is stated for justifications $j'$ that are not
  \emph{residual}, that is, whose generation applies no elaboration to a pair,
  a conjunct or a relabelling involving the first iteration of $\ell$, as made
  precise at the end of the proof. Residual justifications have no image, and
  are discharged in Lemma~\ref{l:gamma-pres-post-futures}.

\end{lemma}

The proof of Lemma~\ref{l:gamma-pres-justs} proceeds inductively over the
generation of justifications: Firstly, $\gamma$ preserves and reflects
pre-justifications. Secondly, assuming that $\gamma$ preserves and reflects a
set of justifications, any elaboration applying to the justifications in the
domain of $\gamma$ also applies to corresponding justifications in the range of
$\gamma$, and vice versa.

Write $\Sigma^0$ for the symbols read in the first iteration of $\ell$ in
$\mathbb E_{n+1}$, that is $\Sigma_{\ell\mapsto 0}\setminus\Sigma_{\ell\mapsto
\emptyset}$, on which $\gamma$ is undefined. The vocabularies
$\Sigma_{\ell\mapsto i}$ of Appendix~\ref{s:app-proofs-restricted} include the
symbols read before the loop, and $\Sigma^0$ does not.

The correspondence between justifications $j'\colon(P',D')\justifies^{\delta'}w'$ above requires
structural properties on $P'$ and $\delta'$ in the form of the following
invariants:

\begin{enumerate}

  \item\label{gamma-pres-just:decomp}
    $P'\equiv R'_0\wedge\gamma^{-1}P$ for a predicate $R'_0$ over
    $\Sigma_{\ell\mapsto 0}$ with $\restrict{R'_0}{\emptyset}\equiv\top$, where
    $\gamma^{-1}P$ renames each symbol of $P$ to its preimage.

  \item\label{gamma-pres-just:restricted-preds}
    $\restrict{P'}{i}=\restrict{P}{i-1}$ for $i\geq 1$.

  \item\label{gamma-pres-justs:delta}\label{gamma-pres-just:delta}

    $\delta=\gamma\delta'$

  \item\label{gamma-pres-just:symbols}

    $\symbols(P')=\gamma^{-1}\symbols(P)\uplus\Sigma^0$.

\end{enumerate}

Invariant~\ref{gamma-pres-just:decomp} relates the predicates themselves, not
only their restrictions: $P'$ is the preimage predicate $P$ under $\gamma$
conjoined with a part $R'_0$ that speaks about the first iteration alone and
constrains nothing before the loop. We establish it for pre-justifications and
carry it through matching elaborations in the domain and range of $\gamma$.

Invariant~\ref{gamma-pres-just:restricted-preds} follows from
Invariants~\ref{gamma-pres-just:decomp} and~\ref{gamma-pres-just:symbols}. Let
$i\geq 1$. $\gamma^{-1}P$ mentions no symbol read in the first iteration, so
projecting those symbols away gives
$\exists\Sigma^0.(R'_0\wedge\gamma^{-1}P)\equiv
\gamma^{-1}P\wedge\exists\Sigma^0.R'_0\equiv\gamma^{-1}P$, the
last step as $\restrict{R'_0}{\emptyset}\equiv\top$: every assignment to the
symbols read before the loop extends to one satisfying $R'_0$.
$\Sigma_{\ell\mapsto i}$ contains no symbol read in the first iteration, so
$\restrict{P'}{i}=\restrict{\gamma^{-1}P}{i}$, which is $\restrict{P}{i-1}$ as
$\Sigma_{\ell\mapsto i}$ in $\mathbb E_{n+1}$ and $\Sigma_{\ell\mapsto i-1}$ in
$\mathbb E_n$ are the same symbols by Property~\ref{prop:de-bruijn-tail} of
Corollary~\ref{l:de-bruijn-props}.

Invariant~\ref{gamma-pres-just:decomp} also makes $R'_0$ \emph{transparent} to
the queries the dependency relations ask. Let $E$ be a formula whose symbols
avoid $\Sigma^0$. Then

\begin{enumerate}

  \item\label{r0-transparent:sat} $P'\wedge E$ is satisfiable iff
    $\gamma^{-1}P\wedge E$ is, and

  \item\label{r0-transparent:ent} $P'\implies E$ iff $\gamma^{-1}P\implies E$.

\end{enumerate}

\noindent In both directions from right to left, $P'$ implies $\gamma^{-1}P$. For
the converse, take a model of $\gamma^{-1}P\wedge E$, or of $\gamma^{-1}P$ with
$E$ false; $R'_0$ mentions only symbols of $\Sigma_{\ell\mapsto 0}$ and
$\restrict{R'_0}{\emptyset}\equiv\top$, so the model's assignment to the symbols
read before the loop extends to one satisfying $R'_0$, and the extension changes
neither $\gamma^{-1}P$ nor $E$, which mention no symbol of $\Sigma^0$ by
Invariant~\ref{gamma-pres-just:symbols} and by assumption.

By Lemma~\ref{l:episodic-symb} the location and value expressions of an event
outside the first iteration of $\ell$ mention no symbol of $\Sigma^0$, so the
two clauses apply to the equalities $\ppoalias$ tests and to the branching
conditions $\ppormw$ tests. This is what carries $\ppo$ and $\pred$ along
$\gamma$ below, for any pair of events in its domain and without restricting the
predicates.

Invariant~\ref{gamma-pres-just:symbols} records that, beyond the symbols of
$\Sigma_{\ell\mapsto i}$ constrained by Invariant~\ref{gamma-pres-just:restricted-preds}, the
two predicates differ only by the first iteration in the domain of $\gamma$. It
holds for pre-justifications, where $P'=\valres(w')$ is the conjunction of the
branching conditions accumulated along $\po$ up to $w'$ per
Definition~\ref{def:gen-es} of the event structure semantics: as $\gamma$
follows branching decisions and maps the $k+1$-st iteration to the $k$-th, the
conditions of the $k+1$-st iteration in the domain and of the $k$-th in the
image coincide under the de Bruijn indexing, and the conditions of the first
iteration in the domain have no counterpart. It is preserved by the
elaborations, which modify $P'$ only through symbols from $\Sigma_{\ell\mapsto i}$ by
Lemma~\ref{l:restrict-elabs}, leaving the remaining conjuncts, and hence
$\Sigma^0$, untouched.

\begin{lemma}[$\gamma$ preserves and reflects $\pred$]\label{l:gamma-pres-pred}

  For all events $e'_1$ and $e'_2$ in the domain of $\gamma$, $\gamma$ preserves
  and reflects $\pred$ as follows

  \begin{equation}
    e'_1\in\pred_{\delta'}(e'_2,P')~\text{iff}~\gamma
    e'_1\in\pred_{\gamma\delta'}(\gamma e'_2,P)
  \end{equation}

\end{lemma}

\begin{proof}

  By Definition~\ref{def:pred}, $e'_1\in\pred_{\delta'}(e'_2,P')$ says that
  $e'_1\ppo^{P'}_{\delta'}e'_2$ and that no event lies $\ppo$-between them.
  Lemma~\ref{l:gamma-pres-ppo} carries both directions of the first, for any
  pair of events in the domain of $\gamma$.

  For the second, an event $\ppo$-between $e'_1$ and $e'_2$ lies between them in
  $\po$, as $\ppo$ refines $\po$. Where both lie in iterations of $\ell$, such an
  event lies in an iteration between theirs, hence outside the first, and
  $\gamma$ is injective, so the events between correspond and neither side has
  one the other lacks.

  Where $e'_1$ lies before the loop and $e'_2$ in an iteration of it, an event of
  the first iteration may lie between them. It has no image, so $\gamma e'_1$ can
  be a $\pred$ of $\gamma e'_2$ where $e'_1$ was not a $\pred$ of $e'_2$: the
  reflection holds, and the preservation is the first-iteration residue that
  Lemma~\ref{l:gamma-pres-post-futures} discharges at histories covering that
  iteration, as it does for the $\DP$-edges the first iteration sources.

\end{proof}

\begin{lemma}[$\gamma$ preserves and reflects forwarding
  relations]\label{l:gamma-pres-fwdrel}

  Both forwarding relations, $F_{j'}$ and $\text{WE}_{j'}$, are contextual in
  justifications $j'\colon(P',D')\justifies^{\delta'}w'$. Assume that $\gamma
  j'=(P,\gamma D')\justifies^{\gamma\delta'}\gamma w'$ is the image of $j'$
  under $\gamma$ per the induction hypothesis in
  Lemma~\ref{l:gamma-pres-just}.
  \begin{enumerate}
    \item $e'_1\fwdrel[j'] e'_2$ iff $\gamma e'_1\fwdrel[\gamma j']\gamma e'_2$
    \item $e'_1\werel[j'] e'_2$ iff $\gamma e'_1\werel[\gamma j']\gamma e'_2$
  \end{enumerate}

\end{lemma}

\begin{proof}

  By Lemma~\ref{l:gamma-pres-pred} $\gamma$ preserves $\pred$, so that
  $e'_1\in\pred_{\delta'}(e'_2,P')$ iff $\gamma
  e'_1\in\pred_{\gamma\delta'}(\gamma e'_2,P)$.
  $\loc(e'_1)\equiv_{P'\land\psi_{\delta'}}\loc(e'_2)$ iff by
  Definition~\ref{def:sem-equiv} of semantic equivalence
  $(\psi_{\delta'}\implies\loc(e'_1)=\loc(e'_2))\equiv_{P'}\top$ iff by
  Invariant~\ref{gamma-pres-just:restricted-preds}
  $(\psi_{\gamma\delta'}\implies\loc(\gamma e'_1)=\loc(\gamma
  e'_2))\equiv_{P}\top$ iff by Definition~\ref{def:sem-equiv}
  $\loc(\gamma e'_1)\equiv_{P\land\psi_{\gamma\delta'}}\loc(\gamma e'_2)$.
  Thus, $\gamma$ preserves and reflects the auxiliary relation $F'$ of
  Definition~\ref{def:fwd-rels}, from which $\fwdrel$ and $\werel$ are cut by
  event type -- the prime belongs to its name and does not mark the domain of
  $\gamma$ -- in that $e'_1\xrightarrow{F'_{j'}}e'_2$ iff $\gamma
  e'_1\xrightarrow{F'_{\gamma j'}}\gamma e'_2$.

  Thus $\gamma$ preserves and reflects both $\fwdrel$ and
  $\werel$, which is a consequence of $\gamma$ additionally preserving the event
  type.

\end{proof}

Relabel equivalences are relations between pairs $P_1\colon e_1$ and
$P_2\colon e_2$ of predicates and events. Relabel equivalences are used in
the definition of the lifting elaboration $\elab{lift}$, which provides the
context to define the lifting of relabel equivalences along $\gamma$. The next
lemma proves that $\gamma$ in fact preserves complete relabel equivalences.

\begin{lemma}\label{l:gamma-pres-relabeq}

  $\gamma$ commutes with complete relabelling equivalences such that
  whenever $P'_1\colon e'_1\relabeq\Lambda{\delta'}^* P'_2\colon e'_2$, then
  $P_1\colon\gamma e'_1\relabeq\Lambda{\gamma\delta'}^* P_2\colon\gamma e'_2$.

\end{lemma}

\begin{proof}

  The proof follows from Lemma~\ref{l:restrict-relabeq} and
  Invariant~\ref{gamma-pres-just:restricted-preds}.

\end{proof}

\medskip

\begin{proof}[of Lemma~\ref{l:gamma-pres-justs}]

  \medskip
  \textbf{Pre-justifications}
  The proof uses episodicity to establish that $\gamma$ preserves expressions
  for memory location and value, so that
  $\loc(w')=\loc(\gamma w')$ and $\val(w')=\val(\gamma w')$, as well as
  $\gamma\origin{\loc(w')}=\origin{\loc(\gamma w')}$ and
  $\gamma\origin{\val(w')}=\origin{\val(\gamma w')}$.

  This covers the three shapes of $\mathbb J_0$ alike, the dependency set
  $\Xi(w')$ of the statement being read off those same expressions in each:
  both for a write, the value expression alone for an allocation, whose size it
  is, and the location expression alone for a deallocation. An allocation is
  the one shape whose justified event introduces a symbol of its own, and that
  symbol is not among its dependencies, $\origin{\alpha'}$ being the allocation
  itself. Nothing is asked of $\gamma$ there beyond what
  Corollary~\ref{c:gamma-constr} already gives: $\alpha'$ occurs in no value
  restriction, so Invariant~\ref{gamma-pres-just:symbols} is untouched by it,
  and the allocation of the $k+1$-st iteration introduces the symbol its image
  in the $k$-th introduces, under the de Bruijn indexing.

  For the predicates, $P'=\valres(w')$ is the conjunction of the branching
  conditions accumulated along $\po$ up to $w'$ per
  Definition~\ref{def:gen-es}: those before the loop, those of the first
  iteration of $\ell$, and those of the iterations after it. As $\gamma$ follows
  branching decisions and maps the $k+1$-st iteration to the $k$-th, the
  conditions before the loop and those of the iterations after the first are,
  read back along $\gamma$, exactly the conjuncts of $\valres(\gamma w')$ by
  Property~\ref{prop:de-bruijn-tail} of Corollary~\ref{l:de-bruijn-props}. Let
  $R'_0$ be the conjunction of the conditions of the branching events of the
  first iteration, including those of loops nested in it, which
  Condition~\ref{iter:nesting} places in the same iteration. Then
  $P'\equiv R'_0\wedge\gamma^{-1}\valres(\gamma w')$,
  $\symbols(R'_0)\subseteq\Sigma_{\ell\mapsto 0}$ by
  Lemma~\ref{l:episodic-symb}, and $\restrict{R'_0}{\emptyset}\equiv\top$ is
  Condition~\ref{episodic:cond} of Definition~\ref{def:episodic} for the first
  iteration. This is Invariant~\ref{gamma-pres-just:decomp}, and
  Invariant~\ref{gamma-pres-just:restricted-preds} follows as shown above.

  \medskip
  \hyperref[def:elab-va]{\textbf{Value assignment}}
  $\elab{va}(j'_1,j')$ is witnessed by $\alpha\equiv_{P'_1}v$. By
  Invariant~\ref{gamma-pres-just:restricted-preds}, this is iff
  $\alpha\equiv_{P_1}v$, so that $\elab{va}(\gamma j'_1,\gamma j')$ with
  $P=P_1$ and $P'=P'_1$, preserving
  Invariants~\ref{gamma-pres-just:decomp}
  and~\ref{gamma-pres-just:restricted-preds}. $\elab{va}$ does not modify the
  forwarding context $\delta'$, and thus maintains
  Invariant~\ref{gamma-pres-justs:delta}.

  \medskip
  \hyperref[def:elab-str]{\textbf{Strengthening}} Let $\elab{str}(j'_1,j')$
  conjoin to $P'_1\equiv R'_0\wedge\gamma^{-1}P_1$ a predicate $Q'$ whose
  symbols are read outside the first iteration of $\ell$. Its image
  $\gamma Q'$ satisfies the side conditions of Definition~\ref{def:elab-str}
  for $\gamma j'_1$, as $\gamma$ preserves and reflects $\ppo$ for $j'_1$ by the
  induction hypothesis of Lemma~\ref{l:gamma-pres-justs}, and $\sqsubseteq$ and
  $\remap$ by construction,
  so $\elab{str}(\gamma j'_1,j)$ with $P\equiv P_1\wedge\gamma Q'$ and
  $P'\equiv R'_0\wedge\gamma^{-1}P$, preserving
  Invariant~\ref{gamma-pres-just:decomp}. A strengthening with a conjunct over
  the first iteration alone changes $R'_0$ only, and has the same image as its
  premise. A strengthening with a conjunct relating a symbol of the first
  iteration to one outside it leaves the form of
  Invariant~\ref{gamma-pres-just:decomp}, and is residual in the sense below.

  \medskip
  \hyperref[def:elab-fwd]{\textbf{Forwarding}} $\elab{fwd}$ within the same
  iteration is preserved by $\gamma$, as $\gamma$ commutes with $\pred$ by
  Lemma~\ref{l:gamma-pres-pred} and $\gamma$ preserves expressions in locations
  and values literally. Across a loop boundary, store forwarding cannot arise by
  Condition~\ref{episodic:mem} of Definition~\ref{def:episodic} of episodic
  loops, and store-store forwarding substitutes within
  $\Sigma_{\ell\mapsto \emptyset}$ -- the same set of symbols in $\mathbb E_{n+1}$
  and $\mathbb E_n$ by Property~\ref{prop:de-bruijn-head} of
  Corollary~\ref{l:de-bruijn-props}. Load forwarding across a boundary is the
  shape excepted in Lemma~\ref{l:restrict-elabs}.

  All three shapes are covered by one argument on the predicates themselves.
  Let the pair $(e'_1,e'_2)$ lie in iterations $\geq 1$ of $\ell$ or before the
  loop, and write $g=[\val(e'_2)\mapsto\val(e'_1)]$. $\gamma$ preserves $\val$ by
  Corollary~\ref{c:gamma-constr}, so the substitution the image elaboration
  applies in $\mathbb E_n$ is $\gamma g$, which is $g$ along $\gamma$. By
  Invariant~\ref{gamma-pres-just:decomp} for $j'_1$,
  \[
    P'~=~\evalenv{P'_1}{g}~\equiv~\evalenv{R'_0}{g}\wedge
    \evalenv{\gamma^{-1}P_1}{g}~=~\evalenv{R'_0}{g}\wedge\gamma^{-1}\evalenv{P_1}{\gamma g}
    ~=~\evalenv{R'_0}{g}\wedge\gamma^{-1}P
  \]
  \noindent If $\val(e'_2)$ is read in an iteration $\geq 1$, $R'_0$ does not
  mention it and $\evalenv{R'_0}{g}=R'_0$. If it is read before the loop, as in
  store-store forwarding, $\evalenv{R'_0}{g}$ still constrains nothing before
  the loop: an assignment to the symbols read before the loop that extends to
  one satisfying $\evalenv{R'_0}{g}$ is obtained from the extension for $R'_0$ at
  the assignment that gives $\val(e'_2)$ the value of $\val(e'_1)$. Either way
  Invariant~\ref{gamma-pres-just:decomp} holds for $j'$, and with it
  Invariant~\ref{gamma-pres-just:restricted-preds}. The argument does not need
  $P'_1$ to separate by iteration: the substitution acts on the whole predicate,
  alike on both sides. A pair with an event in the first iteration of $\ell$ has
  no image and is residual in the sense below.

  A forwarding edge that is carried by $\gamma$ spans the $k$-th and $k+1$-st
  iterations for some $k\geq 1$, and is mapped to the edge spanning the
  $k-1$-st and $k$-th, so that $\remap$ in the image mirrors $\remap$ in the
  domain. Only an edge from the first iteration into the second lies outside the
  domain of $\gamma$, the same residue as in Claim~\ref{gamma-dp:pres} of
  Lemma~\ref{l:gamma-pres-dp}.

  \medskip
  \hyperref[def:elab-we]{\textbf{Write elision}} $\elab{we}$ follows directly
  from Lemma~\ref{l:gamma-pres-fwdrel}.

  \medskip
  \hyperref[def:elab-lift]{\textbf{Lifting}} By Lemma~\ref{l:restrict-relabeq},
  \[
    \begin{array}{cc}
      P'_1\colon w'_1\relabeq\Lambda{\delta'}^* P'_2\colon w'_2 &
      \text{iff} \\
      \restrict{P'_1}{i}\colon
      w'_1\relabeq\Lambda{\delta'}^*\restrict{P'_2}{i}\colon w'_2 &
      \text{iff} \\
      \restrict{P_1}{i-1}\colon\gamma
      w'_1\relabeq\Lambda{\gamma\delta'}^*\restrict{P_2}{i-1}\colon\gamma w'_2 &
      \text{iff} \\
      P_1\colon\gamma w'_1\relabeq\Lambda{\gamma\delta'}^* P_2\colon\gamma w'_2 &
      \\
    \end{array}
\]
  \noindent which transports the side conditions of
  Definition~\ref{def:elab-lift} along $\gamma$. For the lifted predicate, let
  $P'_1\equiv R'_{0,1}\wedge\gamma^{-1}P_1$ and
  $P'_2\equiv R'_{0,2}\wedge\gamma^{-1}P_2$ by
  Invariant~\ref{gamma-pres-just:decomp}. Suppose the relabelling $\Lambda$ fixes
  the symbols read in the first iteration of $\ell$, and the two premises agree
  on the first iteration, $\evalenv{R'_{0,1}}{\Lambda}\equiv R'_{0,2}$, as they
  do for a lift whose two writes follow a branch in an iteration $\geq 1$ or
  after the loop. Then
  \[
    \evalenv{P'_1}{\Lambda}\vee P'_2~\equiv~
    R'_{0,2}\wedge\left(\evalenv{\gamma^{-1}P_1}{\Lambda}\vee\gamma^{-1}P_2\right)
    ~=~R'_{0,2}\wedge\gamma^{-1}\left(\evalenv{P_1}{\gamma\Lambda}\vee P_2\right)
  \]
  \noindent which is $R'_{0,2}\wedge\gamma^{-1}P$ for the predicate $P$ of the
  image lift $\elab{lift}(\gamma j'_1,\gamma j'_2,j)$. So
  Invariant~\ref{gamma-pres-just:decomp} holds for $j'$, and with it
  Invariant~\ref{gamma-pres-just:restricted-preds}. A lift whose premises differ
  on the first iteration, or whose relabelling moves a symbol read in it, is
  residual in the sense below.

  \medskip
  \hyperref[def:elab-weak]{\textbf{Weakening}} $\elab{weak}$ applies global
  guarantees both in the domain and the range of $\gamma$. Weakening removes
  conjuncts. Removing a conjunct of $\gamma^{-1}P_1$ is matched by removing its
  image from $P_1$. Removing a conjunct of $R'_0$ leaves a weaker $R'_0$, which
  still constrains nothing before the loop, and has the same image as its
  premise. Either way Invariant~\ref{gamma-pres-just:decomp} is preserved, and
  Weakening preserves Invariant~\ref{gamma-pres-justs:delta} trivially.

  \medskip
  \textbf{Residual justifications.} Call a justification of an event in the
  domain of $\gamma$ \emph{residual} if its generation applies an elaboration
  outside the cases above: a forwarding along a pair with an event in the first
  iteration of $\ell$, a strengthening with a conjunct relating a symbol read in
  the first iteration to one read outside it, or a lift whose premises differ on
  the first iteration or whose relabelling moves a symbol read in it. The
  correspondence of Lemma~\ref{l:gamma-pres-justs} is established for
  justifications that are not residual. Residual justifications have no image;
  they are discharged at the level of posterior futures in
  Lemma~\ref{l:gamma-pres-post-futures}, whose histories cover the first
  iteration.

\end{proof}

As $\gamma$ preserves and reflects justifications as in
Lemma~\ref{l:gamma-pres-just} above, it suffices to consider the events in the
image of $\gamma$ in order to quantify justifications across all executions.
Taking the iterative image under $\gamma$ gives a finite bound on the events
to consider. Then previous results in~\cite{Richards25SMRD} stating a
finite bound on justifications over finite sets of events, imply the following
result as a corollary.

\begin{corollary}

  The set of justifications across all executions in a program where all
  unbounded loops are episodic is finitely bounded - up to $\gamma$.

\end{corollary}

Preservation of $\ppo$ under $\gamma$ is contextual in the justification set as
in the following lemma.

\begin{lemma}[$\gamma$ preserves and reflects preserved program
  order]\label{l:gamma-pres-ppo}
  For all events $e'_1$ and $e'_2$ in an execution $\mathbb
  X'~=~(X',J',\rf')$ in $\langle\prog\rangle_{n+1~\rho~\kappa~\varphi}$,
  $e'_1\ppo^{P'}_{\delta'} e'_2$ iff $\gamma
  e'_1\ppo^{P}_{\gamma\delta'}\gamma e'_2$ where
  $j'\colon(P',D')\justifies^{\delta'}w'$ and $\gamma
  j'\colon(P,D)\justifies^{\gamma\delta'}\gamma w'$ are justifications with
  $j'\in J'$.
\end{lemma}

\begin{proof}

  By Definition~\ref{def:ppo}, $\ppo$ is the closure of $\pposync$, $\ppormw$
  and $\ppoalias$ under $\remap_\delta$, and the predicate enters only through
  the last two: $\ppoalias$ by the satisfiability of $P\wedge\loc(e_1)=\loc(e_2)$
  and $\ppormw$ by the branching condition it asks $P$ to entail.

  Both are transparent to $R'_0$, by
  Clauses~\ref{r0-transparent:sat} and~\ref{r0-transparent:ent} of the
  transparency of $R'_0$ established with
  Invariant~\ref{gamma-pres-just:decomp}: the expressions compared are those of
  events in the domain of $\gamma$, which mention no symbol of $\Sigma^0$ by
  Lemma~\ref{l:episodic-symb}. So each query has the same answer under $P'$ as
  under $\gamma^{-1}P$, and $\gamma$ preserves the expressions themselves by
  Lemma~\ref{l:gamma-pres-expr} and the event types and program counters by
  Corollary~\ref{c:gamma-constr}, so the same answer again under $P$ for the
  images. $\pposync$ reads only memory order annotations, which $\gamma$ leaves
  unchanged, and $\remap_{\delta'}$ matches $\remap_{\gamma\delta'}$ by
  Invariant~\ref{gamma-pres-justs:delta}. Hence $e'_1\ppo^{P'}_{\delta'} e'_2$
  iff $\gamma e'_1\ppo^{P}_{\gamma\delta'}\gamma e'_2$.

  The argument does not restrict the predicates, so it holds for a pair of
  events in one iteration of $\ell$ and for a pair across a boundary alike, and
  needs neither Lemma~\ref{l:restrict-ppo} nor the $i=\emptyset$ clause that
  Remark~\ref{r:restrict-nested} calls into question.

\end{proof}

Preservation of $\DP$ follows from the preservation of justifications, as $\DP$
is read off the symbols of their predicates and dependencies by
Definition~\ref{def:freeze}. Reflection holds away from the first iteration in
the range of $\gamma$.

\begin{lemma}[$\gamma$ preserves and reflects $\DP$]\label{l:gamma-pres-dp}

  Let $\mathbb X'~=~(X',J',\rf')$ be an execution in $\mathbb E_{n+1}$, let
  $\ell$ be an episodic loop, and let $\mathbb X~=~(X,J,\rf)$ be its image
  $\gamma\mathbb X'$, an execution in $\mathbb E_n$ by
  Lemma~\ref{l:gamma-restr-execs}. Then

  \begin{enumerate}

    \item\label{gamma-dp:refl} $\gamma$ reflects $\DP$, that is
      $(e_1,e_2)\in\DP$ implies $(e'_1,e'_2)\in\DP'$ for the unique
      $e'_1,e'_2\in X'$ with $\gamma e'_1=e_1$ and $\gamma e'_2=e_2$;

    \item\label{gamma-dp:pres} $\gamma$ preserves the $\DP$-edges not sourced in
      the first iteration of $\ell$, that is if $(e'_1,e'_2)\in\DP'$ with
      $\iter(e'_1)(\ell)\neq 0$, then $(\gamma e'_1,\gamma e'_2)\in\DP$.

  \end{enumerate}

\end{lemma}

\begin{proof}
  By Lemma~\ref{l:gamma-pres-just}, $j'\colon(P',D')\justifies^{\delta'}w'$ is
  a justification in $J'$ if and only if $\gamma j'\colon(P,\gamma
  D')\justifies^{\gamma\delta'}\gamma w'$ is a justification in $J$. By
  Definition~\ref{def:freeze} it therefore suffices to compare the origins of
  $\symbols(P'_{j'})\cup\symbols(D'_{j'})$ with the origins of
  $\symbols(P_{\gamma j'})\cup\symbols(D_{\gamma j'})$.

  \medskip\noindent\textbf{Dependencies.} Pre-justifications take $D'$ to be
  the origins of the symbols of the expressions the justified event carries:
  $\origin{x'}\cup\origin{\expr'}$ for a write $(w'\colon W~x'~\expr')$,
  $\origin{\expr'}$ for an allocation $(w'\colon\Allocs~\alpha'~\expr')$, whose
  size expression it is, and $\origin{\expr'}$ for a deallocation
  $(w'\colon\Deallocs~\expr')$, naming the location it frees. Each elaboration
  resets the dependency set to the origins of the symbols of those expressions
  as it rewrites them, per Definitions~\ref{def:elab-va}
  to~\ref{def:elab-weak}; of the six only $\elab{va}$ and $\elab{fwd}$ rewrite
  any, and neither applies to an allocation or a deallocation, whose dependency
  set is therefore the one its pre-justification fixed.
  By Lemma~\ref{l:episodic-symb} these symbols are read in the same iteration as
  $w'$ or before the loop. As $\gamma$ is a mapping from the $i+1$-st iteration
  in $\mathbb E_{n+1}$ to the $i$-th in $\mathbb E_n$ which preserves
  $\origin$ and $\symbols$ by Corollary~\ref{c:gamma-constr} and
  Lemma~\ref{l:gamma-pres-expr},
  $\gamma\origin{\symbols(D'_{j'})}=\origin{\symbols(D_{\gamma j'})}$ for every
  $w'$ in the domain of $\gamma$. The dependencies therefore contribute no edge
  of $\DP'\setminus\gamma^{-1}\DP$.

  \medskip\noindent\textbf{Predicates.} By
  Invariant~\ref{gamma-pres-just:symbols} of Lemma~\ref{l:gamma-pres-just},
  $\symbols(P'_{j'})=\gamma^{-1}\symbols(P_{\gamma j'})\uplus\Sigma^0$ with
  $\Sigma^0$ the symbols read in the first iteration of $\ell$ in $\mathbb
  E_{n+1}$. Applying $\origin{\cdot}$ and using that $\gamma$ commutes with
  $\origin$,

  \begin{equation}
    \origin{\symbols(P'_{j'})}~=~
    \gamma^{-1}\origin{\symbols(P_{\gamma j'})}~\uplus~\origin{\Sigma^0}
  \end{equation}

  \noindent where $\origin{\Sigma^0}$ lies in the first iteration of $\ell$ in
  $\mathbb E_{n+1}$, by the convention that each read event introduces a fresh
  symbol.

  \medskip\noindent Claim~\ref{gamma-dp:refl} follows from the two inclusions
  $\gamma^{-1}\origin{\symbols(P_{\gamma j'})}\subseteq\origin{\symbols(P'_{j'})}$
  and $\gamma\origin{\symbols(D'_{j'})}=\origin{\symbols(D_{\gamma j'})}$, as
  $\gamma$ is injective on its domain. Claim~\ref{gamma-dp:pres} follows as
  every symbol of $P'_{j'}$ outside $\Sigma^0$ lies in
  $\gamma^{-1}\symbols(P_{\gamma j'})$, and $\origin{\Sigma^0}$ lies in the
  first iteration of $\ell$.

  \medskip\noindent Finally, $\gamma$ does not preserve $\DP$ everywhere: the
  events of the first iteration of $\ell$ in $\mathbb E_{n+1}$ are not in the
  domain of $\gamma$, which maps the $k+1$-st iteration to the $k$-th, so
  neither are the edges they source.

\end{proof}

The posterior futures of Definition~\ref{def:post-futures} discard the
$\DP$-edges $\gamma$ does not preserve, as soon as the history covers the first
iteration in the domain of $\gamma$.

\begin{lemma}[$\gamma$ preserves and reflects posterior
  futures]\label{l:gamma-pres-post-futures}

  Let $H'$ be a history in $\mathbb X'$ containing the events of the first
  iteration of $\ell$, i.e.~$\left\{e'\in X'\mid\iter(e')(\ell)=0\right\}
  \subseteq H'$. Then $\gamma H'$ is a history in $\mathbb X$ and

  \begin{equation}
    \gamma\left(\Phi\mid_{H'}\right)~=~\Phi\mid_{\gamma H'}
  \end{equation}

\end{lemma}

\begin{proof}

  We first check that $\gamma H'$ is a history in $\mathbb X$, that is downward
  closed in $\ppoord\cup\DP$ per Definition~\ref{def:histories}. Let $e\in\gamma
  H'$ and $\tilde e~(\ppoord\cup\DP)~e$. Then $e=\gamma e'$ for some $e'\in H'$,
  and by Lemma~\ref{l:gamma-pres-ppo} and Claim~\ref{gamma-dp:refl} of
  Lemma~\ref{l:gamma-pres-dp} there is $\tilde e'\in X'$ with $\gamma\tilde
  e'=\tilde e$ and $\tilde e'~(\ppo'\cup\DP')~e'$. As $H'$ is downward closed,
  $\tilde e'\in H'$ and hence $\tilde e\in\gamma H'$.

  By Definition~\ref{def:post-futures} a posterior future retains a pair
  $(e'_1,e'_2)$ of a future only when $e'_1\notin H'$. By
  Lemma~\ref{l:gamma-pres-ppo} $\gamma$ preserves and reflects $\ppo$. By
  Claims~\ref{gamma-dp:refl} and~\ref{gamma-dp:pres} of
  Lemma~\ref{l:gamma-pres-dp} $\gamma$ reflects $\DP$ and preserves those edges
  of $\DP'$ not sourced in the first iteration of $\ell$. Contraposing
  Claim~\ref{gamma-dp:pres}, an edge of $\DP'$ that is not preserved is sourced
  in the first iteration of $\ell$, hence in $H'$ by assumption, and is
  therefore discarded in $\Phi\mid_{H'}$.

  It remains to account for executions whose justifications are residual in the
  sense of the proof of Lemma~\ref{l:gamma-pres-justs}, which have no image.
  Unlike the $\DP$-edges above, a residual elaboration may \emph{remove} a
  dependency, which the posterior future cannot discard. We show instead that
  each such execution $\mathbb X'$ has the same posterior future at $H'$ as an
  execution without the residual step, which has an image.

  Let $j'\in J'$ justify $w'$ by a residual lift of premises
  $j'_1\colon(P'_1,D'_1)\justifies w'_1$ and
  $j'_2\colon(P'_2,D'_2)\justifies w'_2$, with $w'=w'_2$, whose predicates agree
  outside the first iteration: $P'_1\equiv R'_{0,1}\wedge Q'_1$ and
  $P'_2\equiv R'_{0,2}\wedge Q'_2$ with $\evalenv{Q'_1}{\Lambda}\equiv Q'_2$ and
  $R'_{0,1},R'_{0,2}$ over $\Sigma_{\ell\mapsto 0}$, each constraining nothing
  before the loop, and $\Lambda$ maps the symbols read in the first iteration
  among themselves. The lifted predicate is
  $(\evalenv{R'_{0,1}}{\Lambda}\vee R'_{0,2})\wedge Q'_2$. Let $J'_2$ be $J'$ with
  $j'_2$ in place of $j'$. Then $(X',J'_2,\rf')$ is an execution with the same
  posterior future at $H'$:

  \begin{itemize}

    \item $\DP$: the two predicates differ only in their first-iteration part,
      so the edges into $w'$ in which the two frozen relations differ are
      sourced in the first iteration, hence in $H'$, and discarded.

    \item $\ppo$: the part in which they differ is over $\Sigma_{\ell\mapsto 0}$
      and constrains nothing before the loop, so for expressions over symbols
      read before the loop or in iterations $\geq 1$ -- the locations of the
      events $\ppoalias$ compares with $w'$ outside $H'$, by
      Lemma~\ref{l:episodic-symb} -- satisfiability is the same under either
      predicate, by the projection argument for
      Invariant~\ref{gamma-pres-just:restricted-preds}. The same holds for the
      branching conditions $\ppormw$ tests.

    \item Consistency: $J'_2$ adds only edges from the first iteration into
      $w'$, which lies in an iteration $\geq 1$. A cycle through one of them
      would need a path from $w'$ back into the first iteration, which
      Condition~\ref{episodic:events} and acyclicity of $\mathbb X'$ exclude.
      $P'_2$ is consistent with the value restrictions of $X'$, as $X'$ takes
      the branch whose conditions $R'_{0,2}$ records.

  \end{itemize}

  \noindent A residual strengthening $\elab{str}(j'_1,j')$ is replaced by its
  premise $j'_1$ likewise, as far as the $\DP$-edges it contributes go. The
  replacement removes the edges from the origins the strengthening newly
  constrains. Those in the first iteration lie in $H'$. For those read in a
  later iteration than the first but in an earlier one than $w'$,
  Remark~\ref{r:str-earlier} gives the ordering without the strengthening, and
  as a history is downward closed, a path from an event outside $H'$ to $w'$
  ends in an edge whose source is outside $H'$, so $w'$ is outside the horizon
  either way. For those read in a later iteration than $w'$,
  Remark~\ref{r:str-earlier} shows the strengthened justification is not used in
  a consistent execution at all.
  A residual forwarding along a pair from the first iteration into the second
  is the case already discharged above for $\DP$.

  Replacing every residual justification in this way yields an execution whose
  justifications all have images, with the same posterior future at $H'$. Thus
  $\gamma\left(\Phi\mid_{H'}\right)=\Phi\mid_{\gamma H'}$.

\end{proof}

Taking $H'$ to be the events of the iterations before the $i+1$-st, $\gamma H'$
is the set of events of the iterations before the $i$-th, which is the
correspondence between histories used in Lemma~\ref{l:post-futures-narrowed}:
there $H_{i+1}$ and $H_i$ are exactly these two sets. The hypothesis that $H'$
cover the first iteration of $\ell$ is met for every $i\geq 0$, as $H'$ collects
the iterations $0$ through $i$. In the boundary case $i=0$, $H'$ is the first
iteration alone and $\gamma H'$ is empty, matching $H_0$, since $\gamma$ is
undefined on that iteration.

The histories this correspondence compares all cover the first iteration, and
that is where residual justifications are discharged. Histories that end inside
the first iteration are never compared, and Theorem~\ref{t:finite-post-futures}
needs no $\gamma$ for them: by Condition~\ref{episodic:events} every event of a
later iteration of $\ell$ has a $(\ppoord\cup\DP)$-predecessor in the first
iteration outside such a history -- as a history is downward closed, a path
from outside it ends in an edge from outside it -- so of the events of $\ell$
their horizons contain only events of the first iteration, of which there are
finitely many.

Iterating the lemma bounds how much of a history the semantics needs to retain.
In an episodic loop the histories $H_{i+1}$ of
Lemma~\ref{l:post-futures-narrowed} form a chain under $\gamma$, and after $i$
applications only the events before the loop and those of one iteration remain.

\begin{corollary}[Histories collapse to one iteration]\label{c:gamma-history-reset}

  Let $H_{i+1}$ be the history of the events of the iterations of $\ell$ before
  the $i+1$-st, together with the events before the loop. Then

  \begin{equation}
    \gamma^i H_{i+1}~=~H_1
  \end{equation}

  \noindent the events before the loop together with those of the $i$-th
  iteration, re-indexed as the first, and
  $\gamma^i\left(\Phi\mid_{H_{i+1}}\right)=\Phi\mid_{H_1}$.

\end{corollary}

\begin{proof}

  $\gamma$ maps the $k$-th iteration to the $k-1$-st and is undefined on the
  first, so $\gamma H_{j+1}=H_j$ for every $j\geq 0$, and $i$ applications give
  $\gamma^i H_{i+1}=H_1$. The events of the $i$-th iteration are carried to the
  first by the same count. $\gamma$ is the identity before the loop by
  Definition~\ref{def:gamma-constr} of $\gamma$, so those events are retained
  throughout. The equality of posterior futures follows by $i$ applications of
  Lemma~\ref{l:gamma-pres-post-futures}, whose hypothesis is met at each step as
  $H_{j+1}$ covers the first iteration for every $j\geq 0$.

\end{proof}

\noindent This is the shape of the reset performed by the operational semantics
at a loop boundary, discussed in Appendix~\ref{app:opsem-sound-complete}.

Lemma~\ref{l:gamma-restr-execs} establishes that $\gamma\mathbb X'$ is an
execution of $\prog$ under step-counter $n$, but not that it is
\emph{consistent} in the sense of Paragraph~\ref{def:mem-model-axiom}. Since
Definition~\ref{def:post-futures} of the posterior future set quantifies over
complete executions, and Lemma~\ref{l:post-futures-narrowed} transports
horizons along $\gamma$, the image must be admitted by the memory model. It is,
and for a reason particular to the direction of $\gamma$: $\gamma$ restricts an
execution rather than extending one, so it introduces no edges, and every edge
of the image is the image of an edge of the domain.

\begin{lemma}[$\gamma$ preserves consistency]\label{l:gamma-pres-consistency}

  Let $\mathbb X'~=~(X',J',\rf')$ be a consistent execution in $\mathbb
  E_{n+1}$, that is one satisfying Axioms~\ref{def:mem-model-axiom:nta}
  and~\ref{def:mem-model-axiom:co}. Then $\gamma\mathbb X'$ is a consistent
  execution in $\mathbb E_n$.

\end{lemma}

\begin{proof}

  By Lemma~\ref{l:gamma-restr-execs}, $\gamma\mathbb X'$ is an execution in
  $\mathbb E_n$, whose read-from relation is the image $\gamma\rf'$. It remains
  to establish the two axioms.

  $\gamma$ reflects each relation the axioms constrain: $\ppo$ by
  Lemma~\ref{l:gamma-pres-ppo}, $\DP$ by Claim~\ref{gamma-dp:refl} of
  Lemma~\ref{l:gamma-pres-dp}, and $\rf$ by construction. As $\gamma$ is
  injective on its domain, $\gamma^{-1}$ is a function on $\gamma\mathbb X'$, so
  every edge of $\DP\cup\ppoord\cup\rf$ in $\gamma\mathbb X'$ is the image of an
  edge of $\DP'\cup\ppoord'\cup\rf'$ in $\mathbb X'$.

  Suppose $\gamma\mathbb X'$ carried a cycle in $\DP\cup\ppoord\cup\rf$. Taking
  $\gamma^{-1}$ of each of its edges yields a cycle in
  $\DP'\cup\ppoord'\cup\rf'$ in $\mathbb X'$, contradicting
  Axiom~\ref{def:mem-model-axiom:nta} for $\mathbb X'$.

  For Axiom~\ref{def:mem-model-axiom:co}, let $\CO'$ be a coherence order
  witnessing the axiom for $\mathbb X'$, and take $\gamma\CO'$ as the coherence
  order on $\gamma\mathbb X'$. It is again a total order on the writes at each
  location, as $\gamma$ is injective and preserves $\loc$ by
  Lemma~\ref{l:gamma-pres-expr}. The synchronises-with component of $\HB$ is
  carried along too: $\SW$ is $\rf$ restricted to a releasing write and an
  acquiring read, $\gamma$ preserves and reflects $\rf$ by construction, and it
  leaves an event's action, and hence its memory order annotation, unchanged, so
  $\SW$ in $\gamma\mathbb X'$ is the image of $\SW'$. Hence $\FR$, $\ECO$ and
  $\HB$ in $\gamma\mathbb X'$ are the images of their counterparts in $\mathbb
  X'$, and a cycle in $\ECO\cup\HB$ pulls back as above.

\end{proof}

Note that no argument about cycles crossing the boundary between iterations is
needed. Such an argument would be required of a map that \emph{extends} an
execution by an iteration, where the added events may close a cycle with the
existing ones; the reachability of a further iteration would then have to be
assumed, and a cycle forced across one loop boundary would recur across every
other by the symmetry of episodic loops. $\gamma$ runs the other way.

The following lemma allows us to construct executions by exhaustively following
posterior future horizons.

\begin{lemma}\label{l:pfh-maximal}

  Exhaustively following posterior future horizons produces maximal consistent
  sets of events.

\end{lemma}

\begin{proof}

  Posterior future horizons are consistent with their respective histories. For
  every event in the posterior future horizon, there is a posterior future and
  thus also a future and an execution containing history and event. If the set
  of events produced by exhaustively following posterior future horizons was not
  maximal, there would be an event consistent with the constructed set missing.
  The same event would occur in an execution and thus in the posterior future
  horizon for a history which is also a prefix to the constructed set of events.

\end{proof}

\subsection{Deciding Equality of Memory Locations}\label{s:app-proofs-locations}

Two of the definitions ask whether two events access the same memory location:
$\ppoalias$ of Definition~\ref{def:ppo}, which orders the accesses a predicate
admits as aliasing, and the constraint $\varphi$ of Definition~\ref{def:freeze}
of freezing justifications, which keeps the locations of distinct allocations
apart. Both are the satisfiability of
\[
  P\wedge\loc(e_1)=\loc(e_2)
\]
\noindent for the predicate $P$ of the justifications in play.

Comparing symbolic memory locations means tracing the assignment of pointers
through the program. In the fragment of Definition~\ref{def:prog-syntax}, the
expressions of Definition~\ref{def:expressions} admit multiplication between two
non-constant operands, so a location expression may be a multivariate polynomial
over symbols, and the query above decides the solvability of a Diophantine
equation. That is undecidable in general -- Hilbert's tenth
problem~\cite{Matiyasevich1993Hilbert10} -- and the undecidability sits in the
semantics rather than in an analysis layered on top of it, as the same query
decides which sets of justifications freeze into an execution. It does not
depend on pointer arithmetic: a polynomial in a branching condition enters $P$
through the value restriction of the accesses it guards.

The programs whose locations we can decide are those that do neither.

\begin{definition}[Constant pointer offsets]\label{def:const-offsets}

  Call a symbol a \emph{base} symbol if it is used as the base of a memory
  location, and call an expression a \emph{location expression} if it is
  $\alpha+c$ for a base symbol $\alpha$ and a constant $c\in\Val$. A program has
  \emph{constant pointer offsets} if

  \begin{enumerate}

    \item\label{const-offsets:loc} every expression it uses as a memory location
      is a location expression,

    \item\label{const-offsets:val} every expression it writes to memory either is
      a location expression or mentions no base symbol, and

    \item\label{const-offsets:cmp} every branching condition mentioning a base
      symbol is an equality or disequality between two location expressions.

  \end{enumerate}

\end{definition}

\noindent Condition~\ref{const-offsets:loc} admits the array access
$\code{*(rcu}+\code{tid)}$, whose offset is fixed for the accessing thread, and
rules out indexing at a computed position.
Condition~\ref{const-offsets:val} is what makes the first condition stable under
reading back what a program stores: a pointer is written as it stands, or
shifted by a constant, and never combined into a larger expression, so a read of
such a location again holds a location expression.
Condition~\ref{const-offsets:cmp} keeps the base symbols out of the arithmetic of
the program, where the reduction above would place a polynomial over them.

\begin{lemma}[Location equality is decidable under constant
  offsets]\label{l:loc-decidable}

  Let $\prog$ have constant pointer offsets, let $P$ be the predicate of a
  justification generated per Definition~\ref{def:gen-just}, together with the
  constraints $\varphi$ of Definition~\ref{def:freeze}, and let $e_1$ and $e_2$
  be events. Then satisfiability of $P\wedge\loc(e_1)=\loc(e_2)$ is decidable.

\end{lemma}

\begin{proof}

  By Condition~\ref{const-offsets:loc}, $\loc(e_1)=\alpha+c$ and
  $\loc(e_2)=\beta+d$ for base symbols $\alpha,\beta$ and constants $c,d$, so the
  equation is $\alpha-\beta=d-c$.

  We first check that every conjunct in play that mentions a base symbol is an
  equality or disequality between location expressions. A base symbol is
  introduced by an allocation event, or by a read of a location that holds one.
  The conjuncts of a justification's predicate are the branching conditions
  accumulated along $\po$ and the value restrictions of Strengthening, by
  Definition~\ref{def:gen-just}, and those are branching conditions again;
  Condition~\ref{const-offsets:cmp} gives the shape for these. The constraints
  $\varphi$ of Definition~\ref{def:freeze} contribute the equality of the value
  of each read with the value of the write it reads from, and the disjointness of
  the locations of distinct allocations. By Condition~\ref{const-offsets:val} the
  value a write puts at a location that holds a pointer is itself a location
  expression, so the first is an equality between location expressions, and the
  second a disequality between the allocation symbols themselves. Value
  assignment substitutes a constant for a symbol under $\alpha\equiv_Pv$, which
  preserves the shape. So no conjunct relates a base symbol to an arithmetic
  expression over other symbols.

  Split $P$ into the conjuncts that mention a base symbol and those that do not,
  $P\equiv Q\wedge R$, so that every conjunct of $Q$ has the shape just
  established. $R$ mentions no base symbol, and $P$ is satisfiable, so
  $P\wedge\loc(e_1)=\loc(e_2)$ is satisfiable exactly when
  $Q\wedge\alpha-\beta=d-c$ is: the two conjunctions share no symbol, and a
  satisfying assignment of each extends to one of both.

  $Q\wedge\alpha-\beta=d-c$ is a finite conjunction of equalities and
  disequalities between symbols shifted by constants, that is a
  quantifier-free formula of the theory of equality with integer offsets.
  Satisfiability in that theory is decidable, by congruence closure over the
  finitely many base symbols with the offsets carried along the equalities.

\end{proof}

\noindent The four algorithms of Section~\ref{s:identifying} have constant
pointer offsets. In RCU and hazard pointers the only modifiable pointer-valued
locations are \code{C} and, in hazard pointers, the entries of \code{hp}; each is
only ever assigned a basic memory location, never an expression over symbols,
and the only conditions over them are the equality tests of the \cas{} and of the
\code{while} in the protect loop. The per-thread flags are read and written
through $\code{rcu}+\code{tid}$ and $\code{hp}+\code{tid}$, an offset fixed for
the accessing thread, and the value stored is a flag rather than a pointer. In
seqlock, \code{rseq} and \code{rdata} are fixed at allocation and accessed
without offsets, and in spinlock the location of \code{mutex} is not
modified. Their location queries are therefore decidable, and this is what
\mordor{} decides: it discharges $\ppoalias$ by asking its solver whether the
locations of two accesses can be equal under the predicate and the constraints
of the execution. Outside the fragment the tool inherits the undecidability of
the general query, and falls back on the syntactic over-approximation of
Section~\ref{s:identifying}.

\vfill

\pagebreak
\section{Appendix: Operational Semantics}

Section~\ref{s:opsem} gives an operational semantics for programs in \smrd{}
whose loops are episodic: it steps a configuration through the next enabled
actions a future offers. Futures are what join the two semantics, read off the
event structure and consumed by the rules, and this appendix makes the join
explicit.

Section~\ref{s:app-opsem-states} sets out what a configuration holds and the
rules that rewrite it. Section~\ref{app:opsem-finite} bounds the state space: a
program with episodic loops reaches finitely many configurations, once
timestamps are identified up to order-isomorphism. That is the operational
counterpart of Theorem~\ref{t:finite-post-futures}, which bounds the next
enabled actions in the event structure semantics; here the bound is on the
states themselves, which is what lets a search for a safety violation terminate,
under the side conditions Corollary~\ref{c:opsem-lfp} collects.
Section~\ref{app:opsem-consistency} shows that the executions the derivations
construct are consistent with the memory model, and
Section~\ref{app:opsem-sound-complete} establishes the correspondence with the
event structure semantics in both directions.

\subsection{Configurations}\label{s:app-opsem-states}

\paragraph{Histories.}
A history $H$ tracks progression through the program. In the event structure
semantics we introduced histories as sets of events. As events are occurrences
of actions, we conceive histories in the operational semantics equivalently as
sets of actions with control labels. Similarly future sets are sets of sets of
pairs of labelled actions, and posterior future horizons are sets of labelled
actions.

\paragraph{Events as labelled actions.}
Events are occurrences of actions, i.e.~pairs $l:a$ of control label
$l\in\Labels$ and actions. Labels $l$ are associated with a line in code and thus an
instruction in the program, a program counter $\pc(l)$, loop indices
$\ell\in\loopfun(l)$, and for each loop index $\ell$ an iteration
$\iter(l)(\ell)$. The latter allows us to compare labels $l$ and $l'$, $l\leq
l'$ lexicographically over all loop indices $\ell$ in $\loopfun(l)$ and
$\loopfun(l')$, respectively. The comparison $\leq$ extends naturally to labels
and histories, $l\leq H$, which is used to detect a change of loop boundaries in
the rules of the operational semantics.

\paragraph{Register states.}
A register state $\rho$ in the operational
semantics augments the register state in the event structure semantics with the
index of the loop in which the last write occurred. The loop index serves to
correctly reset the register state at boundaries of loop iterations.

\paragraph{Constraints.}
A constraint $\varphi$ is a conjunction, where
conjuncts are augmented with loop indices. The latter ensure that $\varphi$ can
be reset correctly at boundaries of loop iterations.

\paragraph{Global guarantees.}
$\Omega$ is a meta-predicate which gets instantiated as predicates once symbols
are introduced during read and allocation events. $\Omega$ is used in the
Weakening elaboration $\elab{weak}$ and describes global guarantees such as
value ranges. For the sake of a simple presentation we will ignore $\Omega$ in
the operational semantics, and assume that the predicates instantiated from
$\Omega$ are implicitly folded into $\varphi$.

\paragraph{Timestamped writes.}
Writes are write actions stamped with a rational timestamp from $\tstamps$.
Internally, they are stored as lists of tuples of symbolic memory location,
write action, timestamp and loop index. The loop index is needed to reset the
writes at boundaries of loop iterations. The symbolic memory location can only
be evaluated relative to a predicate $\varphi$ from $\Expressions_B$, so the
observable writes
$\OW:\Threads\times\Expressions_A\to\Set(\Writes\times\tstamps)$ are taken
relative to the $\varphi$ retrieved from the program state $\sigma$: they are
the writes at a location $x'$ with $x\equiv_\varphi x'$ whose timestamp is at
least that of the reading thread's viewfront at $x$,
\[
  \sigma.\OW(t,x)~\triangleq~\left\{
    (w,q)\in\sigma.\writeset
    \mid
    \loc(w)\equiv_{\sigma.\varphi}x \wedge
    \tst(\sigma.\tview(x))\leq q
  \right\}.
\]
Timestamps totally order the writes at each location, and that order is the
modification order: $(w_1,w_2)\in\CO$ iff their locations are
$\varphi$-equivalent and $\tst(\hat w_1)<\tst(\hat w_2)$. Only this induced
order is significant -- program states that differ by an order-isomorphism of
timestamps are identified as in Definition~\ref{def:opsem-ts-iso}, which we
rely on for Theorem~\ref{t:opsem-finite}.

\paragraph{Viewfronts.}
A viewfront is a map from symbolic memory locations to timestamped writes, and
two families of them are recorded. $\sigma.\tview$ is the viewfront of thread
$t$; it bounds the observable writes $\OW$ above as written. $\sigma.\mview$ is
the viewfront of the write $\hat w$, fixed at the value the writing thread's
viewfront had when $\hat w$ was performed. The rules compare viewfront entries
only through their timestamps, in $\OW$ and in $\viewcomb$, so an entry is
determined by the timestamp it carries and may be a bare timestamp with no write
of the state at it -- which is what the reset at a loop boundary leaves behind
where it removes the write an entry pointed at. Viewfronts are combined
pointwise by taking the later write at each location,
\[ (v_1\viewcomb v_2)(x)~\triangleq~ v_1(x)\text{ if
}\tst(v_2(x))\leq\tst(v_1(x)),\text{ and }v_2(x)\text{ otherwise,} \]
which is how an acquiring read of a releasing write takes on the writer's view.
This is the only construct in the operational semantics by which one thread's
state influences what another thread may read; in particular the future set
$\Phi$ carries no inter-thread edge (see Section~\ref{s:opsem}).

\paragraph{Allocated memory.}\label{par:opsem-memory}
$\memory\subseteq\Symbols$ collects symbolic memory locations allocated. That
the allocated addresses are mutually disjoint, and distinct from the global
locations, is recorded in the program state's $\varphi$ predicate, alongside the
distinctness of the globals themselves that $\sigma_0.\varphi$ carries. None of
these conjuncts carries a thread or a loop index, so the reset
$\varphi|^t_{\setminus\ell}$ at a loop boundary retains them.

\paragraph{Read-from relation.}
The read-from relation $\RF$ of a program state assigns write actions from
$\Writes$ to read actions from $\Reads$. It is written $\RF$ throughout, and is
not to be confused with the $\rf$ of an execution in the event structure
semantics: $\RF$ is what a derivation has recorded so far and is reset at loop
boundaries, whereas $\rf$ is fixed for a complete execution. In order to reset
read-from edges at boundaries of loop iterations, we need to track threads from
$\Threads$ and control labels from $\Labels$ for the read actions.

\begin{definition}[Program States]\label{def:prog-state}

  Program states $\sigma$ contain
  \begin{itemize}

    \item timestamped writes
      $\writeset\subseteq(\Labels\times\Writes\times\Threads)\times\tstamps$,
      with observable writes $\sigma.\OW(t,x)$ defined relative to
      $\sigma.\varphi$ and $\sigma.\tview$ as above

    \item for each thread $t\in\Threads$ a viewfront
      $\tview:\Expressions_A\to\sigma.\writeset\cup\tstamps$

    \item for each $\hat w\in\sigma.\writeset$ a viewfront
      $\mview:\Expressions_A\to\sigma.\writeset\cup\tstamps$

    \item for each thread $t\in\Threads$ the \emph{anchors}
      $\anchors_t\subseteq\sigma.\writeset$, the writes Rule~\ref{fig:step:lb}~(lb)
      retains at loop boundaries $t$ has crossed, which $t$ may no longer read

    \item allocated memory $\memory=\Set(\Symbols)$

    \item a predicate $\varphi$ capturing branching decisions and allocation
      constraints

    \item a read-from relation
      $\RF\subseteq\Writes\times(\Labels\times\Reads\times\Threads)$ assigning
      labelled write actions to labelled read actions

  \end{itemize}

  $\varphi_\RF$ is a predicate which captures the read-from relation such that

  \[
    \varphi_\RF
    =
    \bigwedge\limits_{(\Writes~x~\expr,\Reads~x'~\alpha)\in\RF}(\alpha=\expr)
  \]

\end{definition}

\paragraph{Configurations.}
Configurations are tuples $(\sigma,\rho,H)$ of program state $\sigma$, register
state $\rho$, and history $H$. The initial configuration $(\sigma_0,\rho_0,H_0)$
consists of the initial program state $\sigma_0$ with $\sigma_0.\writeset$
holding the initialising write at each global location at timestamp $0$, every
$\sigma_0.\tview$ and $\sigma_0.\mview$ mapping each location to its
initialising write, $\sigma_0.\memory$ and every $\sigma_0.\anchors_t$ empty,
and $\sigma_0.\varphi$ the conjunction of $x\neq y$ over distinct global
locations $x$ and $y$, which is the disjointness
Definition~\ref{def:freeze} imposes in the event structure semantics; register
state $\rho_0$ and history $H_0$ both empty.

The semantics unpacks an instruction in three layers, each with its own
judgement and each rewriting a different component of a configuration, as the
headers of Figure~\ref{fig:opsem:overview} record. We give the innermost first.

\paragraph{Action Semantics.}
The action one-step semantics $\sigma\overset{a}{\rightsquigarrow}_t\sigma'$
applies a single action of thread $t$ to the program state $\sigma$. The
following rules define the semantics of actions as events on the program state.

The rules are Rule~\ref{fig:action:write}~(write),
 Rule~\ref{fig:action:read}~(read), Rule~\ref{fig:action:allocate}~(allocate)
and Rule~\ref{fig:action:deallocate}~(deallocate) of the overview in
 Figure~\ref{fig:opsem:overview}.

\paragraph{Command Semantics.}
The command one-step semantics $\sigma\vdash\left(l\colon
c,\rho\right)\xrightarrow{\overline a}_t\left(\skipcmd,\rho'\right)$ interprets
an instruction into the actions $\overline a$ it performs and rewrites the
register state $\rho$; the program state is context here, the actions reaching
it only when future stepping applies them. The following rules define the
semantics of individual commands on the actions they perform, that is events
they emit.

The rules are Rule~\ref{fig:command:set}~(set),
 Rule~\ref{fig:command:read-var}~(read-var),
 Rule~\ref{fig:command:read-ref}~(read-ref),
 Rule~\ref{fig:command:read-ptr}~(read-ptr),
 Rule~\ref{fig:command:write-var}~(write-var),
 Rule~\ref{fig:command:write-ptr}~(write-ptr),
 Rule~\ref{fig:command:alloc}~(malloc), Rule~\ref{fig:command:free}~(free),
 Rule~\ref{fig:command:fence}~(fence) and
 Rule~\ref{fig:command:faa}~(faa) of the overview in
Figure~\ref{fig:opsem:overview}.

\paragraph{Branching Semantics and Path-based
reasoning.}
Branching adds both alternative branches to the derivation. In order to verify a
safety property, both branches need to be followed. We employ path-based
reasoning over non-deterministic rules, such that a safety property holds if it
can be verified in all derivation paths.

Recall from Definition~\ref{def:executions} that branching events are filtered
from the executions. In order to detect the transition over a branching
instruction in an $\code{if}$-statement into either $\code{then}$- or
$\code{else}$-branch, we use auxiliary functions $\enterThen$ and $\enterElse$
defined on the executing thread $t$, the control label, and the syntactic
program $\prog$. The functions
$\enterThen$ and $\enterElse$ take the place of the test of loop boundaries in
 Rule~\ref{fig:step:non-lb_non-branch}~(non-lb/non-branch). The
function $\ifCond$ shall return the syntactic branching condition for the label
on thread $t$.
For simplicity of notation we assume that $\code{if}$-branchings are not nested.

The rules are Rule~\ref{fig:step:branch-then}~(then) and
 Rule~\ref{fig:step:branch-else}~(else) of the overview in
Figure~\ref{fig:opsem:overview}.

\paragraph{Safety properties.} A safety property over the program state
$\sigma$ is a predicate which holds for a program if any configuration
reachable from the initial configuration satisfies it.

\paragraph{Operational Command Semantics of CAS.}
We define the semantics of CAS with two non-deterministic rules for a path-based
reasoning approach to the operational semantics as for $\code{if}$-statements
above. As there, the outcome of the test is accumulated in $\sigma.\varphi$: each
rule emits a branching action for its outcome, $\alpha=\evalreg{\expr_s}{\rho}$
or its negation, which the branch rule of
Rule~\ref{fig:action:fence-branch}~(branch)
conjoins to $\sigma.\varphi$ and admits only if the conjunction is satisfiable.

The rules are Rule~\ref{fig:command:cas-success}~(CAS success) and
 Rule~\ref{fig:command:cas-failure}~(CAS failure) of the overview in
Figure~\ref{fig:opsem:overview}.

\paragraph{Future-Stepping Rules.}
Outermost, future stepping
$\overline\prog\vdash\left(\sigma,\rho,H\right)\rightarrow\left(\sigma',\rho',H'\right)$
advances a whole configuration along the next enabled actions in $\horizon\Phi$,
running a command through the two layers below, applying the actions it emits to
the program state, and recording them in the history. It is also where the
program state is reset at a loop boundary. Over the operational semantics
defined in \cite{Wright2023OpSem}, we add to the configuration histories $H$ to
model progress in the program and predicates $\varphi$ to model value
restrictions in the context of symbolic MRD.

In the following we use several shorthand notations to support multiple actions
as needed for the semantics of RMW operations:
\begin{itemize}

  \item $\overline a~=~a_0\ldots a_n$ is a list of actions such as modelling
    \fadd{}
    ($a_r;a_w$) or \cas{} ($a_r;a_b;a_w$ or $a_r;a_b$)

  \item $(l\colon\overline a)\in^*\horizon\Phi_H$ tests if $(l\colon
    a_0)\in\horizon\Phi_H$ and $(l\colon a_1)\in\horizon\Phi_{H\cup\{(l\colon
    a_0)\}}$ and so forth, skipping branching and fence actions: executions
    contain no events of $\Branches\cup\Fences$
    (Definition~\ref{def:executions}), so neither does $\Phi$, and such an
    action is enabled by the command alone. For the same reason a history
    records only the actions of $\overline a\setminus(\Branches\cup\Fences)$
    (Definition~\ref{def:histories}); the branching action of \cas{} acts on
    $\sigma.\varphi$ but leaves no trace in $H$.

  \item $\iter(l)(\ell)>\iter(H)(\ell)$ if $l$ is in a loop $\ell$, there is a
    pair $(l',\_)$ in $H$, and $l$ is in a later iteration than any such
    $(l',\_)$.

\end{itemize}

The rule is Rule~\ref{fig:step:non-lb_non-branch}~(non-lb/non-branch) of
the overview in Figure~\ref{fig:opsem:overview}.

At boundaries of loop iterations the timestamped writes, register state, value
restrictions, and history are reset to the beginning of the loop; the viewfronts
over the writes are not:

\begin{itemize}

  \item Let $\writeset|^t_{\setminus\ell}$ for a loop index $\ell$ denote
    $\writeset$ with the writes of thread $t$ with loop index $\ell$ removed,
    except that at each location the one of greatest timestamp is retained. We
    call the retained writes the \emph{anchors} of the boundary; the rule adds
    them to $\anchors_t$ and drops from it any anchor of $t$ at the same
    location they supersede. An anchor is thus at most one write per thread,
    location and loop, and the locations are finitely many by
    Lemma~\ref{l:opsem-locs-across-boundary}, so retaining anchors leaves
    Theorem~\ref{t:opsem-finite} untouched.

  \item An anchor is retained for the position in the coherence order it holds,
    not to be read again: Rule~\ref{fig:action:read}~(read) denies thread $t$
    the writes in $\sigma.\anchors_t$, which is where
    Condition~\ref{episodic:mem} of Definition~\ref{def:episodic} forbids $t$ to
    take a value from an iteration it has closed. Other threads are unaffected:
    an anchor is observable to them after the boundary exactly as it was before
    it.

  \item No viewfront entry moves. An entry pointing at a write the previous item
    removes keeps that write's timestamp, entries being read only through their
    timestamps; we write $\tview|^t_{\setminus\ell}$ and
    $\mview|^t_{\setminus\ell}$ for the viewfronts so understood, the $\mview$s
    of removed writes being discarded with them. $\OW$ is therefore nowhere
    enlarged by the boundary: it shrinks by the removed writes, and for $t$ by
    its anchors.

    Retaining the anchor is what keeps the crossing thread's own viewfront
    pointing at a write of the state rather than at a bare timestamp
    (Lemma~\ref{l:opsem-anchor}), and it is what gives the writes of a closed
    iteration a witness in every later configuration, which
    Lemma~\ref{l:opsem-co-glue} assembles the modification order from.

  \item Recall from Condition~\ref{episodic:reg} in
    Definition~\ref{def:episodic} that registers must not be read from in
    episodic loops unless they have been written to in the same iteration of the
    loop. Translated to the reset at boundaries of loop iterations in the
    operational semantics, registers are removed if their last write occurred in
    the current loop $\ell$. We denote the result as $\rho|^t_{\setminus\ell}$.

  \item Similarly, let $\varphi|^t_{\setminus\ell}$ denote the conjunction
    $\varphi$ with all conjuncts of thread $t$ and loop index $\ell$ removed.

  \item And, let $\RF|^t_{\setminus\ell}$ denote the read-from relation $\RF$
    with all edges $(r,w)$ from loop $l$ of thread $t$ removed.

  \item On a program state, $\sigma|^t_{\setminus\ell}$ denotes the
    componentwise restriction: $\writeset$, each $\tview$, each surviving
    $\mview$, $\anchors_t$, $\RF$ and $\varphi$ restricted as above, with
    $\memory$ unchanged. This is the form Rule~\ref{fig:step:lb}~(lb) uses.

\end{itemize}

The rule is Rule~\ref{fig:step:lb}~(lb) of the overview in
Figure~\ref{fig:opsem:overview}.

\subsection{Finite Bound on Operational Semantics}\label{app:opsem-finite}

A loop accesses memory through address expressions, and what a given expression
denotes may differ from one iteration to the next. The first result fixes which
locations a loop can return to, and so how much of a configuration a boundary
has to carry across: an iteration cannot use a location carried over in a
register from the iteration before, only one created either $\po$-before in
the same iteration or before the loop, or read as a value from another thread.

\begin{lemma}[Locations crossing a loop boundary]\label{l:opsem-locs-across-boundary}

  Let $\ell$ be an episodic loop of a thread $t$. The value of an address
  expression evaluated by $t$ in an iteration of $\ell$ depends only on values
  fixed before the loop and values read in that same iteration from another
  thread. Consequently $t$ accesses a location in two iterations of $\ell$ only
  by retrieving it in each of them from the same source outside the loop, and
  the locations $t$ accesses across all iterations of $\ell$ are drawn from one
  finite set fixed by the program.

\end{lemma}

\begin{proof}

  Address expressions are evaluated in the register state $\rho$ under the
  constraint $\varphi$. By Condition~\ref{episodic:reg} of
  Definition~\ref{def:episodic} a register read in an iteration was written in
  that iteration or before the loop, so an address expression draws only on
  values read in the iteration and values fixed before the loop. By
  Condition~\ref{episodic:mem} a read in the iteration takes a $\po$-earlier
  write of the same iteration, a write from before the loop, a write of another
  thread whose value is not derived from a write of $t$, or a
  read-don't-modify-write resolving to one of these. None of the cases lets the
  value depend on a write of an earlier iteration: the first two lie within the
  iteration or before the loop, the third excludes derivation from $t$ and hence
  from what $t$ wrote in an earlier iteration, and the fourth reduces to them.
  A thread may thus read a location from outside itself and write to it, but the
  next iteration can only address that location by reading it from the same
  outside source again -- it cannot take it from the iteration that has closed.
   Rule~\ref{fig:step:lb}~(lb) is the operational counterpart: resetting $\rho$
  and $\varphi$ leaves no register and no conjunct of the closed iteration by
  which a location it computed could be named.

  Theorem~\ref{t:opsem-finite} below hinges on the finiteness of the set of
  symbolic memory locations. Symbols are a function of the program counter by
   Rule~\ref{fig:command:read-ptr}(read-ptr),
  \ref{fig:command:read-var}(read-var), \ref{fig:command:faa}(faa) and
  \ref{fig:command:cas-success}(CAS success), so every iteration reads the same
  finitely many symbols, and the address expressions of the program are finitely
  many; each denotes one location per iteration, and by the above two iterations
  agree on it exactly when the reads it is built from return the same values
  from outside the loop.

\end{proof}

The next result bounds the state space of our operational semantics: the
configurations a program can reach are finitely many, once timestamps are
identified up to order-isomorphism.

\begin{definition}[Order-isomorphism of timestamps]\label{def:opsem-ts-iso}

  For a program state $\sigma$ let $Q(\sigma)\subseteq\tstamps$ be the finite
  set of timestamps occurring in $\sigma$: those of the writes in
  $\sigma.\writeset$, of $\sigma.\tview$, and of each $\mview$, an entry being a
  bare timestamp where the reset at a loop boundary left no write at it.
  Configurations $(\sigma,\rho,H)$ and $(\sigma',\rho',H')$ are \emph{identified
  up to order-isomorphism of timestamps} if $\rho=\rho'$, $H=H'$, and there is a
  bijection $f\colon Q(\sigma)\to Q(\sigma')$ with $q<q'$ iff $f(q)<f(q')$ such
  that replacing every timestamp $q$ in $\sigma$ by $f(q)$ yields $\sigma'$, the
  labelled write actions, the symbolic locations they address, $\anchors_t$,
  $\memory$, $\RF$ and $\varphi$ being equal.

\end{definition}

$Q(\sigma)$ is finite because the symbolic memory locations are. Viewfronts map
each location to a write, by the initial state of
Appendix~\ref{s:app-opsem-states}, so $Q(\sigma)$ holds a timestamp per location
before it holds one per write, and finiteness is not simply a matter of having
taken finitely many steps. The count is the one that opens the proof of
Theorem~\ref{t:opsem-finite}: symbols are a function of the program counter, the
reset of Rule~\ref{fig:step:lb}~(lb) leaves finitely many register states,
hence finitely many expressions in instructions and finitely many locations they
denote. That count uses this definition nowhere, so the argument is not
circular, as Lemma~\ref{l:opsem-locs-across-boundary} also notes.

The bijection is one order-isomorphism of $Q(\sigma)$, not of individual memory
locations.
Per location would follow the rules more closely, which compare timestamps only
at a location; but a location is symbolic, and which locations are
$\varphi$-equivalent is settled by the $\varphi$ the configuration carries, which
moves as the history does. An identification indexed by location would be
reindexed at every branch and boundary; one of $Q(\sigma)$ is not. The canonical
representative of a class replaces each timestamp by its rank in $Q(\sigma)$, as
in Example~\ref{ex:opsem-ts-iso} below.

\begin{example}[Ranking a program state]\label{ex:opsem-ts-iso}

  Let $x$ and $y$ be locations $\varphi$ does not identify, and let a state of
  thread $t$ hold

  \[
    \begin{array}{r@{~}c@{~}l}
      \sigma.\writeset &=& \{(w_1\colon\Writes~x~0,~0),~
                            (w_2\colon\Writes~x~1,~\tfrac{3}{4}),\\
                         && \phantom{\{}(w_3\colon\Writes~y~0,~0),~
                            (w_4\colon\Writes~y~1,~5)\}\\[2pt]
      \sigma.\tview &=& \{x\mapsto(w_2,\tfrac{3}{4}),~y\mapsto(w_3,0)\}.
    \end{array}
  \]

  \noindent Then $Q(\sigma)=\{0,\tfrac{3}{4},5\}$, ranking sends $0\mapsto0$,
  $\tfrac{3}{4}\mapsto1$ and $5\mapsto2$, and the representative of the class of
  $\sigma$ is

  \[
    \begin{array}{r@{~}c@{~}l}
      \writeset &=& \{(w_1,0),~(w_2,1),~(w_3,0),~(w_4,2)\}\\[2pt]
      \tview &=& \{x\mapsto(w_2,1),~y\mapsto(w_3,0)\}.
    \end{array}
  \]

  \noindent A state carrying $0,\tfrac{1}{2},9$ where this one carries
  $0,\tfrac{3}{4},5$ has the same representative, and a derivation reaching it
  has returned to a configuration it has already explored. Two things the
  example shows. The writes $w_1$ and $w_3$ keep the timestamp they share,
  ranking being of $Q(\sigma)$ and not of the writes, and $\freshts$ separating
  only writes at $\varphi$-equivalent locations. And the viewfront is carried
  along by the same relabelling rather than ranked on its own, so
  $\OW(t,x)=\{(w_2,\tfrac{3}{4})\}$ before and $\{(w_2,1)\}$ after, with $w_1$
  below the viewfront either way.

\end{example}

\begin{lemma}[The identification is a congruence]\label{l:opsem-ts-congruence}

  Let $f$ identify $(\sigma,\rho,H)$ with $(\sigma',\rho',H')$ as in
  Definition~\ref{def:opsem-ts-iso}. If
  $(\sigma,\rho,H)\rightarrow(\sigma_1,\rho_1,H_1)$ by a rule of thread $t$,
  then $(\sigma',\rho',H')\rightarrow(\sigma_1',\rho_1',H_1')$ by the same rule
  on the same action for some $(\sigma_1',\rho_1',H_1')$ identified with
  $(\sigma_1,\rho_1,H_1)$. The identification being symmetric, it is a
  bisimulation.

\end{lemma}

\begin{proof}

  The rules read timestamps only through $\leq$, and only at timestamps of
  $Q(\sigma)$: $\OW(t,x)$ compares $\tst(\sigma.\tview(x))$ with the timestamps
  of writes at locations $\varphi$-equivalent to $x$, $\viewcomb$ compares the
  entries of two viewfronts at one location, and $\freshts(x,q,q')$ compares
  $q'$ with $q$ and with the timestamps at $x$ above $q$. As $f$ is a monotone
  bijection of $Q(\sigma)$ onto $Q(\sigma')$ and $\varphi=\varphi'$, every such
  comparison has the same value in $\sigma'$ as in $\sigma$. Hence
  $\sigma'.\OW(t,x)$ is the $f$-image of $\sigma.\OW(t,x)$ at every thread and
  location, and $\sigma'.\anchors_t$ the $f$-image of $\sigma.\anchors_t$.

  A rule that touches neither $\writeset$ nor the viewfronts acts on $\rho$,
  $\varphi$ and $H$ alone, which the two configurations share, and $f$ itself
  identifies the results. Rule~\ref{fig:action:read}~(read) takes an
  observable write other than an anchor of $t$: the step from $\sigma'$ takes
  its $f$-image, observable and not an anchor by the paragraph above, and
  advances the viewfronts by $\viewcomb$, whose result is again an $f$-image.
   Rule~\ref{fig:step:lb}~(lb) removes the writes of $t$ carrying the closing
  loop index except the one of greatest timestamp at each location, and resets
  the register state, the value restrictions and the history; which write is
  greatest is settled by the order alone, so the removal commutes with $f$, and
  the entries the reset leaves behind carry $f$-images.

   Rule~\ref{fig:action:write}~(write) is the only rule to introduce a
  timestamp. Let $q'$ be the one it chooses, and $q$ the timestamp of the
  writing thread's viewfront at $x$, so that $q<q'$. If $q'\in Q(\sigma)$ --
  possible, as $\freshts$ separates $q'$ only from the timestamps at $x$ -- the
  step from $\sigma'$ takes $f(q')$ and $f_1=f$. Otherwise $q'$ falls strictly
  between two adjacent elements of $Q(\sigma)$, or above all of them; the step
  from $\sigma'$ takes any $q_1'$ in the corresponding gap of $Q(\sigma')$,
  which is inhabited because $\tstamps$ is dense and unbounded above, and
  $f_1=f\cup\{q'\mapsto q_1'\}$ is again monotone. Either way
  $\freshts(x,q,q_1')$ holds in $\sigma'$, since $q_1'$ lies above $f(q)$ and
  below the image of the least timestamp at $x$ exceeding $q$, both of which are
  outside the gap. The successors are identified by $f_1$.

\end{proof}

\begin{lemma}[Counting the classes]\label{l:opsem-ts-count}

  Suppose the write actions, address expressions, register states, value
  restrictions and histories of $\prog$ are finitely many, and let $W$ bound the
  writes a reachable $\sigma.\writeset$ holds and $L$ its address expressions.
  Then $|Q(\sigma)|\leq W+(|\Threads|+W)\cdot L$, and writing
  $\mathit{skel}(\sigma)$ for $\sigma$ with every timestamp replaced by its rank
  in $Q(\sigma)$, two configurations are identified exactly when their
  $\mathit{skel}$, register state and history agree. The reachable
  configurations therefore fall into finitely many classes.

\end{lemma}

\begin{proof}

  A timestamp of $\sigma$ is carried by a write, of which there are at most $W$,
  or by an entry of one of the $|\Threads|$ thread viewfronts or of the at most
  $W$ write viewfronts, each with at most $L$ entries; whence the bound on
  $|Q(\sigma)|$. An order-isomorphism carries the $i$-th element of $Q(\sigma)$
  to the $i$-th of $Q(\sigma')$, so it leaves ranks fixed and
  $\mathit{skel}(\sigma)=\mathit{skel}(\sigma')$. Conversely that equality
  exhibits an isomorphism, namely the map on $Q(\sigma)$ sending the $i$-th
  element to the $i$-th element of $Q(\sigma')$, monotone and carrying $\sigma$
  to $\sigma'$ because the two agree once ranked.

  The count is then over states whose timestamps are $0,\dots,|Q(\sigma)|-1$,
  and the sets these are built from are finite by hypothesis, which the proof of
  Theorem~\ref{t:opsem-finite} discharges. Counting instead the timestamps of
  the writes as a tuple, with the ties it may carry, replaces the one ranked
  arrangement per size by the orbits of $\mathrm{Aut}(\tstamps,<)$ on tuples of
  that length, of which there are the ordered Bell number many. That this is
  finite is the oligomorphy of the group, and it is all the finiteness of the
  timestamp component comes to; the work is the bound on $W$ and $L$.

\end{proof}

\opsemfinite*

The proof is an exercise in program analysis, and relies on the resets at
boundaries of loop iterations at the beginning of bodies in episodic loops.

\begin{proof}

  There are only finitely many symbols read in programs with episodic loops, as
  symbols are a function of program counter per
   Rule~\ref{fig:command:read-ptr}(read-ptr),
  \ref{fig:command:read-var}(read-var), \ref{fig:command:faa}(faa), and
  \ref{fig:command:cas-success}(CAS success).

  As the register state is reset at boundaries of loop iterations in
   Rule~\ref{fig:step:lb}~(lb), there are then only finitely many register
  states.

  Expressions in instructions are made up of constants and references to
  register values. As there are only finitely many register states, there are
  only finitely many expressions in instructions. As there are thus only
  finitely many expressions in branching instructions, there are then only
  finitely many $\varphi$ produced by Rule~\ref{fig:step:branch-then}(then) and
  \ref{fig:step:branch-else}(else).

  As there are only finitely many expressions and thus symbolic memory
  locations, there are only finitely many write actions, and thus only finitely
  many possible read-from pairs in $\RF$, and thus only finitely many
  constraints produced from $\RF$ in Rule~\ref{fig:action:read}~(read).

  Timestamps are drawn from $\tstamps$, which is infinite, so
   Rule~\ref{fig:action:write}~(write) has infinitely many choices of a fresh
  $q'$ at each step. This does not make the reachable configurations infinite,
  because program states are identified up to order-isomorphism of timestamps
  (Definition~\ref{def:opsem-ts-iso}, Lemmas~\ref{l:opsem-ts-congruence}
  and~\ref{l:opsem-ts-count}):
  $\freshts(x,q,q')$ constrains $q'$ only by where it falls in the order at $x$,
  and the rules read timestamps only through $\leq$ -- in $\OW$, in $\freshts$,
  and in $\viewcomb$. Two states differing by an order-isomorphism therefore
  have the same $\OW$ at every thread and location and step to states again
  related by one. Up to that identification a state is determined by the finite
  set of write actions together with a total preorder on them, total at each
  location, and, for each thread and each write, a position in that preorder at
  each location as its viewfront -- a position rather than a write, a viewfront
  entry left behind by a loop boundary needing no write at it: finitely many,
  since the writes are.

   Rule~\ref{fig:step:lb}~(lb) retains one write per location of the block it
  closes as an anchor rather than removing the block outright, and this does not
  disturb the count. An anchor is at most one write per thread, location and
  loop; by Lemma~\ref{l:opsem-locs-across-boundary} the locations a loop
  accesses across its iterations are finitely many, so the anchors are finitely
  many; and the anchor a boundary retains at a location supersedes the one
  retained there by the previous boundary, so they do not accumulate with the
  iterations. An anchor at a location the loop never addresses again is simply
  never consulted.

\end{proof}

\paragraph{Revisit detection.} The theorem bounds the state space and not the
length of derivations: the semantics has no final configuration and imposes no
fairness condition, so a program with an unbounded loop has derivations of every
length. A finite state space with runs of unbounded length makes a safety
property decidable by exhaustive search only if the search recognises that it
has returned to a configuration it has already explored, and
Lemma~\ref{l:opsem-ts-count} is that test for the timestamps: rank them, and
compare the ranked forms, at the cost of a sort. The other components are
compared as they stand, with two exceptions. A constraint $\varphi$ is a
conjunction of atoms, and two configurations may carry logically equivalent ones
written differently; and each $\tview$ is indexed by address expressions where
the location meant is the $\equiv_\varphi$-class, so a configuration is free to
spell one location two ways. Comparing them syntactically instead leaves the
test sound but incomplete: it distinguishes configurations the identification
merges, and the search then explores more classes than the theorem counts.
Termination is not at risk -- the conjuncts and the address expressions are
finitely many by the proof above.

\begin{corollary}[Reachability as a least fixed point]\label{c:opsem-lfp}

  Write $\mathcal C$ for the configurations of $\prog$ reachable from
  $(\sigma_0,\rho_0,H_0)$, taken modulo the identification of
  Definition~\ref{def:opsem-ts-iso}, and put
  \[
    F(S)~\triangleq~\left\{\left[(\sigma_0,\rho_0,H_0)\right]\right\}
    \cup\mathit{post}(S)
    \qquad\text{for }S\subseteq\mathcal C.
  \]
  Then $F$ is monotone on a finite lattice, its least fixed point $\mu F$ is
  $\mathcal C$, and $\emptyset\subseteq F(\emptyset)\subseteq
  F^2(\emptyset)\subseteq\dots$ reaches it in at most $|\mathcal C|$ steps. A
  safety property whose violating configurations are a union of classes $B$
  holds of $\prog$ exactly when $\mu F\cap B=\emptyset$.

\end{corollary}

\begin{proof}

  Successors are well defined on classes by Lemma~\ref{l:opsem-ts-congruence}:
  identified configurations step to identified configurations, so the class of a
  successor depends only on the class stepped from and $\mathit{post}$ lifts to
  $\mathcal C$. That $\mathcal C$ is finite is Theorem~\ref{t:opsem-finite}, so
  $(\Set(\mathcal C),\subseteq)$ is a finite complete lattice, on which $F$ -- a
  constant joined to a monotone image -- is monotone. Knaster-Tarski gives $\mu
  F$, and the ascending chain of the iteration stabilises within $|\mathcal C|$
  steps for want of room; its limit is the set of classes of configurations
  reachable in finitely many steps, which is $\mathcal C$. The property holds
  iff no reachable configuration violates it, and $B$ being a union of classes,
  that is $\mu F\cap B=\emptyset$.

\end{proof}

\noindent The corollary requires three conditions to be met.
\emph{(1)} $B$ must be a union of classes: a condition read through $\leq$ on
timestamps is one, a condition naming a timestamp is not.
\emph{(2)} The step relation must be decidable, which here is the satisfiability
of the constraints the rules conjoin to $\varphi$.
\emph{(3)} The iteration is the search above, so it inherits the membership test
there, sound under a syntactic comparison and of the theorem's own size only
under an equivalence one.

Safety is the greatest fixed point of the operator taking $S$ to the classes
outside $B$ all of whose successors lie in $S$, whose post-fixed points are
exactly the inductive invariants of $\prog$; the least fixed point above is the
forward reading of the same check, and the one the use-after-free of
Section~\ref{sec:bug} calls for.

The following result contributes to an upper bound on the size of derivations in the
operational semantics.

\begin{lemma}

  Read-from relations multiply reachable configurations in the operational
  semantics over the number of threads.

\end{lemma}

\begin{proof}

  The read action Rule~\ref{fig:action:read} assigns an observable write from
  $\OW(t,x)$ to the read action, and thereby adds a constraint equating the
  fresh symbol read and the write value expression to $\varphi$. In the worst
  case, the constraints are incompatible between different choices of writes.
  Note that $\OW(t,x)$ may hold more than one write -- reading a stale write is
  what makes the semantics non multi-copy atomic -- so the branching here is
  over the observable writes, and a thread's viewfront is what prunes it.

\end{proof}

\subsection{Consistency of Constructed Executions}\label{app:opsem-consistency}

Given a derivation in the operational semantics, soundness has to produce not
merely an execution but a \emph{consistent} one, satisfying
Axioms~\ref{def:mem-model-axiom:nta} and~\ref{def:mem-model-axiom:co}. Neither
axiom appears in the rules of this appendix: no rule mentions
a cycle, and no component of a configuration is a modification order. Both are
met for structural reasons instead, and the three lemmas of this subsection
isolate them. The first concerns the order in which a derivation performs
actions and needs no timestamps at all; the second and third concern the
timestamps, which are what the operational semantics carries in place of a
modification order.

The third is the only one that has to say anything about loop boundaries, and it
is worth being explicit about why, since
Lemma~\ref{l:gamma-pres-consistency} does not.  That lemma needs no argument
about cycles crossing a boundary because $\gamma$ \emph{restricts} an execution,
so every cycle in the image pulls back along it.  Soundness runs the other way,
from a derivation that has been collapsed at every boundary to an execution that
has not, and that is exactly the extending direction the remark following
Lemma~\ref{l:gamma-pres-consistency} names as the one requiring a cross-boundary
argument.  Lemma~\ref{l:opsem-co-glue} supplies it.

\begin{lemma}[Derivations are thin-air free by construction]\label{l:opsem-nta}

  Let $D$ be a derivation from $(\sigma_0,\rho_0,H_0)$ and let $\mathbb X$ be
  the execution it constructs.  Then $\DP\cup\ppoord\cup\rf$ is acyclic on $\mathbb
  X$, so $\mathbb X$ satisfies Axiom~\ref{def:mem-model-axiom:nta}.

\end{lemma}

\begin{proof}

  $D$ performs one labelled action per step and thereby linearly orders the
  events of $\mathbb X$; write $<_D$ for that order.  As a subrelation of a
  strict linear order is acyclic, it suffices that each of the three relations
  is contained in $<_D$.

  \emph{$\ppo$ and $\DP$ within an iteration.}  Every future-stepping rule of
   Rule~\ref{fig:step:non-lb_non-branch} and~\ref{fig:step:lb} selects its
  action from the posterior future horizon $\horizon\Phi_H$, whose members are
  by Definition~\ref{def:horizons} the $\phi$-minimal events not yet in $H$.  By
  Definition~\ref{def:futures} $\phi_{\mathbb X}=X^2\cap(\ppoord\cup\DP)$, so an
  event with an outstanding $\ppo$- or $\DP$-predecessor is not minimal and is
  not stepped.  Both relations therefore point forward in $<_D$.  This is the
  one place where $\Phi$ being built from exactly $\ppo$ and $\DP$ is used, and
  it is why the no-thin-air axiom is discharged by the shape of the rules rather
  than by a side condition on them.

  \emph{$\ppo$ and $\DP$ across a boundary.}  After Rule~\ref{fig:step:lb}~(lb)
  the history no longer records the events of the iteration just closed, so the
  minimality argument no longer sees them.  It does not have to: by
  Condition~\ref{episodic:events} of Definition~\ref{def:episodic} every event of
  an earlier iteration is $(\ppoord\cup\DP)^+$-before every event of a later one,
  and $D$ performed the earlier iteration's events at earlier steps, so
  cross-boundary edges of $\ppoord\cup\DP$ agree with $<_D$ as well.

  \emph{$\rf$.} Rule~\ref{fig:action:read} chooses $\hat w\in\sigma.\OW(t,x)$,
  and $\OW(t,x)\subseteq\sigma.\writeset$, to which a write is added only by
   Rule~\ref{fig:action:write} at the step performing it.  The write of an $\rf$
  edge has therefore been stepped before the read taking it, and
  $\rf\subseteq{<_D}$.

\end{proof}

\paragraph{The boundary and the viewfronts.}
Both remaining lemmas rest on a viewfront never moving back, and the reset at a
loop boundary is the one point in the semantics at which it might.  It does not,
and the anchors of the future-stepping rules above are what secure that for the
thread crossing the boundary.

\begin{lemma}[Viewfronts at loop boundaries]\label{l:opsem-anchor}

  Let a thread $t$ cross a boundary of a loop $\ell$ in a state $\sigma$.  Then
  $\tview|^t_{\setminus\ell}=\sigma.\tview$, and this viewfront points at a
  write of $\sigma.\writeset|^t_{\setminus\ell}$ at every location.  No entry of
  any other viewfront moves either.  Hence
  $\sigma|^t_{\setminus\ell}.\OW(t',x)\subseteq\sigma.\OW(t',x)$ for every
  thread $t'$ and location $x$.

\end{lemma}

\begin{proof}

  By induction on the derivation reaching $\sigma$, with the invariant that a
  thread's viewfront at a location is at or above every write that thread has
  performed there. Rule~\ref{fig:action:write}~(write) advances $t$'s
  viewfront at $x$ to the write it adds, whose timestamp $\freshts$ places above
  that of an observable write and hence above $\tst(\sigma.\tview(x))$;
   Rule~\ref{fig:action:read}~(read) advances it too, $\viewcomb$ taking the
  later entry at each location; and the boundaries before this one leave it
  where it stands, by the induction hypothesis and what follows.

  Fix a location $x$ and let $\hat u=\sigma.\tview(x)$.  Suppose first that
  $\hat u$ is a write of $t$ carrying loop index $\ell$.  It is then the
  greatest such write at $x$, by the invariant and being one of them.  It is
  therefore the anchor the boundary retains at $x$, and the restriction leaves
  the viewfront pointing at it.  If $\hat u$ is not a write of $t$ with loop
  index $\ell$, the boundary does not remove it, and again the viewfront does
  not move.

  For the other viewfronts, an entry whose write the boundary removes keeps that
  write's timestamp, so no entry is assigned a smaller one.  $\OW(t',x)$ is
  determined by $\writeset$ and the timestamp of $t'$'s viewfront at $x$, of
  which the boundary shrinks the first and leaves the second, so it can only
  shrink.

\end{proof}

\begin{corollary}[Coherence is monotone in the loop index]\label{c:opsem-co-monotone}

  Let $t$ be a thread, $x$ a location, and $i<j$ iterations of a loop $\ell$ of
  $t$ in a derivation $D$.  Then every write of $t$ at $x$ in iteration $i$ is
  $\CO$-before every write of $t$ at $x$ in iteration $j$.

\end{corollary}

\begin{proof}

  Take $j=i+1$ first.  By Lemma~\ref{l:opsem-anchor} $t$'s viewfront at $x$
  stands after the boundary where it stood before it, at the anchor at $x$,
  which is the $\CO$-greatest write of $t$ at $x$ in iteration $i$.
   Rule~\ref{fig:action:write}~(write) draws the timestamp of a write of iteration
  $i+1$ at $x$ from $\freshts$ above the timestamp of a write in $\OW(t,x)$,
  hence above the viewfront, hence above the anchor and above every write of
  iteration $i$ at $x$.  The anchor retained at $x$ by the boundary closing
  $i+1$ is a write of $i+1$ where $t$ wrote $x$ in it and the anchor of $i$
  otherwise, so the argument iterates and the general case follows by induction
  on $j-i$.

\end{proof}

Monotonicity follows from the properties of episodic loops per
Definition~\ref{def:episodic}:  By Condition~\ref{episodic:events} the events of
iteration $i$ precede those of iteration $j$ in $(\ppoord\cup\DP)^+$ and hence
in $\HB$, so a write of $j$ $\CO$-before a write of $i$ at the same location
would close an $\HB;\CO$ cycle, which Axiom~\ref{def:mem-model-axiom:co}
forbids. Corollary~\ref{c:opsem-co-monotone} is the counterpart of that in the
operational semantics.

The anchor enables the monotonicity result, carrying the timestamp of the last
write visible in the viewfront across the loop boundary. Without the anchor the
restriction would drop $\sigma.\tview(x)$ to the greatest surviving write below
the write $t$ made in the iteration -- possibly the initialising write at $x$,
which Condition~\ref{episodic:mem} case~\ref{episodic-casea} expressly permits a
later iteration to read.  $t$ would then be free to take a write its own
viewfront had passed, which is precisely the $\FR$ edge back into the closed
iteration that Lemma~\ref{l:opsem-co-glue} has to exclude below, and which
Definition~\ref{def:episodic} does not forbid on its own.  Retaining the anchor
keeps the viewfront where it stands; denying $t$ the anchor in
 Rule~\ref{fig:action:read}~(read) keeps Condition~\ref{episodic:mem}, which
removing the write was there to implement; and leaving the anchor readable by
the other threads keeps what they could observe of the iteration before the
boundary.  The same slack for a thread that is not crossing the boundary -- its
viewfront pointing at a write of $t$ that the boundary removes -- is closed by
the entry keeping its timestamp instead of being pulled back to a surviving
write, there being no bound on how many writes a boundary would have to retain
to keep every viewfront of every thread on a write of the state.

\paragraph{Coherence within a block.}
The coherence Axiom~\ref{def:mem-model-axiom:co} is defined on a modification
order, which is read off the timestamps as in Section~\ref{s:app-opsem-states}.

\begin{lemma}[Timestamps are an adequate coherence
  order]\label{l:opsem-co-adequate}

  Let $\sigma$ be a program state reachable from $\sigma_0$, and let
  $\CO_\sigma$ relate $\hat w_1$ to $\hat w_2$ when their locations are
  $\sigma.\varphi$-equivalent and $\tst(\hat w_1)<\tst(\hat w_2)$.  Then

  \begin{enumerate}

    \item\label{co-adq:total} $\CO_\sigma$ is a strict total order on the writes
      of $\sigma.\writeset$ at each location;

    \item\label{co-adq:view} for every thread $t$, location $x$, and event $e$
      of $t$, at the moment $e$ is performed $\tst(\sigma.\tview(x))$ is at
      least $\tst(\hat w)$ for every $\hat w\in\sigma.\writeset$ at $x$ which
      $\HB$-precedes $e$, and for every $\hat w\in\sigma.\writeset$ at $x$ read
      by an action that $\HB$-precedes $e$ or is $e$ itself;

    \item\label{co-adq:ww} if $\hat w_1,\hat w_2\in\sigma.\writeset$ write
      $\varphi$-equivalent locations and $(w_1,w_2)\in\HB$, then $(\hat w_1,\hat
      w_2)\in\CO_\sigma$;

    \item\label{co-adq:rd} if a read $r$ by $t$ at $x$ takes $\hat w$, then
      $\tst(\hat w)$ is at least $\tst(\hat w')$ for every $\hat
      w'\in\sigma.\writeset$ at $x$ that $\HB$-precedes $r$ or is read by an
      action $\HB$-preceding $r$.

  \end{enumerate}

\end{lemma}

\begin{proof}

  Claim~\ref{co-adq:total}.  Timestamps are rationals. $<$ orders them
  totally, such that no two writes at $\varphi$-equivalent locations carry
  the same timestamp.  $\sigma_0.\writeset$ holds one write per location at
  timestamp $0$, and $\freshts(x,q,q')$ requires $q<q'$ together with $q'<q''$
  for every $q''$ at a location $\varphi$-equivalent to $x$ above $q$, so $q'$
  falls strictly between two adjacent existing timestamps of that class and
  coincides with neither.  Removing writes at a boundary does not disturb this.

  The classes are compared under the $\varphi$ of the state the write is made
  in, and $\varphi$ grows along a derivation, so it remains that no class
  absorbs another after its timestamps have been chosen.  The initialising
  writes are held apart by $\sigma_0.\varphi$, which keeps distinct global
  locations distinct, and the locations of allocations by the allocation
  constraints of Paragraph~\ref{par:opsem-memory}; both survive the reset at a
  loop boundary, which drops only conjuncts carrying a thread and a loop index.
  A branching action conjoins $b$ to $\varphi$ only where $\varphi\wedge
  b\not\equiv\bot$, so no branch, and no test of a $\cas$, can identify two
  locations these constraints hold apart.

  Claim~\ref{co-adq:view}, by induction on the length of the derivation.  The
  initial state satisfies the claim trivially.  For the inductive step, note
  first that no rule ever lowers $\sigma.\tview(x)$, the boundary included:
  Rule~\ref{fig:action:write}~(write) advances it to a write of strictly greater
  timestamp, and Rule~\ref{fig:action:read}~(read) sets it to
  $\funup{\sigma.\syncview_t(\hat w,a)}{x}{\hat w}$, where $\viewcomb$ takes the
  pointwise later write and $\hat w\in\OW(t,x)$ lies at or above the current
  value at $x$.  The hypothesis is therefore preserved by any step adding no new
  $\HB$-predecessors, and it remains to check the steps that do, and the
  boundary.

  If $e$ is $\ppo$- or $\DP$-after an earlier event of $t$, its new
  $\HB$-predecessors are those of that earlier event, and $\tview(x)$ has not
  decreased since, so the bound carries over.  If $e$ is a read taking $\hat w$,
  then $\tview(x)$ is set at or above $\hat w$, discharging the clause for
  writes read by $e$ itself.  If moreover $w\in\Wrel$ and $e\in\Racq$, the new
  $\HB$-predecessors are those of $w$ together with $w$, and $\syncview$ folds
  $\mview$ into $\tview$; $\mview$ was fixed to the writing thread's viewfront
  when $w$ was performed, which by induction bounded everything $\HB$-before $w$
  and everything read by an action $\HB$-before-or-equal $w$, and $\viewcomb$
  retains those bounds.  No other rule introduces an $\HB$ edge, $\SW$ being the
  only inter-thread constituent of $\HB$.

  At a boundary the viewfront stays where it stands: an entry keeps the
  timestamp it carries whether or not the write at it survives, and for the
  crossing thread it keeps the write itself by Lemma~\ref{l:opsem-anchor}.  The
  obligation is weakened alongside, as it ranges over $\sigma.\writeset$ and the
  removed writes leave it.  So the bound is preserved for every write that
  remains.

  Claim~\ref{co-adq:ww}.  When $w_2$ is performed, $\tst(\tview(x))\geq\tst(\hat
  w_1)$ by Claim~\ref{co-adq:view}, and $\freshts$ places $\hat w_2$ strictly
  above the observable write it extends, hence above $\tview(x)$ and so above
  $\hat w_1$.

  Claim~\ref{co-adq:rd}.  The read takes $\hat w\in\OW(t,x)$, so $\tst(\hat
  w)\geq\tst(\tview(x))$, which by Claim~\ref{co-adq:view} bounds every write
  named in the claim.

\end{proof}

\paragraph{Coherence across loop boundaries.}
Lemma~\ref{l:opsem-co-adequate} speaks only of the writes a configuration
retains, and it has to: Rule~\ref{fig:step:lb}~(lb) removes from
$\sigma.\writeset$ every write of the crossing thread carrying the closing
iteration's loop index bar the anchor at each location, so at any configuration
the timestamps order only the writes present there.  A derivation with $k$
iterations of $\ell$ constructs an execution whose writes include those of all
$k$, and Axiom~\ref{def:mem-model-axiom:co} asks for one modification order over
all of them.  No single timestamp assignment carries it; it has to be assembled
from the per-configuration ones.

\begin{lemma}[Assembling the coherence order across boundaries]\label{l:opsem-co-glue}

  Let $D$ be a derivation from $(\sigma_0,\rho_0,H_0)$, and let $\mathbb X$ be
  the execution it constructs.  There is a modification order $\CO$ on $\mathbb
  X$ restricting at each configuration $\sigma$ of $D$ to $\CO_\sigma$ on the
  writes retained there, and $\ECO\cup\HB$ is acyclic on $\mathbb X$ under it.
  Hence $\mathbb X$ satisfies Axiom~\ref{def:mem-model-axiom:co}.

\end{lemma}

\begin{proof}

  Call the writes a thread $t$ performs within one iteration of $\ell$ the
  \emph{block} of that thread and iteration; one application of
   Rule~\ref{fig:step:lb}~(lb) removes a block bar its anchors, one write per
  location, which stay until a later boundary supersedes them.  The writes
  outside any loop, the initialising writes among them, are removed by no reset,
  are therefore present in every $\sigma.\writeset$, and keep the timestamps
  they were assigned, no rule altering a timestamp once given; call them the
  \emph{spine}.  Blocks of different threads interleave freely, threads crossing
  their boundaries independently and each reset being per thread, so the
  assembly cannot proceed by concatenating blocks in a global order.  It
  proceeds per location.

  \emph{Construction.}
  Fix a location $x$ up to $\varphi$-equivalence and let $B_1,\ldots,B_m$
  enumerate the blocks holding a write at $x$, indexed in the order in which $D$
  opens their iterations.  All $\CO_\sigma$ agree on the spine, which they order
  by its fixed timestamps.  Each $B_i$ was present alongside the whole spine at
  the configurations between its opening and its reset, and $\CO_\sigma$ there
  places $B_i$'s writes among themselves and against the spine; that placement
  does not vary with $\sigma$, for the same reason.  Define $\CO$ at $x$ by
  taking the spine in its order, inserting each $B_i$ at the positions its
  timestamps give it, and ordering two writes of distinct blocks $B_i$ and $B_j$
  that fall between the same two adjacent spine writes by $i<j$.  This is a
  strict total order at $x$ by Claim~\ref{co-adq:total} of
  Lemma~\ref{l:opsem-co-adequate}, and it restricts to $\CO_\sigma$ at every
  configuration.  At any configuration at most one open block per thread is
  present, together with the anchors of the blocks that thread has closed, so
  the tie-breaking clause is consulted only for writes of distinct blocks of one
  thread that are never simultaneously present.  Where an anchor of $B_i$ is
  present alongside a later block $B_j$ of the same thread, $\CO_\sigma$ orders
  the two by their timestamps, which agrees with $i<j$ by
  Corollary~\ref{c:opsem-co-monotone}.

  \emph{No return into a closed block.}
  Let $B$ be the block of $t$ and iteration $i$, closed at the boundary $\beta$,
  and let $e$ be an event performed after $\beta$.  We claim no edge of
  $\rf\cup\CO\cup\FR\cup\HB$ runs from $e$ or a later event into $B$.  For
  $\rf$: an $\rf$ edge into $B$ would have a write of $B$ as target, and $\rf$
  targets reads.  For $\CO$: by construction every $\CO$ edge between $B$ and a
  block opened later points out of $B$, and $B$'s placement against the spine is
  fixed, so a $\CO$ edge into $B$ has its source in a block opened earlier or in
  the spine, neither of which the cycle can reach from $e$ without a further
  edge into a closed block.  For $\HB$: $\HB\subseteq{<_D}$ by
  Lemma~\ref{l:opsem-nta}, and every event of $B$ precedes $\beta$ in $<_D$.
  For $\FR$: an $\FR$ edge from a read $r$ into a write $\hat w$ of $B$ requires
  $r$ to read some $\hat w'$ at $x$ with $(\hat w',\hat w)\in\CO$.  If $r$
  belongs to $t$ then $\hat w'$ is at or above $t$'s viewfront at $x$ after
  $\beta$, which by Lemma~\ref{l:opsem-anchor} is where it stood before $\beta$:
  at the anchor at $x$ if $t$ wrote $x$ in the iteration, and hence at or above
  every write of $B$ at $x$, the anchor being the greatest of them and itself
  denied to $r$ by Rule~\ref{fig:action:read}~(read); and at a write not in
  $B$ otherwise, when no write of $B$ is at $x$ at all.  Either way $\hat w'$ is
  not $\CO$-below a write of $B$.  If $r$ belongs to another thread $t'$, the
  only write of $B$ still available to it after $\beta$ is the anchor at $x$,
  the rest having been removed, and the anchor is the $\CO$-greatest write of
  $B$ at $x$, so taking it opens no $\FR$ edge into $B$.  Any other write $r$
  reads is $\CO$-placed against $B$ by the construction, from which $\hat
  w'\CO$-below $\hat w$ would require $t'$ to have observed $\hat w$ and then
  read below it, contradicting Claim~\ref{co-adq:rd} of
  Lemma~\ref{l:opsem-co-adequate} at the configuration where both were present.

  \emph{Acyclicity.}
  Suppose $\ECO\cup\HB$ carried a cycle, and write $\ECO$ in the normal form
  $\rf\cup\CO\cup\FR\cup\CO;\rf\cup\FR;\rf$.  If every write the cycle visits is
  present at one configuration $\sigma$ -- which holds in particular when the
  cycle visits at most one block per thread -- then all its edges are edges of
  $\CO_\sigma$, $\FR_\sigma$, $\rf$ and $\HB$ there, and we show it excluded by
  taking the six shapes of $\HB$ composed with $\ECO$ or nothing in turn,
  writing $x$ for the location involved.

  \begin{itemize}

    \item $\HB$ alone: $\HB=(\DP\cup\ppoord\cup\SW)^+$ with $\SW\subseteq\rf$, so
      $\HB\subseteq{<_D}$ by Lemma~\ref{l:opsem-nta} and $<_D$ is strict.

    \item $\HB;\rf$: a read $r$ and a write $w$ with $(w,r)\in\rf$ and
      $(r,w)\in\HB$ give $w<_D r<_D w$.

    \item $\HB;\CO$: writes with $(\hat w_1,\hat w_2)\in\CO$ and
      $(w_2,w_1)\in\HB$; Claim~\ref{co-adq:ww} of
      Lemma~\ref{l:opsem-co-adequate} gives $(\hat w_2,\hat w_1)\in\CO$,
      contradicting Claim~\ref{co-adq:total}.

    \item $\HB;\FR$: a write $w$ and a read $r$ with $(w,r)\in\HB$ and $r$
      reading some $\hat w'$ with $(\hat w',\hat w)\in\CO$; but $w$
      $\HB$-precedes $r$, so Claim~\ref{co-adq:rd} gives $\tst(\hat
      w')\geq\tst(\hat w)$.

    \item $\HB;\CO;\rf$: a read $r$ and writes $w_1,w_2$ with $(r,w_1)\in\HB$,
      $(\hat w_1,\hat w_2)\in\CO$ and $(w_2,r)\in\rf$.  Then $\hat w_2$ is read
      by $r$ and $r$ $\HB$-precedes $w_1$, so Claim~\ref{co-adq:view} bounds the
      viewfront at $x$ of $w_1$'s thread below by $\tst(\hat w_2)$ when $w_1$ is
      performed, and $\freshts$ puts $\hat w_1$ above it, giving $(\hat w_2,\hat
      w_1)\in\CO$.

    \item $\HB;\FR;\rf$: reads $r_1,r_2$ and writes $\hat w,\hat w'$ with
      $(r_2,r_1)\in\HB$, $r_1$ reading $\hat w'$, $(\hat w',\hat w)\in\CO$, and
      $r_2$ reading $\hat w$.  Here $\hat w$ is read by $r_2$ and $r_2$
      $\HB$-precedes $r_1$, so Claim~\ref{co-adq:rd} applied to $r_1$ gives
      $\tst(\hat w')\geq\tst(\hat w)$.

  \end{itemize}

  Otherwise the cycle visits writes of two blocks $B_i$, $B_j$ of the same
  thread with $i<j$, never simultaneously present.  Every event of $B_i$
  precedes in $<_D$ the boundary closing it, and every event of $B_j$ follows
  it, so the cycle must return from an event after that boundary into $B_i$,
  which the previous paragraph excludes.

  The remark following Lemma~\ref{l:gamma-pres-consistency} observes that a map
  extending an execution by an iteration would need precisely this argument, and
  that a cycle forced across one boundary would recur across every other by the
  symmetry of episodic loops.  That symmetry is what keeps the argument finite:
  it is discharged once, at an arbitrary boundary, and the episodicity
  conditions make every boundary alike.

\end{proof}

Because of the finiteness claim of Theorem~\ref{t:opsem-finite} established
above, the operational semantics is only able to discern programs by safety
properties, refuted by a finite prefix if they are refuted at all, and the
prefix-level correspondence transfers them in both directions -- completeness to
carry a proof in the operational semantics to every execution of the event
structure semantics, soundness to carry the use-after-free witness of
Section~\ref{s:opsem-og} back to a behaviour the program has.

\subsection{Soundness and Completeness of Operational
Semantics}\label{app:opsem-sound-complete}

The proofs of soundness and completeness establish a correspondence between
derivations $D$ in the operational semantics and executions in the event
structure semantics, such that derivations $D$ transition between configurations
$(\sigma,\rho,H)$ which successively enable labelled actions $(l\colon
a)\in\Phi|_H$ which correspond to events in traces of executions $\mathbb X$.

The correspondence relies on derivations $D$ following posterior future
horizons. The latter traverse future sets $\Phi_\mathbb{X}$ defined on
executions $\mathbb{X}$ along histories, as long as these are consistent with
respect to branching decisions. The predicate $\varphi$ in program states
$\sigma$ ensures the consistency with respect to branching decisions and
provides a context to evaluate symbolic values. The remaining crucial point of
the proofs is to show that the resets in the Step
 Rule~\ref{fig:step:lb}~(lb) at boundaries of loop iterations accurately
reflect the structure of $\gamma$ on event structures per
Definition~\ref{def:gamma}.

\opsemcompleteness*

\begin{proof}

  $\mathbb X$ defines a future set $\Phi_\mathbb{X}$. We show that there is a
  derivation $D$ from the initial configuration $(\sigma_0,\rho_0,H_0)$
  traversing $\Phi_\mathbb{X}$: (1) the resets at boundaries of loop iterations
  in Rule~\ref{fig:step:lb}~(lb) are compatible with executions, and (2)
  read-from relations establish $\RF$-pairs in the operational semantics in
   Rule~\ref{fig:action:read}~(read).

  \emph{(1) Boundaries of loop iterations} In order to show that
   Rule~\ref{fig:step:lb}~(lb) is correct, we need to show that it does not
  restrict configurations in a way that breaks compatibility with the event
  structure semantics. Therefore, we need to show that the resets in the rule
  either subsume episodicity conditions in Definition~\ref{def:episodic} of
  episodic loops or correspond to $\gamma$, which identifies states across loop
  iterations in the event structure semantics.

  \begin{itemize}

    \item Resetting the timestamped writes in Rule~\ref{fig:step:lb}~(lb) makes
      writes from previous loop iterations unavailable for assignment in
      Rule~\ref{fig:action:read}~(read), and thus implements
      Condition~\ref{episodic:mem} of episodic loops in
      Definition~\ref{def:episodic}. The anchors the reset retains are exempt
      from the removal but not from the condition:
      Rule~\ref{fig:action:read}~(read) denies them to the thread that wrote
      them, which is the thread Condition~\ref{episodic:mem} speaks about, while
      leaving them available to the other threads exactly as they were before
      the boundary.

    \item The viewfronts are not reset, which keeps the previous item from
      admitting behaviours the event structure semantics forbids. A viewfront
      moved back would enlarge $\OW$, putting writes the reading thread had
      passed back within its reach; which Lemma~\ref{l:opsem-anchor} excludes.
      The crossing thread's viewfront points after the boundary at the write it
      pointed at before it -- the anchor, where the thread wrote the location in
      the iteration it closed -- and the viewfront of any other thread keeps its
      timestamp where the write at it is removed. $\OW$ therefore only shrinks
      at a boundary, by the removed writes and, for the crossing thread, by its
      anchors, and every read the operational semantics offers after the
      boundary was on offer before it.

      This is where the episodicity conditions do the work of the
      correspondence. What the boundary removes are the writes of the closed
      iteration, which by Condition~\ref{episodic:mem} no read of the thread may
      take in a later one; what it retains at each location is the one write
      that the thread's own viewfront, and the coherence order, still stand on.
      Without Condition~\ref{episodic:mem} the removal would not be
      behaviour-preserving, and without the anchor the crossing thread's
      viewfront would fall back to a write of an earlier iteration or to the
      initialising write, which Condition~\ref{episodic:mem}
      case~\ref{episodic-casea} permits it to read.

    \item Resetting the register state in Rule~\ref{fig:step:lb}~(lb) makes
      register assignments in previous loop iterations unavailable for
      evaluation in expressions, which implements Condition~\ref{episodic:reg}
      of episodic loops.

    \item Resetting $\varphi$ makes branching decisions and constraints from
      read-from relations in previous loop iterations unavailable as context for
      the evaluation of expressions. By Conditions~\ref{episodic:reg}
      and~\ref{episodic:mem}, expressions do not use symbols read in previous
      loop iterations. $\gamma$ eliminates constraints from previous loop
      iterations by Lemma~\ref{l:restrict-elabs} using that the branching
      conditions of an iteration do not, jointly, constrain symbols read before
      the loop by Condition~\ref{episodic:cond}. Thus resetting $\varphi$
      mirrors the event structure semantics.

    \item Resetting histories in Rule~\ref{fig:step:lb}~(lb) implements
      $\gamma$. The rule forms $H^-$ by discarding every labelled action of the
      loop,
      $H^-=H\setminus\{(l'\colon\_)\mid\loopfun(l)\subseteq\loopfun(l')\}$,
      retaining the actions before the loop, and then re-indexes the actions of
      the iteration it opens to the first, $l^-=(\pc(l),0)$. That is the map of
      Corollary~\ref{c:gamma-history-reset}: for a history $H_{i+1}$ of the
      events of the iterations before the $i+1$-st, $\gamma^i H_{i+1}=H_1$
      retains the events before the loop together with those of the $i$-th
      iteration re-indexed as the first. By the same corollary
      $\gamma^i\left(\Phi\mid_{H_{i+1}}\right)=\Phi\mid_{H_1}$, so the reset
      leaves the posterior futures, and hence the next enabled actions the rule
      selects from $\horizon\Phi_H$, unchanged. This is what bounds the
      histories the operational semantics need represent, as the events of at
      most one iteration are retained at any point.

  \end{itemize}

  \emph{(2) Read-from relations} The read-from relation in complete executions
  $\mathbb X$ assigns a write event to every read event. By
  Lemma~\ref{l:co-axiom-rf-viswrite}, $\rf$ will assign to a read event only a
  write that is observable to the reading thread, matching the choice of $\hat
  w\in\sigma.\OW(t,x)$ in Rule~\ref{fig:action:read}~(read). Concretely, the
  coherence axiom forbids a read from a write that the reading thread's
  viewfront has passed: the viewfront is advanced past a write $w_1$ at $x$ only
  by reading a write $w_2$ at $x$ with $(w_1,w_2)\in\CO$, or by synchronising
  with a write whose $\mview$ has been so advanced, and in either case $w_2$
  $\HB$-precedes the read, so reading $w_1$ would close the $\FR;\CO$ cycle of
  Lemma~\ref{l:co-axiom-rf-viswrite}. Conversely every write the execution's
  $\rf$ selects is $\CO$-after the viewfront and hence in $\OW$, so the
  derivation can follow $\mathbb X$'s choice at each read.

  \emph{(3) Synchronisation} $\mathbb X$'s happens-before is
  $\HB=(\DP\cup\ppoord\cup\SW)^+$ with $\SW=\rf\cap(\Wrel\times\Racq)$, and the
  $\DP$ and $\ppo$ components are followed by future stepping, since $\Phi$ is
  built from them. The $\SW$ component is not in $\Phi$ -- no inter-thread edge
  is -- and is realised instead by the viewfront combination $\viewcomb$ in
  Rule~\ref{fig:action:read}~(read): when the rule takes a $w\in\Wrel$ by an
  $a\in\Racq$ it folds $\mview$ into the reader's viewfront, so every write the
  writer had observed at the release is observable to the reader afterwards. As
  $\mview$ was fixed to the writer's viewfront at the write, and a viewfront is
  advanced by exactly the writes its thread has performed or observed, the
  writes made observable are exactly those $\HB$-before $w$. The derivation
  therefore realises each $\SW$ edge of $\mathbb X$ without $\Phi$ ordering the
  two events, which is what allows $\Phi$ to remain a per-thread order.

  \emph{(4) Branching decisions} The branch action of
  Rule~\ref{fig:action:fence-branch}~(branch) fires only where $\sigma.\varphi\wedge
  b\not\equiv\bot$, so a derivation cannot take an outcome contradicting the
  decisions its history already records. It remains to check that this gate
  prunes no execution of the event structure semantics. Let $\mathbb X$ contain
  an event $e$ whose value restriction records the outcome $b$ of the branch at
  hand; by Definition~\ref{def:gen-es}, $b$ is one of the conjuncts of
  $\valres(e)$. Definition~\ref{def:executions} admits $J$ only consistent with
  $\bigwedge_{e\in X}\valres(e)$, and Condition~(3) of
  Definition~\ref{def:freeze} asks $P\wedge\varphi_\rf$ to be satisfiable, where
  $P$ conjoins the predicates of the justifications and a pre-justification
  carries the value restriction of the event it justifies. Along a derivation
  following $\mathbb X$, every conjunct of $\sigma.\varphi$ is a branching
  condition taken by this thread in this iteration or an equality contributed at
  a read, and both occur among the conjuncts of $P\wedge\varphi_\rf$. So
  $\sigma.\varphi\wedge b$ is a sub-conjunction of a satisfiable conjunction,
  and hence satisfiable. The premise therefore holds wherever $\mathbb X$ takes
  the branch, and the derivation can follow it. What the gate excludes is the
  converse case -- an outcome whose conjunct contradicts the history -- and an
  execution taking it would carry an unsatisfiable predicate, which
  Condition~(3) already denies it.

\end{proof}

\opsemsoundness*

\begin{proof}

  By the definition of future stepping in
  Rule~\ref{fig:step:non-lb_non-branch}~(non-lb/non-branch)
  and~\ref{fig:step:lb}, $D$ follows posterior future horizons in $\Phi$. We
  show that traversing these posterior future horizons produces a history
  consistent with respect to branching decisions as long as we choose one of the
  alternative branches in rules for branching commands such as
  $\code{if}$-statements or $\cas$. The so obtained history forms a consistent
  set of events. Using a separate result in Lemma~\ref{l:pfh-maximal}, the so
  constructed consistent set of events is maximal. Then we show that the
  derivation $D$ selects read-from assignments for read events meeting the
  episodicity criteria and the conditions on $\rf$ in
  Definition~\ref{def:freeze}. The so obtained maximal conflict-free set of
  events then is an execution $\mathbb X~=~(X,J,\rf)$.

  Histories are constructed consistent with respect to branching decisions, and
  what secures this is now a satisfiability premise rather than an entailment
  one. Rule~\ref{fig:step:branch-then}~(then)
  and~\ref{fig:step:branch-else}~(else) fire only where
  $\sigma.\varphi\wedge\evalreg{b}{\rho}\not\equiv\bot$, and $\cas$ in
  Rule~\ref{fig:command:cas-success}~(cas-success)
  and~\ref{fig:command:cas-failure}~(cas-failure) emits a branching action
  carrying the outcome of its test, which
  Rule~\ref{fig:action:fence-branch}~(branch) admits under the same premise and
  records in $\sigma.\varphi$. This is why the $\cas$ rules need no premise
  equating the read symbol with the expected value.

  Satisfiability suffices because $\sigma.\varphi$ is a conjunction. Every
  decision the derivation has already taken is a conjunct of it, so
  $\sigma.\varphi$ \emph{entails} each of them, and the opposite outcome of any
  of them is barred: $\sigma.\varphi\wedge b\wedge\neg b\equiv\bot$ fails the
  premise. A history therefore never holds events of two alternative outcomes
  of one branch, which is consistency in the sense of
  Definition~\ref{def:event-structures}, where events conflict when their value
  restrictions are incompatible. Nothing downstream asks for more:
  Condition~(3) of Definition~\ref{def:freeze} is itself satisfiability of
  $P\wedge\varphi_\rf$, and the value restrictions the corresponding events
  carry are exactly the conjuncts $\sigma.\varphi$ accumulates, as
  Definition~\ref{def:gen-es} passes the value restriction to the continuation
  and so keeps a branch condition on the events after the join.

  Nor does the derivation need a single execution to witness it. The test
  $(l\colon\overline a)\in^*\horizon\Phi_H$ ranges over the horizons of every
  future in $\Phi_H$, and the actions of one $\overline a$ may be enabled in
  different ones: after the branching action of $\cas$ is skipped, the failure
  rule tests only $a_r$, and the success rule tests $a_r$ and then $a_w$. What
  the test settles is which labels may step, that is the $\ppoord\cup\DP$
  ordering, and the labels of $\overline a$ are those of one command of one
  thread. Which of the conflicting events carrying a label the derivation is
  following is settled instead by $\sigma.\varphi$, whose conjuncts are the
  branching decisions taken: of the copies of a label, only those whose value
  restriction is compatible with $\sigma.\varphi$ remain, and as every decision
  on the path to the label is among the conjuncts, one does. The two mechanisms
  together pin a unique event, and the set the derivation builds is shown an
  execution below rather than assumed to lie in one.

  Rule~\ref{fig:action:read}~(read) picks a write action from the observable
  writes $\OW(t,x)$ of the program state, and adds a constraint for the newly
  read symbol. The writes in $\OW(t,x)$ have been added by
  Rule~\ref{fig:action:write}~(write) either in the same loop iteration or
  outside of the loop in the same thread, or added in another thread, as
  otherwise the write would have been reset in Rule~\ref{fig:step:lb}~(lb). Thus
  Condition~\ref{episodic:mem} of Definition~\ref{def:episodic} of episodic
  loops is satisfied. By Rule~\ref{fig:action:read}~(read) an observable write
  must be assigned, moreover observable writes cannot be elided, so that
  Condition~(1) of Equation~\ref{eq:def:freeze} in Definition~\ref{def:freeze}
  is satisfied. The read event is by construction part of the execution, so that
  Condition~(2) is satisfied. As the symbol is new, $\varphi$ is necessarily
  consistent with the constraint, so that Condition~(3) of
  Equation~\ref{eq:def:freeze} is satisfied.

  It remains that the execution so constructed is \emph{consistent}, that is
  that it meets the two axioms of the memory model, neither of which any rule
  mentions.  Both are supplied by
  Section~\ref{app:opsem-consistency}. Axiom~\ref{def:mem-model-axiom:nta} is
  Lemma~\ref{l:opsem-nta}: future stepping takes only $\phi$-minima and
   Rule~\ref{fig:action:read}~(read) only writes already performed, so
  $\DP\cup\ppoord\cup\rf$ lies inside the order in which the derivation performed
  its actions.  Axiom~\ref{def:mem-model-axiom:co} is
  Lemma~\ref{l:opsem-co-glue}, which assembles a modification order on the whole
  execution out of the per-configuration timestamp orders and shows
  $\ECO\cup\HB$ acyclic under it, resting on the viewfront invariant of
  Lemma~\ref{l:opsem-co-adequate}.  The coherence axiom is where the timestamps
  earn their keep, and it was immediate only as long as a read could take the
  single last visible write; it is not immediate now that $\OW(t,x)$ may hold
  several.

  What the axiom leaves unconstrained is a write ordered before a read by $\CO$
  but not by $\HB$, and it is this slack -- a thread reading a write its own
  viewfront has not passed although a $\CO$-later one exists -- that makes the
  semantics non multi-copy atomic, agreeing with the event structure semantics.

\end{proof}

\vfill

\pagebreak
\section{Appendix: RCU Algorithm}\label{app:rcu}

\begin{example}\label{ex:rcu}

  Implementation of a shared counter with Read-Copy-Update (RCU), taken from
  Gotsman et al.~\cite{Gotsman23grace}:

  \begin{minted}[linenos,escapeinside=||]{c}

// global state
bool rcu[N] = {0};
int *C := new int(0);
Set det[N] := {{}};

sync () {
  bool r[N+1] := {0};           // S1|\label{rcuS1}|
  for (int i = 0; i < N; i++)   // S2|\label{rcuS2}|
    r[i] := rcu[i];             // S3|\label{rcuS3}|
  for (int i = 0; i < N; i++)   // S4|\label{rcuS4}|
    if r[i]                     // S5|\label{rcuS5}|
      while rcu[i]; }           // S6|\label{rcuS6}|

void reclaim (int *s) {
  insert(det[tid], s);          // R1|\label{rcuR1}|
  if (nondet()) return;         // R2|\label{rcuR2}|
  sync();                       // R3|\label{rcuR3}|
  while !isEmpty(det[tid])      // R4|\label{rcuR4}|
    free(pop(det[tid])); }      // R5|\label{rcuR5}|

int inc () {
  int v, *n, *s;                // I1,I2,I3|\label{rcuI1}||\label{rcuI2}||\label{rcuI3}|
  n := new int;                 // I4|\label{rcuI4}|
  do {
    rcu[tid] :=|$^{\text{rel}}$| 0;            // I5|\label{rcuI5}|
    rcu[tid] := 1;              // I6|\label{rcuI6}|
    s := FAA|$^{\text{rel,acq}}$|(&C, 0);        // I7|\label{rcuI7}|
    v := *s;                    // I8|\label{rcuI8}|
    *n := v+1;                  // I9|\label{rcuI9}|
    r := CAS|$^{\text{rel,acq}}$|(&C, s, n);    // I10|\label{rcuI10}|
  } while !r;                   // casres|\label{rcucasres}|
  rcu[tid] := 0;                // I11|\label{rcuI11}|
  reclaim(s);                   // I12|\label{rcuI12}|
  return v; }                   // I13|\label{rcuI13}|
\end{minted}

\end{example}

\vfill

\pagebreak
\section{Appendix: Hazard Pointers Algorithm}\label{app:hp}

\begin{example}\label{ex:hazptr}

  Implementation of a shared counter with Hazard Pointers, adopted from Meta's
  Folly Library~\cite{folly}. \mordor{}~\cite{kissig2026mordor} checks
  \code{inc} at $N = 1$ with the functions it calls inlined, as
  \code{programs/episodicity/hp-1.lit}; \code{retire} and \code{scan} follow
  the retry loops and are elided there.
  Folly's \code{hazptr\_holder} reads \code{C} once above the retry loop and
  carries it in a register; reading it at the head of the body, as below, is the
  same protocol without the carried register and without the branch a carry
  needs. The fence at \code{I8} carries the ordering, so the accesses to
  \code{hp[tid]} and the reloads of \code{C} need no annotation of their own:

  \begin{minted}[linenos,escapeinside=||]{c}

// global state
int *C = new int(0);
int *hp[N] = {nullptr};

void scan(tid, rcount, rlist) {
  int* plist[] = {};                  // S1
  for (int i = 0; i < N; ++i) {       // S2
    int *p = hp[i];                   // S3
    if (p != nullptr) {               // S4
      plist.push_back(p);             // S5
    }
  }

  auto tmplist = rlist;               // S6
  rlist.clear();                      // S7
  rcount = 0;                         // S8

  for (const auto node: tmplist) {    // S9
    auto lookup = std::find(plist.begin(), plist.end(), node);  // S10
    if (lookup == plist.end()) {      // S11
      delete[] node;                  // S12
    } else {
      rlist.push_back(node);          // S13
      rcount++;                       // S14
    }
  }
}

void retire(int *p) {
  rlist.push_back(p) ;                // R1
  rcount++ ;                          // R2
  if (rcount >= R) {                  // R3
    scan(tid, rcount, rlist);         // R4
  }
}

void inc() {
  // local state
  int *rlist[] = {} ;
  int rcount = 0;

  int *n, *p, *s, v ;                 // I1,I2,I3,I4
  n = new int[1] {0} ;                // I5

  do {
    do {
      p = C ;                         // I6
      hp[tid] = p ;                   // I7
      // full memory fence
      fence|$^\text{rel,acq}$| ;                     // I8|\label{hpI8}|
    } while (C != p) ;                // I9,I10

    s = hp[tid] ;                     // I11
    v = *s ;                          // I12
    *n = v + 1 ;                      // I13

    cas_succ = CAS|$^\text{rel,acq}$| (&C, s, n);  // I14
  } while (!cas_succ);                // I15

  hp[tid] = nullptr ;                 // I16
  retire(s);                          // I17
}
\end{minted}

\end{example}

\pagebreak
\section{Appendix: Sequence Lock Algorithm}\label{app:seqlock}

\begin{example}\label{ex:seqlock}

  Implementation of sequence lock, adopted from \cite{hemminger2002seqlock}:

  \begin{minted}{c}

rseq  := new int(0);  // The seq counter
rdata := new int(0);  // The data word

// --- write ---
// Lock
rs := *rseq;
*rseq := rs + 1;

// critical section

// Unlock: bump seq to the next even value
rs2 := *rseq;
*rseq := rs2 + 1;

// --- read ---
do {
  rr1 := *rseq; // Sample seq before reading
  rval := *rdata; // Read the data
  rr2 := *rseq; // Sample seq after reading

  // Retry if seq changed or was odd (write in progress)
  rodd  := rr1 & 1;
  rdiff := rr1 ^ rr2; // Zero iff the two samples agree
  rretry := rodd | rdiff;
} while (rretry != 0);

\end{minted}

\end{example}

\section{Appendix: Spin-lock Algorithm}\label{app:spinlock}

\begin{example}\label{ex:spinlock}

  Implementation of spinlock, adopted from \cite{Vafeiadis13RSL}:

  \begin{minted}[linenos,escapeinside=||]{c}

int mutex := 0;

// lock
do {
  rold := CAS|$^\text{acq,rlx}$|(&mutex, 0, 1);
} while (!rold);

// ( critical section )

// unlock
mutex := 0;

\end{minted}

\end{example}

\end{document}